\documentclass[12pt,american]{article}
\usepackage[T1]{fontenc}
\usepackage[utf8]{inputenc}
\usepackage{color}
\usepackage{babel}
\usepackage{mathtools}
\usepackage{amsmath}
\usepackage{amsthm}
\usepackage{amssymb}
\usepackage{graphicx}
\usepackage{esint}
\usepackage[pdfusetitle,
 bookmarks=true,bookmarksnumbered=false,bookmarksopen=false,
 breaklinks=false,pdfborder={0 0 0},pdfborderstyle={},backref=false,colorlinks=true]
 {hyperref}
\hypersetup{
 linkcolor=magenta, urlcolor=blue, citecolor=blue}

\makeatletter
\theoremstyle{plain}
\newtheorem{thm}{\protect\theoremname}
\theoremstyle{definition}
\newtheorem{defn}[thm]{\protect\definitionname}
\theoremstyle{remark}
\newtheorem{rem}[thm]{\protect\remarkname}
\theoremstyle{plain}
\newtheorem{fact}[thm]{\protect\factname}
\theoremstyle{plain}
\newtheorem{lem}[thm]{\protect\lemmaname}
\theoremstyle{plain}
\newtheorem{cor}[thm]{\protect\corollaryname}
\theoremstyle{plain}
\newtheorem{prop}[thm]{\protect\propositionname}

\DeclareMathOperator{\Tr}{Tr}

\DeclareMathOperator{\KM}{KM}
\DeclareMathOperator{\RLD}{RLD}

\allowdisplaybreaks
\numberwithin{equation}{section}

\usepackage{fullpage}
\usepackage{tocbibind}

\makeatother

\providecommand{\corollaryname}{Corollary}
\providecommand{\definitionname}{Definition}
\providecommand{\factname}{Fact}
\providecommand{\lemmaname}{Lemma}
\providecommand{\propositionname}{Proposition}
\providecommand{\remarkname}{Remark}
\providecommand{\theoremname}{Theorem}

\begin{document}
\title{\textbf{Quantum Fisher information matrices from R\'enyi relative
entropies}}
\author{Mark M. Wilde\\
\textit{School of Electrical and Computer Engineering,}\\
\textit{Cornell University, Ithaca, New York 14850, USA}}
\date{\today}
\maketitle
\begin{abstract}
Quantum generalizations of the Fisher information are important in
quantum information science, with applications in high energy and
condensed matter physics and in quantum estimation theory, machine
learning, and optimization. One can derive a quantum generalization
of the Fisher information matrix in a natural way as the Hessian matrix
arising in a Taylor expansion of a smooth divergence. Such an approach
is appealing and intuitive for quantum information theorists, given
the ubiquity of divergences in quantum information theory. In contrast
to the classical case, there is not a unique quantum generalization
of the Fisher information matrix, similar to how there is not a unique
quantum generalization of the relative entropy or the R\'enyi relative
entropy. In this paper, I derive information matrices arising from
the log-Euclidean, $\alpha$-$z$, and geometric R\'enyi relative
entropies, with the main technical tool for doing so being the method
of divided differences for calculating matrix derivatives. Interestingly,
for all non-negative values of the R\'enyi parameter $\alpha$, the
log-Euclidean R\'enyi relative entropy leads to the Kubo--Mori information
matrix, and the geometric R\'enyi relative entropy leads to the right-logarithmic
derivative Fisher information matrix. Thus, the resulting information
matrices obey the data-processing inequality for all non-negative
values of the R\'enyi parameter $\alpha$ even though the original
quantities do not. Additionally, I derive and establish basic properties
of $\alpha$-$z$ information matrices resulting from the $\alpha$-$z$
R\'enyi relative entropies, while leaving it open to determine all
values of $\alpha$ and $z$ for which the $\alpha$-$z$ information
matrices obey the data-processing inequality. For parameterized thermal
states and time-evolved states, I establish formulas for their $\alpha$-$z$
information matrices and hybrid quantum--classical algorithms for
estimating them, with applications in quantum Boltzmann machine learning.
\end{abstract}
\begin{quote}
\textit{This paper is dedicated to Professor Fumio Hiai on the occasion
of his forthcoming 80$^{\text{th}}$ birthday. His seminal insights
in matrix analysis \cite{Hiai2010,Hiai2014}, mathematical physics
\cite{Hiai2018,Hiai2019,Hiai2021}, and quantum information theory
\cite{Hiai1991,Hiai2011,Mosonyi2011} have left indelible impacts
on these fields and have had a profound influence on younger generations
of quantum information scientists.}
\end{quote}
\tableofcontents{}

\section{Introduction}

\subsection{Background}

Fisher information is a core concept in information science \cite[Section~11.10]{Cover2005}.
It arises in a number of contexts, perhaps most prominently in estimation
theory \cite{Korostelev2011} but also in machine learning and optimization
in the form of the natural gradient descent method \cite{Amari1998}.
Additionally, it is the centerpiece of the field of information geometry
\cite{Amari2016,Nielsen2022}, because the Fisher information matrix
defines a Riemannian metric on the parameter space underlying a parameterized
family of probability distributions. In fact, up to a multiplicative
constant, it is the unique Riemannian metric for which the data-processing
inequality holds \cite{Cencov1982}. The fact that it goes by the
moniker ``information'' suggests that every information theorist should
know something about it.

The Fisher information concept has been generalized to the quantum
case \cite{Helstrom1967,Helstrom1969,Braunstein1994,Holevo2011} with
the original goal of developing a theory of estimation relevant for
quantum information science. Since then, it has been broadly applied
in a variety of contexts, including high energy physics \cite{Miyaji2015,Lashkari2016,Banerjee2018,May2018},
condensed matter \cite{CamposVenuti2007,Zanardi2007,Gu2010,Hauke2016,Carollo2020},
quantum machine learning and optimization \cite{Stokes2020,Sbahi2022,Koczor2022,Sohail2024,Patel2024,Minervini2025,Minervini2025a},
and quantum resource theories \cite{Tan2021,Marvian2022,Yamaguchi2025}.
Signifying its fundamental role in quantum information science, several
reviews of quantum Fisher information have appeared to date \cite{Bengtsson2006,Hayashi2017,Liu2019,Sidhu2020,Jarzyna2020,Katariya2021,Meyer2021,Sbahi2022,Scandi2024}.

One of the simplest ways for an information theorist to connect with
the Fisher information concept is through smooth divergences that
measure the distinguishability of nearby probability distributions.
More precisely, the Fisher information matrix is equal to the Hessian
matrix in a Taylor expansion of a smooth divergence evaluated at nearby
probability distributions. Remarkably, the Čencov (Chentsov) theorem
states that, up to a constant prefactor, for every smooth divergence,
the Fisher information matrix is the unique matrix that arises in
this way \cite{Cencov1982}. Thus, in the classical case, we can speak
of a single notion of Fisher information matrix.

Similar to the classical case, a quantum generalization of the Fisher
information matrix arises through smooth divergences that measure
the distinguishability of nearby quantum states. As before, a quantum
generalization of the Fisher information matrix is equal to the Hessian
matrix in a Taylor expansion of a smooth divergence evaluated at nearby
quantum states. However, in distinction to the classical case, there
is not a unique quantum Fisher information matrix, and instead there
are an infinite number of possibilities, similar to how there is not
a unique quantum generalization of the relative entropy or the R\'enyi
relative entropy. There is a characterization of information matrices
that arise in the aforementioned way \cite{Morozova1991,Petz1996}
(see also \cite[Theorem~14.1]{Bengtsson2006}), which is recalled
as Theorem~\ref{thm:q-fisher-info-char} below.

\subsection{Summary of results}

In this paper, I explore quantum generalizations of the Fisher information
matrix that arise from the log-Euclidean, geometric, and $\alpha$-$z$
R\'enyi relative entropies (defined in \eqref{eq:log-euclidean-renyi},
\eqref{eq:geometric-Renyi-def}, and \eqref{eq:alpha-z-renyi-def},
respectively). These are called the log-Euclidean information matrix,
the geometric information matrix, and the $\alpha$-$z$ information
matrix, and are defined in \eqref{eq:log-euclidean-fisher-def}, \eqref{eq:geometric-Renyi-fisher-def},
and \eqref{eq:a-z-fisher-info-def}, respectively.

The main technical results of this paper are as follows:
\begin{itemize}
\item For all $\alpha\in\left(0,1\right)\cup\left(1,\infty\right)$, the
log-Euclidean information matrix is equal to the Kubo--Mori information
matrix (Theorem~\ref{thm:log-Euclidean-information-matrix}), the
latter defined in \eqref{eq:kubo-mori-def}, being the information
matrix that arises from the standard (Umegaki) relative entropy.
\item For all $\alpha\in\left(0,1\right)\cup\left(1,\infty\right)$, the
geometric information matrix is equal to the right logarithmic derivative
(RLD) information matrix (Theorem~\ref{thm:geometric-Renyi-to-RLD}),
the latter defined in \eqref{eq:RLD-def} and well studied in the
literature on quantum estimation theory (see, e.g., \cite[Section~V-E]{Sidhu2020}).
I also prove that the same is true when considering the information
matrix that arises from the Belavkin--Staszewski relative entropy
(Theorem~\ref{thm:belavski-RLD}). These derivations generalize those
from the single-parameter case \cite{Matsumoto2010,Matsumoto2013,Matsumoto2018,Katariya2021}.
\item For all $\alpha\in\left(0,1\right)\cup\left(1,\infty\right)$ and
$z>0$, I derive a formula for the $\alpha$-$z$ information matrix
(Theorem~\ref{thm:fisher-info-from-alpha-z}). This formula, applicable
for the multiparameter case, generalizes the formula reported in \cite[Eq.~(2.12)]{May2018}
for the single-parameter case and that in \cite[Eqs.~(132)--(133)]{Ciaglia2018}
for the Bloch-sphere qubit case. Additionally, the formula demonstrates
that, for all $z>0$, the $\alpha$-$z$ information matrix converges
to the Kubo--Mori information matrix in the limit as $\alpha\to1$.
It also does so for all $\alpha\in\left(0,1\right)\cup\left(1,\infty\right)$
in the limit as $z\to\infty$.
\item The Petz-- and sandwiched R\'enyi information matrices arise as
special cases of the $\alpha$-$z$ information matrix when $z=1$
and $z=\alpha$ (Corollary~\ref{cor:Petz-Renyi-special-case} and
Corollary~\ref{cor:sandwiched-Renyi-special-case}), respectively,
which correspond to the cases in which the underlying R\'enyi relative
entropy is set to be the Petz-- and sandwiched R\'enyi relative
entropies, respectively. I establish integral representations for
these information matrices when $\alpha\in\left(0,1\right)$ (Proposition
\ref{prop:integral-rep-Petz} and Proposition \ref{prop:integral-rep-sandwiched})
and derive simple formulas for them for $\alpha=2$.
\item I establish ordering relations for the Petz-- and sandwiched R\'enyi
information matrices. Namely, as a function of the R\'enyi parameter
$\alpha$, the Petz--R\'enyi information matrix is monotone decreasing
on $\alpha\in\left(0,\frac{1}{2}\right]$ and monotone increasing
on $\alpha\in\left[\frac{1}{2},\infty\right)$ (Theorem~\ref{thm:ordering-Petz-Renyi}),
and the sandwiched R\'enyi information matrix is monotone increasing
on $\alpha\in\left(0,\infty\right)$ (Theorem~\ref{thm:ordering-sandwiched-Renyi}).
\item For all $\alpha\in\left(0,1\right)$ and $z>0$, I derive a formula
for the $\alpha$-$z$ information matrix when the underlying family
of states consists of parameterized thermal states (Theorem~\ref{thm:QBM-a-z-formula}).
Specifically, for $\rho(\theta)\coloneqq\frac{e^{-H(\theta)}}{\Tr\left[e^{-H(\theta)}\right]}$
and $H(\theta)\coloneqq\sum_{j=1}^{L}\theta_{j}H_{j}$, with $\theta_{j}\in\mathbb{R}$
and $H_{j}$ Hermitian for all $j\in\left\{ 1,\ldots,L\right\} $,
the following formula holds for the $\alpha$-$z$ information matrix:
\begin{equation}
\left[I_{\alpha,z}(\theta)\right]_{i,j}=\frac{1}{2}\left\langle \left\{ \Phi_{q_{\alpha,z},\theta}\!\left(H_{i}\right),H_{j}\right\} \right\rangle _{\rho(\theta)}-\left\langle H_{i}\right\rangle _{\rho(\theta)}\left\langle H_{j}\right\rangle _{\rho(\theta)},\label{eq:a-z-formula-thermal-intro}
\end{equation}
where the quantum channel $\Phi_{q_{\alpha,z},\theta}$ is given by
\begin{equation}
\Phi_{q_{\alpha,z},\theta}\!\left(X\right)\coloneqq\int_{-\infty}^{\infty}dt\:q_{\alpha,z}(t)\,e^{-itH(\theta)}Xe^{itH(\theta)},
\end{equation}
and $q_{\alpha,z}\coloneqq p*p_{\alpha,z}$ is a probability density
function equal to the convolution of the probability density functions
$p$ and $p_{\alpha,z}$, each defined on $t\in\mathbb{R}$ as
\begin{equation}
p(t)\coloneqq\frac{2}{\pi}\ln\!\left|\coth\!\left(\frac{\pi t}{2}\right)\right|,\qquad p_{\alpha,z}(t)\coloneqq\frac{z}{2\pi\alpha\left(1-\alpha\right)}\ln\!\left(1+\left(\frac{\sin(\pi\alpha)}{\sinh(\pi zt)}\right)^{2}\right).\label{eq:high-peak-tent-1}
\end{equation}
\item For all $\alpha\in\left(0,1\right)$ and $z>0$, I derive a formula
for the $\alpha$-$z$ information matrix when the underlying parameterized
family consists of time-evolved states (Theorem~\ref{thm:time-evolved-a-z-formula}).
Specifically, for $\sigma(\phi)\coloneqq e^{-iH(\phi)}\rho e^{iH(\phi)}$,
where $H(\phi)\coloneqq\sum_{j=1}^{L}\phi_{j}H_{j}$, with $\phi_{j}\in\mathbb{R}$
and $H_{j}$ Hermitian for all $j\in\left\{ 1,\ldots,L\right\} $,
and $\rho\coloneqq\frac{e^{-G}}{\Tr\left[e^{-G}\right]}$ with $G$
Hermitian, the following formula holds for the $\alpha$-$z$ information
matrix:
\begin{equation}
\left[I_{\alpha,z}(\phi)\right]_{i,j}=\left\langle \left[\Phi_{p_{\alpha,z}}\!\left(\Psi_{\phi}\!\left(H_{i}\right)\right),\left[G,\Psi_{\phi}(H_{j})\right]\right]\right\rangle _{\rho},\label{eq:a-z-formula-time-evolved-intro}
\end{equation}
where the quantum channels $\Psi_{\phi}$ and $\Phi_{p_{\alpha,z}}$
are defined as
\begin{align}
\Psi_{\phi}(X) & \coloneqq\int_{0}^{1}dt\ e^{iH(\phi)t}Xe^{-iH(\phi)t},\\
\Phi_{p_{\alpha,z}}\!\left(X\right) & \coloneqq\int_{-\infty}^{\infty}dt\:p_{\alpha,z}(t)e^{-iGt}Xe^{iGt},
\end{align}
 and $p_{\alpha,z}$ is the probability density function defined in
\eqref{eq:high-peak-tent-1}.
\item I evaluate the $\alpha$-$z$, Kubo--Mori, and RLD information matrices
for parameterized families of classical-quantum states, showing that
they decompose as a sum of a classical part and a quantum part (Theorem~\ref{thm:cq-decomposition}),
thus generalizing the previous finding from \cite[Proposition~7]{Katariya2021}
for the single-parameter case. A direct consequence of this and the
data-processing inequality is that these information matrices are
convex (Corollary~\ref{cor:convexity}).
\end{itemize}
The first two results mentioned above are reminiscent of the classical
Čencov theorem. That is, for all $\alpha\in\left(0,1\right)\cup\left(1,\infty\right)$,
the log-Euclidean information matrices collapse to the same information
matrix, namely, the Kubo--Mori one, and for all $\alpha\in\left(0,1\right)\cup\left(1,\infty\right)$,
the geometric information matrices collapse to the same information
matrix, namely, the  RLD Fisher information matrix. Thus, the underlying
R\'enyi relative entropies belong to two distinct families in this
sense. Additionally, in spite of the fact that the log-Euclidean R\'enyi
relative entropy does not obey the data-processing inequality for
$\alpha>1$, the resulting log-Euclidean information matrix does.
Similarly, even though the geometric R\'enyi relative entropy does
not obey the data-processing inequality for $\alpha>2$, the resulting
RLD Fisher information matrix does.

The result in \eqref{eq:a-z-formula-thermal-intro} is a broad generalization
of the findings reported in \cite[Theorems~1 and 2]{Patel2024} and
\cite[Theorem~14]{Minervini2025}, such that these earlier findings
are now special cases of Theorem~\ref{thm:QBM-a-z-formula}. Specifically,
\begin{itemize}
\item \cite[Theorem~1]{Patel2024} follows from Theorem~\ref{thm:QBM-a-z-formula}
with $\alpha=z=\frac{1}{2}$ ,
\item \cite[Theorem~2]{Patel2024} follows from Theorem~\ref{thm:QBM-a-z-formula}
with $\alpha\to1$ and $z>0$ (or $\alpha>0$ and $z\to\infty$) ,
\item \cite[Theorem~14]{Minervini2025} follows from Theorem~\ref{thm:QBM-a-z-formula}
with $\alpha=\frac{1}{2}$ and $z=1$, and I note here that Theorem~\ref{thm:QBM-a-z-formula}
provides a simpler formula than that given in \cite[Theorem~14]{Minervini2025}. 
\end{itemize}
Given that parameterized thermal states are also known as quantum
Boltzmann machines \cite{Amin2018,Kieferova2017,Benedetti2017}, the
result stated in \eqref{eq:a-z-formula-thermal-intro} is applicable
to quantum Boltzmann machine learning. In particular, if one wishes
to perform quantum optimization using natural gradient descent with
quantum Boltzmann machines and the geometry induced by the $\alpha$-$z$
R\'enyi relative entropy, then the formula in \eqref{eq:a-z-formula-thermal-intro}
is applicable. To estimate the elements of the $\alpha$-$z$ information
matrix by means of a hybrid quantum-classical algorithm, one can employ
a procedure similar to that outlined in \cite[Figure~4(a)]{Minervini2025},
with the underlying probability density replaced by $q_{\alpha,z}$,
as defined in \eqref{eq:q-a-z-prob-dens}. Given the connections between
the $\alpha$-$z$ information matrix and high energy physics, as
considered in \cite{May2018}, the formula in \eqref{eq:a-z-formula-thermal-intro}
and the corresponding hybrid quantum-classical algorithm could find
applications there as well.

Similar to what was stated above for parameterized thermal states,
the result in \eqref{eq:a-z-formula-time-evolved-intro}, regarding
time-evolved states, is also a broad generalization of the findings
reported in \cite[Theorems~11, 15, and 19]{Minervini2025}, such that
these earlier findings are now special cases of Theorem~\ref{thm:time-evolved-a-z-formula}.
Specifically,
\begin{itemize}
\item \cite[Theorem~11]{Minervini2025} follows from Theorem~\ref{thm:time-evolved-a-z-formula}
with $\alpha=z=\frac{1}{2}$.
\item \cite[Theorem~15]{Minervini2025} follows from Theorem~\ref{thm:time-evolved-a-z-formula}
with $\alpha=\frac{1}{2}$ and $z=1$, and I note here that Theorem~\ref{thm:time-evolved-a-z-formula}
provides a simpler formula than that given in \cite[Theorem~15]{Minervini2025}. 
\item \cite[Theorem~19]{Minervini2025} follows from Theorem~\ref{thm:time-evolved-a-z-formula}
with $\alpha\to1$ and $z>0$ (or $\alpha>0$ and $z\to\infty$).
\end{itemize}
Given that time-evolved states are also known as quantum evolution
machines \cite{Minervini2025}, the result stated in \eqref{eq:a-z-formula-time-evolved-intro}
is also applicable to quantum machine learning. In particular, if
one wishes to perform quantum optimization using natural gradient
descent with quantum evolution machines and the geometry induced by
the $\alpha$-$z$ R\'enyi relative entropy, then the formula in
\eqref{eq:a-z-formula-time-evolved-intro} is applicable. To estimate
the elements of the $\alpha$-$z$ information matrix by means of
a hybrid quantum-classical algorithm, one can employ a procedure similar
to that outlined in \cite[Figure~2(b)]{Minervini2025}, with the underlying
probability density $p$ therein replaced by $p_{\alpha,z}$, as defined
in \eqref{eq:high-peak-tent-1}. 

The results in \eqref{eq:a-z-formula-thermal-intro} and \eqref{eq:a-z-formula-time-evolved-intro}
are also reminiscent of the classical Čencov theorem. Indeed, when
considering the $\alpha$-$z$ information matrices of parameterized
thermal states and time-evolved states for $\alpha\in\left(0,1\right)$
and $z>0$, the only $\alpha$-$z$ dependence that these expressions
have is with respect to the underlying probability density, which
is $q_{\alpha,z}$ in \eqref{eq:a-z-formula-thermal-intro} and $p_{\alpha,z}$
in \eqref{eq:a-z-formula-time-evolved-intro}. Thus, these results
represent a near unification of the $\alpha$-$z$ information matrices
for these classes of states.

\subsection{Paper organization}

This paper is structured as follows. Section~\ref{sec:Preliminaries}
reviews preliminary material needed to understand the rest of the
paper, including classical Fisher information and quantum generalizations
thereof. Section~\ref{sec:log-Euclid-to-KM} presents the first result
stated above, that the log-Euclidean information matrix is equal to
the Kubo--Mori information matrix (Theorem~\ref{thm:log-Euclidean-information-matrix}).
Section~\ref{sec:RLD-Fisher-information-from-geometric} presents
the second result stated above, that the geometric information matrix
is equal to the RLD Fisher information matrix (Theorem~\ref{thm:geometric-Renyi-to-RLD}).
Section~\ref{sec:a-z-information-matrices} presents the formula
for the $\alpha$-$z$ information matrix (Theorem~\ref{thm:fisher-info-from-alpha-z}),
Section~\ref{sec:Special-cases-a-z} presents special cases of the
$\alpha$-$z$ information matrices, and Section~\ref{sec:Orderings-and-relations}
presents the ordering relations for the Petz-- and sandwiched R\'enyi
information matrices. Section~\ref{sec:a-z-Information-matrices-thermal-states}
presents the formula for the $\alpha$-$z$ information matrix of
parameterized thermal states (Theorem~\ref{thm:QBM-a-z-formula}),
and Section~\ref{sec:Information-matrices-time-evolved} presents
the formula for the $\alpha$-$z$ information matrix of time-evolved
states (Theorem~\ref{thm:time-evolved-a-z-formula}). Section~\ref{sec:c-q-states}
provides the decomposition formula for information matrices evaluated
on classical-quantum states. Finally, Section~\ref{sec:Conclusion}
concludes with a summary of the results and points to directions for
future work. Appendix~\ref{app:Review:-Matrix-derivatives} provides
a detailed review of matrix derivatives.

\section{Preliminaries}

\label{sec:Preliminaries}Throughout the paper, I use the notation
\begin{equation}
\mathbb{R}_{+}\coloneqq\left(0,\infty\right).
\end{equation}
The rest of this section provides an overview of classical Fisher
information (Section~\ref{subsec:Classical-Fisher-information})
and quantum generalizations of Fisher information (Section~\ref{subsec:Quantum-Fisher}),
from the perspective of smooth divergences. 

\subsection{Classical Fisher information}

\label{subsec:Classical-Fisher-information}One of the simplest ways
for an information theorist to connect with the Fisher information
concept is through smooth divergences that measure the distinguishability
of nearby probability distributions. Examples of smooth divergences
include the relative entropy and the R\'enyi relative entropies,
respectively defined for $\alpha\in\left(0,1\right)\cup\left(1,\infty\right)$
and probability distributions $p$ and $q$ over a finite alphabet
$\mathcal{X}$ as follows:
\begin{align}
D(p\|q) & \coloneqq\sum_{x\in\mathcal{X}}p(x)\ln\!\left(\frac{p(x)}{q(x)}\right),\\
D_{\alpha}(p\|q) & \coloneqq\frac{1}{\alpha-1}\ln\sum_{x\in\mathcal{X}}p(x)^{\alpha}q(x)^{1-\alpha}.
\end{align}
To see this in more detail, suppose that $\left(p_{\theta}\right)_{\theta\in\Theta}$
is a second-order differentiable parameterized family of probability
distributions over a finite alphabet $\mathcal{X}$, where $\Theta\subseteq\mathbb{R}^{L}$
is open and $L\in\mathbb{N}$, and suppose that $\boldsymbol{D}$
is a smooth divergence. By the latter, I mean that
\begin{enumerate}
\item $\boldsymbol{D}$ is second-order differentiable, so that $\frac{\partial^{2}}{\partial\varepsilon_{i}\partial\varepsilon_{j}}\boldsymbol{D}(p_{\theta}\|p_{\theta+\varepsilon})$
exists for all $\theta,\varepsilon\in\Theta$.
\item $\boldsymbol{D}$ obeys the data-processing inequality, so that the
following inequality holds for every classical channel $N\equiv\left(N(y|x)\right)_{x\in\mathcal{X},y\in\mathcal{Y}}$
and every pair $\left(p,q\right)$ of probability distributions:
\begin{equation}
\boldsymbol{D}(p\|q)\geq\boldsymbol{D}(N(p)\|N(q)),
\end{equation}
where $N(p)$ denotes the probability distribution $\left(\sum_{x\in\mathcal{X}}N(y|x)p(x)\right)_{y\in\mathcal{Y}}$,
with a similar meaning for $N(q)$.
\item $\boldsymbol{D}$ is faithful, so that the following holds for probability
distributions $p$ and $q$:
\begin{equation}
\boldsymbol{D}(p\|q)=0\quad\iff\quad p=q.
\end{equation}
\end{enumerate}
Taken together, data processing and faithfulness imply that
\begin{equation}
\boldsymbol{D}(p\|q)\geq0
\end{equation}
for every pair $\left(p,q\right)$ of probability distributions, so
that the minimum value of $\boldsymbol{D}$ is equal to zero. Under
all three assumptions, the following Taylor expansion holds for $\varepsilon\in\Theta$
such that $\left\Vert \varepsilon\right\Vert \ll1$:
\begin{align}
\boldsymbol{D}(p_{\theta}\|p_{\theta+\varepsilon}) & =\boldsymbol{D}(p_{\theta}\|p_{\theta})+\varepsilon^{T}(\nabla\boldsymbol{D})(\theta)+\frac{1}{2}\varepsilon^{T}I_{\boldsymbol{D}}(\theta)\varepsilon+o\!\left(\left\Vert \varepsilon\right\Vert ^{2}\right)\label{eq:smoothness-divergence}\\
 & =\frac{1}{2}\varepsilon^{T}I_{\boldsymbol{D}}(\theta)\varepsilon+o\!\left(\left\Vert \varepsilon\right\Vert ^{2}\right),\label{eq:taylor-expand-smooth-divergence}
\end{align}
where $(\nabla\boldsymbol{D})(\theta)$ denotes the gradient vector
and $I_{\boldsymbol{D}}(\theta)$ the Hessian matrix, defined in terms
of their elements as follows: 
\begin{align}
\left[(\nabla\boldsymbol{D})(\theta)\right]_{i} & \coloneqq\left.\frac{\partial}{\partial\varepsilon_{i}}\boldsymbol{D}(p_{\theta}\|p_{\theta+\varepsilon})\right|_{\varepsilon=0},\\
\left[I_{\boldsymbol{D}}(\theta)\right]_{i,j} & \coloneqq\left.\frac{\partial^{2}}{\partial\varepsilon_{i}\partial\varepsilon_{j}}\boldsymbol{D}(p_{\theta}\|p_{\theta+\varepsilon})\right|_{\varepsilon=0}.
\end{align}
The equality in \eqref{eq:smoothness-divergence} follows from the
smoothness assumption, and the equality in \eqref{eq:taylor-expand-smooth-divergence}
follows from all three assumptions. Indeed, $\boldsymbol{D}(p_{\theta}\|p_{\theta})=0$
by faithfulness and $(\nabla\boldsymbol{D})(\theta)=0$ because we
have expanded the function $\varepsilon\mapsto\boldsymbol{D}(p_{\theta}\|p_{\theta+\varepsilon})$
about a critical point $\varepsilon=0$ (as previously stated, $\boldsymbol{D}(p_{\theta}\|p_{\theta+\varepsilon})$
takes its minimum value at $\varepsilon=0$). The Hessian matrix $I_{\boldsymbol{D}}(\theta)$
is also known as ``information matrix'' in this context. From the
equality in \eqref{eq:taylor-expand-smooth-divergence} and assumed
properties of the smooth divergence $\boldsymbol{D}$, one can deduce
that $I_{\boldsymbol{D}}(\theta)$ is positive semi-definite and obeys
the data-processing inequality (as a matrix inequality) for all $\theta\in\Theta$
(see, e.g., \cite[Proposition~4]{Minervini2025}).

Remarkably, the Čencov theorem states that, up to a constant prefactor,
the Fisher information matrix is the unique matrix such that \eqref{eq:taylor-expand-smooth-divergence}
holds \cite{Cencov1982}. That is, for every smooth divergence $\boldsymbol{D}$,
there exists a constant $\kappa>0$ such that the following equality
holds:
\begin{equation}
\left.\frac{\partial^{2}}{\partial\varepsilon_{i}\partial\varepsilon_{j}}\boldsymbol{D}(p_{\theta}\|p_{\theta+\varepsilon})\right|_{\varepsilon=0}=\kappa\left[I_{F}(\theta)\right]_{i,j},
\end{equation}
where the Fisher information matrix $I_{F}(\theta)$ is defined in
terms of its elements as 
\begin{equation}
\left[I_{F}(\theta)\right]_{i,j}\coloneqq\sum_{x\in\mathcal{X}}\frac{1}{p(x)}\left(\frac{\partial}{\partial\theta_{i}}p_{\theta}(x)\right)\left(\frac{\partial}{\partial\theta_{j}}p_{\theta}(x)\right).\label{eq:classical-fisher-information-matrix}
\end{equation}
As a special case, it is known that the following equality holds for
all $\alpha\in\left(0,1\right)\cup\left(1,\infty\right)$ when choosing
the smooth divergence $\boldsymbol{D}$ to be the R\'enyi relative
entropy $D_{\alpha}$ \cite[Eq.~(50)]{Erven2014}:
\begin{equation}
\frac{1}{\alpha}\left.\frac{\partial^{2}}{\partial\varepsilon_{i}\partial\varepsilon_{j}}D_{\alpha}(p_{\theta}\|p_{\theta+\varepsilon})\right|_{\varepsilon=0}=\left[I_{F}(\theta)\right]_{i,j}.\label{eq:classical-renyi-fisher}
\end{equation}
Additionally, the following equality holds \cite[Eq.~(49)]{Erven2014}:
\begin{equation}
\left.\frac{\partial^{2}}{\partial\varepsilon_{i}\partial\varepsilon_{j}}D(p_{\theta}\|p_{\theta+\varepsilon})\right|_{\varepsilon=0}=\left[I_{F}(\theta)\right]_{i,j},\label{eq:classical-rel-ent-fisher}
\end{equation}
which is consistent with the well known fact that $\lim_{\alpha\to1}D_{\alpha}(p\|q)=D(p\|q)$
for every pair $\left(p,q\right)$ of probability distributions. See
Appendix \ref{app:classical-fisher-from-renyi} for brief proofs of
\eqref{eq:classical-renyi-fisher} and \eqref{eq:classical-rel-ent-fisher},
which are helpful guidelines for how analogous proofs will proceed
for the more complicated case of quantum generalizations of Fisher
information. 

\subsection{Quantum generalizations of Fisher information}

\label{subsec:Quantum-Fisher}Several of the developments above can
be generalized to the quantum case, with the main exception being
that there is no longer a unique quantum generalization of Fisher
information, similar to how there are not unique generalizations of
relative entropy and R\'enyi relative entropy in quantum information
theory, and instead there are several sensible generalizations (see,
e.g., \cite[Chapter~7]{KW2024book}).

Let us briefly review quantum generalizations of the concepts from
Section~\ref{subsec:Classical-Fisher-information}. Now suppose that
$\left(\rho(\theta)\right)_{\theta\in\Theta}$ is a second-order differentiable
parameterized family of density operators acting on finite-dimensional
Hilbert space, where $\Theta\subseteq\mathbb{R}^{L}$ is open and
$L\in\mathbb{N}$, and suppose that $\boldsymbol{D}$ is a smooth
quantum divergence. By the latter, I mean that
\begin{enumerate}
\item $\boldsymbol{D}$ is second-order differentiable, so that $\frac{\partial^{2}}{\partial\varepsilon_{i}\partial\varepsilon_{j}}\boldsymbol{D}(\rho(\theta)\|\rho(\theta+\varepsilon))$
exists for all $\theta,\varepsilon\in\Theta$.
\item $\boldsymbol{D}$ obeys the data-processing inequality, so that the
following inequality holds for every quantum channel $\mathcal{N}$
(a completely positive, trace-preserving superoperator) and every
pair $\left(\rho,\sigma\right)$ of states:
\begin{equation}
\boldsymbol{D}(\rho\|\sigma)\geq\boldsymbol{D}(\mathcal{N}(\rho)\|\mathcal{N}(\sigma)).
\end{equation}
\item $\boldsymbol{D}$ is faithful, so that the following holds for states
$\rho$ and $\sigma$:
\begin{equation}
\boldsymbol{D}(\rho\|\sigma)=0\quad\iff\quad\rho=\sigma.
\end{equation}
\end{enumerate}
As in the classical case, under all three assumptions, the following
Taylor expansion holds for $\varepsilon\in\Theta$ such that $\left\Vert \varepsilon\right\Vert \ll1$:
\begin{align}
\boldsymbol{D}(\rho(\theta)\|\rho(\theta+\varepsilon)) & =\frac{1}{2}\varepsilon^{T}I_{\boldsymbol{D}}(\theta)\varepsilon+o\!\left(\left\Vert \varepsilon\right\Vert ^{2}\right)\label{eq:smoothness-divergence-1}
\end{align}
where $I_{\boldsymbol{D}}(\theta)$ denotes the Hessian matrix, defined
in terms of its elements as follows: 
\begin{align}
\left[I_{\boldsymbol{D}}(\theta)\right]_{i,j} & \coloneqq\left.\frac{\partial^{2}}{\partial\varepsilon_{i}\partial\varepsilon_{j}}\boldsymbol{D}(\rho(\theta)\|\rho(\theta+\varepsilon))\right|_{\varepsilon=0}.\label{eq:fisher-info-from-div-quantum}
\end{align}
The reasoning for \eqref{eq:smoothness-divergence-1} is the same
as that given for \eqref{eq:taylor-expand-smooth-divergence} (see,
e.g., \cite[Eq.~(36)]{Minervini2025}). Furthermore, the information
matrix $I_{\boldsymbol{D}}(\theta)$ is positive semi-definite and
obeys the data-processing inequality (see, e.g., \cite[Proposition~4]{Minervini2025}).
Additionally, additivity of the information matrix $I_{\boldsymbol{D}}(\theta)$
with respect to tensor-product families of states follows if the underlying
smooth divergence is additive. Indeed, for a tensor-product family
$\left(\rho(\theta)\otimes\sigma(\theta)\right)_{\theta\in\Theta}$,
\begin{align}
\left[I_{\boldsymbol{D}}\!\left(\theta;\left(\rho(\theta)\otimes\sigma(\theta)\right)_{\theta\in\Theta}\right)\right]_{i,j} & \coloneqq\left.\frac{\partial^{2}}{\partial\varepsilon_{i}\partial\varepsilon_{j}}\boldsymbol{D}(\rho(\theta)\otimes\sigma(\theta)\|\rho(\theta+\varepsilon)\otimes\sigma(\theta+\varepsilon))\right|_{\varepsilon=0}\\
 & =\left.\frac{\partial^{2}}{\partial\varepsilon_{i}\partial\varepsilon_{j}}\left[\boldsymbol{D}(\rho(\theta)\|\rho(\theta+\varepsilon))+\boldsymbol{D}(\sigma(\theta)\|\sigma(\theta+\varepsilon))\right]\right|_{\varepsilon=0}\\
 & =\left.\frac{\partial^{2}}{\partial\varepsilon_{i}\partial\varepsilon_{j}}\boldsymbol{D}(\rho(\theta)\|\rho(\theta+\varepsilon))\right|_{\varepsilon=0}+\left.\frac{\partial^{2}}{\partial\varepsilon_{i}\partial\varepsilon_{j}}\boldsymbol{D}(\sigma(\theta)\|\sigma(\theta+\varepsilon))\right|_{\varepsilon=0}\\
 & =\left[I_{\boldsymbol{D}}\!\left(\theta;\left(\rho(\theta)\right)_{\theta\in\Theta}\right)\right]_{i,j}+\left[I_{\boldsymbol{D}}\!\left(\theta;\left(\sigma(\theta)\right)_{\theta\in\Theta}\right)\right]_{i,j},
\end{align}
where the second equality follows from the assumption that $\boldsymbol{D}$
is additive on tensor-product states.

As stated previously, the main distinction between the classical and
quantum cases is that there is no longer a unique quantum generalization
of Fisher information. One of the main reasons for this is that a
state $\rho(\theta)$ and each derivative $\frac{\partial}{\partial\theta_{i}}\rho(\theta)$
need not commute. Regardless, there is still a remarkable characterization
due to \cite{Morozova1991,Petz1996} (see also \cite[Theorem~14.1]{Bengtsson2006}),
which I now recall.
\begin{thm}[\cite{Morozova1991,Petz1996}]
\label{thm:q-fisher-info-char}Let $\left(\rho(\theta)\right)_{\theta\in\Theta}$
be a second-order differentiable parameterized family of positive
definite states, where $\Theta\subseteq\mathbb{R}^{L}$ is open and
$L\in\mathbb{N}$. For every smooth quantum divergence $\boldsymbol{D}$,
there exists a function $\zeta\colon\mathbb{R}_{+}\times\mathbb{R}_{+}\to\mathbb{R}_{+}$
satisfying properties 1--4 below, such that the following equality
holds:
\begin{equation}
\left.\frac{\partial^{2}}{\partial\varepsilon_{i}\partial\varepsilon_{j}}\boldsymbol{D}(\rho(\theta)\|\rho(\theta+\varepsilon))\right|_{\varepsilon=0}=\left[I_{\zeta}(\theta)\right]_{i,j},\label{eq:q-fisher-info-petz}
\end{equation}
where the quantum information matrix $I_{\zeta}(\theta)$ is defined
in terms of its elements as 
\begin{equation}
\left[I_{\zeta}(\theta)\right]_{i,j}\coloneqq\sum_{k,\ell}\zeta(\lambda_{k},\lambda_{\ell})\Tr\!\left[\Pi_{k}\left(\partial_{i}\rho(\theta)\right)\Pi_{\ell}\left(\partial_{j}\rho(\theta)\right)\right],
\end{equation}
$\partial_{i}\equiv\frac{\partial}{\partial\theta_{i}}$, $\rho(\theta)$
has a spectral decomposition as $\rho(\theta)=\sum_{k}\lambda_{k}\Pi_{k}$,
(which suppresses the dependence of each eigenvalue $\lambda_{k}$
and each eigenprojection $\Pi_{k}$ on $\theta$), and the function
$\zeta$ satisfies the following:
\begin{enumerate}
\item There exists a constant $\kappa>0$ such that $\zeta(x,x)=\frac{\kappa}{x}$
for all $x>0$.
\item $\zeta$ is symmetric, i.e., $\zeta(x,y)=\zeta(y,x)$ for all $x,y>0$.
\item $\zeta$ satisfies $\zeta(sx,sy)=s^{-1}\zeta(x,y)$ for all $s,x,y>0$.
\item The function $t\mapsto f(t)\coloneqq\frac{1}{\zeta(t,1)}$ is operator
monotone on $t\in\left(0,\infty\right)$.
\end{enumerate}
\end{thm}

\section{Kubo--Mori information matrix from log-Euclidean R\'enyi relative
entropies}

\label{sec:log-Euclid-to-KM}The main result of this section is Theorem~\ref{thm:log-Euclidean-information-matrix},
which asserts that the information matrix resulting from the log-Euclidean
R\'enyi relative entropy is equal to the Kubo--Mori information
matrix, for all values of the R\'enyi parameter $\alpha$.

The log-Euclidean R\'enyi relative entropy of positive definite states
$\rho$ and $\sigma$ is defined for all $\alpha\in\left(0,1\right)\cup\left(1,\infty\right)$
as
\begin{equation}
D_{\alpha}^{\flat}(\rho\|\sigma)\coloneqq\frac{1}{\alpha-1}\ln\Tr[\exp(\alpha\ln\rho+(1-\alpha)\ln\sigma)].\label{eq:log-euclidean-renyi}
\end{equation}
See \cite[Section~4]{Audenaert2015} and \cite[Eq.~(17)]{Mosonyi2017}.
It is equal to the standard (Umegaki) quantum relative entropy \cite{Umegaki1962}
in the limit $\alpha\to1$:
\begin{equation}
\lim_{\alpha\to1}D_{\alpha}^{\flat}(\rho\|\sigma)=D(\rho\|\sigma)\coloneqq\Tr\!\left[\rho\left(\ln\rho-\ln\sigma\right)\right].
\end{equation}
The data-processing inequality holds for $D_{\alpha}^{\flat}$ if
$\alpha\in(0,1)$, and it does not hold if $\alpha>1$, as a special
case of Fact~\ref{fact:a-z-data-proc} with $z\to\infty$.

Let $\left(\rho(\theta)\right)_{\theta\in\mathbb{R}^{L}}$ be a parameterized
family of positive definite states, where $\theta$ is an $L$-dimensional
real parameter vector. The log-Euclidean information matrix is defined
in terms of \eqref{eq:log-euclidean-renyi} and for all $\alpha\in\left(0,1\right)\cup\left(1,\infty\right)$
as follows:
\begin{equation}
\left[I_{\alpha}^{\flat}(\theta)\right]_{i,j}\coloneqq\frac{1}{\alpha}\left(\left.\frac{\partial^{2}}{\partial\varepsilon_{i}\partial\varepsilon_{j}}D_{\alpha}^{\flat}(\rho(\theta)\|\rho(\theta+\varepsilon))\right|_{\varepsilon=0}\right).\label{eq:log-euclidean-fisher-def}
\end{equation}

\begin{defn}[Kubo--Mori information matrix]
The Kubo--Mori information matrix $I^{\KM}(\theta)$ is defined
in terms of its matrix elements as
\begin{equation}
\left[I_{\KM}(\theta)\right]_{i,j}\coloneqq\left(\left.\frac{\partial^{2}}{\partial\varepsilon_{i}\partial\varepsilon_{j}}D(\rho(\theta)\|\rho(\theta+\varepsilon))\right|_{\varepsilon=0}\right)\label{eq:kubo-mori-def}
\end{equation}
and it has the following alternative expressions:
\begin{align}
\left[I_{\KM}(\theta)\right]_{i,j} & =\Tr\!\left[\left(\partial_{i}\rho(\theta)\right)\left(\partial_{j}\ln\rho(\theta)\right)\right],\label{eq:kubo-mori-elements}\\
 & =\int_{0}^{\infty}dt\ \Tr\!\left[\left(\partial_{i}\rho(\theta)\right)\left(\rho(\theta)+tI\right)^{-1}\left(\partial_{j}\rho(\theta)\right)\left(\rho(\theta)+tI\right)^{-1}\right]\label{eq:kubo-mori-elements-2}\\
 & =\sum_{k,\ell}f_{\ln}^{\left[1\right]}(\lambda_{k},\lambda_{\ell})\Tr\!\left[\left(\partial_{i}\rho(\theta)\right)\Pi_{k}\left(\partial_{j}\rho(\theta)\right)\Pi_{\ell}\right],\label{eq:kubo-mori-elements-3}
\end{align}
where $\partial_{i}\equiv\frac{\partial}{\partial\theta_{i}}$, a
spectral decomposition of $\rho(\theta)$ is given by $\rho(\theta)=\sum_{k}\lambda_{k}\Pi_{k}$,
and the first divided difference $f_{\ln}^{\left[1\right]}(x,y)$
of the logarithm function $x\mapsto\ln x$ is defined for all $x,y>0$
as
\begin{equation}
f_{\ln}^{\left[1\right]}(x,y)\coloneqq\begin{cases}
\frac{\ln x-\ln y}{x-y} & :x\neq y\\
\frac{1}{x} & :x=y
\end{cases}.
\end{equation}
\end{defn}

The equalities in \eqref{eq:kubo-mori-elements-2} and \eqref{eq:kubo-mori-elements-3}
follow from, e.g., \cite[Eqs.~(B9), (B12), (B22), (B23)]{Sbahi2022},
and the first equality follows from applying Proposition \ref{prop:derivative-log}
to \eqref{eq:kubo-mori-elements-2}. See also Theorem~\ref{thm:divided-difference-matrix-deriv}
for concluding \eqref{eq:kubo-mori-elements-3} from \eqref{eq:kubo-mori-elements}.
\begin{thm}
\label{thm:log-Euclidean-information-matrix}The log-Euclidean information
matrix in \eqref{eq:log-euclidean-fisher-def} is equal to the Kubo--Mori
information matrix for all $\alpha\in\left(0,1\right)\cup\left(1,\infty\right)$:
\begin{align}
I_{\alpha}^{\flat}(\theta) & =I_{\KM}(\theta).\label{eq:log-euclidean-fisher-elements}
\end{align}
\end{thm}

\begin{proof}
Defining
\begin{equation}
\omega(\alpha,\theta,\varepsilon)\coloneqq\exp(\alpha\ln\rho(\theta)+(1-\alpha)\ln\rho(\theta+\varepsilon)),
\end{equation}
consider that
\begin{align}
\frac{\partial^{2}}{\partial\varepsilon_{i}\partial\varepsilon_{j}}D_{\alpha}^{\flat}(\rho(\theta)\|\rho(\theta+\varepsilon)) & =\frac{\partial^{2}}{\partial\varepsilon_{i}\partial\varepsilon_{j}}\left(\frac{1}{\alpha-1}\ln\Tr[\omega(\alpha,\theta,\varepsilon)]\right)\\
 & =\frac{1}{\alpha-1}\frac{\partial}{\partial\varepsilon_{i}}\left(\frac{\partial}{\partial\varepsilon_{j}}\ln\Tr[\omega(\alpha,\theta,\varepsilon)]\right)\\
 & =\frac{1}{\alpha-1}\frac{\partial}{\partial\varepsilon_{i}}\left(\frac{\frac{\partial}{\partial\varepsilon_{j}}\Tr[\omega(\alpha,\theta,\varepsilon)]}{\Tr[\omega(\alpha,\theta,\varepsilon)]}\right).\label{eq:proof-step-1}
\end{align}
Now consider that
\begin{align}
\frac{\partial}{\partial\varepsilon_{j}}\Tr[\omega(\alpha,\theta,\varepsilon)] & =\frac{\partial}{\partial\varepsilon_{j}}\Tr[\exp(\alpha\ln\rho(\theta)+(1-\alpha)\ln\rho(\theta+\varepsilon))]\label{eq:proof-step-3}\\
 & =\Tr\!\left[\begin{array}{c}
\exp(\alpha\ln\rho(\theta)+(1-\alpha)\ln\rho(\theta+\varepsilon))\times\\
\frac{\partial}{\partial\varepsilon_{j}}\left(\alpha\ln\rho(\theta)+(1-\alpha)\ln\rho(\theta+\varepsilon)\right)
\end{array}\right]\\
 & =(1-\alpha)\Tr\!\left[\omega(\alpha,\theta,\varepsilon)\left(\frac{\partial}{\partial\varepsilon_{j}}\ln\rho(\theta+\varepsilon)\right)\right],\label{eq:proof-step-2}
\end{align}
where the penultimate equality follows from Corollary~\ref{cor:derivative-in-trace}.
Then, plugging \eqref{eq:proof-step-2} into \eqref{eq:proof-step-1},
we find that 
\begin{align}
 & \frac{\partial^{2}}{\partial\varepsilon_{i}\partial\varepsilon_{j}}D_{\alpha}^{\flat}(\rho(\theta)\|\rho(\theta+\varepsilon))\nonumber \\
 & =-\frac{\partial}{\partial\varepsilon_{i}}\left(\frac{\Tr\!\left[\omega(\alpha,\theta,\varepsilon)\left(\frac{\partial}{\partial\varepsilon_{j}}\ln\rho(\theta+\varepsilon)\right)\right]}{\Tr[\omega(\alpha,\theta,\varepsilon)]}\right)\\
 & =\frac{\left(\frac{\partial}{\partial\varepsilon_{i}}\Tr[\omega(\alpha,\theta,\varepsilon)]\right)\Tr\!\left[\omega(\alpha,\theta,\varepsilon)\left(\frac{\partial}{\partial\varepsilon_{j}}\ln\rho(\theta+\varepsilon)\right)\right]}{\left(\Tr[\omega(\alpha,\theta,\varepsilon)]\right)^{2}}\nonumber \\
 & \qquad-\frac{\frac{\partial}{\partial\varepsilon_{i}}\Tr\!\left[\omega(\alpha,\theta,\varepsilon)\left(\frac{\partial}{\partial\varepsilon_{j}}\ln\rho(\theta+\varepsilon)\right)\right]}{\Tr[\omega(\alpha,\theta,\varepsilon)]}\\
 & =\left(1-\alpha\right)\frac{\Tr\!\left[\omega(\alpha,\theta,\varepsilon)\left(\frac{\partial}{\partial\varepsilon_{i}}\ln\rho(\theta+\varepsilon)\right)\right]\Tr\!\left[\omega(\alpha,\theta,\varepsilon)\left(\frac{\partial}{\partial\varepsilon_{j}}\ln\rho(\theta+\varepsilon)\right)\right]}{\left(\Tr[\omega(\alpha,\theta,\varepsilon)]\right)^{2}}\nonumber \\
 & \qquad-\frac{\Tr\!\left[\left(\frac{\partial}{\partial\varepsilon_{i}}\omega(\alpha,\theta,\varepsilon)\right)\left(\frac{\partial}{\partial\varepsilon_{j}}\ln\rho(\theta+\varepsilon)\right)\right]}{\Tr[\omega(\alpha,\theta,\varepsilon)]}\nonumber \\
 & \qquad-\frac{\Tr\!\left[\omega(\alpha,\theta,\varepsilon)\left(\frac{\partial^{2}}{\partial\varepsilon_{i}\partial\varepsilon_{j}}\ln\rho(\theta+\varepsilon)\right)\right]}{\Tr[\omega(\alpha,\theta,\varepsilon)]}.
\end{align}
The last equality follows from \eqref{eq:proof-step-3}--\eqref{eq:proof-step-2}
and the product rule for derivatives. By observing that
\begin{align}
\omega(\alpha,\theta,0) & \coloneqq\exp(\alpha\ln\rho(\theta)+(1-\alpha)\ln\rho(\theta))\\
 & =\rho(\theta),\\
\left.\frac{\partial}{\partial\varepsilon_{j}}\ln\rho(\theta+\varepsilon)\right|_{\varepsilon=0} & =\frac{\partial}{\partial\theta_{j}}\ln\rho(\theta),\\
\left.\frac{\partial^{2}}{\partial\varepsilon_{i}\partial\varepsilon_{j}}\ln\rho(\theta+\varepsilon)\right|_{\varepsilon=0} & =\frac{\partial^{2}}{\partial\theta_{i}\partial\theta_{j}}\ln\rho(\theta),
\end{align}
it follows that
\begin{align}
 & \left.\frac{\partial^{2}}{\partial\varepsilon_{i}\partial\varepsilon_{j}}D_{\alpha}^{\flat}(\rho(\theta)\|\rho(\theta+\varepsilon))\right|_{\varepsilon=0}\nonumber \\
 & =\left(1-\alpha\right)\frac{\Tr\!\left[\omega(\alpha,\theta,0)\left(\left.\frac{\partial}{\partial\varepsilon_{i}}\ln\rho(\theta+\varepsilon)\right|_{\varepsilon=0}\right)\right]\Tr\!\left[\omega(\alpha,\theta,0)\left(\left.\frac{\partial}{\partial\varepsilon_{j}}\ln\rho(\theta+\varepsilon)\right|_{\varepsilon=0}\right)\right]}{\left(\Tr[\omega(\alpha,\theta,0)]\right)^{2}}\nonumber \\
 & \qquad-\frac{\Tr\!\left[\left(\left.\frac{\partial}{\partial\varepsilon_{i}}\omega(\alpha,\theta,\varepsilon)\right|_{\varepsilon=0}\right)\left(\left.\frac{\partial}{\partial\varepsilon_{j}}\ln\rho(\theta+\varepsilon)\right|_{\varepsilon=0}\right)\right]}{\Tr[\omega(\alpha,\theta,0)]}\nonumber \\
 & \qquad-\frac{\Tr\!\left[\omega(\alpha,\theta,0)\left(\left.\frac{\partial^{2}}{\partial\varepsilon_{i}\partial\varepsilon_{j}}\ln\rho(\theta+\varepsilon)\right|_{\varepsilon=0}\right)\right]}{\Tr[\omega(\alpha,\theta,0)]}\\
 & =\left(1-\alpha\right)\Tr\!\left[\rho(\theta)\left(\frac{\partial}{\partial\theta_{i}}\ln\rho(\theta)\right)\right]\Tr\!\left[\rho(\theta)\left(\frac{\partial}{\partial\theta_{j}}\ln\rho(\theta)\right)\right]\nonumber \\
 & \qquad-\Tr\!\left[\left(\left.\frac{\partial}{\partial\varepsilon_{i}}\omega(\alpha,\theta,\varepsilon)\right|_{\varepsilon=0}\right)\left(\frac{\partial}{\partial\theta_{j}}\ln\rho(\theta)\right)\right]-\Tr\!\left[\rho(\theta)\left(\frac{\partial^{2}}{\partial\theta_{i}\partial\theta_{j}}\ln\rho(\theta)\right)\right]\\
 & =-\Tr\!\left[\left(\left.\frac{\partial}{\partial\varepsilon_{i}}\omega(\alpha,\theta,\varepsilon)\right|_{\varepsilon=0}\right)\left(\frac{\partial}{\partial\theta_{j}}\ln\rho(\theta)\right)\right]-\Tr\!\left[\rho(\theta)\left(\frac{\partial^{2}}{\partial\theta_{i}\partial\theta_{j}}\ln\rho(\theta)\right)\right],\label{eq:log-euclidean-fisher-almost}
\end{align}
where the final equality follows because
\begin{equation}
\Tr\!\left[\rho(\theta)\left(\frac{\partial}{\partial\theta_{i}}\ln\rho(\theta)\right)\right]=\Tr\!\left[\rho(\theta)\left(\frac{\partial}{\partial\theta_{j}}\ln\rho(\theta)\right)\right]=0,
\end{equation}
 due to \cite[Eqs.~(E78)–(E84)]{Minervini2025}. Now consider, from
Proposition \ref{prop:deriv-exp}, that
\begin{align}
 & \frac{\partial}{\partial\varepsilon_{i}}\omega(\alpha,\theta,\varepsilon)\nonumber \\
 & =\frac{\partial}{\partial\varepsilon_{i}}\exp(\alpha\ln\rho(\theta)+(1-\alpha)\ln\rho(\theta+\varepsilon))\\
 & =\int_{0}^{1}dt\ e^{t[\alpha\ln\rho(\theta)+(1-\alpha)\ln\rho(\theta+\varepsilon)]}\left[\frac{\partial}{\partial\varepsilon_{i}}\left(\alpha\ln\rho(\theta)+(1-\alpha)\ln\rho(\theta+\varepsilon)\right)\right]\times\nonumber \\
 & \qquad\qquad e^{\left(1-t\right)[\alpha\ln\rho(\theta)+(1-\alpha)\ln\rho(\theta+\varepsilon)]}\\
 & =\left(1-\alpha\right)\int_{0}^{1}dt\ \left[\omega(\alpha,\theta,\varepsilon)\right]^{t}\left[\frac{\partial}{\partial\varepsilon_{i}}\ln\rho(\theta+\varepsilon)\right]\left[\omega(\alpha,\theta,\varepsilon)\right]^{1-t}.
\end{align}
Then we find that
\begin{align}
\left.\frac{\partial}{\partial\varepsilon_{i}}\omega(\alpha,\theta,\varepsilon)\right|_{\varepsilon=0} & =\left(1-\alpha\right)\int_{0}^{1}dt\ \left[\omega(\alpha,\theta,0)\right]^{t}\left[\left.\frac{\partial}{\partial\varepsilon_{i}}\ln\rho(\theta+\varepsilon)\right|_{\varepsilon=0}\right]\left[\omega(\alpha,\theta,0)\right]^{1-t}\\
 & =\left(1-\alpha\right)\int_{0}^{1}dt\ \rho(\theta)^{t}\left[\frac{\partial}{\partial\theta_{i}}\ln\rho(\theta)\right]\rho(\theta)^{1-t}\label{eq:omega-at-0-simplify-1}\\
 & =\left(1-\alpha\right)\frac{\partial}{\partial\theta_{i}}\rho(\theta).\label{eq:omega-at-0-simplify-final}
\end{align}
where the last equality follows from Proposition \ref{prop:deriv-exp}
because
\begin{align}
\frac{\partial}{\partial\theta_{i}}\rho(\theta) & =\frac{\partial}{\partial\theta_{i}}\left(e^{\ln\rho(\theta)}\right)\\
 & =\int_{0}^{1}dt\ \rho(\theta)^{t}\left[\frac{\partial}{\partial\theta_{i}}\ln\rho(\theta)\right]\rho(\theta)^{1-t}.
\end{align}
Substituting \eqref{eq:omega-at-0-simplify-final} into \eqref{eq:log-euclidean-fisher-almost},
we conclude that
\begin{align}
 & \left.\frac{\partial^{2}}{\partial\varepsilon_{i}\partial\varepsilon_{j}}D_{\alpha}^{\flat}(\rho(\theta)\|\rho(\theta+\varepsilon))\right|_{\varepsilon=0}\nonumber \\
 & =\left(\alpha-1\right)\Tr\!\left[\left(\frac{\partial}{\partial\theta_{i}}\rho(\theta)\right)\left(\frac{\partial}{\partial\theta_{j}}\ln\rho(\theta)\right)\right]-\Tr\!\left[\rho(\theta)\left(\frac{\partial^{2}}{\partial\theta_{i}\partial\theta_{j}}\ln\rho(\theta)\right)\right]\\
 & =\left(\alpha-1\right)\Tr\!\left[\left(\frac{\partial}{\partial\theta_{i}}\rho(\theta)\right)\left(\frac{\partial}{\partial\theta_{j}}\ln\rho(\theta)\right)\right]+\Tr\!\left[\left(\frac{\partial}{\partial\theta_{i}}\rho(\theta)\right)\left(\frac{\partial}{\partial\theta_{j}}\ln\rho(\theta)\right)\right]\\
 & =\alpha\Tr\!\left[\left(\frac{\partial}{\partial\theta_{i}}\rho(\theta)\right)\left(\frac{\partial}{\partial\theta_{j}}\ln\rho(\theta)\right)\right],\label{eq:log-euclidean-proof-almost-done}
\end{align}
where the second equality follows from \cite[Eq.~(E87)]{Minervini2025}.
Applying Proposition \ref{prop:derivative-log} and Theorem~\ref{thm:divided-difference-matrix-deriv},
we conclude that
\begin{align}
 & \Tr\!\left[\left(\frac{\partial}{\partial\theta_{i}}\rho(\theta)\right)\left(\frac{\partial}{\partial\theta_{j}}\ln\rho(\theta)\right)\right]\nonumber \\
 & =\int_{0}^{\infty}dt\ \Tr\!\left[\left(\frac{\partial}{\partial\theta_{i}}\rho(\theta)\right)\left(\rho(\theta)+tI\right)^{-1}\left(\frac{\partial}{\partial\theta_{j}}\rho(\theta)\right)\left(\rho(\theta)+tI\right)^{-1}\right]\\
 & =\sum_{k,\ell}f_{\ln}^{\left[1\right]}(\lambda_{k},\lambda_{\ell})\Tr\!\left[\left(\frac{\partial}{\partial\theta_{i}}\rho(\theta)\right)\Pi_{k}\left(\frac{\partial}{\partial\theta_{j}}\rho(\theta)\right)\Pi_{\ell}\right],
\end{align}
thus completing the proof after dividing \eqref{eq:log-euclidean-proof-almost-done}
by $\alpha$.
\end{proof}
\begin{rem}
As stated previously, the log-Euclidean R\'enyi relative entropy
$D_{\alpha}^{\flat}$ obeys the data-processing inequality for all
$\alpha\in\left(0,1\right)$, and it does not when $\alpha>1$. However,
as indicated by Theorem~\ref{thm:log-Euclidean-information-matrix},
the log-Euclidean information matrix is equal to the Kubo--Mori information
matrix for all $\alpha\in\left(0,1\right)\cup\left(1,\infty\right)$.
As such, the log-Euclidean information matrix obeys the data-processing
inequality for all $\alpha\in\left(0,1\right)\cup\left(1,\infty\right)$
even though the log-Euclidean R\'enyi relative entropy obeys it only
for $\alpha\in\left(0,1\right)$.
\end{rem}

\section{RLD Fisher information matrix from geometric relative entropies}

\label{sec:RLD-Fisher-information-from-geometric}The main results
of this section are Theorem~\ref{thm:geometric-Renyi-to-RLD}, which
asserts that the information matrix resulting from the geometric R\'enyi
relative entropy is equal to the RLD Fisher information matrix, for
all values of the R\'enyi parameter $\alpha$, and Theorem~\ref{thm:belavski-RLD},
which asserts that the information matrix resulting from the Belavkin--Staszewski
relative entropy is also equal to the RLD Fisher information matrix.

\subsection{RLD Fisher information matrix from geometric R\'enyi relative entropies}

The geometric R\'enyi relative entropy is defined for all $\alpha\in\left(0,1\right)\cup\left(1,\infty\right)$
and positive definite states $\rho$ and $\sigma$ as follows \cite{Matsumoto2013,Matsumoto2018}:
\begin{align}
\widehat{D}_{\alpha}(\rho\|\sigma) & \coloneqq\frac{1}{\alpha-1}\ln\Tr\!\left[\sigma\left(\sigma^{-\frac{1}{2}}\rho\sigma^{-\frac{1}{2}}\right)^{\alpha}\right]\label{eq:geometric-Renyi-def}\\
 & =\frac{1}{\alpha-1}\ln\Tr\!\left[\rho\left(\rho^{-\frac{1}{2}}\sigma\rho^{-\frac{1}{2}}\right)^{1-\alpha}\right],
\end{align}
where the second equality follows from \cite[Proposition~7.39]{KW2024book}.
The data-processing inequality holds for $\widehat{D}_{\alpha}$ if
$\alpha\in(0,1)\cup(1,2]$, and it does not for $\alpha>2$ \cite[Example~3.36]{Mosonyi2024}.

The elements of the geometric information matrix induced by this relative
entropy are defined for a second-order differentiable parameterized
family $\left(\rho(\theta)\right)_{\theta\in\mathbb{R}^{L}}$ of positive
definite states as
\begin{equation}
\left[\widehat{I}_{\alpha}(\theta)\right]_{i,j}\coloneqq\frac{1}{\alpha}\left(\left.\frac{\partial^{2}}{\partial\varepsilon_{i}\partial\varepsilon_{j}}\widehat{D}_{\alpha}(\rho(\theta)\|\rho(\theta+\varepsilon))\right|_{\varepsilon=0}\right).\label{eq:geometric-Renyi-fisher-def}
\end{equation}

The elements of the right-logarithmic derivative (RLD) information
matrix are given by
\begin{align}
\left[I_{\RLD}(\theta)\right]_{i,j} & \coloneqq\Re\left[\Tr\!\left[\left(\partial_{i}\rho(\theta)\right)\rho(\theta)^{-1}\left(\partial_{j}\rho(\theta)\right)\right]\right]\label{eq:RLD-def}\\
 & =\frac{1}{2}\Tr\!\left[\left\{ \partial_{i}\rho(\theta),\partial_{j}\rho(\theta)\right\} \rho(\theta)^{-1}\right],
\end{align}
where $\partial_{i}\equiv\frac{\partial}{\partial\theta_{i}}$. This
is actually the real part of what is commonly called the RLD information
matrix (see, e.g., \cite[Eq.~(176)]{Sidhu2020}), but for brevity
and simplicity, I refer to it here and throughout as the RLD information
matrix.
\begin{thm}
\label{thm:geometric-Renyi-to-RLD}Let $\left(\rho(\theta)\right)_{\theta\in\mathbb{R}^{L}}$
be a second-order differentiable parameterized family of positive
definite states. The geometric information matrix $\widehat{I}_{\alpha}(\theta)$
in \eqref{eq:geometric-Renyi-fisher-def} is equal to the RLD information
matrix for all $\alpha\in\left(0,1\right)\cup\left(1,\infty\right)$:
\begin{equation}
\widehat{I}_{\alpha}(\theta)=I_{\RLD}(\theta).\label{eq:GRF-theorem-equality}
\end{equation}
\end{thm}

\begin{proof}
Consider that
\begin{align}
 & \frac{\partial^{2}}{\partial\varepsilon_{i}\partial\varepsilon_{j}}\widehat{D}_{\alpha}(\rho(\theta)\|\rho(\theta+\varepsilon))\nonumber \\
 & =\frac{\partial^{2}}{\partial\varepsilon_{i}\partial\varepsilon_{j}}\frac{1}{\alpha-1}\ln\Tr\!\left[\rho(\theta)\left(\rho(\theta)^{-\frac{1}{2}}\rho(\theta+\varepsilon)\rho(\theta)^{-\frac{1}{2}}\right)^{1-\alpha}\right]\\
 & =\frac{1}{\alpha-1}\frac{\partial}{\partial\varepsilon_{i}}\frac{\partial}{\partial\varepsilon_{j}}\ln\Tr\!\left[\rho(\theta)\left(\rho(\theta)^{-\frac{1}{2}}\rho(\theta+\varepsilon)\rho(\theta)^{-\frac{1}{2}}\right)^{1-\alpha}\right]\\
 & =\frac{1}{\alpha-1}\frac{\partial}{\partial\varepsilon_{i}}\left(\frac{\frac{\partial}{\partial\varepsilon_{j}}\Tr\!\left[\rho(\theta)\left(\rho(\theta)^{-\frac{1}{2}}\rho(\theta+\varepsilon)\rho(\theta)^{-\frac{1}{2}}\right)^{1-\alpha}\right]}{\Tr\!\left[\rho(\theta)\left(\rho(\theta)^{-\frac{1}{2}}\rho(\theta+\varepsilon)\rho(\theta)^{-\frac{1}{2}}\right)^{1-\alpha}\right]}\right)\\
 & =-\frac{1}{\alpha-1}\left(\frac{\begin{array}{c}
\frac{\partial}{\partial\varepsilon_{i}}\Tr\!\left[\rho(\theta)\left(\rho(\theta)^{-\frac{1}{2}}\rho(\theta+\varepsilon)\rho(\theta)^{-\frac{1}{2}}\right)^{1-\alpha}\right]\times\\
\frac{\partial}{\partial\varepsilon_{j}}\Tr\!\left[\rho(\theta)\left(\rho(\theta)^{-\frac{1}{2}}\rho(\theta+\varepsilon)\rho(\theta)^{-\frac{1}{2}}\right)^{1-\alpha}\right]
\end{array}}{\left(\Tr\!\left[\rho(\theta)\left(\rho(\theta)^{-\frac{1}{2}}\rho(\theta+\varepsilon)\rho(\theta)^{-\frac{1}{2}}\right)^{1-\alpha}\right]\right)^{2}}\right)\nonumber \\
 & \qquad+\frac{1}{\alpha-1}\left(\frac{\frac{\partial^{2}}{\partial\varepsilon_{i}\partial\varepsilon_{j}}\Tr\!\left[\rho(\theta)\left(\rho(\theta)^{-\frac{1}{2}}\rho(\theta+\varepsilon)\rho(\theta)^{-\frac{1}{2}}\right)^{1-\alpha}\right]}{\Tr\!\left[\rho(\theta)\left(\rho(\theta)^{-\frac{1}{2}}\rho(\theta+\varepsilon)\rho(\theta)^{-\frac{1}{2}}\right)^{1-\alpha}\right]}\right).
\end{align}
Then we find that
\begin{align}
 & \left.\frac{\partial^{2}}{\partial\varepsilon_{i}\partial\varepsilon_{j}}\widehat{D}_{\alpha}(\rho(\theta)\|\rho(\theta+\varepsilon))\right|_{\varepsilon=0}\nonumber \\
 & =-\frac{1}{\alpha-1}\left(\frac{\begin{array}{c}
\left.\frac{\partial}{\partial\varepsilon_{i}}\Tr\!\left[\rho(\theta)\left(\rho(\theta)^{-\frac{1}{2}}\rho(\theta+\varepsilon)\rho(\theta)^{-\frac{1}{2}}\right)^{1-\alpha}\right]\right|_{\varepsilon=0}\times\\
\left.\frac{\partial}{\partial\varepsilon_{j}}\Tr\!\left[\rho(\theta)\left(\rho(\theta)^{-\frac{1}{2}}\rho(\theta+\varepsilon)\rho(\theta)^{-\frac{1}{2}}\right)^{1-\alpha}\right]\right|_{\varepsilon=0}
\end{array}}{\left(\Tr\!\left[\rho(\theta)\left(\rho(\theta)^{-\frac{1}{2}}\rho(\theta)\rho(\theta)^{-\frac{1}{2}}\right)^{1-\alpha}\right]\right)^{2}}\right)\nonumber \\
 & \qquad+\frac{1}{\alpha-1}\left(\frac{\left.\frac{\partial^{2}}{\partial\varepsilon_{i}\partial\varepsilon_{j}}\Tr\!\left[\rho(\theta)\left(\rho(\theta)^{-\frac{1}{2}}\rho(\theta+\varepsilon)\rho(\theta)^{-\frac{1}{2}}\right)^{1-\alpha}\right]\right|_{\varepsilon=0}}{\Tr\!\left[\rho(\theta)\left(\rho(\theta)^{-\frac{1}{2}}\rho(\theta)\rho(\theta)^{-\frac{1}{2}}\right)^{1-\alpha}\right]}\right)\\
 & =-\frac{1}{\alpha-1}\left(\begin{array}{c}
\left.\frac{\partial}{\partial\varepsilon_{i}}\Tr\!\left[\rho(\theta)\left(\rho(\theta)^{-\frac{1}{2}}\rho(\theta+\varepsilon)\rho(\theta)^{-\frac{1}{2}}\right)^{1-\alpha}\right]\right|_{\varepsilon=0}\times\\
\left.\frac{\partial}{\partial\varepsilon_{j}}\Tr\!\left[\rho(\theta)\left(\rho(\theta)^{-\frac{1}{2}}\rho(\theta+\varepsilon)\rho(\theta)^{-\frac{1}{2}}\right)^{1-\alpha}\right]\right|_{\varepsilon=0}
\end{array}\right)\nonumber \\
 & \qquad+\frac{1}{\alpha-1}\left(\left.\frac{\partial^{2}}{\partial\varepsilon_{i}\partial\varepsilon_{j}}\Tr\!\left[\rho(\theta)\left(\rho(\theta)^{-\frac{1}{2}}\rho(\theta+\varepsilon)\rho(\theta)^{-\frac{1}{2}}\right)^{1-\alpha}\right]\right|_{\varepsilon=0}\right).\label{eq:geometric-renyi-fisher-proof-1st-block}
\end{align}
Let us now prove that
\begin{multline}
\left.\frac{\partial}{\partial\varepsilon_{i}}\Tr\!\left[\rho(\theta)\left(\rho(\theta)^{-\frac{1}{2}}\rho(\theta+\varepsilon)\rho(\theta)^{-\frac{1}{2}}\right)^{1-\alpha}\right]\right|_{\varepsilon=0}\\
=\left.\frac{\partial}{\partial\varepsilon_{j}}\Tr\!\left[\rho(\theta)\left(\rho(\theta)^{-\frac{1}{2}}\rho(\theta+\varepsilon)\rho(\theta)^{-\frac{1}{2}}\right)^{1-\alpha}\right]\right|_{\varepsilon=0}=0.\label{eq:1st-deriv-zero-geometric-renyi}
\end{multline}
To this end, define
\begin{equation}
\Delta_{\varepsilon}\equiv\rho(\theta+\varepsilon)-\rho(\theta).\label{eq:Delta-eps-op-def}
\end{equation}
By the assumption that the positive operator-valued function $\theta\mapsto\rho(\theta)$
is second-order differentiable, it follows that it is continuous.
In particular, for all $\theta\in\mathbb{R}^{L}$ and $\delta>0$,
there exists $\gamma>0$ such that for all $\varepsilon\in\mathbb{R}^{L}$
satisfying $\left\Vert \varepsilon\right\Vert \leq\gamma$, the inequality
$\left\Vert \rho(\theta+\varepsilon)-\rho(\theta)\right\Vert \leq\delta$
holds. Let us choose $\delta\in\left(0,\lambda_{\min}(\rho(\theta))\right)$,
where $\lambda_{\min}(\rho(\theta))$ denotes the minimum eigenvalue
of $\rho(\theta)$. Then by continuity, there exists $\gamma>0$,
such that for all $\varepsilon\in\mathbb{R}^{L}$ satisfying $\left\Vert \varepsilon\right\Vert \leq\gamma$,
the inequality $\left\Vert \rho(\theta+\varepsilon)-\rho(\theta)\right\Vert \leq\delta$
holds. Proceeding with this choice of $\gamma$ and restricting to
$\varepsilon\in\mathbb{R}^{L}$ satisfying $\left\Vert \varepsilon\right\Vert \leq\gamma$,
consider that
\begin{align}
 & \frac{\partial}{\partial\varepsilon_{i}}\Tr\!\left[\rho(\theta)\left(\rho(\theta)^{-\frac{1}{2}}\rho(\theta+\varepsilon)\rho(\theta)^{-\frac{1}{2}}\right)^{1-\alpha}\right]\nonumber \\
 & =\frac{\partial}{\partial\varepsilon_{i}}\Tr\!\left[\rho(\theta)\left(\rho(\theta)^{-\frac{1}{2}}\left(\Delta_{\varepsilon}+\rho(\theta)\right)\rho(\theta)^{-\frac{1}{2}}\right)^{1-\alpha}\right]\\
 & =\frac{\partial}{\partial\varepsilon_{i}}\Tr\!\left[\rho(\theta)\left(I+\rho(\theta)^{-\frac{1}{2}}\Delta_{\varepsilon}\rho(\theta)^{-\frac{1}{2}}\right)^{1-\alpha}\right]\\
 & =\frac{\partial}{\partial\varepsilon_{i}}\Tr\!\left[\rho(\theta)\left[\sum_{n=0}^{\infty}{1-\alpha \choose n}\left(\rho(\theta)^{-\frac{1}{2}}\Delta_{\varepsilon}\rho(\theta)^{-\frac{1}{2}}\right)^{n}\right]\right]\\
 & =\frac{\partial}{\partial\varepsilon_{i}}\left(\sum_{n=0}^{\infty}{1-\alpha \choose n}\Tr\!\left[\rho(\theta)\left(\rho(\theta)^{-\frac{1}{2}}\Delta_{\varepsilon}\rho(\theta)^{-\frac{1}{2}}\right)^{n}\right]\right)\\
 & =\frac{\partial}{\partial\varepsilon_{i}}\left(\begin{array}{c}
1+\left(1-\alpha\right)\Tr\!\left[\Delta_{\varepsilon}\right]+\frac{\left(1-\alpha\right)\left(-\alpha\right)}{2}\Tr\!\left[\Delta_{\varepsilon}\rho(\theta)^{-1}\Delta_{\varepsilon}\right]\\
+\sum_{n=3}^{\infty}{1-\alpha \choose n}\Tr\!\left[\left(\Delta_{\varepsilon}\rho(\theta)^{-1}\right)^{n-1}\Delta_{\varepsilon}\right]
\end{array}\right)\\
 & =\frac{\partial}{\partial\varepsilon_{i}}\left(\begin{array}{c}
\frac{\left(1-\alpha\right)\left(-\alpha\right)}{2}\Tr\!\left[\Delta_{\varepsilon}\rho(\theta)^{-1}\Delta_{\varepsilon}\right]+\sum_{n=3}^{\infty}{1-\alpha \choose n}\Tr\!\left[\left(\Delta_{\varepsilon}\rho(\theta)^{-1}\right)^{n-1}\Delta_{\varepsilon}\right]\end{array}\right).\label{eq:geometric-renyi-series-expansion}
\end{align}
The third equality follows from the binomial series expansion
\begin{equation}
\left(1+x\right)^{1-\alpha}=\sum_{n=0}^{\infty}{1-\alpha \choose n}x^{n},
\end{equation}
which converges for all $x\in\mathbb{R}$ satisfying $\left|x\right|<1$.
To see that the series
\begin{equation}
\sum_{n=0}^{\infty}{1-\alpha \choose n}\left(\rho(\theta)^{-\frac{1}{2}}\Delta_{\varepsilon}\rho(\theta)^{-\frac{1}{2}}\right)^{n}
\end{equation}
indeed converges to
\begin{equation}
\left(I+\rho(\theta)^{-\frac{1}{2}}\Delta_{\varepsilon}\rho(\theta)^{-\frac{1}{2}}\right)^{1-\alpha},
\end{equation}
consider that, by the choice of $\gamma$ mentioned above, the following
inequality holds for all $\varepsilon\in\mathbb{R}^{L}$ satisfying
$\left\Vert \varepsilon\right\Vert \leq\gamma$:
\begin{align}
\left\Vert \rho(\theta)^{-\frac{1}{2}}\Delta_{\varepsilon}\rho(\theta)^{-\frac{1}{2}}\right\Vert  & =\left\Vert \rho(\theta)^{-\frac{1}{2}}\left(\rho(\theta+\varepsilon)-\rho(\theta)\right)\rho(\theta)^{-\frac{1}{2}}\right\Vert \label{eq:Delta-eps-sandwich-small-1}\\
 & \leq\left\Vert \rho(\theta)^{-\frac{1}{2}}\right\Vert \left\Vert \rho(\theta+\varepsilon)-\rho(\theta)\right\Vert \left\Vert \rho(\theta)^{-\frac{1}{2}}\right\Vert \\
 & =\left\Vert \rho(\theta)^{-1}\right\Vert \left\Vert \rho(\theta+\varepsilon)-\rho(\theta)\right\Vert \\
 & =\frac{\left\Vert \rho(\theta+\varepsilon)-\rho(\theta)\right\Vert }{\lambda_{\min}(\rho(\theta))}\\
 & <1,\label{eq:Delta-eps-sandwich-small-last}
\end{align}
where the first inequality follows from submultiplicativity of the
spectral norm. The last equality in \eqref{eq:geometric-renyi-series-expansion}
follows because $\Tr\!\left[\Delta_{\varepsilon}\right]=0$ and $\frac{\partial}{\partial\varepsilon_{i}}(1)=0$.
Now consider that
\begin{equation}
\left.\frac{\partial}{\partial\varepsilon_{i}}\left(\begin{array}{c}
\frac{\left(1-\alpha\right)\left(-\alpha\right)}{2}\Tr\!\left[\Delta_{\varepsilon}\rho(\theta)^{-1}\Delta_{\varepsilon}\right]\\
+\sum_{n=3}^{\infty}{1-\alpha \choose n}\Tr\!\left[\left(\Delta_{\varepsilon}\rho(\theta)^{-1}\right)^{n-1}\Delta_{\varepsilon}\right]
\end{array}\right)\right|_{\varepsilon=0}=0.
\end{equation}
To see this, observe that applying the product rule to the first term
and evaluating it at $\varepsilon=0$ gives
\begin{multline}
\left.\left(\frac{\partial}{\partial\varepsilon_{i}}\Tr\!\left[\Delta_{\varepsilon}\rho(\theta)^{-1}\Delta_{\varepsilon}\right]\right)\right|_{\varepsilon=0}=\\
\left.\left(\Tr\!\left[\left(\frac{\partial}{\partial\varepsilon_{i}}\Delta_{\varepsilon}\right)\rho(\theta)^{-1}\Delta_{\varepsilon}\right]+\Tr\!\left[\Delta_{\varepsilon}\rho(\theta)^{-1}\left(\frac{\partial}{\partial\varepsilon_{i}}\Delta_{\varepsilon}\right)\right]\right)\right|_{\varepsilon=0}=0,
\end{multline}
which follows because $\lim_{\varepsilon\to0}\Delta_{\varepsilon}=0$.
Similar reasoning implies that
\begin{equation}
\left.\frac{\partial}{\partial\varepsilon_{i}}\Tr\!\left[\left(\Delta_{\varepsilon}\rho(\theta)^{-1}\right)^{n-1}\Delta_{\varepsilon}\right]\right|_{\varepsilon=0}=0
\end{equation}
for all $n\geq3$. We thus conclude \eqref{eq:1st-deriv-zero-geometric-renyi}.

Then the expression in \eqref{eq:geometric-renyi-fisher-proof-1st-block}
simplifies as follows:
\begin{multline}
\left.\frac{\partial^{2}}{\partial\varepsilon_{i}\partial\varepsilon_{j}}\widehat{D}_{\alpha}(\rho(\theta)\|\rho(\theta+\varepsilon))\right|_{\varepsilon=0}=\\
\frac{1}{\alpha-1}\left(\left.\frac{\partial^{2}}{\partial\varepsilon_{i}\partial\varepsilon_{j}}\Tr\!\left[\rho(\theta)\left(\rho(\theta)^{-\frac{1}{2}}\rho(\theta+\varepsilon)\rho(\theta)^{-\frac{1}{2}}\right)^{1-\alpha}\right]\right|_{\varepsilon=0}\right).\label{eq:geometric-renyi-fisher-reduction-to-2nd-deriv}
\end{multline}
To evaluate this last term, we can employ the same binomial series
expansion from \eqref{eq:geometric-renyi-series-expansion} to conclude
that
\begin{align}
 & \left.\left(\frac{\partial^{2}}{\partial\varepsilon_{i}\partial\varepsilon_{j}}\Tr\!\left[\rho(\theta)\left(\rho(\theta)^{-\frac{1}{2}}\rho(\theta+\varepsilon)\rho(\theta)^{-\frac{1}{2}}\right)^{1-\alpha}\right]\right)\right|_{\varepsilon=0}\nonumber \\
 & =\left.\left(\begin{array}{c}
\frac{\left(1-\alpha\right)\left(-\alpha\right)}{2}\frac{\partial^{2}}{\partial\varepsilon_{i}\partial\varepsilon_{j}}\Tr\!\left[\Delta_{\varepsilon}\rho(\theta)^{-1}\Delta_{\varepsilon}\right]\\
+\sum_{n=3}^{\infty}{1-\alpha \choose n}\frac{\partial^{2}}{\partial\varepsilon_{i}\partial\varepsilon_{j}}\Tr\!\left[\left(\Delta_{\varepsilon}\rho(\theta)^{-1}\right)^{n-1}\Delta_{\varepsilon}\right]
\end{array}\right)\right|_{\varepsilon=0}\\
 & =\frac{\alpha\left(\alpha-1\right)}{2}\left.\frac{\partial^{2}}{\partial\varepsilon_{i}\partial\varepsilon_{j}}\Tr\!\left[\Delta_{\varepsilon}\rho(\theta)^{-1}\Delta_{\varepsilon}\right]\right|_{\varepsilon=0}.\label{eq:GRF-almost-done}
\end{align}
The last equality follows because, for all terms for which $n\geq3$,
applying the product rule to $\frac{\partial^{2}}{\partial\varepsilon_{i}\partial\varepsilon_{j}}\Tr\!\left[\left(\Delta_{\varepsilon}\rho(\theta)^{-1}\right)^{n-1}\Delta_{\varepsilon}\right]$
leaves an expression featuring a sum of terms, each of which contains
at least one $\Delta_{\varepsilon}$ remaining, which in turn evaluates
to zero after taking the limit $\varepsilon\to0$. Thus, for all $n\geq3$,
\begin{equation}
\left.\frac{\partial^{2}}{\partial\varepsilon_{i}\partial\varepsilon_{j}}\Tr\!\left[\left(\Delta_{\varepsilon}\rho(\theta)^{-1}\right)^{n-1}\Delta_{\varepsilon}\right]\right|_{\varepsilon=0}=0.\label{eq:geo-renyi-higher-terms-vanish-1}
\end{equation}
For example, for $n=3$, we find that
\begin{align}
 & \frac{\partial^{2}}{\partial\varepsilon_{i}\partial\varepsilon_{j}}\Tr\!\left[\left(\Delta_{\varepsilon}\rho(\theta)^{-1}\right)^{2}\Delta_{\varepsilon}\right]\nonumber \\
 & =\frac{\partial^{2}}{\partial\varepsilon_{i}\partial\varepsilon_{j}}\Tr\!\left[\Delta_{\varepsilon}\rho(\theta)^{-1}\Delta_{\varepsilon}\rho(\theta)^{-1}\Delta_{\varepsilon}\right]\\
 & =\frac{\partial}{\partial\varepsilon_{i}}\left(\frac{\partial}{\partial\varepsilon_{j}}\Tr\!\left[\Delta_{\varepsilon}\rho(\theta)^{-1}\Delta_{\varepsilon}\rho(\theta)^{-1}\Delta_{\varepsilon}\right]\right)\\
 & =\frac{\partial}{\partial\varepsilon_{i}}\left(\begin{array}{c}
\Tr\!\left[\left(\frac{\partial}{\partial\varepsilon_{j}}\Delta_{\varepsilon}\right)\rho(\theta)^{-1}\Delta_{\varepsilon}\rho(\theta)^{-1}\Delta_{\varepsilon}\right]+\Tr\!\left[\Delta_{\varepsilon}\rho(\theta)^{-1}\left(\frac{\partial}{\partial\varepsilon_{j}}\Delta_{\varepsilon}\right)\rho(\theta)^{-1}\Delta_{\varepsilon}\right]\\
+\Tr\!\left[\Delta_{\varepsilon}\rho(\theta)^{-1}\Delta_{\varepsilon}\rho(\theta)^{-1}\left(\frac{\partial}{\partial\varepsilon_{j}}\Delta_{\varepsilon}\right)\right]
\end{array}\right)\\
 & =\frac{\partial}{\partial\varepsilon_{i}}\Tr\!\left[\left(\frac{\partial}{\partial\varepsilon_{j}}\Delta_{\varepsilon}\right)\rho(\theta)^{-1}\Delta_{\varepsilon}\rho(\theta)^{-1}\Delta_{\varepsilon}\right]+\frac{\partial}{\partial\varepsilon_{i}}\Tr\!\left[\Delta_{\varepsilon}\rho(\theta)^{-1}\left(\frac{\partial}{\partial\varepsilon_{j}}\Delta_{\varepsilon}\right)\rho(\theta)^{-1}\Delta_{\varepsilon}\right]\nonumber \\
 & \qquad+\frac{\partial}{\partial\varepsilon_{i}}\Tr\!\left[\Delta_{\varepsilon}\rho(\theta)^{-1}\Delta_{\varepsilon}\rho(\theta)^{-1}\left(\frac{\partial}{\partial\varepsilon_{j}}\Delta_{\varepsilon}\right)\right]\\
 & =\Tr\!\left[\left(\frac{\partial^{2}}{\partial\varepsilon_{i}\partial\varepsilon_{j}}\Delta_{\varepsilon}\right)\rho(\theta)^{-1}\Delta_{\varepsilon}\rho(\theta)^{-1}\Delta_{\varepsilon}\right]+\Tr\!\left[\left(\frac{\partial}{\partial\varepsilon_{j}}\Delta_{\varepsilon}\right)\rho(\theta)^{-1}\left(\frac{\partial}{\partial\varepsilon_{i}}\Delta_{\varepsilon}\right)\rho(\theta)^{-1}\Delta_{\varepsilon}\right]\nonumber \\
 & \qquad+\Tr\!\left[\left(\frac{\partial}{\partial\varepsilon_{j}}\Delta_{\varepsilon}\right)\rho(\theta)^{-1}\Delta_{\varepsilon}\rho(\theta)^{-1}\left(\frac{\partial}{\partial\varepsilon_{i}}\Delta_{\varepsilon}\right)\right]+\Tr\!\left[\left(\frac{\partial}{\partial\varepsilon_{i}}\Delta_{\varepsilon}\right)\rho(\theta)^{-1}\left(\frac{\partial}{\partial\varepsilon_{j}}\Delta_{\varepsilon}\right)\rho(\theta)^{-1}\Delta_{\varepsilon}\right]\nonumber \\
 & \qquad+\Tr\!\left[\Delta_{\varepsilon}\rho(\theta)^{-1}\left(\frac{\partial^{2}}{\partial\varepsilon_{i}\partial\varepsilon_{j}}\Delta_{\varepsilon}\right)\rho(\theta)^{-1}\Delta_{\varepsilon}\right]+\Tr\!\left[\Delta_{\varepsilon}\rho(\theta)^{-1}\left(\frac{\partial}{\partial\varepsilon_{j}}\Delta_{\varepsilon}\right)\rho(\theta)^{-1}\left(\frac{\partial}{\partial\varepsilon_{i}}\Delta_{\varepsilon}\right)\right]\nonumber \\
 & \qquad+\Tr\!\left[\left(\frac{\partial}{\partial\varepsilon_{i}}\Delta_{\varepsilon}\right)\rho(\theta)^{-1}\Delta_{\varepsilon}\rho(\theta)^{-1}\left(\frac{\partial}{\partial\varepsilon_{j}}\Delta_{\varepsilon}\right)\right]+\Tr\!\left[\Delta_{\varepsilon}\rho(\theta)^{-1}\left(\frac{\partial}{\partial\varepsilon_{i}}\Delta_{\varepsilon}\right)\rho(\theta)^{-1}\left(\frac{\partial}{\partial\varepsilon_{j}}\Delta_{\varepsilon}\right)\right]\nonumber \\
 & \qquad+\Tr\!\left[\Delta_{\varepsilon}\rho(\theta)^{-1}\Delta_{\varepsilon}\rho(\theta)^{-1}\left(\frac{\partial^{2}}{\partial\varepsilon_{i}\partial\varepsilon_{j}}\Delta_{\varepsilon}\right)\right].\label{eq:higher-derivs-zero-geometric-renyi}
\end{align}
Thus, in the final expression in \eqref{eq:higher-derivs-zero-geometric-renyi},
each term contains $\Delta_{\varepsilon}$, and after taking the $\varepsilon\to0$
limit, each term evaluates to zero. This analysis in fact implies,
by means of the H\"older and triangle inequalities, that
\begin{multline}
\left|\frac{\partial^{2}}{\partial\varepsilon_{i}\partial\varepsilon_{j}}\Tr\!\left[\left(\Delta_{\varepsilon}\rho(\theta)^{-1}\right)^{2}\Delta_{\varepsilon}\right]\right|\leq\\
3d\left\Vert \rho(\theta)^{-1}\right\Vert ^{2}\left\Vert \Delta_{\varepsilon}\right\Vert \left(\left\Vert \Delta_{\varepsilon}\right\Vert \left\Vert \frac{\partial^{2}}{\partial\varepsilon_{i}\partial\varepsilon_{j}}\Delta_{\varepsilon}\right\Vert +2\left\Vert \frac{\partial}{\partial\varepsilon_{i}}\Delta_{\varepsilon}\right\Vert \left\Vert \frac{\partial}{\partial\varepsilon_{j}}\Delta_{\varepsilon}\right\Vert \right),
\end{multline}
where $d$ is the dimension of the underlying Hilbert space. As such,
taking the limit $\varepsilon\to0$ implies that $\lim_{\varepsilon\to0}\left\Vert \Delta_{\varepsilon}\right\Vert =0$,
so that the term on the left-hand side vanishes. For general $n\geq3$,
this upper bound becomes
\begin{multline}
\left|\frac{\partial^{2}}{\partial\varepsilon_{i}\partial\varepsilon_{j}}\Tr\!\left[\left(\Delta_{\varepsilon}\rho(\theta)^{-1}\right)^{n}\Delta_{\varepsilon}\right]\right|\leq\\
nd\left\Vert \rho(\theta)^{-1}\right\Vert ^{n}\left\Vert \Delta_{\varepsilon}\right\Vert ^{n-2}\left(\left\Vert \Delta_{\varepsilon}\right\Vert \left\Vert \frac{\partial^{2}}{\partial\varepsilon_{i}\partial\varepsilon_{j}}\Delta_{\varepsilon}\right\Vert +\left(n-1\right)\left\Vert \frac{\partial}{\partial\varepsilon_{i}}\Delta_{\varepsilon}\right\Vert \left\Vert \frac{\partial}{\partial\varepsilon_{j}}\Delta_{\varepsilon}\right\Vert \right),\label{eq:geo-renyi-higher-terms-vanish-last}
\end{multline}
so that each of the terms indeed vanish in the limit $\varepsilon\to0$
as claimed.

It thus remains to evaluate \eqref{eq:GRF-almost-done}. To this end,
consider that
\begin{align}
 & \left.\frac{\partial^{2}}{\partial\varepsilon_{i}\partial\varepsilon_{j}}\Tr\!\left[\Delta_{\varepsilon}\rho(\theta)^{-1}\Delta_{\varepsilon}\right]\right|_{\varepsilon=0}\nonumber \\
 & =\left.\left(\begin{array}{c}
\Tr\!\left[\left(\frac{\partial^{2}}{\partial\varepsilon_{i}\partial\varepsilon_{j}}\Delta_{\varepsilon}\right)\rho(\theta)^{-1}\Delta_{\varepsilon}\right]+\Tr\!\left[\Delta_{\varepsilon}\rho(\theta)^{-1}\left(\frac{\partial^{2}}{\partial\varepsilon_{i}\partial\varepsilon_{j}}\Delta_{\varepsilon}\right)\right]\\
+\Tr\!\left[\left(\frac{\partial}{\partial\varepsilon_{i}}\Delta_{\varepsilon}\right)\rho(\theta)^{-1}\left(\frac{\partial}{\partial\varepsilon_{j}}\Delta_{\varepsilon}\right)\right]+\Tr\!\left[\left(\frac{\partial}{\partial\varepsilon_{j}}\Delta_{\varepsilon}\right)\rho(\theta)^{-1}\left(\frac{\partial}{\partial\varepsilon_{i}}\Delta_{\varepsilon}\right)\right]
\end{array}\right)\right|_{\varepsilon=0}\label{eq:geo-to-RLD-final-proof-1}\\
 & =\Tr\!\left[\left(\frac{\partial}{\partial\theta_{i}}\rho(\theta)\right)\rho(\theta)^{-1}\left(\frac{\partial}{\partial\theta_{j}}\rho(\theta)\right)\right]+\Tr\!\left[\left(\frac{\partial}{\partial\theta_{j}}\rho(\theta)\right)\rho(\theta)^{-1}\left(\frac{\partial}{\partial\theta_{i}}\rho(\theta)\right)\right]\\
 & =\Tr\!\left[\left\{ \frac{\partial}{\partial\theta_{j}}\rho(\theta),\frac{\partial}{\partial\theta_{i}}\rho(\theta)\right\} \rho(\theta)^{-1}\right].\label{eq:GRF-proof-done}
\end{align}
Finally, putting together the expressions from \eqref{eq:geometric-Renyi-fisher-def},
\eqref{eq:geometric-renyi-fisher-reduction-to-2nd-deriv}, \eqref{eq:GRF-almost-done},
and \eqref{eq:GRF-proof-done}, we conclude the desired equality in
\eqref{eq:GRF-theorem-equality}.
\end{proof}
\begin{rem}
As stated previously, the geometric R\'enyi relative entropy $\widehat{D}_{\alpha}$
obeys the data-processing inequality for all $\alpha\in\left(0,1\right)\cup\left(1,2\right]$,
and it does not when $\alpha>2$. However, as indicated by Theorem~\ref{thm:geometric-Renyi-to-RLD},
the geometric information matrix is equal to the RLD Fisher information
matrix for all $\alpha\in\left(0,1\right)\cup\left(1,\infty\right)$.
As such, the geometric information matrix obeys the data-processing
inequality for all $\alpha\in\left(0,1\right)\cup\left(1,\infty\right)$
even though the geometric R\'enyi relative entropy obeys it only
for $\alpha\in\left(0,1\right)\cup\left(1,2\right]$.
\end{rem}

\subsection{RLD Fisher information matrix from Belavkin--Staszewski relative
entropy}

The Belavkin--Staszewski relative entropy is defined for positive
definite states $\rho$ and $\sigma$ as follows \cite{Belavkin1982}:
\begin{equation}
\widehat{D}(\rho\|\sigma)\coloneqq\Tr\!\left[\rho\ln\!\left(\rho^{\frac{1}{2}}\sigma^{-1}\rho^{\frac{1}{2}}\right)\right].
\end{equation}
The data-processing inequality holds for $\widehat{D}$ \cite{Hiai1991},
and the Belavkin--Staszewski relative entropy is known to be the
following limit of the geometric R\'enyi relative entropy \cite{Matsumoto2013,Matsumoto2018}
(see also \cite[Proposition~79]{Katariya2021}):
\begin{equation}
\lim_{\alpha\to1}\widehat{D}_{\alpha}(\rho\|\sigma)=\widehat{D}(\rho\|\sigma).
\end{equation}

Let us define the Belavkin--Staszewski information matrix to be the
information matrix induced by this relative entropy. That is, for
a second-order differentiable parameterized family $\left(\rho(\theta)\right)_{\theta\in\mathbb{R}^{L}}$
of positive definite states, it is defined as
\begin{equation}
\left[\widehat{I}(\theta)\right]_{i,j}\coloneqq\left(\left.\frac{\partial^{2}}{\partial\varepsilon_{i}\partial\varepsilon_{j}}\widehat{D}(\rho(\theta)\|\rho(\theta+\varepsilon))\right|_{\varepsilon=0}\right).\label{eq:BS-fisher-def}
\end{equation}

\begin{thm}
\label{thm:belavski-RLD}Let $\left(\rho(\theta)\right)_{\theta\in\mathbb{R}^{L}}$
be a second-order differentiable parameterized family of positive
definite states. The Belavkin--Staszewski information matrix $\widehat{I}(\theta)$,
defined in \eqref{eq:BS-fisher-def}, is equal to the RLD information
matrix:
\begin{equation}
\widehat{I}(\theta)=I_{\RLD}(\theta).\label{eq:BSF-theorem-equality}
\end{equation}
\end{thm}

\begin{proof}
Employing the definition $\Delta_{\varepsilon}\equiv\rho(\theta+\varepsilon)-\rho(\theta)$
as in \eqref{eq:Delta-eps-op-def}, consider that
\begin{align}
 & \widehat{D}(\rho(\theta)\|\rho(\theta+\varepsilon))\nonumber \\
 & =\Tr\!\left[\rho(\theta)\ln\!\left(\rho(\theta)^{\frac{1}{2}}\rho(\theta+\varepsilon)^{-1}\rho(\theta)^{\frac{1}{2}}\right)\right]\\
 & =\Tr\!\left[\rho(\theta)\ln\!\left(\left(\rho(\theta)^{-\frac{1}{2}}\rho(\theta+\varepsilon)\rho(\theta)^{-\frac{1}{2}}\right)^{-1}\right)\right]\\
 & =-\Tr\!\left[\rho(\theta)\ln\!\left(\rho(\theta)^{-\frac{1}{2}}\rho(\theta+\varepsilon)\rho(\theta)^{-\frac{1}{2}}\right)\right]\\
 & =-\Tr\!\left[\rho(\theta)\ln\!\left(\rho(\theta)^{-\frac{1}{2}}\left[\Delta_{\varepsilon}+\rho(\theta)\right]\rho(\theta)^{-\frac{1}{2}}\right)\right]\\
 & =-\Tr\!\left[\rho(\theta)\ln\!\left(I+\rho(\theta)^{-\frac{1}{2}}\Delta_{\varepsilon}\rho(\theta)^{-\frac{1}{2}}\right)\right]\\
 & =-\Tr\!\left[\rho(\theta)\sum_{n=1}^{\infty}\frac{\left(-1\right)^{n+1}}{n}\left(\rho(\theta)^{-\frac{1}{2}}\Delta_{\varepsilon}\rho(\theta)^{-\frac{1}{2}}\right)^{n}\right]\label{eq:log-series-expansion}\\
 & =\Tr\!\left[\rho(\theta)\sum_{n=1}^{\infty}\frac{\left(-1\right)^{n}}{n}\left(\rho(\theta)^{-\frac{1}{2}}\Delta_{\varepsilon}\rho(\theta)^{-\frac{1}{2}}\right)^{n}\right]\\
 & =\sum_{n=1}^{\infty}\frac{\left(-1\right)^{n}}{n}\Tr\!\left[\rho(\theta)\left(\rho(\theta)^{-\frac{1}{2}}\Delta_{\varepsilon}\rho(\theta)^{-\frac{1}{2}}\right)^{n}\right]\\
 & =-\Tr\!\left[\rho(\theta)\left(\rho(\theta)^{-\frac{1}{2}}\Delta_{\varepsilon}\rho(\theta)^{-\frac{1}{2}}\right)\right]+\frac{1}{2}\Tr\!\left[\rho(\theta)\left(\rho(\theta)^{-\frac{1}{2}}\Delta_{\varepsilon}\rho(\theta)^{-\frac{1}{2}}\right)^{2}\right]\nonumber \\
 & \qquad+\sum_{n=3}^{\infty}\frac{\left(-1\right)^{n}}{n}\Tr\!\left[\rho(\theta)\left(\rho(\theta)^{-\frac{1}{2}}\Delta_{\varepsilon}\rho(\theta)^{-\frac{1}{2}}\right)^{n}\right]\\
 & =\frac{1}{2}\Tr\!\left[\Delta_{\varepsilon}\rho(\theta)^{-1}\Delta_{\varepsilon}\right]+\sum_{n=3}^{\infty}\frac{\left(-1\right)^{n}}{n}\Tr\!\left[\rho(\theta)\left(\rho(\theta)^{-\frac{1}{2}}\Delta_{\varepsilon}\rho(\theta)^{-\frac{1}{2}}\right)^{n}\right].
\end{align}
The equality in \eqref{eq:log-series-expansion} follows from the
series expansion $\ln(1+x)=\sum_{n=1}^{\infty}\frac{\left(-1\right)^{n+1}}{n}x^{n}$,
which converges for $\left|x\right|<1$. The fact that $\left\Vert \rho(\theta)^{-\frac{1}{2}}\Delta_{\varepsilon}\rho(\theta)^{-\frac{1}{2}}\right\Vert <1$
holds is justified by the same arguments given in \eqref{eq:Delta-eps-sandwich-small-1}--\eqref{eq:Delta-eps-sandwich-small-last}.
So then
\begin{align}
 & \left.\left(\frac{\partial^{2}}{\partial\varepsilon_{i}\partial\varepsilon_{j}}\widehat{D}(\rho(\theta)\|\rho(\theta+\varepsilon))\right)\right|_{\varepsilon=0}\nonumber \\
 & =\frac{1}{2}\left.\left(\frac{\partial^{2}}{\partial\varepsilon_{i}\partial\varepsilon_{j}}\Tr\!\left[\Delta_{\varepsilon}\rho(\theta)^{-1}\Delta_{\varepsilon}\right]\right)\right|_{\varepsilon=0}\nonumber \\
 & \qquad+\sum_{n=3}^{\infty}\frac{\left(-1\right)^{n}}{n}\left.\left(\frac{\partial^{2}}{\partial\varepsilon_{i}\partial\varepsilon_{j}}\Tr\!\left[\rho(\theta)\left(\rho(\theta)^{-\frac{1}{2}}\Delta_{\varepsilon}\rho(\theta)^{-\frac{1}{2}}\right)^{n}\right]\right)\right|_{\varepsilon=0}\label{eq:BS-conv-proof-almost-done}\\
 & =\Tr\!\left[\left\{ \frac{\partial}{\partial\theta_{j}}\rho(\theta),\frac{\partial}{\partial\theta_{i}}\rho(\theta)\right\} \rho(\theta)^{-1}\right].
\end{align}
The second term in \eqref{eq:BS-conv-proof-almost-done} vanishes
for the same reasons given in \eqref{eq:geo-renyi-higher-terms-vanish-1}--\eqref{eq:geo-renyi-higher-terms-vanish-last}.
The final equality follows from \eqref{eq:geo-to-RLD-final-proof-1}--\eqref{eq:GRF-proof-done}.
\end{proof}
\begin{rem}
In the special case of a single parameter, Theorem~\ref{thm:belavski-RLD}
was established by \cite[Section~6.4]{Matsumoto2013,Matsumoto2018}
(see also \cite[Proposition~53]{Katariya2021}).
\end{rem}

\section{$\alpha$-$z$ Information matrices from $\alpha$-$z$ R\'enyi
relative entropies}

\label{sec:a-z-information-matrices}The main result of this section
is Theorem~\ref{thm:fisher-info-from-alpha-z}, which provides a
formula for the information matrix resulting from the $\alpha$-$z$
R\'enyi relative entropies.

Recall that the $\alpha$-$z$ R\'enyi relative entropies are defined
for all $\alpha\in\left(0,1\right)\cup\left(1,\infty\right)$ and
$z>0$ and all positive definite states $\rho$ and $\sigma$ as follows
\cite{Audenaert2015}: 
\begin{align}
D_{\alpha,z}(\rho\|\sigma) & \coloneqq\frac{1}{\alpha-1}\ln\Tr\!\left[\left(\sigma^{\frac{1-\alpha}{2z}}\rho^{\frac{\alpha}{z}}\sigma^{\frac{1-\alpha}{2z}}\right)^{z}\right]\label{eq:alpha-z-renyi-def}\\
 & =\frac{1}{\alpha-1}\ln\Tr\!\left[\left(\rho^{\frac{\alpha}{2z}}\sigma^{\frac{1-\alpha}{z}}\rho^{\frac{\alpha}{2z}}\right)^{z}\right].
\end{align}

The data-processing inequality is the statement that the following
inequality holds for all states $\rho$ and $\sigma$ and every channel
$\mathcal{N}$:
\begin{equation}
D_{\alpha,z}(\rho\|\sigma)\geq D_{\alpha,z}(\mathcal{N}(\rho)\|\mathcal{N}(\sigma)).
\end{equation}
The full range of $\alpha,z>0$ for which the data-processing inequality
holds was established in \cite[Theorem~1.2]{Zhang2020} and is recalled
as Fact \ref{fact:a-z-data-proc} below:
\begin{fact}
\label{fact:a-z-data-proc}The data-processing inequality holds for
$D_{\alpha,z}$ if and only if one of the following conditions holds
\cite[Theorem~1.2]{Zhang2020}:
\begin{itemize}
\item $0<\alpha<1$ and $z\geq\max\left\{ \alpha,1-\alpha\right\} $,
\item $\alpha>1$ and $\alpha-1\leq z\leq\alpha\leq2z$.
\end{itemize}
\end{fact}

It should be stated that the data-processing inequality for $0<\alpha<1$
and $z\geq\max\left\{ \alpha,1-\alpha\right\} $ follows from \cite[Theorem~2.1]{Hiai2013},
as observed in \cite[Theorem~1]{Audenaert2015}.

The $\alpha$-$z$ information matrix is defined for a parameterized
family $\left(\rho(\theta)\right)_{\theta\in\mathbb{R}^{L}}$ of positive
definite states as follows:
\begin{equation}
\left[I_{\alpha,z}(\theta)\right]_{i,j}\coloneqq\frac{1}{\alpha}\left.\frac{\partial^{2}}{\partial\varepsilon_{i}\partial\varepsilon_{j}}D_{\alpha,z}(\rho(\theta)\|\rho(\theta+\varepsilon))\right|_{\varepsilon=0}.\label{eq:a-z-fisher-info-def}
\end{equation}

\begin{thm}
\label{thm:fisher-info-from-alpha-z}For a parameterized family $\left(\rho(\theta)\right)_{\theta\in\mathbb{R}^{L}}$
of second-order differentiable, positive definite states and for all
$\alpha\in\left(0,1\right)\cup\left(1,\infty\right)$ and $z>0$,
the following equality holds:
\begin{equation}
\left[I_{\alpha,z}(\theta)\right]_{i,j}=\sum_{k,\ell}\zeta_{\alpha,z}(\lambda_{k},\lambda_{\ell})\Tr\!\left[\Pi_{k}\left(\partial_{i}\rho(\theta)\right)\Pi_{\ell}\left(\partial_{j}\rho(\theta)\right)\right],\label{eq:a-z-fisher-formula}
\end{equation}
where $\partial_{i}\equiv\frac{\partial}{\partial\theta_{i}}$, the
spectral decomposition of $\rho(\theta)$ is given by $\rho(\theta)=\sum_{k}\lambda_{k}\Pi_{k}$,
and for all $x,y>0$,
\begin{equation}
\zeta_{\alpha,z}(x,y)\coloneqq\begin{cases}
\frac{z}{\alpha\left(1-\alpha\right)}\left(\frac{x^{\frac{1-\alpha}{z}}-y^{\frac{1-\alpha}{z}}}{x-y}\right)\left(\frac{x^{\frac{\alpha}{z}}-y^{\frac{\alpha}{z}}}{x^{\frac{1}{z}}-y^{\frac{1}{z}}}\right) & :x\neq y\\
\frac{1}{x} & :x=y
\end{cases}.\label{eq:limit-for-a-z-eigenval-func}
\end{equation}
\end{thm}

\begin{proof}
Our first reduction is to prove that
\begin{multline}
\left.\left(\frac{\partial^{2}}{\partial\varepsilon_{i}\partial\varepsilon_{j}}D_{\alpha,z}(\rho(\theta)\|\rho(\theta+\varepsilon))\right)\right|_{\varepsilon=0}=\\
\frac{1}{\alpha-1}\left.\left(\frac{\partial^{2}}{\partial\varepsilon_{i}\partial\varepsilon_{j}}\Tr\!\left[\left(\rho(\theta)^{\frac{\alpha}{2z}}\rho(\theta+\varepsilon)^{\frac{1-\alpha}{z}}\rho(\theta)^{\frac{\alpha}{2z}}\right)^{z}\right]\right)\right|_{\varepsilon=0}.\label{eq:first-reduction-a-z}
\end{multline}
To this end, consider that
\begin{align}
 & \frac{\partial^{2}}{\partial\varepsilon_{i}\partial\varepsilon_{j}}D_{\alpha,z}(\rho(\theta)\|\rho(\theta+\varepsilon))\nonumber \\
 & =\frac{1}{\alpha-1}\frac{\partial}{\partial\varepsilon_{i}}\left[\frac{\partial}{\partial\varepsilon_{j}}\ln\Tr\!\left[\left(\rho(\theta)^{\frac{\alpha}{2z}}\rho(\theta+\varepsilon)^{\frac{1-\alpha}{z}}\rho(\theta)^{\frac{\alpha}{2z}}\right)^{z}\right]\right]\\
 & =\frac{1}{\alpha-1}\frac{\partial}{\partial\varepsilon_{i}}\left[\frac{\frac{\partial}{\partial\varepsilon_{j}}\Tr\!\left[\left(\rho(\theta)^{\frac{\alpha}{2z}}\rho(\theta+\varepsilon)^{\frac{1-\alpha}{z}}\rho(\theta)^{\frac{\alpha}{2z}}\right)^{z}\right]}{\Tr\!\left[\left(\rho(\theta)^{\frac{\alpha}{2z}}\rho(\theta+\varepsilon)^{\frac{1-\alpha}{z}}\rho(\theta)^{\frac{\alpha}{2z}}\right)^{z}\right]}\right]\\
 & =-\frac{1}{\alpha-1}\frac{\left[\begin{array}{c}
\frac{\partial}{\partial\varepsilon_{i}}\Tr\!\left[\left(\rho(\theta)^{\frac{\alpha}{2z}}\rho(\theta+\varepsilon)^{\frac{1-\alpha}{z}}\rho(\theta)^{\frac{\alpha}{2z}}\right)^{z}\right]\times\\
\frac{\partial}{\partial\varepsilon_{j}}\Tr\!\left[\left(\rho(\theta)^{\frac{\alpha}{2z}}\rho(\theta+\varepsilon)^{\frac{1-\alpha}{z}}\rho(\theta)^{\frac{\alpha}{2z}}\right)^{z}\right]
\end{array}\right]}{\left(\Tr\!\left[\left(\rho(\theta)^{\frac{\alpha}{2z}}\rho(\theta+\varepsilon)^{\frac{1-\alpha}{z}}\rho(\theta)^{\frac{\alpha}{2z}}\right)^{z}\right]\right)^{2}}\nonumber \\
 & \qquad+\frac{1}{\alpha-1}\left[\frac{\frac{\partial^{2}}{\partial\varepsilon_{i}\partial\varepsilon_{j}}\Tr\!\left[\left(\rho(\theta)^{\frac{\alpha}{2z}}\rho(\theta+\varepsilon)^{\frac{1-\alpha}{z}}\rho(\theta)^{\frac{\alpha}{2z}}\right)^{z}\right]}{\Tr\!\left[\left(\rho(\theta)^{\frac{\alpha}{2z}}\rho(\theta+\varepsilon)^{\frac{1-\alpha}{z}}\rho(\theta)^{\frac{\alpha}{2z}}\right)^{z}\right]}\right].
\end{align}
Then it follows that
\begin{align}
 & \left.\frac{\partial^{2}}{\partial\varepsilon_{i}\partial\varepsilon_{j}}D_{\alpha,z}(\rho(\theta)\|\rho(\theta+\varepsilon))\right|_{\varepsilon=0}\nonumber \\
 & =-\frac{1}{\alpha-1}\frac{\left[\begin{array}{c}
\left.\frac{\partial}{\partial\varepsilon_{i}}\Tr\!\left[\left(\rho(\theta)^{\frac{\alpha}{2z}}\rho(\theta+\varepsilon)^{\frac{1-\alpha}{z}}\rho(\theta)^{\frac{\alpha}{2z}}\right)^{z}\right]\right|_{\varepsilon=0}\times\\
\left.\frac{\partial}{\partial\varepsilon_{j}}\Tr\!\left[\left(\rho(\theta)^{\frac{\alpha}{2z}}\rho(\theta+\varepsilon)^{\frac{1-\alpha}{z}}\rho(\theta)^{\frac{\alpha}{2z}}\right)^{z}\right]\right|_{\varepsilon=0}
\end{array}\right]}{\left(\Tr\!\left[\left(\rho(\theta)^{\frac{\alpha}{2z}}\rho(\theta)^{\frac{1-\alpha}{z}}\rho(\theta)^{\frac{\alpha}{2z}}\right)^{z}\right]\right)^{2}}\nonumber \\
 & \qquad+\frac{1}{\alpha-1}\left[\frac{\left.\frac{\partial^{2}}{\partial\varepsilon_{i}\partial\varepsilon_{j}}\Tr\!\left[\left(\rho(\theta)^{\frac{\alpha}{2z}}\rho(\theta+\varepsilon)^{\frac{1-\alpha}{z}}\rho(\theta)^{\frac{\alpha}{2z}}\right)^{z}\right]\right|_{\varepsilon=0}}{\Tr\!\left[\left(\rho(\theta)^{\frac{\alpha}{2z}}\rho(\theta)^{\frac{1-\alpha}{z}}\rho(\theta)^{\frac{\alpha}{2z}}\right)^{z}\right]}\right]\\
 & =-\frac{1}{\alpha-1}\left[\begin{array}{c}
\left.\frac{\partial}{\partial\varepsilon_{i}}\Tr\!\left[\left(\rho(\theta)^{\frac{\alpha}{2z}}\rho(\theta+\varepsilon)^{\frac{1-\alpha}{z}}\rho(\theta)^{\frac{\alpha}{2z}}\right)^{z}\right]\right|_{\varepsilon=0}\times\\
\left.\frac{\partial}{\partial\varepsilon_{j}}\Tr\!\left[\left(\rho(\theta)^{\frac{\alpha}{2z}}\rho(\theta+\varepsilon)^{\frac{1-\alpha}{z}}\rho(\theta)^{\frac{\alpha}{2z}}\right)^{z}\right]\right|_{\varepsilon=0}
\end{array}\right]\nonumber \\
 & \qquad+\frac{1}{\alpha-1}\left[\left.\frac{\partial^{2}}{\partial\varepsilon_{i}\partial\varepsilon_{j}}\Tr\!\left[\left(\rho(\theta)^{\frac{\alpha}{2z}}\rho(\theta+\varepsilon)^{\frac{1-\alpha}{z}}\rho(\theta)^{\frac{\alpha}{2z}}\right)^{z}\right]\right|_{\varepsilon=0}\right].
\end{align}
Lemma \ref{lem:first-deriv-zero} implies that
\begin{equation}
\left.\frac{\partial}{\partial\varepsilon_{i}}\Tr\!\left[\left(\rho(\theta)^{\frac{\alpha}{2z}}\rho(\theta+\varepsilon)^{\frac{1-\alpha}{z}}\rho(\theta)^{\frac{\alpha}{2z}}\right)^{z}\right]\right|_{\varepsilon=0}=0
\end{equation}
for all $i$, so that \eqref{eq:first-reduction-a-z} follows.

Now proceeding to analyze \eqref{eq:first-reduction-a-z}, consider
that
\begin{align}
 & \left.\frac{\partial^{2}}{\partial\varepsilon_{i}\partial\varepsilon_{j}}D_{\alpha,z}(\rho(\theta)\|\rho(\theta+\varepsilon))\right|_{\varepsilon=0}\nonumber \\
 & =\frac{1}{\alpha-1}\left[\left.\frac{\partial^{2}}{\partial\varepsilon_{i}\partial\varepsilon_{j}}\Tr\!\left[\left(\rho(\theta)^{\frac{\alpha}{2z}}\rho(\theta+\varepsilon)^{\frac{1-\alpha}{z}}\rho(\theta)^{\frac{\alpha}{2z}}\right)^{z}\right]\right|_{\varepsilon=0}\right]\\
 & =\frac{z}{\alpha-1}\left.\left(\frac{\partial}{\partial\varepsilon_{i}}\Tr\!\left[\begin{array}{c}
\left(\rho(\theta)^{\frac{\alpha}{2z}}\rho(\theta+\varepsilon)^{\frac{1-\alpha}{z}}\rho(\theta)^{\frac{\alpha}{2z}}\right)^{z-1}\times\\
\rho(\theta)^{\frac{\alpha}{2z}}\left(\frac{\partial}{\partial\varepsilon_{j}}\rho(\theta+\varepsilon)^{\frac{1-\alpha}{z}}\right)\rho(\theta)^{\frac{\alpha}{2z}}
\end{array}\right]\right)\right|_{\varepsilon=0}\\
 & =\frac{z}{\alpha-1}\left.\Tr\!\left[\begin{array}{c}
\left(\frac{\partial}{\partial\varepsilon_{i}}\left(\rho(\theta)^{\frac{\alpha}{2z}}\rho(\theta+\varepsilon)^{\frac{1-\alpha}{z}}\rho(\theta)^{\frac{\alpha}{2z}}\right)^{z-1}\right)\times\\
\rho(\theta)^{\frac{\alpha}{2z}}\left(\frac{\partial}{\partial\varepsilon_{j}}\rho(\theta+\varepsilon)^{\frac{1-\alpha}{z}}\right)\rho(\theta)^{\frac{\alpha}{2z}}
\end{array}\right]\right|_{\varepsilon=0}\nonumber \\
 & \qquad+\frac{z}{\alpha-1}\left.\Tr\!\left[\begin{array}{c}
\left(\left(\rho(\theta)^{\frac{\alpha}{2z}}\rho(\theta+\varepsilon)^{\frac{1-\alpha}{z}}\rho(\theta)^{\frac{\alpha}{2z}}\right)^{z-1}\right)\times\\
\rho(\theta)^{\frac{\alpha}{2z}}\left(\frac{\partial^{2}}{\partial\varepsilon_{i}\partial\varepsilon_{j}}\rho(\theta+\varepsilon)^{\frac{1-\alpha}{z}}\right)\rho(\theta)^{\frac{\alpha}{2z}}
\end{array}\right]\right|_{\varepsilon=0}\\
 & =\frac{z}{\alpha-1}\Tr\!\left[\begin{array}{c}
\left(\left.\frac{\partial}{\partial\varepsilon_{i}}\left(\rho(\theta)^{\frac{\alpha}{2z}}\rho(\theta+\varepsilon)^{\frac{1-\alpha}{z}}\rho(\theta)^{\frac{\alpha}{2z}}\right)^{z-1}\right|_{\varepsilon=0}\right)\times\\
\rho(\theta)^{\frac{\alpha}{2z}}\left(\frac{\partial}{\partial\theta_{j}}\rho(\theta)^{\frac{1-\alpha}{z}}\right)\rho(\theta)^{\frac{\alpha}{2z}}
\end{array}\right]\nonumber \\
 & \qquad+\frac{z}{\alpha-1}\Tr\!\left[\begin{array}{c}
\left(\left(\rho(\theta)^{\frac{\alpha}{2z}}\rho(\theta)^{\frac{1-\alpha}{z}}\rho(\theta)^{\frac{\alpha}{2z}}\right)^{z-1}\right)\times\\
\rho(\theta)^{\frac{\alpha}{2z}}\left(\frac{\partial^{2}}{\partial\theta_{i}\partial\theta_{j}}\rho(\theta)^{\frac{1-\alpha}{z}}\right)\rho(\theta)^{\frac{\alpha}{2z}}
\end{array}\right]\\
 & =\frac{z}{\alpha-1}\Tr\!\left[\begin{array}{c}
\left(\left.\frac{\partial}{\partial\varepsilon_{i}}\left(\rho(\theta)^{\frac{\alpha}{2z}}\rho(\theta+\varepsilon)^{\frac{1-\alpha}{z}}\rho(\theta)^{\frac{\alpha}{2z}}\right)^{z-1}\right|_{\varepsilon=0}\right)\times\\
\rho(\theta)^{\frac{\alpha}{2z}}\left(\frac{\partial}{\partial\theta_{j}}\rho(\theta)^{\frac{1-\alpha}{z}}\right)\rho(\theta)^{\frac{\alpha}{2z}}
\end{array}\right]\nonumber \\
 & \qquad+\frac{z}{\alpha-1}\Tr\!\left[\rho(\theta)^{\frac{z-1+\alpha}{z}}\left(\frac{\partial^{2}}{\partial\theta_{i}\partial\theta_{j}}\rho(\theta)^{\frac{1-\alpha}{z}}\right)\right].\label{eq:alpha-z-proof-progress}
\end{align}
For the penultimate equality, we applied the following:
\begin{align}
\left.\frac{\partial}{\partial\varepsilon_{j}}\rho(\theta+\varepsilon)^{\frac{1-\alpha}{z}}\right|_{\varepsilon=0} & =\frac{\partial}{\partial\theta_{j}}\rho(\theta)^{\frac{1-\alpha}{z}},\\
\left.\frac{\partial^{2}}{\partial\varepsilon_{i}\partial\varepsilon_{j}}\rho(\theta+\varepsilon)^{\frac{1-\alpha}{z}}\right|_{\varepsilon=0} & =\frac{\partial^{2}}{\partial\theta_{i}\partial\theta_{j}}\rho(\theta)^{\frac{1-\alpha}{z}}.
\end{align}

Let us handle each of the terms in \eqref{eq:alpha-z-proof-progress}
individually. Beginning with the first term in \eqref{eq:alpha-z-proof-progress}
and applying Proposition \ref{prop:power-function-all-powers} with
the substitutions $r\to z-1$, $\frac{\partial}{\partial x}\to\frac{\partial}{\partial\varepsilon_{i}}$,
and $A(x)\to\rho(\theta)^{\frac{\alpha}{2z}}\rho(\theta+\varepsilon)^{\frac{1-\alpha}{z}}\rho(\theta)^{\frac{\alpha}{2z}}$,
consider that
\begin{align}
 & \left.\frac{\partial}{\partial\varepsilon_{i}}\left(\rho(\theta)^{\frac{\alpha}{2z}}\rho(\theta+\varepsilon)^{\frac{1-\alpha}{z}}\rho(\theta)^{\frac{\alpha}{2z}}\right)^{z-1}\right|_{\varepsilon=0}\nonumber \\
 & =\left.\left(z-1\right)\int_{0}^{1}dt\int_{0}^{\infty}ds\ \left[\begin{array}{c}
\frac{\left(\rho(\theta)^{\frac{\alpha}{2z}}\rho(\theta+\varepsilon)^{\frac{1-\alpha}{z}}\rho(\theta)^{\frac{\alpha}{2z}}\right)^{\left(z-1\right)t}}{\rho(\theta)^{\frac{\alpha}{2z}}\rho(\theta+\varepsilon)^{\frac{1-\alpha}{z}}\rho(\theta)^{\frac{\alpha}{2z}}+sI}\times\\
\left(\frac{\partial}{\partial\varepsilon_{i}}\left(\rho(\theta)^{\frac{\alpha}{2z}}\rho(\theta+\varepsilon)^{\frac{1-\alpha}{z}}\rho(\theta)^{\frac{\alpha}{2z}}\right)\right)\times\\
\frac{\left(\rho(\theta)^{\frac{\alpha}{2z}}\rho(\theta+\varepsilon)^{\frac{1-\alpha}{z}}\rho(\theta)^{\frac{\alpha}{2z}}\right)^{\left(z-1\right)\left(1-t\right)}}{\rho(\theta)^{\frac{\alpha}{2z}}\rho(\theta+\varepsilon)^{\frac{1-\alpha}{z}}\rho(\theta)^{\frac{\alpha}{2z}}+sI}
\end{array}\right]\right|_{\varepsilon=0}\\
 & =\left(z-1\right)\int_{0}^{1}dt\int_{0}^{\infty}ds\ \left[\begin{array}{c}
\frac{\left(\rho(\theta)^{\frac{\alpha}{2z}}\rho(\theta)^{\frac{1-\alpha}{z}}\rho(\theta)^{\frac{\alpha}{2z}}\right)^{\left(z-1\right)t}}{\rho(\theta)^{\frac{\alpha}{2z}}\rho(\theta)^{\frac{1-\alpha}{z}}\rho(\theta)^{\frac{\alpha}{2z}}+sI}\times\\
\rho(\theta)^{\frac{\alpha}{2z}}\left.\left(\frac{\partial}{\partial\varepsilon_{i}}\left(\rho(\theta+\varepsilon)^{\frac{1-\alpha}{z}}\right)\right)\right|_{\varepsilon=0}\rho(\theta)^{\frac{\alpha}{2z}}\times\\
\frac{\left(\rho(\theta)^{\frac{\alpha}{2z}}\rho(\theta)^{\frac{1-\alpha}{z}}\rho(\theta)^{\frac{\alpha}{2z}}\right)^{\left(z-1\right)\left(1-t\right)}}{\rho(\theta)^{\frac{\alpha}{2z}}\rho(\theta)^{\frac{1-\alpha}{z}}\rho(\theta)^{\frac{\alpha}{2z}}+sI}
\end{array}\right]\\
 & =\left(z-1\right)\int_{0}^{1}dt\int_{0}^{\infty}ds\ \left[\begin{array}{c}
\frac{\rho(\theta)^{\left(\frac{z-1}{z}\right)t}}{\rho(\theta)^{\frac{1}{z}}+sI}\rho(\theta)^{\frac{\alpha}{2z}}\times\\
\left.\left(\frac{\partial}{\partial\varepsilon_{i}}\left(\rho(\theta+\varepsilon)^{\frac{1-\alpha}{z}}\right)\right)\right|_{\varepsilon=0}\times\\
\rho(\theta)^{\frac{\alpha}{2z}}\frac{\rho(\theta)^{\left(\frac{z-1}{z}\right)\left(1-t\right)}}{\rho(\theta)^{\frac{1}{z}}+sI}
\end{array}\right]\\
 & =\left(z-1\right)\int_{0}^{1}dt\int_{0}^{\infty}ds\ \left[\begin{array}{c}
\frac{\rho(\theta)^{\left(\frac{z-1}{z}\right)t}}{\rho(\theta)^{\frac{1}{z}}+sI}\rho(\theta)^{\frac{\alpha}{2z}}\left(\frac{\partial}{\partial\theta_{i}}\rho(\theta)^{\frac{1-\alpha}{z}}\right)\times\\
\rho(\theta)^{\frac{\alpha}{2z}}\frac{\rho(\theta)^{\left(\frac{z-1}{z}\right)\left(1-t\right)}}{\rho(\theta)^{\frac{1}{z}}+sI}
\end{array}\right].\label{eq:alpha-z-proof-progress-1st-term}
\end{align}
In the last line, we applied the identity
\begin{equation}
\left.\left(\frac{\partial}{\partial\varepsilon_{i}}\left(\rho(\theta+\varepsilon)^{\frac{1-\alpha}{z}}\right)\right)\right|_{\varepsilon=0}=\frac{\partial}{\partial\theta_{i}}\rho(\theta)^{\frac{1-\alpha}{z}}.
\end{equation}
Substituting \eqref{eq:alpha-z-proof-progress-1st-term} into the
first term of \eqref{eq:alpha-z-proof-progress}, we find that
\begin{align}
 & \Tr\!\left[\begin{array}{c}
\left.\frac{\partial}{\partial\varepsilon_{i}}\left(\rho(\theta)^{\frac{\alpha}{2z}}\rho(\theta+\varepsilon)^{\frac{1-\alpha}{z}}\rho(\theta)^{\frac{\alpha}{2z}}\right)^{z-1}\right|_{\varepsilon=0}\times\\
\rho(\theta)^{\frac{\alpha}{2z}}\left(\frac{\partial}{\partial\theta_{j}}\rho(\theta)^{\frac{1-\alpha}{z}}\right)\rho(\theta)^{\frac{\alpha}{2z}}
\end{array}\right]\nonumber \\
 & =\left(z-1\right)\int_{0}^{1}dt\int_{0}^{\infty}ds\ \Tr\!\left[\begin{array}{c}
\frac{\rho(\theta)^{\left(\frac{z-1}{z}\right)t}}{\rho(\theta)^{\frac{1}{z}}+sI}\rho(\theta)^{\frac{\alpha}{2z}}\left(\frac{\partial}{\partial\theta_{i}}\rho(\theta)^{\frac{1-\alpha}{z}}\right)\times\\
\rho(\theta)^{\frac{\alpha}{2z}}\frac{\rho(\theta)^{\left(\frac{z-1}{z}\right)\left(1-t\right)}}{\rho(\theta)^{\frac{1}{z}}+sI}\times\\
\rho(\theta)^{\frac{\alpha}{2z}}\left(\frac{\partial}{\partial\theta_{j}}\rho(\theta)^{\frac{1-\alpha}{z}}\right)\rho(\theta)^{\frac{\alpha}{2z}}
\end{array}\right]\\
 & =\left(z-1\right)\int_{0}^{1}dt\int_{0}^{\infty}ds\ \Tr\!\left[\begin{array}{c}
\frac{\rho(\theta)^{\left(\frac{z-1}{z}\right)t+\frac{\alpha}{z}}}{\rho(\theta)^{\frac{1}{z}}+sI}\left(\frac{\partial}{\partial\theta_{i}}\rho(\theta)^{\frac{1-\alpha}{z}}\right)\times\\
\frac{\rho(\theta)^{\left(\frac{z-1}{z}\right)\left(1-t\right)+\frac{\alpha}{z}}}{\rho(\theta)^{\frac{1}{z}}+sI}\left(\frac{\partial}{\partial\theta_{j}}\rho(\theta)^{\frac{1-\alpha}{z}}\right)
\end{array}\right].
\end{align}
Now substituting the spectral decomposition $\rho(\theta)=\sum_{k}\lambda_{k}\Pi_{k}$,
consider that
\begin{align}
 & \left(z-1\right)\int_{0}^{1}dt\:\int_{0}^{\infty}ds\:\Tr\!\left[\begin{array}{c}
\frac{\rho(\theta)^{\left(\frac{z-1}{z}\right)t+\frac{\alpha}{z}}}{\rho(\theta)^{\frac{1}{z}}+sI}\left(\frac{\partial}{\partial\theta{}_{i}}\rho(\theta)^{\frac{1-\alpha}{z}}\right)\times\\
\frac{\rho(\theta)^{\left(\frac{z-1}{z}\right)\left(1-t\right)+\frac{\alpha}{z}}}{\rho(\theta)^{\frac{1}{z}}+sI}\left(\frac{\partial}{\partial\theta_{j}}\rho(\theta)^{\frac{1-\alpha}{z}}\right)
\end{array}\right]\nonumber \\
 & =\left(z-1\right)\int_{0}^{1}dt\:\int_{0}^{\infty}ds\:\Tr\!\left[\begin{array}{c}
\sum_{k}\frac{\lambda_{k}^{\left(\frac{z-1}{z}\right)t+\frac{\alpha}{z}}}{\lambda_{k}^{\frac{1}{z}}+s}\Pi_{k}\left(\frac{\partial}{\partial\theta{}_{i}}\rho(\theta)^{\frac{1-\alpha}{z}}\right)\times\\
\sum_{\ell}\frac{\lambda_{\ell}^{\left(\frac{z-1}{z}\right)\left(1-t\right)+\frac{\alpha}{z}}}{\lambda_{\ell}^{\frac{1}{z}}+s}\Pi_{\ell}\left(\frac{\partial}{\partial\theta_{j}}\rho(\theta)^{\frac{1-\alpha}{z}}\right)
\end{array}\right]\\
 & =\left(z-1\right)\sum_{k,\ell}\int_{0}^{1}dt\:\int_{0}^{\infty}ds\:\left(\frac{\lambda_{k}^{\left(\frac{z-1}{z}\right)t+\frac{\alpha}{z}}}{\lambda_{k}^{\frac{1}{z}}+s}\right)\left(\frac{\lambda_{\ell}^{\left(\frac{z-1}{z}\right)\left(1-t\right)+\frac{\alpha}{z}}}{\lambda_{\ell}^{\frac{1}{z}}+s}\right)\times\nonumber \\
 & \qquad\Tr\!\left[\Pi_{k}\left(\frac{\partial}{\partial\theta{}_{i}}\rho(\theta)^{\frac{1-\alpha}{z}}\right)\Pi_{\ell}\left(\frac{\partial}{\partial\theta_{j}}\rho(\theta)^{\frac{1-\alpha}{z}}\right)\right]\\
 & =\left(z-1\right)\sum_{k,\ell}\left(\lambda_{k}\lambda_{\ell}\right)^{\frac{\alpha}{z}}\int_{0}^{1}dt\:\lambda_{k}^{\left(\frac{z-1}{z}\right)t}\lambda_{\ell}^{\left(\frac{z-1}{z}\right)\left(1-t\right)}\int_{0}^{\infty}ds\:\left(\frac{1}{\lambda_{k}^{\frac{1}{z}}+s}\right)\left(\frac{1}{\lambda_{\ell}^{\frac{1}{z}}+s}\right)\times\nonumber \\
 & \qquad\Tr\!\left[\Pi_{k}\left(\frac{\partial}{\partial\theta{}_{i}}\rho(\theta)^{\frac{1-\alpha}{z}}\right)\Pi_{\ell}\left(\frac{\partial}{\partial\theta_{j}}\rho(\theta)^{\frac{1-\alpha}{z}}\right)\right]\\
 & =\left(z-1\right)\sum_{k,\ell}\left(\lambda_{k}\lambda_{\ell}\right)^{\frac{\alpha}{z}}\frac{1}{\frac{z-1}{z}}\left(\frac{\lambda_{k}^{\frac{z-1}{z}}-\lambda_{\ell}^{\frac{z-1}{z}}}{\ln\lambda_{k}-\ln\lambda_{\ell}}\right)\left(\frac{\ln\lambda_{k}^{\frac{1}{z}}-\ln\lambda_{\ell}^{\frac{1}{z}}}{\lambda_{k}^{\frac{1}{z}}-\lambda_{\ell}^{\frac{1}{z}}}\right)\times\nonumber \\
 & \qquad\Tr\!\left[\Pi_{k}\left(\frac{\partial}{\partial\theta{}_{i}}\rho(\theta)^{\frac{1-\alpha}{z}}\right)\Pi_{\ell}\left(\frac{\partial}{\partial\theta_{j}}\rho(\theta)^{\frac{1-\alpha}{z}}\right)\right]\\
 & =\sum_{k,\ell}\left(\lambda_{k}\lambda_{\ell}\right)^{\frac{\alpha}{z}}\left(\frac{\lambda_{k}^{\frac{z-1}{z}}-\lambda_{\ell}^{\frac{z-1}{z}}}{\lambda_{k}^{\frac{1}{z}}-\lambda_{\ell}^{\frac{1}{z}}}\right)\Tr\!\left[\Pi_{k}\left(\frac{\partial}{\partial\theta{}_{i}}\rho(\theta)^{\frac{1-\alpha}{z}}\right)\Pi_{\ell}\left(\frac{\partial}{\partial\theta_{j}}\rho(\theta)^{\frac{1-\alpha}{z}}\right)\right].
\end{align}
For the penultimate equality, we applied the following integrals,
which hold for $x,y>0$:
\begin{align}
\int_{0}^{1}dt\:x^{\left(\frac{z-1}{z}\right)t}y^{\left(\frac{z-1}{z}\right)\left(1-t\right)} & =\frac{1}{\frac{z-1}{z}}\left(\frac{x^{\frac{z-1}{z}}-y^{\frac{z-1}{z}}}{\ln x-\ln y}\right),\\
\int_{0}^{\infty}ds\:\left(\frac{1}{x^{\frac{1}{z}}+s}\right)\left(\frac{1}{y^{\frac{1}{z}}+s}\right) & =\frac{\ln x^{\frac{1}{z}}-\ln y^{\frac{1}{z}}}{x^{\frac{1}{z}}-y^{\frac{1}{z}}},
\end{align}
and we have left it implicit above that one evaluates these expressions
in the limit $x\to y$ when $x=y$. Now applying Theorem~\ref{thm:divided-difference-matrix-deriv}
to evaluate $\frac{\partial}{\partial\theta{}_{i}}\rho(\theta)^{\frac{1-\alpha}{z}}$
and $\frac{\partial}{\partial\theta_{j}}\rho(\theta)^{\frac{1-\alpha}{z}}$,
consider that 
\begin{align}
 & \sum_{k,\ell}\left(\lambda_{k}\lambda_{\ell}\right)^{\frac{\alpha}{z}}\left(\frac{\lambda_{k}^{\frac{z-1}{z}}-\lambda_{\ell}^{\frac{z-1}{z}}}{\lambda_{k}^{\frac{1}{z}}-\lambda_{\ell}^{\frac{1}{z}}}\right)\Tr\!\left[\Pi_{k}\left(\frac{\partial}{\partial\theta{}_{i}}\rho(\theta)^{\frac{1-\alpha}{z}}\right)\Pi_{\ell}\left(\frac{\partial}{\partial\theta_{j}}\rho(\theta)^{\frac{1-\alpha}{z}}\right)\right]\nonumber \\
 & =\sum_{k,\ell}\left(\lambda_{k}\lambda_{\ell}\right)^{\frac{\alpha}{z}}\left(\frac{\lambda_{k}^{\frac{z-1}{z}}-\lambda_{\ell}^{\frac{z-1}{z}}}{\lambda_{k}^{\frac{1}{z}}-\lambda_{\ell}^{\frac{1}{z}}}\right)\times\nonumber \\
 & \qquad\Tr\!\left[\begin{array}{c}
\Pi_{k}\left(\sum_{m,n}\left(\frac{\lambda_{m}^{\frac{1-\alpha}{z}}-\lambda_{n}^{\frac{1-\alpha}{z}}}{\lambda_{m}-\lambda_{n}}\right)\Pi_{m}\left(\frac{\partial}{\partial\theta_{i}}\rho(\theta)\right)\Pi_{n}\right)\times\\
\Pi_{\ell}\left(\sum_{p,r}\left(\frac{\lambda_{p}^{\frac{1-\alpha}{z}}-\lambda_{r}^{\frac{1-\alpha}{z}}}{\lambda_{p}-\lambda_{r}}\right)\Pi_{p}\left(\frac{\partial}{\partial\theta_{j}}\rho(\theta)\right)\Pi_{r}\right)
\end{array}\right]\\
 & =\sum_{k,\ell,m,n,p,r}\left(\lambda_{k}\lambda_{\ell}\right)^{\frac{\alpha}{z}}\left(\frac{\lambda_{k}^{\frac{z-1}{z}}-\lambda_{\ell}^{\frac{z-1}{z}}}{\lambda_{k}^{\frac{1}{z}}-\lambda_{\ell}^{\frac{1}{z}}}\right)\left(\frac{\lambda_{m}^{\frac{1-\alpha}{z}}-\lambda_{n}^{\frac{1-\alpha}{z}}}{\lambda_{m}-\lambda_{n}}\right)\times\nonumber \\
 & \qquad\left(\frac{\lambda_{p}^{\frac{1-\alpha}{z}}-\lambda_{r}^{\frac{1-\alpha}{z}}}{\lambda_{p}-\lambda_{r}}\right)\Tr\!\left[\Pi_{k}\Pi_{m}\left(\frac{\partial}{\partial\theta_{i}}\rho(\theta)\right)\Pi_{n}\Pi_{\ell}\Pi_{p}\left(\frac{\partial}{\partial\theta_{j}}\rho(\theta)\right)\Pi_{r}\right]\\
 & =\sum_{k,\ell}\left(\lambda_{k}\lambda_{\ell}\right)^{\frac{\alpha}{z}}\left(\frac{\lambda_{k}^{\frac{z-1}{z}}-\lambda_{\ell}^{\frac{z-1}{z}}}{\lambda_{k}^{\frac{1}{z}}-\lambda_{\ell}^{\frac{1}{z}}}\right)\left(\frac{\lambda_{k}^{\frac{1-\alpha}{z}}-\lambda_{\ell}^{\frac{1-\alpha}{z}}}{\lambda_{k}-\lambda_{\ell}}\right)\left(\frac{\lambda_{\ell}^{\frac{1-\alpha}{z}}-\lambda_{k}^{\frac{1-\alpha}{z}}}{\lambda_{\ell}-\lambda_{k}}\right)\times\nonumber \\
 & \qquad\Tr\!\left[\Pi_{k}\left(\frac{\partial}{\partial\theta_{i}}\rho(\theta)\right)\Pi_{\ell}\left(\frac{\partial}{\partial\theta_{j}}\rho(\theta)\right)\right]\\
 & =\sum_{k,\ell}\left(\lambda_{k}\lambda_{\ell}\right)^{\frac{\alpha}{z}}\left(\frac{\lambda_{k}^{\frac{z-1}{z}}-\lambda_{\ell}^{\frac{z-1}{z}}}{\lambda_{k}^{\frac{1}{z}}-\lambda_{\ell}^{\frac{1}{z}}}\right)\left(\frac{\lambda_{k}^{\frac{1-\alpha}{z}}-\lambda_{\ell}^{\frac{1-\alpha}{z}}}{\lambda_{k}-\lambda_{\ell}}\right)^{2}\times\nonumber \\
 & \qquad\Tr\!\left[\Pi_{k}\left(\frac{\partial}{\partial\theta_{i}}\rho(\theta)\right)\Pi_{\ell}\left(\frac{\partial}{\partial\theta_{j}}\rho(\theta)\right)\right].\label{eq:first-term-a-z-simplified}
\end{align}

Let us now simplify the second term in \eqref{eq:alpha-z-proof-progress}.
Before doing so, consider that
\begin{multline}
\frac{\partial}{\partial\theta_{i}}\Tr\!\left[\rho(\theta)^{\frac{z-1+\alpha}{z}}\left(\frac{\partial}{\partial\theta_{j}}\rho(\theta)^{\frac{1-\alpha}{z}}\right)\right]=\Tr\!\left[\left(\frac{\partial}{\partial\theta_{i}}\rho(\theta)^{\frac{z-1+\alpha}{z}}\right)\left(\frac{\partial}{\partial\theta_{j}}\rho(\theta)^{\frac{1-\alpha}{z}}\right)\right]\\
+\Tr\!\left[\rho(\theta)^{\frac{z-1+\alpha}{z}}\left(\frac{\partial^{2}}{\partial\theta_{i}\partial\theta_{j}}\rho(\theta)^{\frac{1-\alpha}{z}}\right)\right].\label{eq:no-2nd-deriv-trick-1}
\end{multline}
Now observe that
\begin{align}
 & \Tr\!\left[\rho(\theta)^{\frac{z-1+\alpha}{z}}\left(\frac{\partial}{\partial\theta_{j}}\rho(\theta)^{\frac{1-\alpha}{z}}\right)\right]\nonumber \\
 & =\Tr\!\left[\sum_{k}\lambda_{k}^{\frac{z-1+\alpha}{z}}\Pi_{k}\left(\sum_{m,\ell}\left(\frac{\lambda_{m}^{\frac{1-\alpha}{z}}-\lambda_{\ell}^{\frac{1-\alpha}{z}}}{\lambda_{m}-\lambda_{\ell}}\right)\Pi_{m}\left(\frac{\partial}{\partial\theta_{j}}\rho(\theta)\right)\Pi_{\ell}\right)\right]\label{eq:no-2nd-deriv-trick-2}\\
 & =\sum_{k,m,\ell}\lambda_{k}^{\frac{z-1+\alpha}{z}}\left(\frac{\lambda_{m}^{\frac{1-\alpha}{z}}-\lambda_{\ell}^{\frac{1-\alpha}{z}}}{\lambda_{m}-\lambda_{\ell}}\right)\Tr\!\left[\Pi_{k}\Pi_{m}\left(\frac{\partial}{\partial\theta_{j}}\rho(\theta)\right)\Pi_{\ell}\right]\\
 & =\sum_{k}\lambda_{k}^{\frac{z-1+\alpha}{z}}\left(\frac{1-\alpha}{z}\right)\lambda_{k}^{\frac{1-\alpha}{z}-1}\Tr\!\left[\Pi_{k}\left(\frac{\partial}{\partial\theta_{j}}\rho(\theta)\right)\right]\\
 & =\left(\frac{1-\alpha}{z}\right)\sum_{k}\Tr\!\left[\Pi_{k}\left(\frac{\partial}{\partial\theta_{j}}\rho(\theta)\right)\right]\\
 & =\left(\frac{1-\alpha}{z}\right)\Tr\!\left[\frac{\partial}{\partial\theta_{j}}\rho(\theta)\right]\\
 & =\left(\frac{1-\alpha}{z}\right)\frac{\partial}{\partial\theta_{j}}\Tr\!\left[\rho(\theta)\right]\\
 & =0.\label{eq:no-2nd-deriv-trick-3}
\end{align}
By combining \eqref{eq:no-2nd-deriv-trick-1} and \eqref{eq:no-2nd-deriv-trick-2}--\eqref{eq:no-2nd-deriv-trick-3},
we thus conclude that
\begin{equation}
\Tr\!\left[\rho(\theta)^{\frac{z-1+\alpha}{z}}\left(\frac{\partial^{2}}{\partial\theta_{i}\partial\theta_{j}}\rho(\theta)^{\frac{1-\alpha}{z}}\right)\right]=-\Tr\!\left[\left(\frac{\partial}{\partial\theta_{i}}\rho(\theta)^{\frac{z-1+\alpha}{z}}\right)\left(\frac{\partial}{\partial\theta_{j}}\rho(\theta)^{\frac{1-\alpha}{z}}\right)\right].\label{eq:no-2nd-deriv}
\end{equation}
Then
\begin{align}
 & \Tr\!\left[\rho(\theta)^{\frac{z-1+\alpha}{z}}\left(\frac{\partial^{2}}{\partial\theta_{i}\partial\theta_{j}}\rho(\theta)^{\frac{1-\alpha}{z}}\right)\right]\nonumber \\
 & =-\Tr\!\left[\left(\frac{\partial}{\partial\theta_{i}}\rho(\theta)^{\frac{z-1+\alpha}{z}}\right)\left(\frac{\partial}{\partial\theta_{j}}\rho(\theta)^{\frac{1-\alpha}{z}}\right)\right]\\
 & =-\Tr\!\left[\begin{array}{c}
\left(\sum_{k,\ell}\left(\frac{\lambda_{k}^{\frac{z-1+\alpha}{z}}-\lambda_{\ell}^{\frac{z-1+\alpha}{z}}}{\lambda_{k}-\lambda_{\ell}}\right)\Pi_{k}\left(\frac{\partial}{\partial\theta_{i}}\rho(\theta)\right)\Pi_{\ell}\right)\times\\
\left(\sum_{m,n}\left(\frac{\lambda_{m}^{\frac{1-\alpha}{z}}-\lambda_{n}^{\frac{1-\alpha}{z}}}{\lambda_{m}-\lambda_{n}}\right)\Pi_{m}\left(\frac{\partial}{\partial\theta_{j}}\rho(\theta)\right)\Pi_{n}\right)
\end{array}\right]\\
 & =-\sum_{k,\ell,m,n}\left(\frac{\lambda_{k}^{\frac{z-1+\alpha}{z}}-\lambda_{\ell}^{\frac{z-1+\alpha}{z}}}{\lambda_{k}-\lambda_{\ell}}\right)\left(\frac{\lambda_{m}^{\frac{1-\alpha}{z}}-\lambda_{n}^{\frac{1-\alpha}{z}}}{\lambda_{m}-\lambda_{n}}\right)\times\nonumber \\
 & \qquad\qquad\Tr\!\left[\Pi_{k}\left(\frac{\partial}{\partial\theta_{i}}\rho(\theta)\right)\Pi_{\ell}\Pi_{m}\left(\frac{\partial}{\partial\theta_{j}}\rho(\theta)\right)\Pi_{n}\right]\\
 & =-\sum_{k,\ell}\left(\frac{\lambda_{k}^{\frac{z-1+\alpha}{z}}-\lambda_{\ell}^{\frac{z-1+\alpha}{z}}}{\lambda_{k}-\lambda_{\ell}}\right)\left(\frac{\lambda_{\ell}^{\frac{1-\alpha}{z}}-\lambda_{k}^{\frac{1-\alpha}{z}}}{\lambda_{\ell}-\lambda_{k}}\right)\Tr\!\left[\Pi_{k}\left(\frac{\partial}{\partial\theta_{i}}\rho(\theta)\right)\Pi_{\ell}\left(\frac{\partial}{\partial\theta_{j}}\rho(\theta)\right)\right]\\
 & =-\sum_{k,\ell}\left(\frac{\lambda_{k}^{\frac{z-1+\alpha}{z}}-\lambda_{\ell}^{\frac{z-1+\alpha}{z}}}{\lambda_{k}-\lambda_{\ell}}\right)\left(\frac{\lambda_{k}^{\frac{1-\alpha}{z}}-\lambda_{\ell}^{\frac{1-\alpha}{z}}}{\lambda_{k}-\lambda_{\ell}}\right)\Tr\!\left[\Pi_{k}\left(\frac{\partial}{\partial\theta_{i}}\rho(\theta)\right)\Pi_{\ell}\left(\frac{\partial}{\partial\theta_{j}}\rho(\theta)\right)\right].\label{eq:second-term-a-z-simplified}
\end{align}
Combining both terms in \eqref{eq:alpha-z-proof-progress}, using
\eqref{eq:first-term-a-z-simplified} and \eqref{eq:second-term-a-z-simplified}
(while omitting the prefactor $\frac{z}{\alpha-1}$ for now), we conclude
that
\begin{align}
 & \Tr\!\left[\begin{array}{c}
\left(\left.\frac{\partial}{\partial\varepsilon_{i}}\left(\rho(\theta)^{\frac{\alpha}{2z}}\rho(\theta+\varepsilon)^{\frac{1-\alpha}{z}}\rho(\theta)^{\frac{\alpha}{2z}}\right)^{z-1}\right|_{\varepsilon=0}\right)\times\\
\rho(\theta)^{\frac{\alpha}{2z}}\left(\frac{\partial}{\partial\theta_{j}}\rho(\theta)^{\frac{1-\alpha}{z}}\right)\rho(\theta)^{\frac{\alpha}{2z}}
\end{array}\right]+\Tr\!\left[\rho(\theta)^{\frac{z-1+\alpha}{z}}\left(\frac{\partial^{2}}{\partial\theta_{i}\partial\theta_{j}}\rho(\theta)^{\frac{1-\alpha}{z}}\right)\right]\nonumber \\
 & =\sum_{k,\ell}\left(\lambda_{k}\lambda_{\ell}\right)^{\frac{\alpha}{z}}\left(\frac{\lambda_{k}^{\frac{z-1}{z}}-\lambda_{\ell}^{\frac{z-1}{z}}}{\lambda_{k}^{\frac{1}{z}}-\lambda_{\ell}^{\frac{1}{z}}}\right)\left(\frac{\lambda_{k}^{\frac{1-\alpha}{z}}-\lambda_{\ell}^{\frac{1-\alpha}{z}}}{\lambda_{k}-\lambda_{\ell}}\right)^{2}\Tr\!\left[\Pi_{k}\left(\frac{\partial}{\partial\theta_{i}}\rho(\theta)\right)\Pi_{\ell}\left(\frac{\partial}{\partial\theta_{j}}\rho(\theta)\right)\right]\nonumber \\
 & \qquad-\sum_{k,\ell}\left(\frac{\lambda_{k}^{\frac{z-1+\alpha}{z}}-\lambda_{\ell}^{\frac{z-1+\alpha}{z}}}{\lambda_{k}-\lambda_{\ell}}\right)\left(\frac{\lambda_{k}^{\frac{1-\alpha}{z}}-\lambda_{\ell}^{\frac{1-\alpha}{z}}}{\lambda_{k}-\lambda_{\ell}}\right)\Tr\!\left[\Pi_{k}\left(\frac{\partial}{\partial\theta_{i}}\rho(\theta)\right)\Pi_{\ell}\left(\frac{\partial}{\partial\theta_{j}}\rho(\theta)\right)\right]\\
 & =\sum_{k,\ell}\left[\begin{array}{c}
\left(\lambda_{k}\lambda_{\ell}\right)^{\frac{\alpha}{z}}\left(\frac{\lambda_{k}^{\frac{z-1}{z}}-\lambda_{\ell}^{\frac{z-1}{z}}}{\lambda_{k}^{\frac{1}{z}}-\lambda_{\ell}^{\frac{1}{z}}}\right)\left(\frac{\lambda_{k}^{\frac{1-\alpha}{z}}-\lambda_{\ell}^{\frac{1-\alpha}{z}}}{\lambda_{k}-\lambda_{\ell}}\right)^{2}\\
-\left(\frac{\lambda_{k}^{\frac{z-1+\alpha}{z}}-\lambda_{\ell}^{\frac{z-1+\alpha}{z}}}{\lambda_{k}-\lambda_{\ell}}\right)\left(\frac{\lambda_{k}^{\frac{1-\alpha}{z}}-\lambda_{\ell}^{\frac{1-\alpha}{z}}}{\lambda_{k}-\lambda_{\ell}}\right)
\end{array}\right]\Tr\!\left[\Pi_{k}\left(\frac{\partial}{\partial\theta_{i}}\rho(\theta)\right)\Pi_{\ell}\left(\frac{\partial}{\partial\theta_{j}}\rho(\theta)\right)\right].\label{eq:alpha-z-proof-almost-done-almost}
\end{align}
Appendix \ref{app:proof-algebra-a-z-eigenvals-simplify} provides
a long sequence of algebraic steps proving that
\begin{multline}
\left(\lambda_{k}\lambda_{\ell}\right)^{\frac{\alpha}{z}}\left(\frac{\lambda_{k}^{\frac{z-1}{z}}-\lambda_{\ell}^{\frac{z-1}{z}}}{\lambda_{k}^{\frac{1}{z}}-\lambda_{\ell}^{\frac{1}{z}}}\right)\left(\frac{\lambda_{k}^{\frac{1-\alpha}{z}}-\lambda_{\ell}^{\frac{1-\alpha}{z}}}{\lambda_{k}-\lambda_{\ell}}\right)^{2}-\left(\frac{\lambda_{k}^{\frac{z-1+\alpha}{z}}-\lambda_{\ell}^{\frac{z-1+\alpha}{z}}}{\lambda_{k}-\lambda_{\ell}}\right)\left(\frac{\lambda_{k}^{\frac{1-\alpha}{z}}-\lambda_{\ell}^{\frac{1-\alpha}{z}}}{\lambda_{k}-\lambda_{\ell}}\right)\\
=-\left(\frac{\lambda_{k}^{\frac{1-\alpha}{z}}-\lambda_{\ell}^{\frac{1-\alpha}{z}}}{\lambda_{k}-\lambda_{\ell}}\right)\left(\frac{\lambda_{k}^{\frac{\alpha}{z}}-\lambda_{\ell}^{\frac{\alpha}{z}}}{\lambda_{k}^{\frac{1}{z}}-\lambda_{\ell}^{\frac{1}{z}}}\right).\label{eq:a-z-eigenvals-simplified}
\end{multline}
After combining with \eqref{eq:alpha-z-proof-almost-done-almost}
and incorporating the prefactor $\frac{z}{\alpha-1}$, the proof of
\eqref{eq:a-z-fisher-formula} is concluded. For a proof of the case
$x=y$ in \eqref{eq:limit-for-a-z-eigenval-func}, see Appendix \ref{app:limit-for-a-z-eigenval-func}.
\end{proof}
\begin{lem}
\label{lem:first-deriv-zero}The following equality holds for all
$\alpha,z>0$:
\begin{equation}
\left.\frac{\partial}{\partial\varepsilon_{j}}\Tr\!\left[\left(\rho(\theta)^{\frac{\alpha}{2z}}\rho(\theta+\varepsilon)^{\frac{1-\alpha}{z}}\rho(\theta)^{\frac{\alpha}{2z}}\right)^{z}\right]\right|_{\varepsilon=0}=0.
\end{equation}
\end{lem}

\begin{proof}
If $1-\alpha=0$, then
\begin{align}
\left.\frac{\partial}{\partial\varepsilon_{j}}\Tr\!\left[\left(\rho(\theta)^{\frac{\alpha}{2z}}\rho(\theta+\varepsilon)^{\frac{1-\alpha}{z}}\rho(\theta)^{\frac{\alpha}{2z}}\right)^{z}\right]\right|_{\varepsilon=0} & =\left.\frac{\partial}{\partial\varepsilon_{j}}\Tr\!\left[\left(\rho(\theta)^{\frac{\alpha}{2z}}\rho(\theta)^{\frac{\alpha}{2z}}\right)^{z}\right]\right|_{\varepsilon=0}\\
 & =0.
\end{align}
So suppose that $1-\alpha\neq0$, and consider that
\begin{align}
 & \frac{\partial}{\partial\varepsilon_{j}}\Tr\!\left[\left(\rho(\theta)^{\frac{\alpha}{2z}}\rho(\theta+\varepsilon)^{\frac{1-\alpha}{z}}\rho(\theta)^{\frac{\alpha}{2z}}\right)^{z}\right]\nonumber \\
 & =z\Tr\!\left[\left(\rho(\theta)^{\frac{\alpha}{2z}}\rho(\theta+\varepsilon)^{\frac{1-\alpha}{z}}\rho(\theta)^{\frac{\alpha}{2z}}\right)^{z-1}\frac{\partial}{\partial\varepsilon_{j}}\left(\rho(\theta)^{\frac{\alpha}{2z}}\rho(\theta+\varepsilon)^{\frac{1-\alpha}{z}}\rho(\theta)^{\frac{\alpha}{2z}}\right)\right]\\
 & =z\Tr\!\left[\left(\rho(\theta)^{\frac{\alpha}{2z}}\rho(\theta+\varepsilon)^{\frac{1-\alpha}{z}}\rho(\theta)^{\frac{\alpha}{2z}}\right)^{z-1}\rho(\theta)^{\frac{\alpha}{2z}}\left(\frac{\partial}{\partial\varepsilon_{j}}\rho(\theta+\varepsilon)^{\frac{1-\alpha}{z}}\right)\rho(\theta)^{\frac{\alpha}{2z}}\right]\\
 & =z\Tr\!\left[\rho(\theta)^{\frac{\alpha}{2z}}\left(\rho(\theta)^{\frac{\alpha}{2z}}\rho(\theta+\varepsilon)^{\frac{1-\alpha}{z}}\rho(\theta)^{\frac{\alpha}{2z}}\right)^{z-1}\rho(\theta)^{\frac{\alpha}{2z}}\left(\frac{\partial}{\partial\varepsilon_{j}}\rho(\theta+\varepsilon)^{\frac{1-\alpha}{z}}\right)\right],
\end{align}
where the first equality follows from Corollary~\ref{cor:derivative-in-trace}.
Now applying Proposition \ref{prop:deriv-exp}, consider that
\begin{align}
 & \frac{\partial}{\partial\varepsilon_{j}}\rho(\theta+\varepsilon)^{\frac{1-\alpha}{z}}\nonumber \\
 & =\frac{\partial}{\partial\varepsilon_{j}}e^{\left(\frac{1-\alpha}{z}\right)\ln\rho(\theta+\varepsilon)}\\
 & =\int_{0}^{1}dt\,e^{\left(\frac{1-\alpha}{z}\right)t\ln\rho(\theta+\varepsilon)}\left(\frac{\partial}{\partial\varepsilon_{j}}\left(\left(\frac{1-\alpha}{z}\right)\ln\rho(\theta+\varepsilon)\right)\right)e^{\left(\frac{1-\alpha}{z}\right)\left(1-t\right)\ln\rho(\theta+\varepsilon)}\\
 & =\frac{1-\alpha}{z}\int_{0}^{1}dt\,\rho(\theta+\varepsilon)^{\left(\frac{1-\alpha}{z}\right)t}\left(\frac{\partial}{\partial\varepsilon_{j}}\ln\rho(\theta+\varepsilon)\right)\rho(\theta+\varepsilon)^{\left(\frac{1-\alpha}{z}\right)\left(1-t\right)}.
\end{align}
It then follows that
\begin{align}
 & \frac{\partial}{\partial\varepsilon_{j}}\Tr\!\left[\left(\rho(\theta)^{\frac{\alpha}{2z}}\rho(\theta+\varepsilon)^{\frac{1-\alpha}{z}}\rho(\theta)^{\frac{\alpha}{2z}}\right)^{z}\right]\nonumber \\
 & =\left(1-\alpha\right)\int_{0}^{1}dt\,\Tr\!\left[\begin{array}{c}
\rho(\theta)^{\frac{\alpha}{2z}}\left(\rho(\theta)^{\frac{\alpha}{2z}}\rho(\theta+\varepsilon)^{\frac{1-\alpha}{z}}\rho(\theta)^{\frac{\alpha}{2z}}\right)^{z-1}\rho(\theta)^{\frac{\alpha}{2z}}\times\\
\rho(\theta+\varepsilon)^{\left(\frac{1-\alpha}{z}\right)t}\left(\frac{\partial}{\partial\varepsilon_{j}}\ln\rho(\theta+\varepsilon)\right)\rho(\theta+\varepsilon)^{\left(\frac{1-\alpha}{z}\right)\left(1-t\right)}
\end{array}\right]
\end{align}
We then find that
\begin{align}
 & \frac{1}{1-\alpha}\left.\frac{\partial}{\partial\varepsilon_{i}}\Tr\!\left[\left(\rho(\theta)^{\frac{\alpha}{2z}}\rho(\theta+\varepsilon)^{\frac{1-\alpha}{z}}\rho(\theta)^{\frac{\alpha}{2z}}\right)^{z}\right]\right|_{\varepsilon=0}\nonumber \\
 & =\left.\int_{0}^{1}dt\,\Tr\!\left[\begin{array}{c}
\rho(\theta)^{\frac{\alpha}{2z}}\left(\rho(\theta)^{\frac{\alpha}{2z}}\rho(\theta+\varepsilon)^{\frac{1-\alpha}{z}}\rho(\theta)^{\frac{\alpha}{2z}}\right)^{z-1}\rho(\theta)^{\frac{\alpha}{2z}}\times\\
\rho(\theta+\varepsilon)^{\left(\frac{1-\alpha}{z}\right)t}\left(\frac{\partial}{\partial\varepsilon_{j}}\ln\rho(\theta+\varepsilon)\right)\rho(\theta+\varepsilon)^{\left(\frac{1-\alpha}{z}\right)\left(1-t\right)}
\end{array}\right]\right|_{\varepsilon=0}\\
 & =\int_{0}^{1}dt\,\Tr\!\left[\begin{array}{c}
\rho(\theta)^{\frac{\alpha}{2z}}\left(\rho(\theta)^{\frac{\alpha}{2z}}\rho(\theta)^{\frac{1-\alpha}{z}}\rho(\theta)^{\frac{\alpha}{2z}}\right)^{z-1}\rho(\theta)^{\frac{\alpha}{2z}}\times\\
\rho(\theta)^{\left(\frac{1-\alpha}{z}\right)t}\left.\left(\frac{\partial}{\partial\varepsilon_{j}}\ln\rho(\theta+\varepsilon)\right)\right|_{\varepsilon=0}\rho(\theta)^{\left(\frac{1-\alpha}{z}\right)\left(1-t\right)}
\end{array}\right]\\
 & =\int_{0}^{1}dt\,\Tr\!\left[\begin{array}{c}
\rho(\theta)^{\frac{\alpha}{2z}}\rho(\theta)^{\frac{1-\alpha}{z}\left(1-t\right)}\rho(\theta)^{\frac{z-1}{z}}\rho(\theta)^{\frac{1-\alpha}{z}t}\rho(\theta)^{\frac{\alpha}{2z}}\times\\
\left(\frac{\partial}{\partial\theta_{j}}\ln\rho(\theta)\right)
\end{array}\right]\\
 & =\Tr\!\left[\rho(\theta)^{\frac{\alpha}{2z}}\rho(\theta)^{\frac{1-\alpha}{z}}\rho(\theta)^{\frac{z-1}{z}}\rho(\theta)^{\frac{\alpha}{2z}}\left(\frac{\partial}{\partial\theta_{j}}\ln\rho(\theta)\right)\right]\\
 & =\Tr\!\left[\rho(\theta)\left(\frac{\partial}{\partial\theta_{j}}\ln\rho(\theta)\right)\right]\\
 & =0.
\end{align}
The last equality follows from (E78)--(E84) of \cite{Minervini2025}.
\end{proof}
\begin{rem}
For the special case of a single parameter, Theorem~\ref{thm:fisher-info-from-alpha-z}
was stated in \cite{May2018}. Therein, a proof was not given and
instead a brief suggestion for establishing the proof was provided
in \cite[Appendix~A]{May2018}, with the authors indicating that second
derivatives of matrix powers are needed to carry out the calculation.
In contrast, in the approach given above, there is no need to compute
a second derivative, due to the equality established in \eqref{eq:no-2nd-deriv}.
As such, the proof given above is presumably simpler than the approach
described in \cite[Appendix~A]{May2018}.
\begin{rem}
For two pure states $\psi\equiv|\psi\rangle\!\langle\psi|$ and $\phi\equiv|\phi\rangle\!\langle\phi|$
and $\alpha\in\left(0,1\right)$, the $\alpha$-$z$ R\'enyi relative
entropy reduces as follows:
\begin{align}
D_{\alpha,z}(\psi\|\phi) & \coloneqq\frac{1}{\alpha-1}\ln\Tr\!\left[\left(\phi^{\frac{1-\alpha}{2z}}\psi^{\frac{\alpha}{z}}\phi^{\frac{1-\alpha}{2z}}\right)^{z}\right]\\
 & =\frac{1}{\alpha-1}\ln\Tr\!\left[\left(|\phi\rangle\langle\phi|\psi\rangle\langle\psi|\phi\rangle\langle\phi|\right)^{z}\right]\\
 & =\frac{1}{\alpha-1}\ln\Tr\!\left[\left(|\phi\rangle\left|\langle\phi|\psi\rangle\right|^{2}\langle\phi|\right)^{z}\right]\\
 & =\frac{1}{\alpha-1}\ln\left|\langle\phi|\psi\rangle\right|^{2z}\Tr\!\left[\left(|\phi\rangle\!\langle\phi|\right)^{z}\right]\\
 & =\frac{1}{\alpha-1}\ln\left|\langle\phi|\psi\rangle\right|^{2z}\Tr\!\left[|\phi\rangle\!\langle\phi|\right]\\
 & =\frac{1}{\alpha-1}\ln\left|\langle\phi|\psi\rangle\right|^{2z}\\
 & =\frac{z}{1-\alpha}\left(-\ln\left|\langle\phi|\psi\rangle\right|^{2}\right)\\
 & =\frac{z}{1-\alpha}\left(-\ln F(\psi,\phi)\right),
\end{align}
where $F(\psi,\phi)\coloneqq\left|\langle\phi|\psi\rangle\right|^{2}$
is the fidelity. As expected, for $\psi\neq\phi$, the limit $D_{\alpha,z}(\psi\|\phi)\to\infty$
holds as $\alpha\to1$ for fixed $z>0$ or $D_{\alpha,z}(\psi\|\phi)\to\infty$
as $z\to\infty$ for fixed $\alpha\in\left(0,1\right)$, consistent
with the fact that $D_{\alpha,z}$ converges to the standard (Umegaki)
relative entropy in these limits and this latter quantity being equal
to $\infty$ for two distinct pure states. Thus, for a parameterized
family $\left(\psi(\theta)\right)_{\theta\in\mathbb{R}^{L}}$ of pure
states, the $\alpha$-$z$ information matrix reduces to
\begin{align}
\left[I_{\alpha,z}(\theta)\right]_{i,j} & \coloneqq\frac{1}{\alpha}\left.\frac{\partial^{2}}{\partial\varepsilon_{i}\partial\varepsilon_{j}}D_{\alpha,z}(\psi(\theta)\|\psi(\theta+\varepsilon))\right|_{\varepsilon=0}.\\
 & =\frac{z}{\alpha\left(1-\alpha\right)}\left.\frac{\partial^{2}}{\partial\varepsilon_{i}\partial\varepsilon_{j}}\left(-\ln F(\psi(\theta),\psi(\theta+\varepsilon))\right)\right|_{\varepsilon=0}\\
 & =\frac{2z}{\alpha\left(1-\alpha\right)}\Re\left[\langle\partial_{i}\psi(\theta)|\left(I-|\psi(\theta)\rangle\!\langle\psi(\theta)|\right)|\partial_{j}\psi(\theta)\rangle\right],
\end{align}
consistent with the known result \cite[Theorem~2.5]{Liu2019} that
\begin{equation}
\left.\frac{\partial^{2}}{\partial\varepsilon_{i}\partial\varepsilon_{j}}\left(-2\ln F(\psi(\theta),\psi(\theta+\varepsilon))\right)\right|_{\varepsilon=0}=4\Re\left[\langle\partial_{i}\psi(\theta)|\left(I-|\psi(\theta)\rangle\!\langle\psi(\theta)|\right)|\partial_{j}\psi(\theta)\rangle\right].\label{eq:fisher-bures-pure-states}
\end{equation}
Appendix~\ref{app:Fisher-Bures-pure} reviews the proof of \eqref{eq:fisher-bures-pure-states}.
\end{rem}

\end{rem}

\begin{cor}
\label{cor:operator-monotone-a-z}As a consequence of Theorem~\ref{thm:q-fisher-info-char}
and Theorem~\ref{thm:fisher-info-from-alpha-z}, the following function
is operator monotone on $x\in\left(0,\infty\right)$ for the values
of $\alpha$ and $z$ stated in Fact \ref{fact:a-z-data-proc}:
\begin{equation}
x\mapsto\frac{\alpha\left(1-\alpha\right)}{z}\frac{\left(x-1\right)\left(x^{\frac{1}{z}}-1\right)}{\left(x^{\frac{1-\alpha}{z}}-1\right)\left(x^{\frac{\alpha}{z}}-1\right)}.
\end{equation}
\end{cor}

\section{Special cases of $\alpha$-$z$ information matrices}

\label{sec:Special-cases-a-z}In this section, I consider special
cases of the $\alpha$-$z$ information matrices, as previously done
as well in \cite{May2018,Ciaglia2018}, which amounts to evaluating
them for particular values or in various limits. In particular, the
Kubo--Mori information matrix arises for all $\alpha\in\left(0,1\right)\cup\left(1,\infty\right)$
in the limit $z\to\infty$ or for all $z>0$ in the limit $\alpha\to1$
(Section~\ref{subsec:Kubo=002013Mori-information-matrix-from-a-z}).
The sandwiched R\'enyi information matrix arises when $z=\alpha$
(Section~\ref{subsec:Sandwiched-Rnyi-information-from-a-z}), and
the Petz--R\'enyi information matrix arises when $z=1$ (Section~\ref{subsec:Petz=002013Rnyi-information-matrices-from-a-z}).
All information matrices considered in this section are evaluated
for a second-order differentiable family of positive definite states,
$\left(\rho(\theta)\right)_{\theta\in\mathbb{R}^{L}}$.

\subsection{Kubo--Mori information matrix}

\label{subsec:Kubo=002013Mori-information-matrix-from-a-z}
\begin{prop}
For all $x,y>0$ such that $x\neq y$ and for all $\alpha\in\left(0,1\right)\cup\left(1,\infty\right)$,
\begin{equation}
\lim_{z\to\infty}\zeta_{\alpha,z}(x,y)=\frac{\ln x-\ln y}{x-y}.\label{eq:z-infty-limit-prop-statement}
\end{equation}
Thus, for all $\alpha\in\left(0,1\right)\cup\left(1,\infty\right)$,
\begin{equation}
\lim_{z\to\infty}I_{\alpha,z}(\theta)=I_{\KM}(\theta).\label{eq:a-z-to-KM-z-infty}
\end{equation}
\end{prop}

\begin{proof}
Consider that
\begin{align}
\lim_{z\to\infty}\zeta_{\alpha,z}(x,y) & =\lim_{z\to\infty}\frac{z}{\alpha\left(1-\alpha\right)}\left(\frac{x^{\frac{1-\alpha}{z}}-y^{\frac{1-\alpha}{z}}}{x-y}\right)\left(\frac{x^{\frac{\alpha}{z}}-y^{\frac{\alpha}{z}}}{x^{\frac{1}{z}}-y^{\frac{1}{z}}}\right)\\
 & =\lim_{z\to\infty}\frac{1}{\alpha\left(1-\alpha\right)}\left(\frac{1}{x-y}\right)\left(\frac{x^{\frac{1-\alpha}{z}}-y^{\frac{1-\alpha}{z}}}{\frac{1}{z}}\right)\left(\frac{x^{\frac{\alpha}{z}}-y^{\frac{\alpha}{z}}}{x^{\frac{1}{z}}-y^{\frac{1}{z}}}\right)\\
 & =\frac{1}{\alpha\left(1-\alpha\right)}\left(\frac{1}{x-y}\right)\left[\lim_{z\to\infty}\left(\frac{x^{\frac{1-\alpha}{z}}-y^{\frac{1-\alpha}{z}}}{\frac{1}{z}}\right)\right]\left[\lim_{z\to\infty}\left(\frac{x^{\frac{\alpha}{z}}-y^{\frac{\alpha}{z}}}{x^{\frac{1}{z}}-y^{\frac{1}{z}}}\right)\right]\\
 & =\frac{1}{\alpha\left(1-\alpha\right)}\left(\frac{1}{x-y}\right)\left[\lim_{h\to0}\left(\frac{x^{\left(1-\alpha\right)h}-y^{\left(1-\alpha\right)h}}{h}\right)\right]\left[\lim_{h\to0}\left(\frac{x^{\alpha h}-y^{\alpha h}}{x^{h}-y^{h}}\right)\right].\label{eq:z-infty-limit-proof-1}
\end{align}
Now consider that
\begin{align}
\lim_{h\to0}\frac{x^{\left(1-\alpha\right)h}-y^{\left(1-\alpha\right)h}}{h} & =\lim_{h\to0}\frac{x^{\left(1-\alpha\right)h}}{h}-\lim_{h\to0}\frac{y^{\left(1-\alpha\right)h}}{h}\\
 & =\left(1-\alpha\right)\ln x-\left(1-\alpha\right)\ln y\\
 & =\left(1-\alpha\right)\left(\ln x-\ln y\right)\label{eq:z-infty-limit-proof-2}
\end{align}
and
\begin{align}
\lim_{h\to0}\frac{x^{\alpha h}-y^{\alpha h}}{x^{h}-y^{h}} & =\lim_{h\to0}\frac{x^{\alpha h}\alpha\ln x-y^{\alpha h}\alpha\ln y}{x^{h}\ln x-y^{h}\ln y}\\
 & =\alpha\lim_{h\to0}\frac{x^{\alpha h}\ln x-y^{\alpha h}\ln y}{x^{h}\ln x-y^{h}\ln y}\\
 & =\alpha\left(\frac{\ln x-\ln y}{\ln x-\ln y}\right)\\
 & =\alpha.\label{eq:z-infty-limit-proof-3}
\end{align}
Then putting together \eqref{eq:z-infty-limit-proof-1}, \eqref{eq:z-infty-limit-proof-2},
and \eqref{eq:z-infty-limit-proof-3}, we conclude \eqref{eq:z-infty-limit-prop-statement}.
The equality in \eqref{eq:a-z-to-KM-z-infty} follows from \eqref{eq:z-infty-limit-prop-statement}
and \eqref{eq:kubo-mori-elements-3}.
\end{proof}
Consistent with \cite[Proposition~3]{Lin2015}, the following holds:
\begin{prop}
For all $x,y,z>0$ such that $x\neq y$,
\begin{equation}
\lim_{\alpha\to1}\zeta_{\alpha,z}(x,y)=\frac{\ln x-\ln y}{x-y}.\label{eq:alpha-1-limit-prop-statement}
\end{equation}
Thus, for all $z>0$,
\begin{equation}
\lim_{\alpha\to1}I_{\alpha,z}(\theta)=I_{\KM}(\theta).\label{eq:a-z-to-KM-a-1}
\end{equation}
\end{prop}

\begin{proof}
Consider that
\begin{align}
\lim_{\alpha\to1}\zeta_{\alpha,z}(x,y) & =\lim_{\alpha\to1}\frac{z}{\alpha\left(1-\alpha\right)}\left(\frac{x^{\frac{1-\alpha}{z}}-y^{\frac{1-\alpha}{z}}}{x-y}\right)\left(\frac{x^{\frac{\alpha}{z}}-y^{\frac{\alpha}{z}}}{x^{\frac{1}{z}}-y^{\frac{1}{z}}}\right)\\
 & =\left(\frac{z}{x-y}\right)\lim_{\alpha\to1}\frac{1}{\alpha}\left(\frac{x^{\frac{1-\alpha}{z}}-y^{\frac{1-\alpha}{z}}}{1-\alpha}\right)\left(\frac{x^{\frac{\alpha}{z}}-y^{\frac{\alpha}{z}}}{x^{\frac{1}{z}}-y^{\frac{1}{z}}}\right)\\
 & =\left(\frac{z}{x-y}\right)\left(\lim_{\alpha\to1}\frac{1}{\alpha}\right)\left(\lim_{\alpha\to1}\frac{x^{\frac{1-\alpha}{z}}-y^{\frac{1-\alpha}{z}}}{1-\alpha}\right)\left(\lim_{\alpha\to1}\frac{x^{\frac{\alpha}{z}}-y^{\frac{\alpha}{z}}}{x^{\frac{1}{z}}-y^{\frac{1}{z}}}\right)\\
 & =\left(\frac{z}{x-y}\right)\left(\lim_{\alpha\to1}\frac{x^{\frac{1-\alpha}{z}}-y^{\frac{1-\alpha}{z}}}{1-\alpha}\right)\left(\frac{x^{\frac{1}{z}}-y^{\frac{1}{z}}}{x^{\frac{1}{z}}-y^{\frac{1}{z}}}\right)\\
 & =\left(\frac{z}{x-y}\right)\left(\lim_{\alpha\to1}\frac{x^{\frac{1-\alpha}{z}}-y^{\frac{1-\alpha}{z}}}{1-\alpha}\right)\\
 & =\left(\frac{z}{x-y}\right)\left(\lim_{\alpha\to1}\frac{x^{\frac{1-\alpha}{z}}\left(-\frac{1}{z}\ln x\right)-y^{\frac{1-\alpha}{z}}\left(-\frac{1}{z}\ln y\right)}{-1}\right)\\
 & =\left(\frac{z}{x-y}\right)\frac{1}{z}\left(\ln x-\ln y\right)\\
 & =\frac{\ln x-\ln y}{x-y}.\label{eq:alpha-1-limit-proof-1}
\end{align}
The equality in \eqref{eq:a-z-to-KM-a-1} follows from \eqref{eq:alpha-1-limit-prop-statement}
and \eqref{eq:kubo-mori-elements-3}.
\end{proof}

\subsection{Petz--R\'enyi information matrices}

\label{subsec:Petz=002013Rnyi-information-matrices-from-a-z}The Petz--R\'enyi
relative entropy is defined for positive definite states $\rho$ and
$\sigma$ and $\alpha\in\left(0,1\right)\cup\left(1,\infty\right)$
as \cite{Petz1985,Petz1986}
\begin{equation}
\overline{D}_{\alpha}(\rho\|\sigma)\coloneqq\frac{1}{\alpha-1}\ln\Tr\!\left[\rho^{\alpha}\sigma^{1-\alpha}\right].
\end{equation}
It obeys the data-processing inequality for $\alpha\in\left(0,1\right)\cup\left(1,2\right]$
\cite{Petz1985,Petz1986}. It is a special case of the $\alpha$-$z$
R\'enyi relative entropy when $z=1$.

The elements of the Petz--R\'enyi information matrix, considered
previously in \cite{Hasegawa1993}, are defined for all $\alpha\in\left(0,1\right)\cup\left(1,\infty\right)$
as
\begin{align}
\left[\overline{I}_{\alpha}(\theta)\right]_{i,j} & \coloneqq\frac{1}{\alpha}\left.\frac{\partial^{2}}{\partial\varepsilon_{i}\partial\varepsilon_{j}}\overline{D}_{\alpha}(\rho(\theta)\|\rho(\theta+\varepsilon))\right|_{\varepsilon=0}\\
 & =\frac{1}{\alpha}\left.\frac{\partial^{2}}{\partial\varepsilon_{i}\partial\varepsilon_{j}}D_{\alpha,1}(\rho(\theta)\|\rho(\theta+\varepsilon))\right|_{\varepsilon=0}
\end{align}
Given that, for all $x,y>0$ such that $x\neq y$,
\begin{align}
\lim_{z\to1}\zeta_{\alpha,z}(x,y) & =\lim_{z\to1}\frac{z}{\alpha\left(1-\alpha\right)}\left(\frac{x^{\frac{1-\alpha}{z}}-y^{\frac{1-\alpha}{z}}}{x-y}\right)\left(\frac{x^{\frac{\alpha}{z}}-y^{\frac{\alpha}{z}}}{x^{\frac{1}{z}}-y^{\frac{1}{z}}}\right)\\
 & =\frac{1}{\alpha\left(1-\alpha\right)}\left(\frac{x^{1-\alpha}-y^{1-\alpha}}{x-y}\right)\left(\frac{x^{\alpha}-y^{\alpha}}{x-y}\right)\\
 & =\frac{1}{\alpha\left(1-\alpha\right)}\frac{\left(x^{\alpha}-y^{\alpha}\right)\left(x^{1-\alpha}-y^{1-\alpha}\right)}{\left(x-y\right)^{2}},
\end{align}
we can conclude the following corollary of Theorem~\ref{thm:fisher-info-from-alpha-z}:
\begin{cor}
\label{cor:Petz-Renyi-special-case}For all $\alpha\in\left(0,1\right)\cup\left(1,\infty\right)$,
the elements of the Petz--R\'enyi information matrix are as follows:
\begin{equation}
\left[\overline{I}_{\alpha}(\theta)\right]_{i,j}=\sum_{k,\ell}\overline{\zeta}_{\alpha}(\lambda_{k},\lambda_{\ell})\Tr\!\left[\Pi_{k}\left(\partial_{i}\rho(\theta)\right)\Pi_{\ell}\left(\partial_{j}\rho(\theta)\right)\right],
\end{equation}
where $\partial_{i}\equiv\frac{\partial}{\partial\theta_{i}}$, the
spectral decomposition of $\rho(\theta)$ is given by $\rho(\theta)=\sum_{k}\lambda_{k}\Pi_{k}$,
and for all $x,y>0$,
\begin{equation}
\overline{\zeta}_{\alpha}(x,y)\coloneqq\begin{cases}
\frac{1}{\alpha\left(1-\alpha\right)}\frac{\left(x^{\alpha}-y^{\alpha}\right)\left(x^{1-\alpha}-y^{1-\alpha}\right)}{\left(x-y\right)^{2}} & :x\neq y\\
\frac{1}{x} & :x=y
\end{cases}.\label{eq:limit-for-a-z-eigenval-func-1}
\end{equation}
\end{cor}

Interestingly, the elements of the Petz--R\'enyi information matrix
$\overline{I}_{\alpha}(\theta)$ are $\alpha$-dependent, in contrast
to the log-Euclidean and geometric information matrices (recall Theorem~\ref{thm:log-Euclidean-information-matrix}
and Theorem~\ref{thm:geometric-Renyi-to-RLD}). I explore this point
further in Section~\ref{sec:Orderings-and-relations}.

For $\alpha\in\left(0,1\right)$, Proposition \ref{prop:integral-rep-Petz}
gives an integral representation for the elements of the Petz--R\'enyi
information matrix, which is basis independent, and for $\alpha=2$,
Corollary~\ref{cor:Petz-Renyi-a-2-special} states that the Petz--R\'enyi
information matrix is equal to the RLD information matrix.
\begin{prop}
\label{prop:integral-rep-Petz}For all $\alpha\in\left(0,1\right)$,
the following integral representation holds for the Petz--R\'enyi
information matrix:
\begin{multline}
\left[\overline{I}_{\alpha}(\theta)\right]_{i,j}=\\
\frac{\sin^{2}(\alpha\pi)}{\alpha\left(1-\alpha\right)\pi^{2}}\int_{0}^{\infty}\int_{0}^{\infty}ds\ dt\ s^{\alpha}t^{1-\alpha}\Tr\!\left[\begin{array}{c}
\left(\rho(\theta)+sI\right)^{-1}\left(\rho(\theta)+tI\right)^{-1}\left(\partial_{i}\rho(\theta)\right)\times\\
\left(\rho(\theta)+sI\right)^{-1}\left(\rho(\theta)+tI\right)^{-1}\left(\partial_{j}\rho(\theta)\right)
\end{array}\right].
\end{multline}
\end{prop}

\begin{proof}
Using the integral representation from \eqref{eq:integral-rep-div-diff-x-r},
it follows that, for $x,y>0$ such that $x\neq y$,
\begin{align}
\overline{\zeta}_{\alpha}(x,y) & =\frac{1}{\alpha\left(1-\alpha\right)}\frac{\left(x^{\alpha}-y^{\alpha}\right)\left(x^{1-\alpha}-y^{1-\alpha}\right)}{\left(x-y\right)^{2}}\\
 & =\frac{\sin\!\left(\alpha\pi\right)}{\alpha\pi}\frac{\sin\!\left(\left(1-\alpha\right)\pi\right)}{\left(1-\alpha\right)\pi}\int_{0}^{\infty}ds\ \frac{s^{\alpha}}{\left(x+s\right)\left(y+s\right)}\int_{0}^{\infty}dt\ \frac{t^{1-\alpha}}{\left(x+t\right)\left(y+t\right)}\\
 & =\frac{\sin^{2}(\alpha\pi)}{\alpha\left(1-\alpha\right)\pi^{2}}\int_{0}^{\infty}\int_{0}^{\infty}ds\ dt\ \frac{s^{\alpha}t^{1-\alpha}}{\left(x+s\right)\left(x+t\right)\left(y+s\right)\left(y+t\right)},
\end{align}
which implies, after defining $g(\alpha)\equiv\frac{\sin^{2}(\alpha\pi)}{\alpha\left(1-\alpha\right)\pi^{2}}$,
that
\begin{align}
 & \left[\overline{I}_{\alpha}(\theta)\right]_{i,j}\nonumber \\
 & =\sum_{k,\ell}\overline{\zeta}_{\alpha}(\lambda_{k},\lambda_{\ell})\Tr\!\left[\Pi_{k}\left(\partial_{i}\rho(\theta)\right)\Pi_{\ell}\left(\partial_{j}\rho(\theta)\right)\right]\\
 & =\sum_{k,\ell}g(\alpha)\int_{0}^{\infty}\int_{0}^{\infty}ds\ dt\ \frac{s^{\alpha}t^{1-\alpha}}{\left(\lambda_{k}+s\right)\left(\lambda_{k}+t\right)\left(\lambda_{\ell}+s\right)\left(\lambda_{\ell}+t\right)}\Tr\!\left[\Pi_{k}\left(\partial_{i}\rho(\theta)\right)\Pi_{\ell}\left(\partial_{j}\rho(\theta)\right)\right]\\
 & =g(\alpha)\int_{0}^{\infty}\int_{0}^{\infty}ds\ dt\ s^{\alpha}t^{1-\alpha}\sum_{k,\ell}\frac{\Tr\!\left[\Pi_{k}\left(\partial_{i}\rho(\theta)\right)\Pi_{\ell}\left(\partial_{j}\rho(\theta)\right)\right]}{\left(\lambda_{k}+s\right)\left(\lambda_{k}+t\right)\left(\lambda_{\ell}+s\right)\left(\lambda_{\ell}+t\right)}\\
 & =g(\alpha)\int_{0}^{\infty}\int_{0}^{\infty}ds\ dt\ s^{\alpha}t^{1-\alpha}\Tr\!\left[\begin{array}{c}
\sum_{k}\left(\frac{1}{\left(\lambda_{k}+s\right)\left(\lambda_{k}+t\right)}\right)\Pi_{k}\left(\partial_{i}\rho(\theta)\right)\times\\
\sum_{\ell}\left(\frac{1}{\left(\lambda_{\ell}+s\right)\left(\lambda_{\ell}+t\right)}\right)\Pi_{\ell}\left(\partial_{j}\rho(\theta)\right)
\end{array}\right]\\
 & =g(\alpha)\int_{0}^{\infty}\int_{0}^{\infty}ds\ dt\ s^{\alpha}t^{1-\alpha}\Tr\!\left[\begin{array}{c}
\left(\rho(\theta)+sI\right)^{-1}\left(\rho(\theta)+tI\right)^{-1}\left(\partial_{i}\rho(\theta)\right)\times\\
\left(\rho(\theta)+sI\right)^{-1}\left(\rho(\theta)+tI\right)^{-1}\left(\partial_{j}\rho(\theta)\right)
\end{array}\right],
\end{align}
thus concluding the proof.
\end{proof}
\begin{cor}
\label{cor:Petz-Renyi-a-2-special}For $\alpha=2$, the following
equality holds:
\begin{equation}
\overline{I}_{2}(\theta)=I_{\RLD}(\theta).
\end{equation}
\end{cor}

\begin{proof}
Considering that, for all $x,y>0$ such that $x\neq y$,
\begin{align}
\lim_{\alpha\to2}\overline{\zeta}_{\alpha}(x,y) & =\lim_{\alpha\to2}\frac{1}{\alpha\left(1-\alpha\right)}\frac{\left(x^{\alpha}-y^{\alpha}\right)\left(x^{1-\alpha}-y^{1-\alpha}\right)}{\left(x-y\right)^{2}}\\
 & =-\frac{1}{2}\frac{\left(x^{2}-y^{2}\right)\left(x^{-1}-y^{-1}\right)}{\left(x-y\right)^{2}}\\
 & =-\frac{1}{2}\frac{\left(x+y\right)\left(x-y\right)x^{-1}y^{-1}\left(y-x\right)}{\left(x-y\right)^{2}}\\
 & =\frac{1}{2}\left(x+y\right)x^{-1}y^{-1}\\
 & =\frac{1}{2}\left(\frac{1}{x}+\frac{1}{y}\right),
\end{align}
it follows that
\begin{align}
\left[\overline{I}_{2}(\theta)\right]_{i,j} & =\sum_{k,\ell}\frac{1}{2}\left(\frac{1}{\lambda_{k}}+\frac{1}{\lambda_{\ell}}\right)\Tr\!\left[\Pi_{k}\left(\partial_{i}\rho(\theta)\right)\Pi_{\ell}\left(\partial_{j}\rho(\theta)\right)\right]\\
 & =\frac{1}{2}\sum_{k,\ell}\frac{1}{\lambda_{k}}\Tr\!\left[\Pi_{k}\left(\partial_{i}\rho(\theta)\right)\Pi_{\ell}\left(\partial_{j}\rho(\theta)\right)\right]\nonumber \\
 & \qquad+\frac{1}{2}\sum_{k,\ell}\frac{1}{\lambda_{\ell}}\Tr\!\left[\Pi_{k}\left(\partial_{i}\rho(\theta)\right)\Pi_{\ell}\left(\partial_{j}\rho(\theta)\right)\right]\\
 & =\frac{1}{2}\Tr\!\left[\left(\sum_{k}\frac{1}{\lambda_{k}}\Pi_{k}\right)\left(\partial_{i}\rho(\theta)\right)\left(\sum_{\ell}\Pi_{\ell}\right)\left(\partial_{j}\rho(\theta)\right)\right]\nonumber \\
 & \qquad+\frac{1}{2}\Tr\!\left[\left(\sum_{k}\Pi_{k}\right)\left(\partial_{i}\rho(\theta)\right)\left(\sum_{\ell}\frac{1}{\lambda_{\ell}}\Pi_{\ell}\right)\left(\partial_{j}\rho(\theta)\right)\right]\\
 & =\frac{1}{2}\Tr\!\left[\rho(\theta)^{-1}\left(\partial_{i}\rho(\theta)\right)\left(\partial_{j}\rho(\theta)\right)\right]+\frac{1}{2}\Tr\!\left[\left(\partial_{i}\rho(\theta)\right)\rho(\theta)^{-1}\left(\partial_{j}\rho(\theta)\right)\right]\\
 & =\frac{1}{2}\Tr\!\left[\left\{ \partial_{i}\rho(\theta),\partial_{j}\rho(\theta)\right\} \rho(\theta)^{-1}\right]\\
 & =\left[I_{\RLD}(\theta)\right]_{i,j},
\end{align}
thus concluding the proof.
\end{proof}
\begin{rem}
Corollary~\ref{cor:Petz-Renyi-a-2-special} is consistent with the
fact that $\overline{D}_{2}(\rho\|\sigma)=\widehat{D}_{2}(\rho\|\sigma)$
and Theorem~\ref{thm:geometric-Renyi-to-RLD}.
\end{rem}

\subsection{Sandwiched R\'enyi information matrices}

\label{subsec:Sandwiched-Rnyi-information-from-a-z}The sandwiched
R\'enyi relative entropy is defined for positive definite states
$\rho$ and $\sigma$ and $\alpha\in\left(0,1\right)\cup\left(1,\infty\right)$
as \cite{MuellerLennert2013,Wilde2014}
\begin{equation}
\widetilde{D}_{\alpha}(\rho\|\sigma)\coloneqq\frac{1}{\alpha-1}\ln\Tr\!\left[\left(\sigma^{\frac{1-\alpha}{2\alpha}}\rho\sigma^{\frac{1-\alpha}{2\alpha}}\right)^{\alpha}\right].
\end{equation}
It obeys the data-processing inequality for $\alpha\in\left[\frac{1}{2},1\right)\cup\left(1,\infty\right)$
\cite{Frank2013} (see also \cite{Wilde2018}). It is a special case
of the $\alpha$-$z$ R\'enyi relative entropy when $z=\alpha$.

The elements of the sandwiched R\'enyi information matrix, considered
previously in \cite{Takahashi2017}, are defined for all $\alpha\in\left(0,1\right)\cup\left(1,\infty\right)$
as
\begin{align}
\left[\widetilde{I}_{\alpha}(\theta)\right]_{i,j} & \coloneqq\frac{1}{\alpha}\left.\frac{\partial^{2}}{\partial\varepsilon_{i}\partial\varepsilon_{j}}\widetilde{D}_{\alpha}(\rho(\theta)\|\rho(\theta+\varepsilon))\right|_{\varepsilon=0}\\
 & =\frac{1}{\alpha}\left.\frac{\partial^{2}}{\partial\varepsilon_{i}\partial\varepsilon_{j}}D_{\alpha,\alpha}(\rho(\theta)\|\rho(\theta+\varepsilon))\right|_{\varepsilon=0}
\end{align}
Given that, for all $x,y>0$ such that $x\neq y$,
\begin{align}
\lim_{z\to\alpha}\zeta_{\alpha,z}(x,y) & =\lim_{z\to\alpha}\frac{z}{\alpha\left(1-\alpha\right)}\left(\frac{x^{\frac{1-\alpha}{z}}-y^{\frac{1-\alpha}{z}}}{x-y}\right)\left(\frac{x^{\frac{\alpha}{z}}-y^{\frac{\alpha}{z}}}{x^{\frac{1}{z}}-y^{\frac{1}{z}}}\right)\\
 & =\frac{1}{1-\alpha}\left(\frac{x^{\frac{1-\alpha}{\alpha}}-y^{\frac{1-\alpha}{\alpha}}}{x-y}\right)\left(\frac{x^{\frac{\alpha}{\alpha}}-y^{\frac{\alpha}{\alpha}}}{x^{\frac{1}{z}}-y^{\frac{1}{z}}}\right)\\
 & =\frac{1}{1-\alpha}\left(\frac{x^{\frac{1-\alpha}{\alpha}}-y^{\frac{1-\alpha}{\alpha}}}{x^{\frac{1}{\alpha}}-y^{\frac{1}{\alpha}}}\right),
\end{align}
we conclude the following corollary of Theorem~\ref{thm:fisher-info-from-alpha-z},
related to what was previously reported as \cite[Lemma~4]{Takahashi2017}:
\begin{cor}
\label{cor:sandwiched-Renyi-special-case}For all $\alpha\in\left(0,1\right)\cup\left(1,\infty\right)$,
the elements of the sandwiched R\'enyi information matrix are as
follows:
\begin{equation}
\left[\widetilde{I}_{\alpha}(\theta)\right]_{i,j}=\sum_{k,\ell}\widetilde{\zeta}_{\alpha}(\lambda_{k},\lambda_{\ell})\Tr\!\left[\Pi_{k}\left(\partial_{i}\rho(\theta)\right)\Pi_{\ell}\left(\partial_{j}\rho(\theta)\right)\right],
\end{equation}
where $\partial_{i}\equiv\frac{\partial}{\partial\theta_{i}}$, the
spectral decomposition of $\rho(\theta)$ is given by $\rho(\theta)=\sum_{k}\lambda_{k}\Pi_{k}$,
and for all $x,y>0$,
\begin{equation}
\widetilde{\zeta}_{\alpha}(x,y)\coloneqq\begin{cases}
\frac{1}{1-\alpha}\left(\frac{x^{\frac{1-\alpha}{\alpha}}-y^{\frac{1-\alpha}{\alpha}}}{x^{\frac{1}{\alpha}}-y^{\frac{1}{\alpha}}}\right) & :x\neq y\\
\frac{1}{x} & :x=y
\end{cases}.\label{eq:limit-for-sandwiched-Renyi-eigenval-func}
\end{equation}
\end{cor}

Similar to the Petz--R\'enyi information matrix, the elements of
the sandwiched R\'enyi information matrix $\widetilde{I}_{\alpha}(\theta)$
are $\alpha$-dependent, again in contrast to the log-Euclidean and
geometric information matrices (recall Theorem~\ref{thm:log-Euclidean-information-matrix}
and Theorem~\ref{thm:geometric-Renyi-to-RLD}). I explore this point
further in Section~\ref{sec:Orderings-and-relations}.

For $\alpha\in\left(0,1\right)$, Proposition \ref{prop:integral-rep-sandwiched}
gives an integral representation for the elements of the sandwiched
R\'enyi information matrix, and for $\alpha=2$, Corollary~\ref{cor:sandwiched-Renyi-a-2-special}
gives a simple form for them. Both of the expressions given are basis
independent.
\begin{prop}
\label{prop:integral-rep-sandwiched}For all $\alpha\in\left(0,1\right)$,
the following integral representation holds for the sandwiched R\'enyi
information matrix:
\begin{multline}
\left[\widetilde{I}_{\alpha}(\theta)\right]_{i,j}=\\
\frac{\sin\!\left(\left(1-\alpha\right)\pi\right)}{\left(1-\alpha\right)\pi}\int_{0}^{\infty}dt\ t^{1-\alpha}\Tr\!\left[\left(\rho(\theta)^{\frac{1}{\alpha}}+tI\right)^{-1}\left(\partial_{i}\rho(\theta)\right)\left(\rho(\theta)^{\frac{1}{\alpha}}+tI\right)^{-1}\left(\partial_{j}\rho(\theta)\right)\right].
\end{multline}
\end{prop}

\begin{proof}
Using the integral representation from \eqref{eq:integral-rep-div-diff-x-r},
it follows that, for $x,y>0$ such that $x\neq y$,
\begin{align}
\widetilde{\zeta}_{\alpha}(x,y) & =\frac{1}{1-\alpha}\left(\frac{x^{\frac{1-\alpha}{\alpha}}-y^{\frac{1-\alpha}{\alpha}}}{x^{\frac{1}{\alpha}}-y^{\frac{1}{\alpha}}}\right)\\
 & =\frac{\sin\!\left(\left(1-\alpha\right)\pi\right)}{\left(1-\alpha\right)\pi}\int_{0}^{\infty}dt\ \frac{t^{1-\alpha}}{\left(x^{\frac{1}{\alpha}}+t\right)\left(y^{\frac{1}{\alpha}}+t\right)},
\end{align}
which implies that
\begin{align}
 & \left[\widetilde{I}_{\alpha}(\theta)\right]_{i,j}\nonumber \\
 & =\sum_{k,\ell}\widetilde{\zeta}_{\alpha}(\lambda_{k},\lambda_{\ell})\Tr\!\left[\Pi_{k}\left(\partial_{i}\rho(\theta)\right)\Pi_{\ell}\left(\partial_{j}\rho(\theta)\right)\right]\\
 & =\sum_{k,\ell}\frac{\sin\!\left(\left(1-\alpha\right)\pi\right)}{\left(1-\alpha\right)\pi}\int_{0}^{\infty}dt\ \frac{t^{1-\alpha}}{\left(\lambda_{k}^{\frac{1}{\alpha}}+t\right)\left(\lambda_{\ell}^{\frac{1}{\alpha}}+t\right)}\Tr\!\left[\Pi_{k}\left(\partial_{i}\rho(\theta)\right)\Pi_{\ell}\left(\partial_{j}\rho(\theta)\right)\right]\\
 & =\frac{\sin\!\left(\left(1-\alpha\right)\pi\right)}{\left(1-\alpha\right)\pi}\int_{0}^{\infty}dt\ t^{1-\alpha}\sum_{k,\ell}\frac{1}{\left(\lambda_{k}^{\frac{1}{\alpha}}+t\right)\left(\lambda_{\ell}^{\frac{1}{\alpha}}+t\right)}\Tr\!\left[\Pi_{k}\left(\partial_{i}\rho(\theta)\right)\Pi_{\ell}\left(\partial_{j}\rho(\theta)\right)\right]\\
 & =\frac{\sin\!\left(\left(1-\alpha\right)\pi\right)}{\left(1-\alpha\right)\pi}\int_{0}^{\infty}dt\ t^{1-\alpha}\Tr\!\left[\sum_{k}\left(\frac{1}{\lambda_{k}^{\frac{1}{\alpha}}+t}\right)\Pi_{k}\left(\partial_{i}\rho(\theta)\right)\sum_{\ell}\left(\frac{1}{\lambda_{\ell}^{\frac{1}{\alpha}}+t}\right)\Pi_{\ell}\left(\partial_{j}\rho(\theta)\right)\right]\\
 & =\frac{\sin\!\left(\left(1-\alpha\right)\pi\right)}{\left(1-\alpha\right)\pi}\int_{0}^{\infty}dt\ t^{1-\alpha}\Tr\!\left[\left(\rho(\theta)^{\frac{1}{\alpha}}+tI\right)^{-1}\left(\partial_{i}\rho(\theta)\right)\left(\rho(\theta)^{\frac{1}{\alpha}}+tI\right)^{-1}\left(\partial_{j}\rho(\theta)\right)\right],
\end{align}
thus concluding the proof.
\end{proof}
\begin{cor}
\label{cor:sandwiched-Renyi-a-2-special}For $\alpha=2$, the following
equality holds:
\begin{equation}
\left[\widetilde{I}_{2}(\theta)\right]_{i,j}=\Tr\!\left[\rho(\theta)^{-\frac{1}{2}}\left(\partial_{i}\rho(\theta)\right)\rho(\theta)^{-\frac{1}{2}}\left(\partial_{j}\rho(\theta)\right)\right].
\end{equation}
\end{cor}

\begin{proof}
Consider that
\begin{align}
\lim_{\alpha\to2}\widetilde{\zeta}_{\alpha}(x,y) & =\lim_{\alpha\to2}\frac{1}{1-\alpha}\left(\frac{x^{\frac{1-\alpha}{\alpha}}-y^{\frac{1-\alpha}{\alpha}}}{x^{\frac{1}{\alpha}}-y^{\frac{1}{\alpha}}}\right)\\
 & =-\left(\frac{x^{-\frac{1}{2}}-y^{-\frac{1}{2}}}{x^{\frac{1}{2}}-y^{\frac{1}{2}}}\right)\\
 & =-x^{-\frac{1}{2}}y^{-\frac{1}{2}}\left(\frac{y^{\frac{1}{2}}-x^{\frac{1}{2}}}{x^{\frac{1}{2}}-y^{\frac{1}{2}}}\right)\\
 & =\frac{1}{\sqrt{xy}}.
\end{align}
Then consider that, from Corollary~\ref{cor:sandwiched-Renyi-special-case},
\begin{align}
\left[\widetilde{I}_{2}(\theta)\right]_{i,j} & =\sum_{k,\ell}\widetilde{\zeta}_{2}(\lambda_{k},\lambda_{\ell})\Tr\!\left[\Pi_{k}\left(\partial_{i}\rho(\theta)\right)\Pi_{\ell}\left(\partial_{j}\rho(\theta)\right)\right]\\
 & =\sum_{k,\ell}\frac{1}{\sqrt{\lambda_{k}\lambda_{\ell}}}\Tr\!\left[\Pi_{k}\left(\partial_{i}\rho(\theta)\right)\Pi_{\ell}\left(\partial_{j}\rho(\theta)\right)\right]\\
 & =\Tr\!\left[\left(\sum_{k}\frac{1}{\sqrt{\lambda_{k}}}\Pi_{k}\right)\left(\partial_{i}\rho(\theta)\right)\left(\sum_{\ell}\frac{1}{\sqrt{\lambda_{\ell}}}\Pi_{\ell}\right)\left(\partial_{j}\rho(\theta)\right)\right]\\
 & =\Tr\!\left[\rho(\theta)^{-\frac{1}{2}}\left(\partial_{i}\rho(\theta)\right)\rho(\theta)^{-\frac{1}{2}}\left(\partial_{j}\rho(\theta)\right)\right],
\end{align}
thus concluding the proof.
\end{proof}

\section{Orderings and relations between $\alpha$-$z$ information matrices}

\label{sec:Orderings-and-relations}In this section, I establish ordering
relations for the Petz-- and sandwiched R\'enyi information matrices.
I begin with some orderings that are direct consequences of previous
results. After that, Section~\ref{subsec:Ordering-general-info-matrices}
provides a lemma for orderings of general information matrices, which
was mentioned in passing in \cite[Eq.~(60)]{Jarzyna2020}. Sections~\ref{subsec:Ordering-of-Petz-Renyi}
and \ref{subsec:Ordering-of-sandwiched-Renyi} then provide orderings
for the Petz-- and sandwiched R\'enyi information matrices, respectively.

Let us begin by recalling that orderings of smooth divergences imply
orderings of information matrices, as reviewed in \cite[Proposition~4]{Minervini2025}.
As such, we can conclude some orderings as direct consequences of
orderings of R\'enyi relative entropies that have been previously
established in the literature. Given that, for states $\rho$ and
$\sigma$,
\begin{equation}
\widetilde{D}_{\alpha}(\rho\|\sigma)\leq\overline{D}_{\alpha}(\rho\|\sigma)
\end{equation}
for all $\alpha>0$ \cite{Wilde2014,Datta2014} and that 
\begin{equation}
\alpha\overline{D}_{\alpha}(\rho\|\sigma)\leq\widetilde{D}_{\alpha}(\rho\|\sigma)
\end{equation}
for all $\alpha\in\left(0,1\right)$ \cite[Corollary~2.3]{Iten2017},
we can divide by $\alpha$ in both of the inequalities and apply \cite[Proposition~4]{Minervini2025}
to conclude the following matrix inequalities for the Petz-- and
sandwiched R\'enyi information matrices for all $\theta\in\mathbb{R}^{L}$:
\begin{equation}
\widetilde{I}_{\alpha}(\theta)\leq\overline{I}_{\alpha}(\theta)
\end{equation}
for all $\alpha>0$ and that
\begin{equation}
\alpha\overline{I}_{\alpha}(\theta)\leq\widetilde{I}_{\alpha}(\theta)
\end{equation}
for all $\alpha\in\left(0,1\right)$. Additionally, it is known that,
for fixed states $\rho$ and $\sigma$ and $\alpha>0$, the function
$z\mapsto D_{\alpha,z}(\rho\|\sigma)$ is monotone increasing on $\mathbb{R}_{+}$
if $\alpha\in\left(0,1\right)$ and it is monotone decreasing on $\mathbb{R}_{+}$
if $\alpha>1$ \cite[Proposition~1]{Lin2015}. A direct consequence
is the following:
\begin{cor}
For all $\theta\in\mathbb{R}^{L}$, the matrix-valued function $z\mapsto I_{\alpha,z}(\theta)$
is monotone increasing on $\mathbb{R}_{+}$ if $\alpha\in\left(0,1\right)$
and it is monotone decreasing on $\mathbb{R}_{+}$ if $\alpha>1$,
with respect to the Loewner order.
\end{cor}

\subsection{Ordering of general information matrices}

\label{subsec:Ordering-general-info-matrices}Suppose that $\zeta\colon\mathbb{R}_{+}\times\mathbb{R}_{+}\to\mathbb{R}_{+}$
is a function satisfying items 1--3 of Theorem~\ref{thm:q-fisher-info-char},
and define $f(t)\coloneqq\frac{1}{\zeta(t,1)}$ for all $t>0$. Recall
that a general information matrix has the following form:
\begin{align}
\left[I_{\zeta}(\theta)\right]_{i,j} & =\sum_{k,\ell}\zeta(\lambda_{k},\lambda_{\ell})\Tr\!\left[\Pi_{k}\left(\partial_{i}\rho(\theta)\right)\Pi_{\ell}\left(\partial_{j}\rho(\theta)\right)\right],\\
 & =\sum_{k,\ell}\frac{1}{\lambda_{\ell}f\!\left(\frac{\lambda_{k}}{\lambda_{\ell}}\right)}\Tr\!\left[\Pi_{k}\left(\partial_{i}\rho(\theta)\right)\Pi_{\ell}\left(\partial_{j}\rho(\theta)\right)\right],\label{eq:gen-info-matrix-f-func}
\end{align}
where the second equality follows because
\begin{equation}
\zeta(\lambda_{k},\lambda_{\ell})=\frac{1}{\lambda_{\ell}}\zeta\!\left(\frac{\lambda_{k}}{\lambda_{\ell}},1\right)=\frac{1}{\lambda_{\ell}f\!\left(\frac{\lambda_{k}}{\lambda_{\ell}}\right)}.
\end{equation}

In this section, I drop the requirement for $f$ to be operator monotone
on $\left(0,\infty\right)$, in which case $I_{\zeta}(\theta)$ need
not satisfy the data-processing inequality. However, the information
matrix $I_{\zeta}(\theta)$ is still positive semi-definite, as seen
in Corollary~\ref{cor:gen-info-PSD} below, and we can still consider
ordering relations between information matrices.

Let us begin with a general statement regarding ordering of information
matrices built from functions $f_{1}$ and $f_{2}$ that satisfy an
ordering relation on $\mathbb{R}_{+}$:
\begin{lem}
\label{lem:ordering-for-info-matrices}Let $\zeta_{1},\zeta_{2}\colon\mathbb{R}_{+}\times\mathbb{R}_{+}\to\mathbb{R}_{+}$
be functions that satisfy items 1--3 of Theorem~\ref{thm:q-fisher-info-char},
and define $f_{i}(t)\coloneqq\frac{1}{\zeta_{i}(t,1)}$ for all $t>0$
and $i\in\left\{ 1,2\right\} $. Suppose that $f_{1}(t)\geq f_{2}(t)$
for all $t>0$. Then the following matrix inequality holds for all
$\theta\in\mathbb{R}^{L}$:
\begin{equation}
I_{\zeta_{1}}(\theta)\leq I_{\zeta_{2}}(\theta).
\end{equation}
\end{lem}

\begin{proof}
To see this, consider that the desired matrix inequality is equivalent
to 
\begin{equation}
v^{T}I_{\zeta_{1}}(\theta)v\leq v^{T}I_{\zeta_{2}}(\theta)v\label{eq:ordering-info-matrices-general}
\end{equation}
holding for all $v\in\mathbb{R}^{L}$. Defining
\begin{align}
W & \coloneqq\sum_{i}v_{i}\partial_{i}\rho(\theta),\\
|W\rangle & \coloneqq\left(W\otimes I\right)|\Gamma\rangle,\\
|\Gamma\rangle & \coloneqq\sum_{i}|i\rangle\otimes|i\rangle,
\end{align}
 consider that, for a general $\zeta$ and $f$, defined in the same
way as $\zeta_{i}$ and $f_{i}$, 
\begin{align}
v^{T}I_{\zeta}(\theta)v & =\sum_{i,j}v_{i}\left(\sum_{k,\ell}\frac{1}{\lambda_{\ell}f\!\left(\frac{\lambda_{k}}{\lambda_{\ell}}\right)}\Tr\!\left[\Pi_{k}\left(\partial_{i}\rho(\theta)\right)\Pi_{\ell}\left(\partial_{j}\rho(\theta)\right)\right]\right)v_{j}\\
 & =\sum_{k,\ell}\frac{1}{\lambda_{\ell}f\!\left(\frac{\lambda_{k}}{\lambda_{\ell}}\right)}\Tr\!\left[\Pi_{k}\left(\sum_{i}v_{i}\partial_{i}\rho(\theta)\right)\Pi_{\ell}\left(\sum_{j}v_{j}\partial_{j}\rho(\theta)\right)\right]\\
 & =\sum_{k,\ell}\frac{1}{\lambda_{\ell}f\!\left(\frac{\lambda_{k}}{\lambda_{\ell}}\right)}\Tr\!\left[\Pi_{k}W\Pi_{\ell}W\right]\\
 & =\sum_{k,\ell}\frac{1}{\lambda_{\ell}f\!\left(\frac{\lambda_{k}}{\lambda_{\ell}}\right)}\langle W|\left(\Pi_{k}\otimes\Pi_{\ell}^{T}\right)|W\rangle\\
 & =\langle W|\left[\sum_{k,\ell}\frac{1}{\lambda_{\ell}f\!\left(\frac{\lambda_{k}}{\lambda_{\ell}}\right)}\left(\Pi_{k}\otimes\Pi_{\ell}^{T}\right)\right]|W\rangle,\label{eq:alt-form-metric-petz-K}
\end{align}
where we used that
\begin{align}
\Tr\!\left[\Pi_{k}W\Pi_{\ell}W\right] & =\Tr\!\left[W\Pi_{k}W\Pi_{\ell}\right]\\
 & =\langle\Gamma|\left(W\Pi_{k}W\Pi_{\ell}\otimes I\right)|\Gamma\rangle\\
 & =\langle\Gamma|\left(W\Pi_{k}W\otimes\Pi_{\ell}^{T}\right)|\Gamma\rangle\\
 & =\langle\Gamma|\left(W\otimes I\right)\left(\Pi_{k}\otimes\Pi_{\ell}^{T}\right)\left(W\otimes I\right)|\Gamma\rangle\\
 & =\langle W|\left(\Pi_{k}\otimes\Pi_{\ell}^{T}\right)|W\rangle.
\end{align}
Then the desired inequality in \eqref{eq:ordering-info-matrices-general}
follows because
\begin{align}
v^{T}I_{\zeta_{1}}(\theta)v & =\langle W|\left[\sum_{k,\ell}\frac{1}{\lambda_{\ell}f_{1}\!\left(\frac{\lambda_{k}}{\lambda_{\ell}}\right)}\left(\Pi_{k}\otimes\Pi_{\ell}^{T}\right)\right]|W\rangle\\
 & \leq\langle W|\left[\sum_{k,\ell}\frac{1}{\lambda_{\ell}f_{2}\!\left(\frac{\lambda_{k}}{\lambda_{\ell}}\right)}\left(\Pi_{k}\otimes\Pi_{\ell}^{T}\right)\right]|W\rangle\\
 & =v^{T}I_{\zeta_{2}}(\theta)v,
\end{align}
where we used the matrix inequality
\begin{equation}
\sum_{k,\ell}\frac{1}{\lambda_{\ell}f_{1}\!\left(\frac{\lambda_{k}}{\lambda_{\ell}}\right)}\left(\Pi_{k}\otimes\Pi_{\ell}^{T}\right)\leq\sum_{k,\ell}\frac{1}{\lambda_{\ell}f_{2}\!\left(\frac{\lambda_{k}}{\lambda_{\ell}}\right)}\left(\Pi_{k}\otimes\Pi_{\ell}^{T}\right),
\end{equation}
which follows from the assumption that $f_{1}(t)\geq f_{2}(t)$ for
all $t>0$.
\end{proof}
\begin{rem}
The expression in \eqref{eq:alt-form-metric-petz-K} is a rewriting
of \cite[Eqs.~(8) \& (11)]{Petz1996} that does not make use of the
left and right multiplication superoperators.
\end{rem}

\begin{cor}
\label{cor:gen-info-PSD}Let $\zeta\colon\mathbb{R}_{+}\times\mathbb{R}_{+}\to\mathbb{R}_{+}$
be a function that satisfies items 1--3 of Theorem~\ref{thm:q-fisher-info-char},
and define $f(t)\coloneqq\frac{1}{\zeta(t,1)}$ for all $t>0$. Then
the following matrix inequality holds for all $\theta\in\mathbb{R}^{L}$:
\begin{equation}
I_{\zeta}(\theta)\geq0.
\end{equation}
\end{cor}

\begin{proof}
By the same reasoning as in the proof of Lemma \ref{lem:ordering-for-info-matrices},
it follows that, for all $v\in\mathbb{R}^{L}$,
\begin{align}
v^{T}I_{\zeta_{1}}(\theta)v & =\langle W|\left[\sum_{k,\ell}\frac{1}{\lambda_{\ell}f\!\left(\frac{\lambda_{k}}{\lambda_{\ell}}\right)}\left(\Pi_{k}\otimes\Pi_{\ell}^{T}\right)\right]|W\rangle.
\end{align}
The matrix $\sum_{k,\ell}\left[\lambda_{\ell}f\!\left(\frac{\lambda_{k}}{\lambda_{\ell}}\right)\right]^{-1}\left(\Pi_{k}\otimes\Pi_{\ell}^{T}\right)$
is positive semi-definite because $\lambda_{\ell}>0$, $f\!\left(\frac{\lambda_{k}}{\lambda_{\ell}}\right)>0$
given that $\lambda_{k},\lambda_{\ell}>0$, and $\Pi_{k}\otimes\Pi_{\ell}^{T}\geq0$
for all $k$ and $\ell$. Then the inequality
\begin{equation}
\langle W|\left[\sum_{k,\ell}\frac{1}{\lambda_{\ell}f\!\left(\frac{\lambda_{k}}{\lambda_{\ell}}\right)}\left(\Pi_{k}\otimes\Pi_{\ell}^{T}\right)\right]|W\rangle\geq0
\end{equation}
holds for all $|W\rangle$, thus concluding the proof.
\end{proof}

\subsection{Ordering of Petz--R\'enyi information matrices}

\label{subsec:Ordering-of-Petz-Renyi}Recall that, for all $\alpha\in\left(0,1\right)\cup\left(1,\infty\right)$,
the Petz--R\'enyi information matrix has the form given in Corollary~\ref{cor:Petz-Renyi-special-case},
which we can rewrite as follows by employing \eqref{eq:gen-info-matrix-f-func}:
\begin{equation}
\left[\overline{I}_{\alpha}(\theta)\right]_{i,j}=\sum_{k,\ell}\frac{1}{\lambda_{\ell}\overline{f}\!\left(\frac{\lambda_{k}}{\lambda_{\ell}},\alpha\right)}\Tr\!\left[\Pi_{k}\left(\partial_{i}\rho(\theta)\right)\Pi_{\ell}\left(\partial_{j}\rho(\theta)\right)\right],
\end{equation}
where
\begin{equation}
\overline{f}(x,\alpha)\coloneqq\alpha\left(1-\alpha\right)\left(\frac{\left(x-1\right)^{2}}{\left(x^{\alpha}-1\right)\left(x^{1-\alpha}-1\right)}\right).
\end{equation}
As a consequence of Lemma \ref{lem:ordering-for-info-matrices} and
Lemma \ref{lem:key-lem-monotone-Petz-Renyi}, the following theorem
holds:
\begin{thm}
\label{thm:ordering-Petz-Renyi}For all $\theta\in\mathbb{R}^{L}$,
the matrix-valued function $\alpha\mapsto\overline{I}_{\alpha}(\theta)$
is monotone decreasing on $\alpha\in\left(0,\frac{1}{2}\right]$ and
monotone increasing on $\alpha\in\left[\frac{1}{2},\infty\right)$,
with respect to the Loewner order.
\end{thm}

\begin{lem}
\label{lem:key-lem-monotone-Petz-Renyi}For all $x>0$, the function
\begin{align}
\alpha & \mapsto\overline{f}(x,\alpha)\coloneqq\alpha\left(1-\alpha\right)\left(\frac{\left(x-1\right)^{2}}{\left(x^{\alpha}-1\right)\left(x^{1-\alpha}-1\right)}\right)\label{eq:petz-renyi-MC-function-f}
\end{align}
is monotone increasing on $\alpha\in\left(0,\frac{1}{2}\right]$ and
monotone decreasing on $\alpha\in\left[\frac{1}{2},\infty\right)$.
\end{lem}

\begin{proof}
The plan is to calculate the derivative of the function in \eqref{eq:petz-renyi-MC-function-f}
and show that it is non-negative for all $\alpha\in\left(0,\frac{1}{2}\right]$
and non-positive for all $\alpha\geq\frac{1}{2}$. To begin with,
consider that $\overline{f}(x,\alpha)>0$ for all $\alpha>0$ and
$x>0$ (one can verify this by considering various cases $\alpha\in\left(0,1\right)$,
$\alpha>1$, $x\in\left(0,1\right)$, and $x>1$). The derivative
of $\overline{f}(x,\alpha)$ with respect to $\alpha$ is given by
\begin{align}
 & \frac{\partial}{\partial\alpha}\overline{f}(x,\alpha)\nonumber \\
 & =\frac{\partial}{\partial\alpha}\left(\frac{\alpha\left(1-\alpha\right)\left(x-1\right)^{2}}{\left(x^{\alpha}-1\right)\left(x^{1-\alpha}-1\right)}\right)\\
 & =\frac{\frac{\partial}{\partial\alpha}\left[\alpha\left(1-\alpha\right)\right]\left(x-1\right)^{2}}{\left(x^{\alpha}-1\right)\left(x^{1-\alpha}-1\right)}-\left(\frac{\alpha\left(1-\alpha\right)\left(x-1\right)^{2}}{\left[\left(x^{\alpha}-1\right)\left(x^{1-\alpha}-1\right)\right]^{2}}\frac{\partial}{\partial\alpha}\left[\left(x^{\alpha}-1\right)\left(x^{1-\alpha}-1\right)\right]\right)\\
 & =\frac{\left(1-2\alpha\right)\left(x-1\right)^{2}}{\left(x^{\alpha}-1\right)\left(x^{1-\alpha}-1\right)}-\alpha\left(1-\alpha\right)\left(x-1\right)^{2}\frac{-x^{\alpha}\ln x+x^{1-\alpha}\ln x}{\left[\left(x^{\alpha}-1\right)\left(x^{1-\alpha}-1\right)\right]^{2}}\\
 & =\frac{\left(x-1\right)^{2}}{\left(x^{\alpha}-1\right)\left(x^{1-\alpha}-1\right)}\left[\left(1-2\alpha\right)-\alpha\left(1-\alpha\right)\frac{x^{1-\alpha}-x^{\alpha}}{\left(x^{\alpha}-1\right)\left(x^{1-\alpha}-1\right)}\ln x\right].
\end{align}
Let us analyze the term on the right in square brackets:
\begin{align}
 & \left(1-2\alpha\right)-\alpha\left(1-\alpha\right)\frac{x^{1-\alpha}-x^{\alpha}}{\left(x^{\alpha}-1\right)\left(x^{1-\alpha}-1\right)}\ln x\nonumber \\
 & =\left(1-2\alpha\right)-\alpha\left(1-\alpha\right)\frac{x^{1-\alpha}-1-\left(x^{\alpha}-1\right)}{\left(x^{\alpha}-1\right)\left(x^{1-\alpha}-1\right)}\ln x\\
 & =\left(1-2\alpha\right)-\alpha\left(1-\alpha\right)\left(\frac{1}{x^{\alpha}-1}-\frac{1}{x^{1-\alpha}-1}\right)\ln x\\
 & =1-\alpha-\alpha-\left(1-\alpha\right)\frac{\alpha\ln x}{x^{\alpha}-1}-\alpha\frac{\left(1-\alpha\right)\ln x}{x^{1-\alpha}-1}\\
 & =\left(1-\alpha\right)\left(1-\frac{\alpha\ln x}{x^{\alpha}-1}\right)-\alpha\left(1-\frac{\left(1-\alpha\right)\ln x}{x^{1-\alpha}-1}\right)\\
 & =\alpha\left(1-\alpha\right)\left[\frac{1}{\alpha}\left(1-\frac{\alpha\ln x}{x^{\alpha}-1}\right)-\frac{1}{1-\alpha}\left(1-\frac{\left(1-\alpha\right)\ln x}{x^{1-\alpha}-1}\right)\right].
\end{align}
Then it follows that
\begin{equation}
\frac{\partial}{\partial\alpha}\overline{f}(x,\alpha)=\overline{f}(x,\alpha)\left[\frac{1}{\alpha}\left(1-\frac{\alpha\ln x}{x^{\alpha}-1}\right)-\frac{1}{1-\alpha}\left(1-\frac{\left(1-\alpha\right)\ln x}{x^{1-\alpha}-1}\right)\right].
\end{equation}
Since $\overline{f}(x,\alpha)>0$ for all $x>0$ and $\alpha>0$,
it suffices to prove that
\begin{equation}
\frac{1}{\alpha}\left(1-\frac{\alpha\ln x}{x^{\alpha}-1}\right)\geq\frac{1}{1-\alpha}\left(1-\frac{\left(1-\alpha\right)\ln x}{x^{1-\alpha}-1}\right)\label{eq:revised-func-petz-renyi-inc}
\end{equation}
for all $\alpha\in\left(0,\frac{1}{2}\right]$ and that
\begin{equation}
\frac{1}{\alpha}\left(1-\frac{\alpha\ln x}{x^{\alpha}-1}\right)\leq\frac{1}{1-\alpha}\left(1-\frac{\left(1-\alpha\right)\ln x}{x^{1-\alpha}-1}\right)\label{eq:revised-func-petz-renyi-inc-1}
\end{equation}
for all $\alpha\geq\frac{1}{2}$. Since the function $\alpha\mapsto\frac{1}{\alpha}\left(1-\frac{\alpha\ln x}{x^{\alpha}-1}\right)$
is symmetric about $\alpha=\frac{1}{2}$ on $\alpha\in\left(0,1\right)$,
it suffices to prove the inequality in \eqref{eq:revised-func-petz-renyi-inc-1}
for all $\alpha\geq\frac{1}{2}$. The inequality in \eqref{eq:revised-func-petz-renyi-inc-1}
will follow if we prove that the function $\alpha\mapsto\frac{1}{\alpha}\left(1-\frac{\alpha\ln x}{x^{\alpha}-1}\right)$
is decreasing on $\alpha\in\left(-\infty,\infty\right)$. To do so,
let us evaluate the derivative of this function as follows and prove
that it is non-positive for all $x>0$:
\begin{align}
\frac{\partial}{\partial\alpha}\left(\frac{1}{\alpha}\left(1-\frac{\alpha\ln x}{x^{\alpha}-1}\right)\right) & =\frac{\partial}{\partial\alpha}\left(\frac{1}{\alpha}-\frac{\ln x}{x^{\alpha}-1}\right)\\
 & =-\frac{1}{\alpha^{2}}+\frac{\ln x}{\left(x^{\alpha}-1\right)^{2}}x^{\alpha}\ln x\\
 & =-\frac{1}{\alpha^{2}}+\left(\frac{\ln x}{x^{\alpha}-1}\right)^{2}x^{\alpha}
\end{align}
The desired inequality $\frac{\partial}{\partial\alpha}\left(\frac{1}{\alpha}\left(1-\frac{\alpha\ln x}{x^{\alpha}-1}\right)\right)\leq0$
is then equivalent to
\begin{align}
-\frac{1}{\alpha^{2}}+\left(\frac{\ln x}{x^{\alpha}-1}\right)^{2}x^{\alpha} & \leq0\\
\Longleftrightarrow\qquad\left(\frac{\ln x}{x^{\alpha}-1}\right)^{2}x^{\alpha} & \leq\frac{1}{\alpha^{2}}\\
\Longleftrightarrow\qquad\left(\ln x\right)^{2}\alpha^{2} & \leq\frac{\left(x^{\alpha}-1\right)^{2}}{x^{\alpha}}\\
 & =x^{\alpha}-2+x^{-\alpha}
\end{align}
Now set $x=e^{\lambda}$ for $\lambda\in\mathbb{R}$ and observe that
\begin{align}
\left(\ln x\right)^{2}\alpha^{2} & \leq x^{\alpha}-2+x^{-\alpha}\\
\Longleftrightarrow\qquad2+\left(\lambda\alpha\right)^{2} & \leq e^{\lambda\alpha}+e^{-\lambda\alpha}\\
\Longleftrightarrow\qquad1+\frac{1}{2}\left(\lambda\alpha\right)^{2} & \leq\cosh(\lambda\alpha)\\
 & =1+\frac{1}{2}\left(\lambda\alpha\right)^{2}+\sum_{k=2}^{\infty}\frac{\left(\lambda\alpha\right)^{2k}}{(2k)!}\\
\Longleftrightarrow\qquad0 & \leq\sum_{k=2}^{\infty}\frac{\left(\lambda\alpha\right)^{2k}}{(2k)!}.
\end{align}
Thus, the function $\alpha\mapsto\frac{1}{\alpha}\left(1-\frac{\alpha\ln x}{x^{\alpha}-1}\right)$
is indeed decreasing on $\alpha\in\left(-\infty,\infty\right)$ for
all $x>0$, concluding the proof.
\end{proof}

\subsection{Ordering of sandwiched R\'enyi information matrices}

\label{subsec:Ordering-of-sandwiched-Renyi}Recall that, for all $\alpha\in\left(0,1\right)\cup\left(1,\infty\right)$,
the sandwiched R\'enyi information matrix has the form given in Corollary~\ref{cor:sandwiched-Renyi-special-case},
which we can rewrite as follows by employing \eqref{eq:gen-info-matrix-f-func}:
\begin{equation}
\left[\widetilde{I}_{\alpha}(\theta)\right]_{i,j}=\sum_{k,\ell}\frac{1}{\lambda_{\ell}\widetilde{f}\!\left(\frac{\lambda_{k}}{\lambda_{\ell}},\alpha\right)}\Tr\!\left[\Pi_{k}\left(\partial_{i}\rho(\theta)\right)\Pi_{\ell}\left(\partial_{j}\rho(\theta)\right)\right],
\end{equation}
where
\begin{equation}
\widetilde{f}(x,\alpha)\coloneqq\left(1-\alpha\right)\left(\frac{x^{\frac{1}{\alpha}}-1}{x^{\frac{1-\alpha}{\alpha}}-1}\right).
\end{equation}
As a consequence of Lemma \ref{lem:ordering-for-info-matrices} and
Lemma \ref{lem:key-lem-monotone-Petz-Renyi}, the following theorem
holds:
\begin{thm}
\label{thm:ordering-sandwiched-Renyi}For all $\theta\in\mathbb{R}^{L}$,
the matrix-valued function $\alpha\mapsto\widetilde{I}_{\alpha}(\theta)$
is monotone increasing on $\alpha\in\left(0,\infty\right)$, with
respect to the Loewner order.
\end{thm}

\begin{lem}
For all $x>0$, the function 
\begin{align}
\alpha & \mapsto\widetilde{f}(x,\alpha)\coloneqq\left(1-\alpha\right)\left(\frac{x^{\frac{1}{\alpha}}-1}{x^{\frac{1-\alpha}{\alpha}}-1}\right)\label{eq:sandwiched-renyi-MC-function-f}
\end{align}
is monotone decreasing on $\alpha\in\left(0,\infty\right)$.
\end{lem}

\begin{proof}
The plan is to calculate the derivative of the function in \eqref{eq:sandwiched-renyi-MC-function-f}
and show that it is non-negative for all $\alpha\in\left(0,\infty\right)$.
To begin with, consider that the inequality $\widetilde{f}(x,\alpha)>0$
holds for all $\alpha>0$ and $x>0$ (one can verify this by considering
various cases $\alpha\in\left(0,1\right)$, $\alpha>1$, $x\in\left(0,1\right)$,
and $x>1$). The derivative of $\widetilde{f}(x,\alpha)$ with respect
to $\alpha$ is given by
\begin{align}
 & \frac{\partial}{\partial\alpha}\widetilde{f}(x,\alpha)\nonumber \\
 & =\frac{\partial}{\partial\alpha}\left(\left(1-\alpha\right)\left(\frac{x^{\frac{1}{\alpha}}-1}{x^{\frac{1-\alpha}{\alpha}}-1}\right)\right)\\
 & =-\left(\frac{x^{\frac{1}{\alpha}}-1}{x^{\frac{1-\alpha}{\alpha}}-1}\right)+\left(1-\alpha\right)\left(\frac{-\frac{1}{\alpha^{2}}x^{\frac{1}{\alpha}}\ln x}{x^{\frac{1-\alpha}{\alpha}}-1}\right)+\left(1-\alpha\right)\left(\frac{x^{\frac{1}{\alpha}}-1}{\left(x^{\frac{1-\alpha}{\alpha}}-1\right)^{2}}\right)\frac{1}{\alpha^{2}}x^{\frac{1-\alpha}{\alpha}}\ln x\\
 & =-\left(\frac{x^{\frac{1}{\alpha}}-1}{x^{\frac{1-\alpha}{\alpha}}-1}\right)-\left(1-\alpha\right)\left(\frac{x^{\frac{1}{\alpha}}\ln x}{\alpha^{2}\left(x^{\frac{1-\alpha}{\alpha}}-1\right)}\right)+\left(1-\alpha\right)\left(\frac{x^{\frac{1}{\alpha}}-1}{x^{\frac{1-\alpha}{\alpha}}-1}\right)\frac{x^{\frac{1-\alpha}{\alpha}}\ln x}{\alpha^{2}\left(x^{\frac{1-\alpha}{\alpha}}-1\right)}\\
 & =\frac{\left(1-\alpha\right)}{\alpha^{2}}\left(\frac{x^{\frac{1}{\alpha}}-1}{x^{\frac{1-\alpha}{\alpha}}-1}\right)\left[-\frac{\alpha^{2}}{1-\alpha}-\frac{x^{\frac{1}{\alpha}}\ln x}{x^{\frac{1}{\alpha}}-1}+\frac{x^{\frac{1-\alpha}{\alpha}}\ln x}{x^{\frac{1-\alpha}{\alpha}}-1}\right]\\
 & =\frac{\widetilde{f}(x,\alpha)}{\alpha^{2}}\left[\frac{1}{\frac{1}{\alpha}}-\frac{1}{\frac{1-\alpha}{\alpha}}-\frac{\ln x}{1-x^{-\frac{1}{\alpha}}}+\frac{\ln x}{1-x^{-\left(\frac{1-\alpha}{\alpha}\right)}}\right]\\
 & =\frac{\widetilde{f}(x,\alpha)}{\alpha^{2}}\left[-\left(\frac{\ln x}{1-x^{-\frac{1}{\alpha}}}-\frac{1}{\frac{1}{\alpha}}\right)+\frac{\ln x}{1-x^{-\left(\frac{1-\alpha}{\alpha}\right)}}-\frac{1}{\frac{1-\alpha}{\alpha}}\right].
\end{align}
Since $\frac{\widetilde{f}(x,\alpha)}{\alpha^{2}}>0$ for all $\alpha>0$
and $x>0$, the desired inequality $\frac{\partial}{\partial\alpha}\widetilde{f}(x,\alpha)\leq0$
will follow if we prove that
\begin{align}
\frac{\ln x}{1-x^{-\left(\frac{1-\alpha}{\alpha}\right)}}-\frac{1}{\frac{1-\alpha}{\alpha}} & \leq\frac{\ln x}{1-x^{-\frac{1}{\alpha}}}-\frac{1}{\frac{1}{\alpha}}\label{eq:sandwiched-monotone-ineq}
\end{align}
for all $\alpha\in\left(0,\infty\right)$. Defining $\gamma\coloneqq\frac{1-\alpha}{\alpha}$,
it follows that $\frac{1}{\alpha}=\gamma+1$, so that the inequality
in \eqref{eq:sandwiched-monotone-ineq} is equivalent to
\begin{equation}
\frac{\ln x}{1-x^{-\gamma}}-\frac{1}{\gamma}\leq\frac{\ln x}{1-x^{-\left(\gamma+1\right)}}-\frac{1}{\gamma+1}\label{eq:sandwiched-monotone-ineq-gamma}
\end{equation}
for all $\gamma\in\left(-\infty,\infty\right)$. The inequality in
\eqref{eq:sandwiched-monotone-ineq-gamma} will follow if we prove
that the function
\begin{equation}
\gamma\mapsto\frac{\ln x}{1-x^{-\gamma}}-\frac{1}{\gamma}
\end{equation}
is increasing on $\gamma\in\left(-\infty,\infty\right)$. Note that
$\lim_{\gamma\to0}\left(\frac{\ln x}{1-x^{-\gamma}}-\frac{1}{\gamma}\right)=\frac{\ln x}{2}$.
Taking the derivative of $\frac{\ln x}{1-x^{-\gamma}}-\frac{1}{\gamma}$
with respect to $\gamma$, we find that
\begin{align}
\frac{\partial}{\partial\gamma}\left(\frac{\ln x}{1-x^{-\gamma}}-\frac{1}{\gamma}\right) & =\frac{\ln x}{\left(1-x^{-\gamma}\right)^{2}}\left(-x^{-\gamma}\ln x\right)+\frac{1}{\gamma^{2}}\\
 & =-\frac{x^{-\gamma}\left(\ln x\right)^{2}}{\left(1-x^{-\gamma}\right)^{2}}+\frac{1}{\gamma^{2}}.
\end{align}
Note that $\lim_{\gamma\to0}\left(-\frac{x^{-\gamma}\left(\ln x\right)^{2}}{\left(1-x^{-\gamma}\right)^{2}}+\frac{1}{\gamma^{2}}\right)=\frac{\left(\ln x\right)^{2}}{12}$,
so that $\frac{\partial}{\partial\gamma}\left(\frac{\ln x}{1-x^{-\gamma}}-\frac{1}{\gamma}\right)>0$
at $\gamma=0$ for all $x>0$. Setting $x=e^{\lambda}$ for $\lambda\in\mathbb{R}$,
the inequality $\frac{\partial}{\partial\gamma}\left(\frac{\ln x}{1-x^{-\gamma}}-\frac{1}{\gamma}\right)\geq0$
for all $\gamma\in\left(-\infty,0\right)\cup\left(0,\infty\right)$
is equivalent to
\begin{align}
-\frac{x^{-\gamma}\left(\ln x\right)^{2}}{\left(1-x^{-\gamma}\right)^{2}}+\frac{1}{\gamma^{2}} & \geq0\\
\Longleftrightarrow\qquad\frac{1}{\gamma^{2}} & \geq\frac{e^{-\lambda\gamma}\left(\lambda\right)^{2}}{\left(1-e^{-\lambda\gamma}\right)^{2}}\\
\Longleftrightarrow\qquad\frac{\left(1-e^{-\lambda\gamma}\right)^{2}}{e^{-\lambda\gamma}} & \geq\left(\gamma\lambda\right)^{2}\\
\Longleftrightarrow\qquad e^{\lambda\gamma}\left(1-2e^{-\lambda\gamma}+e^{-2\lambda\gamma}\right) & \geq\left(\gamma\lambda\right)^{2}\\
\Longleftrightarrow\qquad e^{\lambda\gamma}+e^{-\lambda\gamma} & \geq2+\left(\gamma\lambda\right)^{2}\\
\Longleftrightarrow\qquad\cosh(\lambda\gamma) & \geq1+\frac{1}{2}\left(\gamma\lambda\right)^{2}\\
\Longleftrightarrow\qquad1+\frac{1}{2}\left(\lambda\gamma\right)^{2}+\sum_{k=2}^{\infty}\frac{\left(\lambda\gamma\right)^{2k}}{(2k)!} & \geq1+\frac{1}{2}\left(\gamma\lambda\right)^{2}\\
\Longleftrightarrow\qquad\sum_{k=2}^{\infty}\frac{\left(\lambda\gamma\right)^{2k}}{(2k)!} & \geq0,
\end{align}
thus concluding the proof.
\end{proof}

\section{$\alpha$-$z$ Information matrices of parameterized thermal states}

\label{sec:a-z-Information-matrices-thermal-states}

In this section, I establish a formula for the $\alpha$-$z$ information
matrix of parameterized thermal states, for all $\alpha\in\left(0,1\right)$
and $z>0$ (Theorem~\ref{thm:QBM-a-z-formula}). Like the former
results of \cite{Patel2024,Minervini2025}, this formula leads to
a hybrid quantum--classical algorithm for estimating the elements
of the $\alpha$-$z$ information matrix of parameterized thermal
states, assuming that one has the ability to prepare thermal states
on a quantum computer. As such, this formula has applications in quantum
Boltzmann machine learning, namely, in natural gradient descent algorithms
for performing optimization using quantum Boltzmann machines, as put
forward in \cite{Patel2024,Minervini2025}.

Let $\left(\rho(\theta)\right)_{\theta\in\mathbb{R}^{L}}$ be a parameterized
family of positive definite states. A quantum generalization of the
Fisher information matrix is denoted by $I_{\zeta}(\theta)$, and
it has the form stated in Theorem~\ref{thm:q-fisher-info-char},
wherein the properties of the function $\zeta(x,y)$ are stated. For
convenience, let us restate the form of $I_{\zeta}(\theta)$ here:
\begin{equation}
\left[I_{\zeta}(\theta)\right]_{i,j}=\sum_{k,\ell}\zeta(\lambda_{k},\lambda_{\ell})\Tr\!\left[\Pi_{k}\left(\partial_{i}\rho(\theta)\right)\Pi_{\ell}\left(\partial_{j}\rho(\theta)\right)\right],\label{eq:gen-fisher-info-repeated}
\end{equation}
where $\rho(\theta)=\sum_{k}\lambda_{k}\Pi_{k}$ is a spectral decomposition
of $\rho(\theta)$ and $\zeta(x,y)$ is a function of $x,y>0$ that
satisfies the following properties for all $x,y,s>0$:
\begin{align}
\zeta(x,y) & =\zeta(y,x),\label{eq:zeta-func-prop-1-later}\\
\zeta(sx,sy) & =\frac{1}{s}\zeta(x,y),\\
\zeta(x,x) & =\frac{\kappa}{x}.\label{eq:zeta-kappa-normalization}
\end{align}
with $\kappa>0$ a constant. Also, the function $t\mapsto\frac{1}{\zeta(t,1)}$
is operator monotone on $t\in\left(0,\infty\right)$.

Let us now evaluate the expression in \eqref{eq:gen-fisher-info-repeated}
for parameterized thermal states. In this case, let
\begin{align}
H(\theta) & \coloneqq\sum_{j}\theta_{j}H_{j},\\
\rho(\theta) & \coloneqq\frac{e^{-H(\theta)}}{Z(\theta)},\label{eq:param-thermal-state}\\
Z(\theta) & \coloneqq\Tr\!\left[e^{-H(\theta)}\right],
\end{align}
where $\theta_{j}\in\mathbb{R}$ and $H_{j}$ is Hermitian for all
$j\in\left\{ 1,\ldots,L\right\} $. As a consequence of Lemma \ref{lem:general-fisher-thermal-states}
and Lemma \ref{lem:fourier-trans-a-z}, the following theorem holds
for the $\alpha$-$z$ information matrices of parameterized thermal
states:
\begin{thm}
\label{thm:QBM-a-z-formula}For parameterized thermal states of the
form in \eqref{eq:param-thermal-state}, the following equality holds
for all $\alpha\in\left(0,1\right)$ and $z>0$:
\begin{equation}
\left[I_{\alpha,z}(\theta)\right]_{i,j}=\frac{1}{2}\left\langle \left\{ \Phi_{q_{\alpha,z},\theta}\!\left(H_{i}\right),H_{j}\right\} \right\rangle _{\rho(\theta)}-\left\langle H_{i}\right\rangle _{\rho(\theta)}\left\langle H_{j}\right\rangle _{\rho(\theta)},
\end{equation}
where the channel $\Phi_{q_{\alpha,z},\theta}$ is given by
\begin{equation}
\Phi_{q_{\alpha,z},\theta}\!\left(X\right)\coloneqq\int_{-\infty}^{\infty}dt\:q_{\alpha,z}(t)e^{-itH(\theta)}Xe^{itH(\theta)},
\end{equation}
and $q_{\alpha,z}$ is the probability density function defined in
\eqref{eq:q-a-z-prob-dens}.
\end{thm}

As stated above, Theorem~\ref{thm:QBM-a-z-formula} follows in part
from Lemma \ref{lem:general-fisher-thermal-states}, which establishes
a formula for the information matrix $I_{\zeta}(\theta)$ of parameterized
thermal states, under the assumption that a certain Fourier transform
exists. Additionally, Theorem~\ref{thm:QBM-a-z-formula} follows
from Lemma \ref{lem:fourier-trans-a-z}, which precisely determines
the needed Fourier transform for all $\alpha\in\left(0,1\right)$
and $z>0$ when the function $\zeta(x,y)$ is set to $\zeta_{\alpha,z}\!\left(x,y\right)$,
as defined in \eqref{eq:limit-for-a-z-eigenval-func}.

\subsection{General formula for information matrices of parameterized thermal
states}
\begin{lem}
\label{lem:general-fisher-thermal-states}Let $\zeta(x,y)$ be a function
defined for $x,y>0$ satisfying the properties stated in \eqref{eq:zeta-func-prop-1-later}--\eqref{eq:zeta-kappa-normalization},
and suppose that the Fourier transform of the function
\begin{equation}
\omega\mapsto\zeta\!\left(e^{-\omega},1\right)\frac{\left(e^{-\omega}-1\right)^{2}}{\omega^{2}\left(e^{-\omega}+1\right)}
\end{equation}
exists, where $\omega\in\mathbb{R}$. Then, for parameterized thermal
states of the form in \eqref{eq:param-thermal-state}, the following
equality holds:
\begin{equation}
\left[I_{\zeta}(\theta)\right]_{i,j}=\frac{1}{2}\left\langle \left\{ \Phi_{f,\theta}\!\left(H_{i}\right),H_{j}\right\} \right\rangle _{\rho(\theta)}-\kappa\left\langle H_{i}\right\rangle _{\rho(\theta)}\left\langle H_{j}\right\rangle _{\rho(\theta)},
\end{equation}
where $I_{\zeta}(\theta)$ is defined in \eqref{eq:gen-fisher-info-repeated},
\begin{equation}
\Phi_{f,\theta}\!\left(X\right)\coloneqq\int_{-\infty}^{\infty}dt\:f(t)e^{-itH(\theta)}Xe^{itH(\theta)},
\end{equation}
and $f(t)$ is a real-valued function satisfying the following for
all $\omega\in\mathbb{R}$:
\begin{align}
\int_{-\infty}^{\infty}dt\:f(t)e^{it\omega} & =2\,\zeta\!\left(e^{-\omega},1\right)\frac{\left(e^{-\omega}-1\right)^{2}}{\omega^{2}\left(e^{-\omega}+1\right)},\\
\int_{-\infty}^{\infty}dt\:f(t) & =\kappa.\label{eq:normalization-f-fourier}
\end{align}
\end{lem}

\begin{proof}
Consider that 
\begin{align}
\frac{\partial}{\partial\theta_{i}}\rho(\theta) & =\frac{\partial}{\partial\theta_{i}}\left(\frac{e^{-H(\theta)}}{Z(\theta)}\right)\\
 & =\frac{\frac{\partial}{\partial\theta_{i}}e^{-H(\theta)}}{Z(\theta)}-\frac{e^{-H(\theta)}}{Z(\theta)^{2}}\frac{\partial}{\partial\theta_{i}}Z(\theta)\\
 & =\frac{1}{Z(\theta)}\left(\frac{\partial}{\partial\theta_{i}}e^{-H(\theta)}-\rho(\theta)\frac{\partial}{\partial\theta_{i}}\Tr\!\left[e^{-H(\theta)}\right]\right)\\
 & =\frac{1}{Z(\theta)}\left(\frac{\partial}{\partial\theta_{i}}e^{-H(\theta)}+\rho(\theta)\Tr\!\left[e^{-H(\theta)}\frac{\partial}{\partial\theta_{i}}H(\theta)\right]\right)\\
 & =\frac{1}{Z(\theta)}\left(\frac{\partial}{\partial\theta_{i}}e^{-H(\theta)}+\rho(\theta)\Tr\!\left[e^{-H(\theta)}H_{i}\right]\right)\\
 & =\frac{1}{Z(\theta)}\left(\frac{\partial}{\partial\theta_{i}}e^{-H(\theta)}\right)+\rho(\theta)\Tr\!\left[H_{i}\rho(\theta)\right]\\
 & =\frac{1}{Z(\theta)}\left(\frac{\partial}{\partial\theta_{i}}e^{-H(\theta)}\right)+\rho(\theta)\left\langle H_{i}\right\rangle _{\rho(\theta)},\label{eq:fisher-thermal-general-deriv-1}
\end{align}
where the fourth equality follows from Corollary~\ref{cor:derivative-in-trace}.
Now applying Proposition \ref{prop:deriv-exp}, we find that
\begin{align}
\frac{\partial}{\partial\theta_{i}}e^{-H(\theta)} & =-\int_{0}^{1}dt\ e^{-tH(\theta)}\left(\frac{\partial}{\partial\theta_{i}}H(\theta)\right)e^{-\left(1-t\right)H(\theta)}\\
 & =-\int_{0}^{1}dt\ e^{-tH(\theta)}H_{i}e^{-\left(1-t\right)H(\theta)},
\end{align}
which implies from \eqref{eq:fisher-thermal-general-deriv-1} that
\begin{align}
\frac{\partial}{\partial\theta_{i}}\rho(\theta) & =-\frac{1}{Z(\theta)}\int_{0}^{1}dt\ e^{-tH(\theta)}H_{i}e^{-\left(1-t\right)H(\theta)}+\rho(\theta)\left\langle H_{i}\right\rangle _{\rho(\theta)}\\
 & =-\int_{0}^{1}dt\ \left(\frac{e^{-H(\theta)}}{Z(\theta)}\right)^{t}H_{i}\left(\frac{e^{-H(\theta)}}{Z(\theta)}\right)^{1-t}+\rho(\theta)\left\langle H_{i}\right\rangle _{\rho(\theta)}\\
 & =-\int_{0}^{1}dt\ \rho(\theta)^{t}H_{i}\rho(\theta)^{1-t}+\rho(\theta)\left\langle H_{i}\right\rangle _{\rho(\theta)}.
\end{align}
Then it follows that
\begin{align}
\left[I_{\zeta}(\theta)\right]_{i,j} & =\sum_{k,\ell}\zeta(\lambda_{k},\lambda_{\ell})\Tr\!\left[\Pi_{k}\left(\partial_{i}\rho(\theta)\right)\Pi_{\ell}\left(\partial_{j}\rho(\theta)\right)\right]\\
 & =\sum_{k,\ell}\zeta(\lambda_{k},\lambda_{\ell})\Tr\!\left[\Pi_{k}\left(-\int_{0}^{1}dt\ \rho(\theta)^{t}H_{i}\rho(\theta)^{1-t}\right)\Pi_{\ell}\left(-\int_{0}^{1}ds\ \rho(\theta)^{s}H_{j}\rho(\theta)^{1-s}\right)\right]\nonumber \\
 & \qquad+\sum_{k,\ell}\zeta(\lambda_{k},\lambda_{\ell})\Tr\!\left[\Pi_{k}\left(-\int_{0}^{1}dt\ \rho(\theta)^{t}H_{i}\rho(\theta)^{1-t}\right)\Pi_{\ell}\left(\rho(\theta)\left\langle H_{j}\right\rangle _{\rho(\theta)}\right)\right]\nonumber \\
 & \qquad+\sum_{k,\ell}\zeta(\lambda_{k},\lambda_{\ell})\Tr\!\left[\Pi_{k}\left(\rho(\theta)\left\langle H_{i}\right\rangle _{\rho(\theta)}\right)\Pi_{\ell}\left(-\int_{0}^{1}ds\ \rho(\theta)^{s}H_{j}\rho(\theta)^{1-s}\right)\right]\nonumber \\
 & \qquad+\sum_{k,\ell}\zeta(\lambda_{k},\lambda_{\ell})\Tr\!\left[\Pi_{k}\left(\rho(\theta)\left\langle H_{i}\right\rangle _{\rho(\theta)}\right)\Pi_{\ell}\left(\rho(\theta)\left\langle H_{j}\right\rangle _{\rho(\theta)}\right)\right]\\
 & =\sum_{k,\ell}\zeta(\lambda_{k},\lambda_{\ell})\int_{0}^{1}\int_{0}^{1}dt\ ds\ \Tr\!\left[\Pi_{k}\rho(\theta)^{t}H_{i}\rho(\theta)^{1-t}\Pi_{\ell}\rho(\theta)^{s}H_{j}\rho(\theta)^{1-s}\right]\nonumber \\
 & \qquad-\left\langle H_{j}\right\rangle _{\rho(\theta)}\sum_{k,\ell}\zeta(\lambda_{k},\lambda_{\ell})\int_{0}^{1}dt\ \Tr\!\left[\Pi_{k}\rho(\theta)^{t}H_{i}\rho(\theta)^{1-t}\Pi_{\ell}\rho(\theta)\right]\nonumber \\
 & \qquad-\left\langle H_{i}\right\rangle _{\rho(\theta)}\sum_{k,\ell}\zeta(\lambda_{k},\lambda_{\ell})\int_{0}^{1}ds\ \Tr\!\left[\Pi_{k}\rho(\theta)\Pi_{\ell}\rho(\theta)^{s}H_{j}\rho(\theta)^{1-s}\right]\nonumber \\
 & \qquad+\left\langle H_{i}\right\rangle _{\rho(\theta)}\left\langle H_{j}\right\rangle _{\rho(\theta)}\sum_{k,\ell}\zeta(\lambda_{k},\lambda_{\ell})\Tr\!\left[\Pi_{k}\rho(\theta)\Pi_{\ell}\rho(\theta)\right]\\
 & =\sum_{k,\ell}\zeta(\lambda_{k},\lambda_{\ell})\int_{0}^{1}\int_{0}^{1}dt\ ds\ \Tr\!\left[\Pi_{k}\lambda_{k}^{1-s+t}H_{i}\lambda_{\ell}^{1-t+s}\Pi_{\ell}H_{j}\right]\nonumber \\
 & \qquad-\left\langle H_{j}\right\rangle _{\rho(\theta)}\sum_{k,\ell}\zeta(\lambda_{k},\lambda_{\ell})\int_{0}^{1}dt\ \Tr\!\left[\Pi_{k}\lambda_{k}^{t}H_{i}\Pi_{\ell}\lambda_{\ell}^{2-t}\right]\nonumber \\
 & \qquad-\left\langle H_{i}\right\rangle _{\rho(\theta)}\sum_{k,\ell}\zeta(\lambda_{k},\lambda_{\ell})\int_{0}^{1}ds\ \Tr\!\left[\Pi_{k}\lambda_{k}^{2-s}\Pi_{\ell}\lambda_{\ell}^{s}H_{j}\right]\nonumber \\
 & \qquad+\left\langle H_{i}\right\rangle _{\rho(\theta)}\left\langle H_{j}\right\rangle _{\rho(\theta)}\sum_{k,\ell}\zeta(\lambda_{k},\lambda_{\ell})\Tr\!\left[\Pi_{k}\Pi_{\ell}\lambda_{\ell}^{2}\right]\\
 & =\sum_{k,\ell}\zeta(\lambda_{k},\lambda_{\ell})\lambda_{k}\lambda_{\ell}\int_{0}^{1}\int_{0}^{1}dt\ ds\ \left(\frac{\lambda_{k}}{\lambda_{\ell}}\right)^{t-s}\Tr\!\left[\Pi_{k}H_{i}\Pi_{\ell}H_{j}\right]\nonumber \\
 & \qquad-\left\langle H_{j}\right\rangle _{\rho(\theta)}\sum_{k,\ell}\zeta(\lambda_{k},\lambda_{\ell})\int_{0}^{1}dt\ \lambda_{k}^{t}\lambda_{\ell}^{2-t}\Tr\!\left[\Pi_{\ell}\Pi_{k}H_{i}\right]\nonumber \\
 & \qquad-\left\langle H_{i}\right\rangle _{\rho(\theta)}\sum_{k,\ell}\zeta(\lambda_{k},\lambda_{\ell})\int_{0}^{1}ds\ \lambda_{k}^{2-s}\lambda_{\ell}^{s}\Tr\!\left[\Pi_{k}\Pi_{\ell}H_{j}\right]\nonumber \\
 & \qquad+\left\langle H_{i}\right\rangle _{\rho(\theta)}\left\langle H_{j}\right\rangle _{\rho(\theta)}\sum_{k,\ell}\zeta(\lambda_{k},\lambda_{\ell})\lambda_{\ell}^{2}\Tr\!\left[\Pi_{k}\Pi_{\ell}\right]\\
 & =\sum_{k,\ell}\zeta(\lambda_{k},\lambda_{\ell})\lambda_{k}\lambda_{\ell}\int_{0}^{1}\int_{0}^{1}dt\ ds\ \left(\frac{\lambda_{k}}{\lambda_{\ell}}\right)^{t-s}\Tr\!\left[\Pi_{k}H_{i}\Pi_{\ell}H_{j}\right]\nonumber \\
 & \qquad-\left\langle H_{j}\right\rangle _{\rho(\theta)}\sum_{k}\zeta(\lambda_{k},\lambda_{k})\int_{0}^{1}dt\ \lambda_{k}^{t}\lambda_{k}^{2-t}\Tr\!\left[\Pi_{k}H_{i}\right]\nonumber \\
 & \qquad-\left\langle H_{i}\right\rangle _{\rho(\theta)}\sum_{k}\zeta(\lambda_{k},\lambda_{k})\int_{0}^{1}ds\ \lambda_{k}^{2-s}\lambda_{k}^{s}\Tr\!\left[\Pi_{k}H_{j}\right]\nonumber \\
 & \qquad+\left\langle H_{i}\right\rangle _{\rho(\theta)}\left\langle H_{j}\right\rangle _{\rho(\theta)}\sum_{k}\zeta(\lambda_{k},\lambda_{k})\lambda_{k}^{2}\Tr\!\left[\Pi_{k}\right]\\
 & =\sum_{k,\ell}\zeta(\lambda_{k},\lambda_{\ell})\lambda_{k}\lambda_{\ell}\int_{0}^{1}\int_{0}^{1}dt\ ds\ \left(\frac{\lambda_{k}}{\lambda_{\ell}}\right)^{t-s}\Tr\!\left[\Pi_{k}H_{i}\Pi_{\ell}H_{j}\right]\nonumber \\
 & \qquad-\left\langle H_{j}\right\rangle _{\rho(\theta)}\sum_{k}\frac{\kappa}{\lambda_{k}}\lambda_{k}^{2}\Tr\!\left[\Pi_{k}H_{i}\right]\nonumber \\
 & \qquad-\left\langle H_{i}\right\rangle _{\rho(\theta)}\sum_{k}\frac{\kappa}{\lambda_{k}}\lambda_{k}^{2}\Tr\!\left[\Pi_{k}H_{j}\right]\nonumber \\
 & \qquad+\left\langle H_{i}\right\rangle _{\rho(\theta)}\left\langle H_{j}\right\rangle _{\rho(\theta)}\sum_{k}\frac{\kappa}{\lambda_{k}}\lambda_{k}^{2}\Tr\!\left[\Pi_{k}\right]\\
 & =\sum_{k,\ell}\zeta(\lambda_{k},\lambda_{\ell})\lambda_{k}\lambda_{\ell}\int_{0}^{1}\int_{0}^{1}dt\ ds\ \left(\frac{\lambda_{k}}{\lambda_{\ell}}\right)^{t-s}\Tr\!\left[\Pi_{k}H_{i}\Pi_{\ell}H_{j}\right]\nonumber \\
 & \qquad-\kappa\left\langle H_{j}\right\rangle _{\rho(\theta)}\Tr\!\left[\sum_{k}\lambda_{k}\Pi_{k}H_{i}\right]-\kappa\left\langle H_{i}\right\rangle _{\rho(\theta)}\Tr\!\left[\sum_{k}\lambda_{k}\Pi_{k}H_{j}\right]\nonumber \\
 & \qquad+\kappa\left\langle H_{i}\right\rangle _{\rho(\theta)}\left\langle H_{j}\right\rangle _{\rho(\theta)}\Tr\!\left[\sum_{k}\lambda_{k}\Pi_{k}\right]\\
 & =\sum_{k,\ell}\zeta(\lambda_{k},\lambda_{\ell})\lambda_{k}\lambda_{\ell}\int_{0}^{1}\int_{0}^{1}dt\ ds\ \left(\frac{\lambda_{k}}{\lambda_{\ell}}\right)^{t-s}\Tr\!\left[\Pi_{k}H_{i}\Pi_{\ell}H_{j}\right]\nonumber \\
 & \qquad-\kappa\left\langle H_{j}\right\rangle _{\rho(\theta)}\left\langle H_{i}\right\rangle _{\rho(\theta)}-\kappa\left\langle H_{i}\right\rangle _{\rho(\theta)}\left\langle H_{j}\right\rangle _{\rho(\theta)}+\kappa\left\langle H_{i}\right\rangle _{\rho(\theta)}\left\langle H_{j}\right\rangle _{\rho(\theta)}\\
 & =\sum_{k,\ell}\zeta(\lambda_{k},\lambda_{\ell})\lambda_{k}\lambda_{\ell}\int_{0}^{1}\int_{0}^{1}dt\ ds\ \left(\frac{\lambda_{k}}{\lambda_{\ell}}\right)^{t-s}\Tr\!\left[\Pi_{k}H_{i}\Pi_{\ell}H_{j}\right]\nonumber \\
 & \qquad-\kappa\left\langle H_{i}\right\rangle _{\rho(\theta)}\left\langle H_{j}\right\rangle _{\rho(\theta)}.
\end{align}
Thus, it remains to evaluate the first term in the last expression
above. To this end, considering that $\rho(\theta)$ and $H(\theta)$
commute, they have the same eigenprojections, so that the spectral
decomposition of $H(\theta)$ is given by $\sum_{k}\mu_{k}\Pi_{k}$
and that of $\rho(\theta)$ is given by $\rho(\theta)=\sum_{k}\lambda_{k}\Pi_{k}=\sum_{k}\frac{e^{-\mu_{k}}}{Z}\Pi_{k}$,
where $Z\equiv Z(\theta)$. Then consider that
\begin{align}
 & \sum_{k,\ell}\zeta(\lambda_{k},\lambda_{\ell})\lambda_{k}\lambda_{\ell}\int_{0}^{1}\int_{0}^{1}dt\ ds\ \left(\frac{\lambda_{k}}{\lambda_{\ell}}\right)^{t-s}\Tr\!\left[\Pi_{k}H_{i}\Pi_{\ell}H_{j}\right]\nonumber \\
 & =\sum_{k,\ell}\zeta\!\left(\frac{e^{-\mu_{k}}}{Z},\frac{e^{-\mu_{\ell}}}{Z}\right)\frac{e^{-\mu_{k}}}{Z}\frac{e^{-\mu_{\ell}}}{Z}\int_{0}^{1}\int_{0}^{1}dt\ ds\ \left(e^{-\left(\mu_{k}-\mu_{\ell}\right)}\right)^{t-s}\Tr\!\left[\Pi_{k}H_{i}\Pi_{\ell}H_{j}\right]\\
 & =\sum_{k,\ell}\zeta\!\left(e^{-\mu_{k}},e^{-\mu_{\ell}}\right)e^{-\mu_{k}}\frac{e^{-\mu_{\ell}}}{Z}\int_{0}^{1}\int_{0}^{1}dt\ ds\ e^{-t\left(\mu_{k}-\mu_{\ell}\right)}e^{s\left(\mu_{k}-\mu_{\ell}\right)}\Tr\!\left[\Pi_{k}H_{i}\Pi_{\ell}H_{j}\right]\\
 & =\sum_{k,\ell}\zeta\!\left(e^{-\left(\mu_{k}-\mu_{\ell}\right)},1\right)\frac{e^{-\mu_{k}}}{Z}\left(\int_{0}^{1}dt\ e^{-t\left(\mu_{k}-\mu_{\ell}\right)}\right)\left(\int_{0}^{1}ds\ e^{s\left(\mu_{k}-\mu_{\ell}\right)}\right)\Tr\!\left[\Pi_{k}H_{i}\Pi_{\ell}H_{j}\right].
\end{align}
Observing that
\begin{align}
\int_{0}^{1}dt\ e^{-t\left(\mu_{k}-\mu_{\ell}\right)} & =\left.\frac{e^{-t\left(\mu_{k}-\mu_{\ell}\right)}}{-\left(\mu_{k}-\mu_{\ell}\right)}\right|_{0}^{1}\\
 & =\frac{e^{-\left(\mu_{k}-\mu_{\ell}\right)}}{-\left(\mu_{k}-\mu_{\ell}\right)}-\frac{1}{-\left(\mu_{k}-\mu_{\ell}\right)}\\
 & =\frac{1-e^{-\left(\mu_{k}-\mu_{\ell}\right)}}{\mu_{k}-\mu_{\ell}},\\
\int_{0}^{1}ds\ e^{s\left(\mu_{k}-\mu_{\ell}\right)} & =\left.\frac{e^{s\left(\mu_{k}-\mu_{\ell}\right)}}{\mu_{k}-\mu_{\ell}}\right|_{0}^{1}\\
 & =\frac{e^{\mu_{k}-\mu_{\ell}}}{\mu_{k}-\mu_{\ell}}-\frac{1}{\mu_{k}-\mu_{\ell}}\\
 & =\frac{e^{\mu_{k}-\mu_{\ell}}-1}{\mu_{k}-\mu_{\ell}},
\end{align}
we conclude that
\begin{align}
 & \sum_{k,\ell}\zeta\!\left(e^{-\left(\mu_{k}-\mu_{\ell}\right)},1\right)\frac{e^{-\mu_{k}}}{Z}\left(\int_{0}^{1}dt\ e^{-t\left(\mu_{k}-\mu_{\ell}\right)}\right)\left(\int_{0}^{1}ds\ e^{s\left(\mu_{k}-\mu_{\ell}\right)}\right)\Tr\!\left[\Pi_{k}H_{i}\Pi_{\ell}H_{j}\right]\nonumber \\
 & =\sum_{k,\ell}\zeta\!\left(e^{-\left(\mu_{k}-\mu_{\ell}\right)},1\right)\frac{e^{-\mu_{k}}}{Z}\left(\frac{1-e^{-\left(\mu_{k}-\mu_{\ell}\right)}}{\mu_{k}-\mu_{\ell}}\right)\left(\frac{e^{\mu_{k}-\mu_{\ell}}-1}{\mu_{k}-\mu_{\ell}}\right)\Tr\!\left[\Pi_{k}H_{i}\Pi_{\ell}H_{j}\right]\\
 & =\sum_{k,\ell}\zeta\!\left(e^{-\left(\mu_{k}-\mu_{\ell}\right)},1\right)\frac{e^{-\mu_{k}}}{Z}e^{\mu_{k}-\mu_{\ell}}\left(\frac{1-e^{-\left(\mu_{k}-\mu_{\ell}\right)}}{\mu_{k}-\mu_{\ell}}\right)\left(\frac{1-e^{-\left(\mu_{k}-\mu_{\ell}\right)}}{\mu_{k}-\mu_{\ell}}\right)\Tr\!\left[\Pi_{k}H_{i}\Pi_{\ell}H_{j}\right]\\
 & =\sum_{k,\ell}\zeta\!\left(e^{-\left(\mu_{k}-\mu_{\ell}\right)},1\right)\frac{e^{-\mu_{\ell}}}{Z}\left(\frac{1-e^{-\left(\mu_{k}-\mu_{\ell}\right)}}{\mu_{k}-\mu_{\ell}}\right)^{2}\Tr\!\left[\Pi_{k}H_{i}\Pi_{\ell}H_{j}\right]\\
 & =\sum_{k,\ell}\zeta\!\left(e^{-\left(\mu_{k}-\mu_{\ell}\right)},1\right)\frac{e^{-\mu_{\ell}}}{Z}\left(\frac{e^{-\left(\mu_{k}-\mu_{\ell}\right)}-1}{\mu_{k}-\mu_{\ell}}\right)^{2}\left(\frac{e^{-\left(\mu_{k}-\mu_{\ell}\right)}+1}{e^{-\left(\mu_{k}-\mu_{\ell}\right)}+1}\right)\Tr\!\left[\Pi_{k}H_{i}\Pi_{\ell}H_{j}\right]\\
 & =\frac{1}{2}\sum_{k,\ell}2\,\zeta\!\left(e^{-\left(\mu_{k}-\mu_{\ell}\right)},1\right)\left(\frac{e^{-\left(\mu_{k}-\mu_{\ell}\right)}-1}{\mu_{k}-\mu_{\ell}}\right)^{2}\left(\frac{1}{e^{-\left(\mu_{k}-\mu_{\ell}\right)}+1}\right)\left(\frac{e^{-\mu_{k}}}{Z}+\frac{e^{-\mu_{\ell}}}{Z}\right)\Tr\!\left[\Pi_{k}H_{i}\Pi_{\ell}H_{j}\right].
\end{align}
Now suppose that $f(t)$ is a function satisfying the following Fourier
transform relation:
\begin{equation}
\int_{-\infty}^{\infty}dt\:f(t)e^{it\omega}=2\,\zeta\!\left(e^{-\omega},1\right)\frac{\left(e^{-\omega}-1\right)^{2}}{\omega^{2}\left(e^{-\omega}+1\right)}.
\end{equation}
To establish that $f(t)$ is real-valued, it suffices to prove that
$\omega\mapsto2\,\zeta\!\left(e^{-\omega},1\right)\frac{\left(e^{-\omega}-1\right)^{2}}{\omega^{2}\left(e^{-\omega}+1\right)}$
is an even function. To this end, consider that
\begin{align}
-\omega & \mapsto2\,\zeta\!\left(e^{-\left(-\omega\right)},1\right)\frac{\left(e^{-\left(-\omega\right)}-1\right)^{2}}{\left(-\omega\right)^{2}\left(e^{-\left(-\omega\right)}+1\right)}\\
 & =2\,\zeta\!\left(e^{\omega},1\right)\frac{\left(e^{\omega}-1\right)^{2}}{\omega^{2}\left(e^{\omega}+1\right)}\\
 & =2\,\zeta\!\left(1,e^{\omega}\right)\frac{\left(e^{\omega}-1\right)^{2}}{\omega^{2}\left(e^{\omega}+1\right)}\\
 & =2\,e^{-\omega}\zeta\!\left(e^{-\omega},1\right)\frac{\left(e^{\omega}-1\right)^{2}}{\omega^{2}\left(e^{\omega}+1\right)}\\
 & =2\,\zeta\!\left(e^{-\omega},1\right)\frac{e^{-\omega}\left(e^{\omega}-1\right)e^{-\omega}\left(e^{\omega}-1\right)}{e^{-\omega}\omega^{2}\left(e^{\omega}+1\right)}\\
 & =2\,\zeta\!\left(e^{-\omega},1\right)\frac{\left(1-e^{-\omega}\right)\left(1-e^{-\omega}\right)}{\omega^{2}\left(1+e^{-\omega}\right)}\\
 & =2\,\zeta\!\left(e^{-\omega},1\right)\frac{\left(1-e^{-\omega}\right)^{2}}{\omega^{2}\left(e^{-\omega}+1\right)}\\
 & =2\,\zeta\!\left(e^{-\omega},1\right)\frac{\left(e^{-\omega}-1\right)^{2}}{\omega^{2}\left(e^{-\omega}+1\right)},
\end{align}
thus establishing the claim. Then it follows that
\begin{align}
 & \frac{1}{2}\sum_{k,\ell}2\,\zeta\!\left(e^{-\left(\mu_{k}-\mu_{\ell}\right)},1\right)\left(\frac{e^{-\left(\mu_{k}-\mu_{\ell}\right)}-1}{\mu_{k}-\mu_{\ell}}\right)^{2}\left(\frac{1}{e^{-\left(\mu_{k}-\mu_{\ell}\right)}+1}\right)\left(\frac{e^{-\mu_{k}}}{Z}+\frac{e^{-\mu_{\ell}}}{Z}\right)\Tr\!\left[\Pi_{k}H_{i}\Pi_{\ell}H_{j}\right]\nonumber \\
 & =\frac{1}{2}\sum_{k,\ell}\left(\int_{-\infty}^{\infty}dt\:f(t)e^{it\left(\mu_{k}-\mu_{\ell}\right)}\right)\left(\frac{e^{-\mu_{k}}}{Z}+\frac{e^{-\mu_{\ell}}}{Z}\right)\Tr\!\left[\Pi_{k}H_{i}\Pi_{\ell}H_{j}\right]\\
 & =\frac{1}{2}\int_{-\infty}^{\infty}dt\:f(t)\sum_{k,\ell}\left(\frac{e^{-\mu_{k}}}{Z}+\frac{e^{-\mu_{\ell}}}{Z}\right)\Tr\!\left[\Pi_{k}e^{it\mu_{k}}H_{i}\Pi_{\ell}e^{-it\mu_{\ell}}H_{j}\right]\\
 & =\frac{1}{2}\int_{-\infty}^{\infty}dt\:f(t)\sum_{k,\ell}\Tr\!\left[\Pi_{k}\frac{e^{-\mu_{k}}}{Z}e^{it\mu_{k}}H_{i}\Pi_{\ell}e^{-it\mu_{\ell}}H_{j}\right]\nonumber \\
 & \qquad+\frac{1}{2}\int_{-\infty}^{\infty}dt\:f(t)\sum_{k,\ell}\Tr\!\left[\Pi_{k}e^{it\mu_{k}}H_{i}\Pi_{\ell}\frac{e^{-\mu_{\ell}}}{Z}e^{-it\mu_{\ell}}H_{j}\right]\\
 & =\frac{1}{2}\int_{-\infty}^{\infty}dt\:f(t)\Tr\!\left[\left(\sum_{k}\Pi_{k}\frac{e^{-\mu_{k}}}{Z}e^{it\mu_{k}}\right)H_{i}\left(\sum_{\ell}\Pi_{\ell}e^{-it\mu_{\ell}}\right)H_{j}\right]\nonumber \\
 & \qquad+\frac{1}{2}\int_{-\infty}^{\infty}dt\:f(t)\Tr\!\left[\left(\sum_{k}\Pi_{k}e^{it\mu_{k}}\right)H_{i}\left(\sum_{\ell}\Pi_{\ell}\frac{e^{-\mu_{\ell}}}{Z}e^{-it\mu_{\ell}}\right)H_{j}\right]\\
 & =\frac{1}{2}\int_{-\infty}^{\infty}dt\:f(t)\Tr\!\left[\rho(\theta)e^{itH(\theta)}H_{i}e^{-itH(\theta)}H_{j}\right]\nonumber \\
 & \qquad+\frac{1}{2}\int_{-\infty}^{\infty}dt\:f(t)\Tr\!\left[e^{itH(\theta)}H_{i}e^{-itH(\theta)}\rho(\theta)H_{j}\right]\\
 & =\frac{1}{2}\Tr\!\left[\Phi_{f,\theta}\!\left(H_{i}\right)H_{j}\rho(\theta)\right]+\Tr\!\left[H_{j}\Phi_{f,\theta}\!\left(H_{i}\right)\rho(\theta)\right]\\
 & =\frac{1}{2}\left\langle \left\{ \Phi_{f,\theta}\!\left(H_{i}\right),H_{j}\right\} \right\rangle _{\rho(\theta)}.
\end{align}

Finally, the claim in \eqref{eq:normalization-f-fourier} follows
because
\begin{align}
\int_{-\infty}^{\infty}dt\:f(t) & =\lim_{\omega\to0}2\,\zeta\!\left(e^{-\omega},1\right)\frac{\left(e^{-\omega}-1\right)^{2}}{\omega^{2}\left(e^{-\omega}+1\right)}\\
 & =2\,\zeta\!\left(1,1\right)\left(\lim_{\omega\to0}\frac{e^{-\omega}-1}{\omega}\right)^{2}\frac{1}{2}\\
 & =\kappa\left(\left.\frac{d}{d\omega}e^{-\omega}\right|_{\omega=0}\right)^{2}\\
 & =\kappa,
\end{align}
where the third equality follows from \eqref{eq:zeta-kappa-normalization}.
\end{proof}

\subsection{$\alpha$-$z$ High-peak tent probability densities}

In this subsection, I prove Lemma \ref{lem:fourier-trans-a-z}, which
is the second ingredient needed for the proof of Theorem~\ref{thm:QBM-a-z-formula}.
The probability density in \eqref{eq:high-peak-tent} is known as
the high-peak tent, a name given in \cite{Patel2024a}, due to its
form when plotted (see \cite[Figure~3]{Patel2024a}). It was shown
in \cite[Appendix~C]{Ejima2019}, \cite[Section~5.1, Supplementary Information]{Anshu2021},
and \cite[Lemma~12]{Patel2024a} that the function $p(t)$ is the
Fourier transform of $\tanh\!\left(\frac{\omega}{2}\right)/\left(\frac{\omega}{2}\right)$,
the latter function having been considered in the context of quantum
belief propagation \cite{Hastings2007,Kim2012,Kato2019}.

The family of probability densities in \eqref{eq:high-peak-tent-a-z}--\eqref{eq:2nd-exp-a-z-dist},
parameterized by $\alpha\in\left(0,1\right)$ and $z>0$, have a similar
form to $p(t)$ when plotted, and so let us refer to them as the $\alpha$-$z$
high-peak tent probability densities. 
\begin{lem}
\label{lem:fourier-trans-a-z}For all $\alpha\in\left(0,1\right)$
and $z>0$, the Fourier transform of the function
\begin{equation}
\omega\mapsto2\,\zeta_{\alpha,z}\!\left(e^{-\omega},1\right)\frac{\left(e^{-\omega}-1\right)^{2}}{\omega^{2}\left(e^{-\omega}+1\right)},
\end{equation}
where $\zeta_{\alpha,z}\!\left(e^{-\omega},1\right)$ is defined in
\eqref{eq:limit-for-a-z-eigenval-func}, is equal to the following
probability density function:
\begin{equation}
q_{\alpha,z}(t)\coloneqq\left(p*p_{\alpha,z}\right)(t)=\int_{-\infty}^{\infty}d\tau\:p(\tau)p_{\alpha,z}(t-\tau),\label{eq:q-a-z-prob-dens}
\end{equation}
where the probability density functions $p$ and $p_{\alpha,z}$ are
defined on $t\in\mathbb{R}$ as
\begin{align}
p(t) & \coloneqq\frac{2}{\pi}\ln\!\left|\coth\!\left(\frac{\pi t}{2}\right)\right|,\label{eq:high-peak-tent}\\
p_{\alpha,z}(t) & \coloneqq\frac{z}{2\pi\alpha\left(1-\alpha\right)}\ln\!\left(1+\left(\frac{\sin(\pi\alpha)}{\sinh(\pi zt)}\right)^{2}\right)\label{eq:high-peak-tent-a-z}\\
 & =\frac{z}{2\pi\alpha\left(1-\alpha\right)}\ln\!\left(\coth^{2}\!\left(\pi zt\right)-\left(\frac{\cos(\pi\alpha)}{\sinh\!\left(\pi zt\right)}\right)^{2}\right).\label{eq:2nd-exp-a-z-dist}
\end{align}
Additionally,
\begin{enumerate}
\item The identity $p_{\alpha=\frac{1}{2},z=\frac{1}{2}}(t)=p(t)$ holds
for all $t\in\mathbb{R}$.
\item For all $z>0$, the characteristic function of $p_{\alpha,z}(t)$
converges pointwise to one everywhere in the limit $\alpha\to1$,
implying that, for all $z>0$, the random variable $X_{\alpha,z}\sim q_{\alpha,z}$
converges in distribution to the random variable $X\sim p$ in the
limit $\alpha\to1$.
\item Similarly, for all $\alpha\in\left(0,1\right)$, the characteristic
function of $p_{\alpha,z}(t)$ converges pointwise to one everywhere
in the limit $z\to\infty$, implying that, for all $\alpha\in\left(0,1\right)$,
the random variable $X_{\alpha,z}\sim q_{\alpha,z}$ converges in
distribution to the random variable $X\sim p$ in the limit $z\to\infty$.
\end{enumerate}
\end{lem}

\begin{proof}
Recall from \eqref{eq:limit-for-a-z-eigenval-func} that
\begin{equation}
\zeta_{\alpha,z}\!\left(e^{-\omega},1\right)=\frac{z}{\alpha\left(1-\alpha\right)}\left(\frac{e^{-\left(\frac{1-\alpha}{z}\right)\omega}-1}{e^{-\omega}-1}\right)\left(\frac{e^{-\left(\frac{\alpha}{z}\right)\omega}-1}{e^{-\left(\frac{1}{z}\right)\omega}-1}\right).
\end{equation}
After defining 
\begin{equation}
f_{\alpha,z}(\omega)\coloneqq\frac{z\left(1-e^{-\left(\frac{1-\alpha}{z}\right)\omega}\right)\left(1-e^{-\left(\frac{\alpha}{z}\right)\omega}\right)}{\alpha\left(1-\alpha\right)\omega\left(1-e^{-\left(\frac{1}{z}\right)\omega}\right)},\label{eq:def-f-a-z-func}
\end{equation}
it follows that
\begin{align}
2\,\zeta_{\alpha,z}\!\left(e^{-\omega},1\right)\frac{\left(e^{-\omega}-1\right)^{2}}{\omega^{2}\left(e^{-\omega}+1\right)} & =\frac{2z\left(e^{-\left(\frac{1-\alpha}{z}\right)\omega}-1\right)\left(e^{-\left(\frac{\alpha}{z}\right)\omega}-1\right)\left(e^{-\omega}-1\right)}{\alpha\left(1-\alpha\right)\left(e^{-\left(\frac{1}{z}\right)\omega}-1\right)\omega^{2}\left(e^{-\omega}+1\right)}\\
 & =\frac{z\left(1-e^{-\left(\frac{1-\alpha}{z}\right)\omega}\right)\left(1-e^{-\left(\frac{\alpha}{z}\right)\omega}\right)}{\alpha\left(1-\alpha\right)\omega\left(1-e^{-\left(\frac{1}{z}\right)\omega}\right)}\frac{\left(1-e^{-\omega}\right)}{\frac{\omega}{2}\left(1+e^{-\omega}\right)}\\
 & =f_{\alpha,z}(\omega)\left(\frac{1-e^{-\omega}}{\frac{\omega}{2}\left(1+e^{-\omega}\right)}\right)\\
 & =f_{\alpha,z}(\omega)\left(\frac{e^{\frac{\omega}{2}}-e^{-\frac{\omega}{2}}}{\frac{\omega}{2}\left(e^{\frac{\omega}{2}}+e^{-\frac{\omega}{2}}\right)}\right)\\
 & =f_{\alpha,z}(\omega)\frac{\tanh\!\left(\frac{\omega}{2}\right)}{\frac{\omega}{2}}.
\end{align}
Our goal is to evaluate the following expression for all $t\in\mathbb{R}$:
\begin{equation}
\frac{1}{2\pi}\int_{-\infty}^{\infty}d\omega\ f_{\alpha,z}(\omega)\frac{\tanh\!\left(\frac{\omega}{2}\right)}{\frac{\omega}{2}}e^{i\omega t},
\end{equation}
and by the convolution theorem \cite{Weisstein}, this is equal to
the convolution of the following two Fourier transforms:
\begin{align}
 & \frac{1}{2\pi}\int_{-\infty}^{\infty}d\omega\ f_{\alpha,z}(\omega)e^{i\omega t},\label{eq:fourier-trans-f-a-z}\\
 & \frac{1}{2\pi}\int_{-\infty}^{\infty}d\omega\ \frac{\tanh\!\left(\frac{\omega}{2}\right)}{\frac{\omega}{2}}e^{i\omega t}=\frac{2}{\pi}\ln\!\left|\coth\!\left(\frac{\pi t}{2}\right)\right|,\label{eq:fourier-transform-orig-high-peak}
\end{align}
where \eqref{eq:fourier-transform-orig-high-peak} follows from \cite[Lemma~12]{Patel2024a}
and the function $t\mapsto\frac{2}{\pi}\ln\!\left|\coth\!\left(\frac{\pi t}{2}\right)\right|$
is a probability density function on $\mathbb{R}$ known as the high-peak
tent \cite{Patel2024a}. The equality in \eqref{eq:fourier-transform-orig-high-peak}
was calculated in \cite[Appendix~C]{Ejima2019} and \cite[Section~5.1, Supplementary Information]{Anshu2021}
by means of contour integration. By setting $\alpha=z=\frac{1}{2}$,
the following equality holds:
\begin{equation}
f_{\alpha=\frac{1}{2},z=\frac{1}{2}}(\omega)=\frac{\tanh\!\left(\frac{\omega}{2}\right)}{\frac{\omega}{2}},
\end{equation}
which implies that the Fourier transform in \eqref{eq:fourier-transform-orig-high-peak}
is a special case of that in \eqref{eq:fourier-trans-f-a-z}. Indeed,
consider that
\begin{align}
f_{\alpha=\frac{1}{2},z=\frac{1}{2}}(\omega) & =\lim_{\alpha,z\to\frac{1}{2}}\frac{z\left(1-e^{-\left(\frac{1-\alpha}{z}\right)\omega}\right)\left(1-e^{-\left(\frac{\alpha}{z}\right)\omega}\right)}{\alpha\left(1-\alpha\right)\omega\left(1-e^{-\left(\frac{1}{z}\right)\omega}\right)}\\
 & =\frac{\left(1-e^{-\omega}\right)\left(1-e^{-\omega}\right)}{\frac{\omega}{2}\left(1-e^{-2\omega}\right)}\\
 & =\frac{\left(1-e^{-\omega}\right)\left(1-e^{-\omega}\right)}{\frac{\omega}{2}\left(1-e^{-\omega}\right)\left(1+e^{-\omega}\right)}\\
 & =\frac{1-e^{-\omega}}{\frac{\omega}{2}\left(1+e^{-\omega}\right)}\\
 & =\frac{\tanh\!\left(\frac{\omega}{2}\right)}{\frac{\omega}{2}}.
\end{align}
In what follows, I prove that
\begin{align}
\frac{1}{2\pi}\int_{-\infty}^{\infty}d\omega\ f_{\alpha,z}(\omega)e^{i\omega t} & =\frac{z}{2\pi\alpha\left(1-\alpha\right)}\ln\!\left(1+\left(\frac{\sin(\pi\alpha)}{\sinh(\pi zt)}\right)^{2}\right),\label{eq:fourier-trans-f-a-z-to-p-a-z}
\end{align}
for all $\alpha\in\left(0,1\right)$ and $z>0$, and that
\begin{equation}
\lim_{\alpha,z\to\frac{1}{2}}\frac{z}{2\pi\alpha\left(1-\alpha\right)}\ln\!\left(1+\left(\frac{\sin(\pi\alpha)}{\sinh(\pi zt)}\right)^{2}\right)=\frac{2}{\pi}\ln\!\left|\coth\!\left(\frac{\pi t}{2}\right)\right|.
\end{equation}

\begin{figure}
\centering

\includegraphics[width=0.5\textwidth]{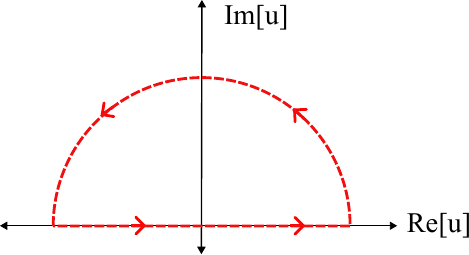}\caption{Contour $\gamma_{R}^{+}$ for $t>0$.}\label{fig:Contour-for-pos-t}

\end{figure}
Let us now evaluate the Fourier transform in \eqref{eq:fourier-trans-f-a-z}
by employing contour integration and the Cauchy residue theorem. We
begin by considering the case when $t>0$, and consider the positively
oriented contour $\gamma_{R}^{+}$ depicted in Figure \ref{fig:Contour-for-pos-t},
where $R>0$ denotes the radius of the semicircle depicted there.
Consider that
\begin{equation}
\lim_{R\to\infty}\oint_{\gamma_{R}^{+}}du\ f_{\alpha,z}(u)e^{iut}=\lim_{R\to\infty}\left[\int_{-R}^{R}d\omega\ f_{\alpha,z}(\omega)e^{i\omega t}+\int_{0}^{\pi}d\theta\ f_{\alpha,z}(Re^{i\theta})e^{iRe^{i\theta}t}\right],
\end{equation}
such that the contour integral is broken up into the line integral
along the real axis and the line integral around the arc of the semicircle.
Let us now prove that the line integral around the arc of the semicircle
evaluates to zero in the limit $R\to\infty$. To this end, consider
that
\begin{align}
 & \left|\int_{0}^{\pi}d\theta\ f_{\alpha,z}(Re^{i\theta})e^{iRe^{i\theta}t}\right|\nonumber \\
 & =\left|\int_{0}^{\pi}d\theta\ f_{\alpha,z}(Re^{i\theta})e^{iR\cos(\theta)t}e^{-R\sin(\theta)t}\right|\label{eq:arc-to-zero-1}\\
 & \leq\int_{0}^{\pi}d\theta\ \left|f_{\alpha,z}(Re^{i\theta})e^{iR\cos(\theta)t}e^{-R\sin(\theta)t}\right|\\
 & =\int_{0}^{\pi}d\theta\ \left|f_{\alpha,z}(Re^{i\theta})\right|e^{-R\sin(\theta)t}\\
 & =\int_{0}^{\pi}d\theta\ \left|\frac{z\left(1-e^{-\left(\frac{1-\alpha}{z}\right)Re^{i\theta}}\right)\left(1-e^{-\left(\frac{\alpha}{z}\right)Re^{i\theta}}\right)}{\alpha\left(1-\alpha\right)Re^{i\theta}\left(1-e^{-\left(\frac{1}{z}\right)Re^{i\theta}}\right)}\right|e^{-R\sin(\theta)t}\\
 & \leq\frac{z}{\alpha\left|1-\alpha\right|R}\int_{0}^{\pi}d\theta\ \left|\frac{\left(1-e^{-\left(\frac{1-\alpha}{z}\right)Re^{i\theta}}\right)\left(1-e^{-\left(\frac{\alpha}{z}\right)Re^{i\theta}}\right)}{\left(1-e^{-\left(\frac{1}{z}\right)Re^{i\theta}}\right)}\right|\\
 & \leq\frac{z}{\alpha\left|1-\alpha\right|R}\int_{0}^{\pi}d\theta\ 2\label{eq:2-bound-of-a-z-sinh-func}\\
 & =\frac{2\pi z}{\alpha\left|1-\alpha\right|R}
\end{align}
Thus, we conclude that
\begin{align}
\lim_{R\to\infty}\left|\int_{0}^{\pi}d\theta\ f_{\alpha,z}(Re^{i\theta})e^{iRe^{i\theta}t}\right| & \leq\lim_{R\to\infty}\frac{\pi z}{\alpha\left|1-\alpha\right|R}.\\
 & =0.
\end{align}
To see \eqref{eq:2-bound-of-a-z-sinh-func}, setting $u=x+iy$, consider
that
\begin{align}
 & \left|\frac{\left(1-e^{-\left(\frac{1-\alpha}{z}\right)u}\right)\left(1-e^{-\left(\frac{\alpha}{z}\right)u}\right)}{\left(1-e^{-\left(\frac{1}{z}\right)u}\right)}\right|\nonumber \\
 & =\left|\frac{\left(1-e^{-\left(\frac{1-\alpha}{z}\right)u}\right)\left(1-e^{-\left(\frac{\alpha}{z}\right)u}\right)}{\left(1-e^{-\left(\frac{1}{z}\right)u}\right)}\right|\\
 & =\left|\frac{e^{-\left(\frac{1-\alpha}{2z}\right)u}\left(e^{\left(\frac{1-\alpha}{2z}\right)u}-e^{-\left(\frac{1-\alpha}{2z}\right)u}\right)e^{-\left(\frac{\alpha}{2z}\right)u}\left(e^{\left(\frac{\alpha}{2z}\right)u}-e^{-\left(\frac{\alpha}{2z}\right)u}\right)}{e^{-\left(\frac{1}{2z}\right)u}\left(e^{\left(\frac{1}{2z}\right)u}-e^{-\left(\frac{1}{2z}\right)u}\right)}\right|\\
 & =\left|\frac{\left(e^{\left(\frac{1-\alpha}{2z}\right)u}-e^{-\left(\frac{1-\alpha}{2z}\right)u}\right)\left(e^{\left(\frac{\alpha}{2z}\right)u}-e^{-\left(\frac{\alpha}{2z}\right)u}\right)}{\left(e^{\left(\frac{1}{2z}\right)u}-e^{-\left(\frac{1}{2z}\right)u}\right)}\right|\\
 & =2\left|\frac{\sinh\!\left(\left(\frac{1-\alpha}{2z}\right)u\right)\sinh\!\left(\left(\frac{\alpha}{2z}\right)u\right)}{\sinh\!\left(\left(\frac{1}{2z}\right)u\right)}\right|\\
 & =2\frac{\left|\sinh\!\left(\left(\frac{1-\alpha}{2z}\right)u\right)\right|\left|\sinh\!\left(\left(\frac{\alpha}{2z}\right)u\right)\right|}{\left|\sinh\!\left(\left(\frac{1}{2z}\right)u\right)\right|}\\
 & =2\frac{\left|\sinh\!\left(\left(\frac{1-\alpha}{2z}\right)x\right)\right|\left|\sinh\!\left(\left(\frac{\alpha}{2z}\right)x\right)\right|}{\left|\sinh\!\left(\left(\frac{1}{2z}\right)x\right)\right|}\frac{\sqrt{\left(1+\sin^{2}\!\left(\left(\frac{1-\alpha}{2z}\right)y\right)\right)}\sqrt{\left(1+\sin^{2}\!\left(\left(\frac{\alpha}{2z}\right)y\right)\right)}}{\sqrt{\left(1+\sin^{2}\!\left(\left(\frac{1}{2z}\right)y\right)\right)}}\nonumber \\
 & \leq2,
\end{align}
where we used the following identities:
\begin{align}
\left|\sinh(x+iy)\right| & =\left|\sinh(x)\cos(y)+i\cosh(x)\sin(y)\right|\\
 & =\sqrt{\sinh^{2}(x)\cos^{2}(y)+\cosh^{2}(x)\sin^{2}(y)}\\
 & =\sqrt{\sinh^{2}(x)\cos^{2}(y)+\left(\sinh^{2}(x)+1\right)\sin^{2}(y)}\\
 & =\sqrt{\sinh^{2}(x)\left(\cos^{2}(y)+\sin^{2}(y)\right)+\sinh^{2}(x)\sin^{2}(y)}\\
 & =\sqrt{\sinh^{2}(x)\left(1+\sin^{2}(y)\right)}\\
 & =\left|\sinh(x)\right|\sqrt{\left(1+\sin^{2}(y)\right)},
\end{align}
and the facts that, for all $x,y\in\mathbb{R}$, $\alpha\in\left(0,1\right)$,
and $z>0$,
\begin{align}
\frac{\sqrt{\left(1+\sin^{2}\!\left(\left(\frac{1-\alpha}{2z}\right)y\right)\right)}\sqrt{\left(1+\sin^{2}\!\left(\left(\frac{\alpha}{2z}\right)y\right)\right)}}{\sqrt{\left(1+\sin^{2}\!\left(\left(\frac{1}{2z}\right)y\right)\right)}} & \leq2,\\
\frac{\left|\sinh\!\left(\left(\frac{1-\alpha}{2z}\right)x\right)\right|\left|\sinh\!\left(\left(\frac{\alpha}{2z}\right)x\right)\right|}{\left|\sinh\!\left(\left(\frac{1}{2z}\right)x\right)\right|} & \leq\frac{1}{2}.\label{eq:arc-to-zero-last}
\end{align}
The inequality in \eqref{eq:arc-to-zero-last} follows because, for
all $\alpha\in\left(0,1\right)$ and $z>0$,
\begin{align}
\lim_{x\to\infty}\frac{\sinh\!\left(\left(\frac{1-\alpha}{2z}\right)x\right)\sinh\!\left(\left(\frac{\alpha}{2z}\right)x\right)}{\sinh\!\left(\left(\frac{1}{2z}\right)x\right)} & =\lim_{x\to\infty}\frac{\frac{e^{\left(\frac{1-\alpha}{2z}\right)x}}{2}\frac{e^{\left(\frac{\alpha}{2z}\right)x}}{2}}{\frac{e^{\left(\frac{1}{2z}\right)x}}{2}}=\frac{1}{2},\\
\lim_{x\to-\infty}\frac{\sinh\!\left(\left(\frac{1-\alpha}{2z}\right)x\right)\sinh\!\left(\left(\frac{\alpha}{2z}\right)x\right)}{\sinh\!\left(\left(\frac{1}{2z}\right)x\right)} & =\lim_{x\to-\infty}\frac{\frac{-e^{-\left(\frac{1-\alpha}{2z}\right)x}}{2}\frac{e^{-\left(\frac{\alpha}{2z}\right)x}}{2}}{\frac{-e^{-\left(\frac{1}{2z}\right)x}}{2}}=-\frac{1}{2},
\end{align}
\begin{align}
 & \lim_{x\to0}\frac{\sinh\!\left(\left(\frac{1-\alpha}{2z}\right)x\right)\sinh\!\left(\left(\frac{\alpha}{2z}\right)x\right)}{\sinh\!\left(\left(\frac{1}{2z}\right)x\right)}\nonumber \\
 & =\lim_{x\to0}\frac{\left(\frac{1-\alpha}{2z}\right)\cosh\!\left(\left(\frac{1-\alpha}{2z}\right)x\right)\sinh\!\left(\left(\frac{\alpha}{2z}\right)x\right)+\sinh\!\left(\left(\frac{1-\alpha}{2z}\right)x\right)\left(\frac{\alpha}{2z}\right)\cosh\!\left(\left(\frac{\alpha}{2z}\right)x\right)}{\left(\frac{1}{2z}\right)\cosh\!\left(\left(\frac{1}{2z}\right)x\right)}\\
 & =0,
\end{align}
and the fact that the function $x\mapsto\sinh\!\left(\left(\frac{1-\alpha}{2z}\right)x\right)\sinh\!\left(\left(\frac{\alpha}{2z}\right)x\right)/\sinh\!\left(\left(\frac{1}{2z}\right)x\right)$
is increasing on $x\in\mathbb{R}$. Indeed, since this latter function
is odd, to see that it is increasing on $x\in\mathbb{R}$, it suffices
to prove that it is increasing on $\left(0,\infty\right)$. To see
this, let us calculate the derivative of this function with respect
to $x$ and prove that it is non-negative on $\left(0,\infty\right)$.
Consider that
\begin{align}
 & 2\frac{\partial}{\partial x}\left(\frac{\sinh\!\left(\left(\frac{1-\alpha}{2z}\right)x\right)\sinh\!\left(\left(\frac{\alpha}{2z}\right)x\right)}{\sinh\!\left(\left(\frac{1}{2z}\right)x\right)}\right)\nonumber \\
 & =\frac{\partial}{\partial x}\left(\frac{\left(e^{\left(\frac{1-\alpha}{2z}\right)x}-e^{-\left(\frac{1-\alpha}{2z}\right)x}\right)\left(e^{\frac{\alpha}{2z}x}-e^{-\frac{\alpha}{2z}x}\right)}{e^{\frac{1}{2z}x}-e^{-\frac{1}{2z}x}}\right)\\
 & =\frac{\partial}{\partial x}\left(\frac{\left(e^{\left(\frac{1-\alpha}{z}\right)x}-1\right)\left(e^{\frac{\alpha}{z}x}-1\right)}{e^{\frac{1}{z}x}-1}\right)\\
 & =\frac{\frac{1-\alpha}{z}e^{\left(\frac{1-\alpha}{z}\right)x}\left(e^{\frac{\alpha}{z}x}-1\right)+\frac{\alpha}{z}\left(e^{\left(\frac{1-\alpha}{z}\right)x}-1\right)e^{\frac{\alpha}{z}x}}{e^{\frac{1}{z}x}-1}-\frac{\frac{1}{z}\left(e^{\left(\frac{1-\alpha}{z}\right)x}-1\right)\left(e^{\frac{\alpha}{z}x}-1\right)e^{\frac{1}{z}x}}{\left(e^{\frac{1}{z}x}-1\right)^{2}}\\
 & =\frac{\frac{1-\alpha}{z}\left(e^{\left(\frac{1-\alpha}{z}\right)x}-1+1\right)\left(e^{\frac{\alpha}{z}x}-1\right)+\frac{\alpha}{z}\left(e^{\left(\frac{1-\alpha}{z}\right)x}-1\right)\left(e^{\frac{\alpha}{z}x}-1+1\right)}{e^{\frac{1}{z}x}-1}\nonumber \\
 & \qquad-\frac{\frac{1}{z}\left(e^{\left(\frac{1-\alpha}{z}\right)x}-1\right)\left(e^{\frac{\alpha}{z}x}-1\right)e^{\frac{1}{z}x}}{\left(e^{\frac{1}{z}x}-1\right)^{2}}\\
 & =\frac{\left[\begin{array}{c}
\frac{1-\alpha}{z}\left(e^{\left(\frac{1-\alpha}{z}\right)x}-1\right)\left(e^{\frac{\alpha}{z}x}-1\right)+\frac{\alpha}{z}\left(e^{\left(\frac{1-\alpha}{z}\right)x}-1\right)\left(e^{\frac{\alpha}{z}x}-1\right)\\
+\frac{1-\alpha}{z}\left(e^{\frac{\alpha}{z}x}-1\right)+\frac{\alpha}{z}\left(e^{\left(\frac{1-\alpha}{z}\right)x}-1\right)
\end{array}\right]}{e^{\frac{1}{z}x}-1}\nonumber \\
 & \qquad-\frac{\frac{1}{z}\left(e^{\left(\frac{1-\alpha}{z}\right)x}-1\right)\left(e^{\frac{\alpha}{z}x}-1\right)}{\left(e^{\frac{1}{z}x}-1\right)}\frac{e^{\frac{1}{z}x}}{\left(e^{\frac{1}{z}x}-1\right)}\\
 & =\frac{\left(e^{\left(\frac{1-\alpha}{z}\right)x}-1\right)\left(e^{\frac{\alpha}{z}x}-1\right)}{\left(e^{\frac{1}{z}x}-1\right)}\left(\frac{1-\alpha}{z}+\frac{\alpha}{z}-\frac{e^{\frac{1}{z}x}}{z\left(e^{\frac{1}{z}x}-1\right)}\right)\nonumber \\
 & \qquad+\frac{\frac{1-\alpha}{z}\left(e^{\frac{\alpha}{z}x}-1\right)+\frac{\alpha}{z}\left(e^{\left(\frac{1-\alpha}{z}\right)x}-1\right)}{e^{\frac{1}{z}x}-1}\\
 & =\frac{\left(e^{\left(\frac{1-\alpha}{z}\right)x}-1\right)\left(e^{\frac{\alpha}{z}x}-1\right)}{\left(e^{\frac{1}{z}x}-1\right)}\frac{1}{z}\left(1-\frac{e^{\frac{1}{z}x}}{e^{\frac{1}{z}x}-1}\right)+\frac{\frac{1-\alpha}{z}\left(e^{\frac{\alpha}{z}x}-1\right)+\frac{\alpha}{z}\left(e^{\left(\frac{1-\alpha}{z}\right)x}-1\right)}{e^{\frac{1}{z}x}-1}\\
 & =-\frac{\left(e^{\left(\frac{1-\alpha}{z}\right)x}-1\right)\left(e^{\frac{\alpha}{z}x}-1\right)}{\left(e^{\frac{1}{z}x}-1\right)}\left(\frac{1}{z\left(e^{\frac{1}{z}x}-1\right)}\right)+\frac{\frac{1-\alpha}{z}\left(e^{\frac{\alpha}{z}x}-1\right)}{e^{\frac{1}{z}x}-1}+\frac{\frac{\alpha}{z}\left(e^{\left(\frac{1-\alpha}{z}\right)x}-1\right)}{e^{\frac{1}{z}x}-1}\\
 & =\frac{\left(e^{\left(\frac{1-\alpha}{z}\right)x}-1\right)\left(e^{\frac{\alpha}{z}x}-1\right)}{z\left(e^{\frac{1}{z}x}-1\right)}\left(\frac{1-\alpha}{e^{\left(\frac{1-\alpha}{z}\right)x}-1}+\frac{\alpha}{e^{\frac{\alpha}{z}x}-1}-\frac{1}{e^{\frac{1}{z}x}-1}\right).
\end{align}
Since $\frac{\left(e^{\left(\frac{1-\alpha}{z}\right)x}-1\right)\left(e^{\frac{\alpha}{z}x}-1\right)}{z\left(e^{\frac{1}{z}x}-1\right)}>0$
for all $x>0$, the desired statement will follow if the inequality
\begin{equation}
\frac{1-\alpha}{e^{\left(\frac{1-\alpha}{z}\right)x}-1}+\frac{\alpha}{e^{\frac{\alpha}{z}x}-1}\geq\frac{1}{e^{\frac{1}{z}x}-1}
\end{equation}
holds for all $\alpha\in\left(0,1\right)$ and $x>0$. Thus, the statement
follows if the function $x\mapsto\frac{1}{e^{\frac{1}{z}x}-1}$ is
convex and decreasing on $x\in\left(0,\infty\right)$. Indeed, if
these properties hold, then
\begin{equation}
\frac{1-\alpha}{e^{\left(\frac{1-\alpha}{z}\right)x}-1}+\frac{\alpha}{e^{\frac{\alpha}{z}x}-1}\geq\frac{1}{e^{\frac{1}{z}\left(\left(1-\alpha\right)^{2}+\alpha^{2}\right)x}-1}\geq\frac{1}{e^{\frac{1}{z}x}-1},
\end{equation}
where the first inequality follows from convexity and the second from
the decreasing property because $\left(1-\alpha\right)^{2}+\alpha^{2}<1$
for $\alpha\in\left(0,1\right)$. To prove these properties, let us
calculate the first and second derivatives of the function $x\mapsto\frac{1}{e^{\frac{1}{z}x}-1}$:
\begin{align}
\frac{\partial}{\partial x}\left(\frac{1}{e^{\frac{1}{z}x}-1}\right) & =-\frac{e^{\frac{x}{z}}}{z\left(e^{\frac{1}{z}x}-1\right)^{2}},\\
\frac{\partial^{2}}{\partial x^{2}}\left(\frac{1}{e^{\frac{1}{z}x}-1}\right) & =\frac{e^{\frac{x}{z}}\left(e^{\frac{x}{z}}+1\right)}{z^{2}\left(e^{\frac{1}{z}x}-1\right)^{3}}.
\end{align}
Since $\frac{\partial}{\partial x}\left(\frac{1}{e^{\frac{1}{z}x}-1}\right)\leq0$
and $\frac{\partial^{2}}{\partial x^{2}}\left(\frac{1}{e^{\frac{1}{z}x}-1}\right)\geq0$
on $x\in\left(0,\infty\right)$, the claim follows. Thus, it follows
that
\begin{align}
\lim_{x\to\pm\infty}\frac{\left|\sinh\!\left(\left(\frac{1-\alpha}{2z}\right)x\right)\right|\left|\sinh\!\left(\left(\frac{\alpha}{2z}\right)x\right)\right|}{\left|\sinh\!\left(\left(\frac{1}{2z}\right)x\right)\right|} & =\frac{1}{2},\\
\lim_{x\to0}\frac{\left|\sinh\!\left(\left(\frac{1-\alpha}{2z}\right)x\right)\right|\left|\sinh\!\left(\left(\frac{\alpha}{2z}\right)x\right)\right|}{\left|\sinh\!\left(\left(\frac{1}{2z}\right)x\right)\right|} & =0,
\end{align}
and that the function $x\mapsto\frac{\left|\sinh\left(\left(\frac{1-\alpha}{2z}\right)x\right)\right|\left|\sinh\left(\left(\frac{\alpha}{2z}\right)x\right)\right|}{\left|\sinh\left(\left(\frac{1}{2z}\right)x\right)\right|}$
is decreasing on $\left(-\infty,0\right]$ and increasing on $\left[0,\infty\right)$.
We then conclude that
\begin{align}
\frac{1}{2\pi}\lim_{R\to\infty}\oint_{\gamma_{R}^{+}}du\ f_{\alpha,z}(u)e^{iut} & =\frac{1}{2\pi}\lim_{R\to\infty}\int_{-R}^{R}d\omega\ f_{\alpha,z}(\omega)e^{i\omega t}\label{eq:cauchy-residue-about-to-happen}\\
 & =\frac{1}{2\pi}\int_{-\infty}^{\infty}d\omega\ f_{\alpha,z}(\omega)e^{i\omega t}.
\end{align}

We can evaluate the expression on the left-hand side of \eqref{eq:cauchy-residue-about-to-happen}
by means of the Cauchy residue theorem. For $u\in\mathbb{C}$ and
$\alpha\in\left(0,1\right)$, the singularities of the function 
\begin{equation}
f_{\alpha,z}(u)e^{iut}=\frac{z\left(1-e^{-\left(\frac{1-\alpha}{z}\right)u}\right)\left(1-e^{-\left(\frac{\alpha}{z}\right)u}\right)}{\alpha\left(1-\alpha\right)u\left(1-e^{-\left(\frac{1}{z}\right)u}\right)}e^{iut}\label{eq:cauchy-residue-to-the-rescue-1}
\end{equation}
in the region enclosed by the contour $\gamma_{R}^{+}$, as $R\to\infty$,
occur at $u\in\left\{ 2\pi izm:m\in\mathbb{N}\right\} $. Indeed,
there is not a singularity of $f_{\alpha,z}(u)e^{iut}$ at $u=0$
because
\begin{align}
 & \lim_{u\to0}f_{\alpha,z}(u)e^{iut}\nonumber \\
 & =\lim_{u\to0}\frac{z\left(1-e^{-\left(\frac{1-\alpha}{z}\right)u}\right)\left(1-e^{-\left(\frac{\alpha}{z}\right)u}\right)}{\alpha\left(1-\alpha\right)u\left(1-e^{-\left(\frac{1}{z}\right)u}\right)}e^{iut}\\
 & =\frac{z}{\alpha\left(1-\alpha\right)}\lim_{u\to0}\frac{\left(1-e^{-\left(\frac{1-\alpha}{z}\right)u}\right)\left(1-e^{-\left(\frac{\alpha}{z}\right)u}\right)}{u\left(1-e^{-\left(\frac{1}{z}\right)u}\right)}\lim_{u\to0}e^{iut}\label{eq:limit-f-a-z-is-1-first}\\
 & =\frac{z}{\alpha\left(1-\alpha\right)}\lim_{u\to0}\frac{\left(\frac{1-\alpha}{z}\right)e^{-\left(\frac{1-\alpha}{z}\right)u}\left(1-e^{-\left(\frac{\alpha}{z}\right)u}\right)+\left(1-e^{-\left(\frac{1-\alpha}{z}\right)u}\right)\left(\frac{\alpha}{z}\right)e^{-\left(\frac{\alpha}{z}\right)u}}{1-e^{-\left(\frac{1}{z}\right)u}+u\left(\frac{1}{z}\right)e^{-\left(\frac{1}{z}\right)u}}\\
 & =\frac{z}{\alpha\left(1-\alpha\right)}\lim_{u\to0}\frac{\left(\frac{1-\alpha}{z}\right)\left(e^{-\left(\frac{1-\alpha}{z}\right)u}-e^{-\left(\frac{1}{z}\right)u}\right)+\left(\frac{\alpha}{z}\right)\left(e^{-\left(\frac{\alpha}{z}\right)u}-e^{-\left(\frac{1}{z}\right)u}\right)}{1-e^{-\left(\frac{1}{z}\right)u}+u\left(\frac{1}{z}\right)e^{-\left(\frac{1}{z}\right)u}}\\
 & =\frac{z}{\alpha\left(1-\alpha\right)}\lim_{u\to0}\frac{\left(\frac{1-\alpha}{z}\right)e^{-\left(\frac{1-\alpha}{z}\right)u}+\left(\frac{\alpha}{z}\right)e^{-\left(\frac{\alpha}{z}\right)u}-\left(\frac{1}{z}\right)e^{-\left(\frac{1}{z}\right)u}}{1-e^{-\left(\frac{1}{z}\right)u}+u\left(\frac{1}{z}\right)e^{-\left(\frac{1}{z}\right)u}}\\
 & =\frac{z}{\alpha\left(1-\alpha\right)}\lim_{u\to0}\frac{-\left(\frac{1-\alpha}{z}\right)^{2}e^{-\left(\frac{1-\alpha}{z}\right)u}-\left(\frac{\alpha}{z}\right)^{2}e^{-\left(\frac{\alpha}{z}\right)u}+\left(\frac{1}{z}\right)^{2}e^{-\left(\frac{1}{z}\right)u}}{\left(\frac{1}{z}\right)e^{-\left(\frac{1}{z}\right)u}+\left(\frac{1}{z}\right)e^{-\left(\frac{1}{z}\right)u}-u\left(\frac{1}{z}\right)^{2}e^{-\left(\frac{1}{z}\right)u}}\\
 & =\frac{z}{\alpha\left(1-\alpha\right)}\frac{-\left(\frac{1-\alpha}{z}\right)^{2}-\left(\frac{\alpha}{z}\right)^{2}+\left(\frac{1}{z}\right)^{2}}{\frac{1}{z}+\frac{1}{z}}\\
 & =\frac{-\left(1-\alpha\right)^{2}-\alpha^{2}+1}{2\alpha\left(1-\alpha\right)}\\
 & =\frac{-\left(1-2\alpha+\alpha^{2}\right)-\alpha^{2}+1}{2\alpha\left(1-\alpha\right)}\\
 & =\frac{2\alpha-2\alpha^{2}}{2\alpha\left(1-\alpha\right)}\\
 & =1.\label{eq:limit-f-a-z-is-1-last}
\end{align}
For all $m\in\mathbb{N}$, we then need to calculate the residues.
Since the poles of the function $u\mapsto f_{\alpha,z}(u)e^{iut}$
at $2\pi izm$ are simple poles, we can use the formula $\text{Res}_{u=c}\!\left[\frac{P(z)}{Q(z)}\right]=\frac{P(c)}{Q'(c)}$
to conclude that
\begin{align}
\text{Res}_{u=2\pi imz}\!\left[f_{\alpha,z}(u)e^{iut}\right] & =\frac{z}{\alpha\left(1-\alpha\right)}\left.\frac{\left(1-e^{-\left(\frac{1-\alpha}{z}\right)u}\right)\left(1-e^{-\left(\frac{\alpha}{z}\right)u}\right)e^{iut}}{\frac{\partial}{\partial u}\left[u\left(1-e^{-\left(\frac{1}{z}\right)u}\right)\right]}\right|_{u=2\pi imz}\\
 & =\frac{z\left(1-e^{-\left(\frac{1-\alpha}{z}\right)2\pi imz}\right)\left(1-e^{-\left(\frac{\alpha}{z}\right)2\pi imz}\right)}{\alpha\left(1-\alpha\right)\left.\left(1-e^{-\left(\frac{1}{z}\right)u}+\frac{u}{z}e^{-\left(\frac{1}{z}\right)u}\right)\right|_{u=2\pi imz}}e^{i\left(2\pi imz\right)t}\\
 & =\frac{z\left(1-e^{-\left(1-\alpha\right)2\pi im}\right)\left(1-e^{-\alpha2\pi im}\right)}{\alpha\left(1-\alpha\right)2\pi im}e^{-2\pi mzt}\\
 & =\frac{z}{\alpha\left(1-\alpha\right)}\left(\frac{2-e^{\alpha2\pi im}-e^{-\alpha2\pi im}}{2\pi im}\right)e^{-2\pi mzt}.
\end{align}
Thus, by applying the Cauchy residue theorem, it follows that
\begin{align}
 & \frac{1}{2\pi}\lim_{R\to\infty}\oint_{\gamma_{R}^{+}}du\ f_{\alpha,z}(u)e^{iut}\nonumber \\
 & =i\sum_{m=1}^{\infty}\text{Res}_{u=i2\pi mz}\!\left[f_{\alpha,z}(u)e^{iut}\right]\\
 & =i\sum_{m=1}^{\infty}\frac{z}{\alpha\left(1-\alpha\right)}\left(\frac{2-e^{\alpha2\pi im}-e^{-\alpha2\pi im}}{2\pi im}\right)e^{-2\pi mzt}\\
 & =\frac{z}{2\pi\alpha\left(1-\alpha\right)}\sum_{m=1}^{\infty}\frac{2e^{-2\pi mzt}}{m}-\frac{\left(e^{2\pi\left(i\alpha-zt\right)}\right)^{m}}{m}-\frac{\left(e^{2\pi\left(-i\alpha-zt\right)}\right)^{m}}{m}\\
 & =\frac{z}{2\pi\alpha\left(1-\alpha\right)}\left(-2\ln\!\left(1-e^{-2\pi zt}\right)+\ln\!\left(1-e^{2\pi\left(i\alpha-zt\right)}\right)+\ln\!\left(1-e^{2\pi\left(-i\alpha-zt\right)}\right)\right)\\
 & =\frac{z}{2\pi\alpha\left(1-\alpha\right)}\ln\!\left(\frac{\left(1-e^{2\pi\left(i\alpha-zt\right)}\right)\left(1-e^{2\pi\left(-i\alpha-zt\right)}\right)}{\left(1-e^{-2\pi zt}\right)^{2}}\right)\\
 & =\frac{z}{2\pi\alpha\left(1-\alpha\right)}\ln\!\left(\frac{1-e^{2\pi\left(i\alpha-zt\right)}-e^{2\pi\left(-i\alpha-zt\right)}+e^{-4\pi zt}}{\left(1-e^{-2\pi zt}\right)^{2}}\right)\\
 & =\frac{z}{2\pi\alpha\left(1-\alpha\right)}\ln\!\left(\frac{1+e^{-4\pi zt}-2e^{-2\pi zt}\cos(2\pi\alpha)}{\left(1-e^{-2\pi zt}\right)^{2}}\right)\\
 & =\frac{z}{2\pi\alpha\left(1-\alpha\right)}\ln\!\left(\frac{\left(1-e^{-2\pi zt}\right)^{2}+2e^{-2\pi zt}\left(1-\cos(2\pi\alpha)\right)}{\left(1-e^{-2\pi zt}\right)^{2}}\right)\\
 & =\frac{z}{2\pi\alpha\left(1-\alpha\right)}\ln\!\left(1+\frac{2e^{-2\pi zt}\left(1-\cos(2\pi\alpha)\right)}{\left(1-e^{-2\pi zt}\right)^{2}}\right)\\
 & =\frac{z}{2\pi\alpha\left(1-\alpha\right)}\ln\!\left(1+\frac{2\left(1-\cos(2\pi\alpha)\right)}{\left(e^{\pi zt}-e^{-\pi zt}\right)^{2}}\right)\\
 & =\frac{z}{2\pi\alpha\left(1-\alpha\right)}\ln\!\left(1+\left(\frac{\sin(\pi\alpha)}{\sinh(\pi zt)}\right)^{2}\right).
\end{align}
where the fourth equality follows from the Taylor expansion $-\ln(1-u)=\sum_{m=1}^{\infty}\frac{u^{m}}{m}$,
which holds for all $u\in\mathbb{C}$ such that $\left|u\right|<1$.
Thus, for $t>0$, we conclude that
\begin{equation}
\frac{1}{2\pi}\int_{-\infty}^{\infty}d\omega\ f_{\alpha,z}(\omega)e^{i\omega t}=\frac{z}{2\pi\alpha\left(1-\alpha\right)}\ln\!\left(1+\left(\frac{\sin(\pi\alpha)}{\sinh(\pi zt)}\right)^{2}\right).\label{eq:cauchy-residue-to-the-rescue-last}
\end{equation}

We can handle the case when $t<0$ by a similar, yet symmetric argument,
instead considering the negatively oriented contour $\gamma_{R}^{-}$
depicted in Figure \ref{fig:Contour-for-neg-t}. To summarize the
argument succinctly, we again use the fact that
\begin{equation}
\lim_{R\to\infty}\oint_{\gamma_{R}^{-}}du\ f_{\alpha,z}(u)e^{iut}=\lim_{R\to\infty}\left[\int_{-R}^{R}d\omega\ f_{\alpha,z}(\omega)e^{i\omega t}+\int_{0}^{-\pi}d\theta\ f_{\alpha,z}(Re^{i\theta})e^{iRe^{i\theta}t}\right],
\end{equation}
and apply a similar argument as in \eqref{eq:arc-to-zero-1}--\eqref{eq:arc-to-zero-last}
to conclude that
\begin{equation}
\lim_{R\to\infty}\left|\int_{0}^{-\pi}d\theta\ f_{\alpha,z}(Re^{i\theta})e^{iRe^{i\theta}t}\right|=0,
\end{equation}
while keeping in mind that $\sin(\theta)\leq0$ for all $\theta\in\left[-\pi,0\right]$.
Then, similar to \eqref{eq:cauchy-residue-to-the-rescue-1}--\eqref{eq:cauchy-residue-to-the-rescue-last},
we apply the Cauchy residue theorem to conclude that
\begin{equation}
\frac{1}{2\pi}\int_{-\infty}^{\infty}d\omega\ f_{\alpha,z}(\omega)e^{i\omega t}=\frac{z}{2\pi\alpha\left(1-\alpha\right)}\ln\!\left(1+\left(\frac{\sin(\pi\alpha)}{\sinh(\pi zt)}\right)^{2}\right)
\end{equation}
for all $t<0$.

\begin{figure}
\centering
\includegraphics[width=0.5\textwidth]{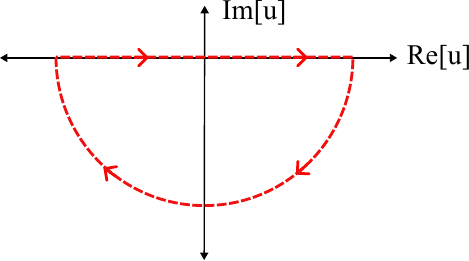}\caption{Contour $\gamma_{R}^{-}$ for $t<0$.}\label{fig:Contour-for-neg-t}
\end{figure}

The function $p_{\alpha,z}(t)$ is a probability density for all $\alpha\in\left(0,1\right)$
and $z>0$ because
\begin{align}
p_{\alpha,z}(t) & =\frac{z}{2\pi\alpha\left(1-\alpha\right)}\ln\!\left(1+\left(\frac{\sin(\pi\alpha)}{\sinh(\pi zt)}\right)^{2}\right)\\
 & \geq\frac{z}{2\pi\alpha\left(1-\alpha\right)}\ln\!\left(1\right)\\
 & =0,
\end{align}
where we used the fact that $\left(\frac{\sin(\pi\alpha)}{\sinh(\pi zt)}\right)^{2}\geq0$
for all $t,z>0$ and $\alpha\in\left(0,1\right)$. Additionally,
\begin{equation}
\int_{-\infty}^{\infty}dt\ p_{\alpha,z}(t)=\lim_{\omega\to0}f_{\alpha,z}(\omega)=1,
\end{equation}
the latter equality already having been argued in \eqref{eq:limit-f-a-z-is-1-first}--\eqref{eq:limit-f-a-z-is-1-last}.

The expression in \eqref{eq:2nd-exp-a-z-dist} follows because
\begin{align}
 & \frac{z}{2\pi\alpha\left(1-\alpha\right)}\ln\!\left(1+\left(\frac{\sin(\pi\alpha)}{\sinh(\pi zt)}\right)^{2}\right)\nonumber \\
 & =\frac{z}{2\pi\alpha\left(1-\alpha\right)}\ln\!\left(\frac{\sinh^{2}(\pi zt)+\sin^{2}(\pi\alpha)}{\sinh^{2}(\pi zt)}\right)\\
 & =\frac{z}{2\pi\alpha\left(1-\alpha\right)}\ln\!\left(\frac{\sinh^{2}(\pi zt)+1-\cos^{2}(\pi\alpha)}{\sinh^{2}(\pi zt)}\right)\\
 & =\frac{z}{2\pi\alpha\left(1-\alpha\right)}\ln\!\left(\frac{\cosh^{2}(\pi zt)-\cos^{2}(\pi\alpha)}{\sinh^{2}(\pi zt)}\right)\\
 & =\frac{z}{2\pi\alpha\left(1-\alpha\right)}\ln\!\left(\coth^{2}\!\left(\pi zt\right)-\left(\frac{\cos(\pi\alpha)}{\sinh\left(\pi zt\right)}\right)^{2}\right).
\end{align}

The equality $p_{\alpha=\frac{1}{2},z=\frac{1}{2}}(t)=p(t)$ holds
because 
\begin{align}
p_{\alpha=\frac{1}{2},z=\frac{1}{2}}(t) & =\frac{\frac{1}{2}}{2\pi\frac{1}{2}\left(1-\frac{1}{2}\right)}\ln\!\left(\coth^{2}\!\left(\pi\left(\frac{1}{2}\right)t\right)-\left(\frac{\cos\!\left(\pi\left(\frac{1}{2}\right)\right)}{\sinh\left(\pi\left(\frac{1}{2}\right)t\right)}\right)^{2}\right)\\
 & =\frac{1}{\pi}\ln\!\left(\coth^{2}\!\left(\frac{\pi t}{2}\right)\right)\\
 & =\frac{2}{\pi}\ln\!\left|\coth\!\left(\frac{\pi t}{2}\right)\right|.
\end{align}

The final statement about the characteristic function $f_{\alpha,z}(\omega)$
of $p_{\alpha,z}(t)$ holds because, for all $z>0$ and $\omega\in\mathbb{R}\backslash\left\{ 0\right\} $,
\begin{align}
 & \lim_{\alpha\to1}f_{\alpha,z}(\omega)\nonumber \\
 & =\lim_{\alpha\to1}\frac{z\left(1-e^{-\left(\frac{1-\alpha}{z}\right)\omega}\right)\left(1-e^{-\left(\frac{\alpha}{z}\right)\omega}\right)}{\alpha\left(1-\alpha\right)\omega\left(1-e^{-\left(\frac{1}{z}\right)\omega}\right)}\\
 & =\frac{z}{\omega\left(1-e^{-\left(\frac{1}{z}\right)\omega}\right)}\lim_{\alpha\to1}\frac{\left(1-e^{-\left(\frac{1-\alpha}{z}\right)\omega}\right)\left(1-e^{-\left(\frac{\alpha}{z}\right)\omega}\right)}{\alpha\left(1-\alpha\right)}\\
 & =\frac{z}{\omega\left(1-e^{-\left(\frac{1}{z}\right)\omega}\right)}\lim_{\alpha\to1}\frac{\left(-\frac{\omega}{z}e^{-\left(\frac{1-\alpha}{z}\right)\omega}\right)\left(1-e^{-\left(\frac{\alpha}{z}\right)\omega}\right)+\left(1-e^{-\left(\frac{1-\alpha}{z}\right)\omega}\right)\left(\frac{\omega}{z}e^{-\left(\frac{\alpha}{z}\right)\omega}\right)}{1-2\alpha}\\
 & =\frac{z}{\omega\left(1-e^{-\left(\frac{1}{z}\right)\omega}\right)}\left(-\left(-\frac{\omega}{z}\right)\left(1-e^{-\left(\frac{1}{z}\right)\omega}\right)\right)\\
 & =1.
\end{align}
In the case that $\omega=0$, we already showed in \eqref{eq:limit-f-a-z-is-1-first}--\eqref{eq:limit-f-a-z-is-1-last}
that $\lim_{\omega\to0}f_{\alpha,z}(\omega)=1$ for all $\alpha\in\left(0,1\right)\cup\left(1,\infty\right)$
and $z>0$, so that $\lim_{\alpha\to1}f_{\alpha,z}(0)=1$. It follows
from Levy's continuity theorem that, for all $z>0$, the random variable
$X_{\alpha,z}\sim q_{\alpha,z}$ converges in distribution to the
random variable $X\sim p$ in the limit $\alpha\to1$.

Similarly, for all $\alpha\in\left(0,1\right)$ and $\omega\in\mathbb{R}\backslash\left\{ 0\right\} $,
\begin{align}
 & \lim_{z\to\infty}f_{\alpha,z}(\omega)\\
 & =\lim_{z\to\infty}\frac{\left(1-e^{-\left(\frac{1-\alpha}{z}\right)\omega}\right)\left(1-e^{-\left(\frac{\alpha}{z}\right)\omega}\right)}{\alpha\left(1-\alpha\right)\omega\left(\frac{1}{z}\right)\left(1-e^{-\left(\frac{1}{z}\right)\omega}\right)}\\
 & =\frac{1}{\alpha\left(1-\alpha\right)\omega}\lim_{h\to0}\frac{\left(1-e^{-h\left(1-\alpha\right)\omega}\right)\left(1-e^{-h\alpha\omega}\right)}{h\left(1-e^{-h\omega}\right)}\\
 & =\frac{1}{\alpha\left(1-\alpha\right)\omega}\lim_{h\to0}\frac{\left(1-\alpha\right)\omega e^{-h\left(1-\alpha\right)\omega}\left(1-e^{-h\alpha\omega}\right)+\left(1-e^{-h\left(1-\alpha\right)\omega}\right)\alpha\omega e^{-h\alpha\omega}}{1-e^{-h\omega}+h\omega e^{-h\omega}}\\
 & =\frac{1}{\alpha\left(1-\alpha\right)\omega}\lim_{h\to0}\frac{\left[\begin{array}{c}
-\left[\left(1-\alpha\right)\omega\right]^{2}e^{-h\left(1-\alpha\right)\omega}\left(1-e^{-h\alpha\omega}\right)+\left(1-\alpha\right)\omega e^{-h\left(1-\alpha\right)\omega}\alpha\omega e^{-h\alpha\omega}\\
+\left(1-\alpha\right)\omega e^{-h\left(1-\alpha\right)\omega}\alpha\omega e^{-h\alpha\omega}-\left(1-e^{-h\left(1-\alpha\right)\omega}\right)\left(\alpha\omega\right)^{2}e^{-h\alpha\omega}
\end{array}\right]}{\omega e^{-h\omega}+\omega e^{-h\omega}-h\omega^{2}e^{-h\omega}}\\
 & =\frac{1}{\alpha\left(1-\alpha\right)\omega}\left(\frac{2\left(1-\alpha\right)\alpha\omega^{2}}{2\omega}\right)\\
 & =1.
\end{align}
In the case that $\omega=0$, we already showed in \eqref{eq:limit-f-a-z-is-1-first}--\eqref{eq:limit-f-a-z-is-1-last}
that $\lim_{\omega\to0}f_{\alpha,z}(\omega)=1$ for all $\alpha\in\left(0,1\right)\cup\left(1,\infty\right)$
and $z>0$, so that $\lim_{z\to\infty}f_{\alpha,z}(0)=1$. It follows
from Levy's continuity theorem that, for all $\alpha\in\left(0,1\right)$,
the random variable $X_{\alpha,z}\sim q_{\alpha,z}$ converges in
distribution to the random variable $X\sim p$ in the limit $z\to\infty$.
\end{proof}

\section{$\alpha$-$z$ Information matrices of time-evolved states}

\label{sec:Information-matrices-time-evolved}In this section, I establish
a formula for the $\alpha$-$z$ information matrix of time-evolved
states, for all $\alpha\in\left(0,1\right)$ and $z>0$ (Theorem~\ref{thm:time-evolved-a-z-formula}).
By appealing to the former results of \cite{Minervini2025}, this
formula leads to a hybrid quantum--classical algorithm for estimating
the elements of the $\alpha$-$z$ information matrix of time-evolved
states. As such, this formula has applications in quantum machine
learning, namely, in natural gradient descent algorithms for performing
optimization using quantum evolution machines, as put forward in \cite{Minervini2025}.

Let us now evaluate the expression in \eqref{eq:gen-fisher-info-repeated}
for time-evolved states. Let
\begin{align}
H(\phi) & \coloneqq\sum_{j=1}^{L}\phi_{j}H_{j},\\
\rho & \coloneqq\frac{e^{-G}}{Z},\label{eq:initial-fixed-thermal-state}\\
\sigma(\phi) & \coloneqq e^{-iH(\phi)}\rho e^{iH(\phi)},\label{eq:time-evolved-state}
\end{align}
where $\phi_{j}\in\mathbb{R}$ and $H_{j}$ is Hermitian for all $j\in\left\{ 1,\ldots,L\right\} $,
$G$ is Hermitian, and $Z\coloneqq\Tr\!\left[e^{-G}\right]$. The
state $\sigma(\phi)$ is what we refer to as a ``time-evolved state.''
For analyzing these states, the following quantum channel $\Psi_{\phi}$
and its corresponding adjoint channel $\Psi_{\phi}^{\dag}$ play a
role:
\begin{align}
\Psi_{\phi}(X) & \coloneqq\int_{0}^{1}dt\ e^{iH(\phi)t}Xe^{-iH(\phi)t}=e^{iH(\phi)}\Psi_{\phi}^{\dag}(H_{i})e^{-iH(\phi)},\label{eq:time-evolved-channel}\\
\Psi_{\phi}^{\dag}(X) & \coloneqq\int_{0}^{1}dt\ e^{-iH(\phi)t}Xe^{iH(\phi)t},
\end{align}
As a consequence of Lemma \ref{lem:general-fisher-time-evolved-states}
and Corollary \ref{cor:fourier-trans-a-z-other}, the following theorem
holds for the $\alpha$-$z$ information matrices of time-evolved
states:
\begin{thm}
\label{thm:time-evolved-a-z-formula}For time-evolved states of the
form in \eqref{eq:time-evolved-state}, the following equality holds
for all $\alpha\in\left(0,1\right)$ and $z>0$:
\begin{equation}
\left[I_{\alpha,z}(\phi)\right]_{i,j}=\left\langle \left[\Phi_{p_{\alpha,z}}\!\left(\Psi_{\phi}\!\left(H_{i}\right)\right),\left[G,\Psi_{\phi}(H_{j})\right]\right]\right\rangle _{\rho},
\end{equation}
where the channel $\Psi_{\phi}$ is defined in \eqref{eq:time-evolved-channel},
the channel $\Phi_{q_{\alpha,z}}$ is defined as
\begin{equation}
\Phi_{p_{\alpha,z}}\!\left(X\right)\coloneqq\int_{-\infty}^{\infty}dt\:p_{\alpha,z}(t)e^{-iGt}Xe^{iGt},\label{eq:time-evolved-p-a-z-channel}
\end{equation}
and $p_{\alpha,z}$ is the probability density function defined in
\eqref{eq:high-peak-tent-a-z}.
\end{thm}

As stated above, Theorem~\ref{thm:time-evolved-a-z-formula} follows
in part from Lemma \ref{lem:general-fisher-time-evolved-states},
which establishes a formula for a general information matrix $I_{\zeta}(\phi)$
of time-evolved states, under the assumption that a certain Fourier
transform exists. Additionally, Theorem~\ref{thm:time-evolved-a-z-formula}
follows from Corollary \ref{cor:fourier-trans-a-z-other}, which precisely
determines the needed Fourier transform for all $\alpha\in\left(0,1\right)$
and $z>0$ when the function $\zeta(x,y)$ is set to $\zeta_{\alpha,z}\!\left(x,y\right)$,
as defined in \eqref{eq:limit-for-a-z-eigenval-func}.

The following corollary was established already in \eqref{eq:fourier-trans-f-a-z-to-p-a-z}
(therein, recall the definition of $f_{\alpha,z}(\omega)$ in \eqref{eq:def-f-a-z-func}):
\begin{cor}
\label{cor:fourier-trans-a-z-other}For all $\alpha\in\left(0,1\right)$
and $z>0$, the Fourier transform of the function
\begin{equation}
\omega\mapsto\frac{\zeta_{\alpha,z}\!\left(e^{-\omega},1\right)\left(1-e^{-\omega}\right)}{\omega},
\end{equation}
where $\zeta_{\alpha,z}\!\left(e^{-\omega},1\right)$ is defined in
\eqref{eq:limit-for-a-z-eigenval-func}, is equal to the probability
density function $p_{\alpha,z}(t)$ defined in \eqref{eq:high-peak-tent-a-z}.
That is, for all $\alpha\in\left(0,1\right)$ and $z>0$,
\begin{multline}
\frac{1}{2\pi}\int_{-\infty}^{\infty}d\omega\ \frac{z\left(1-e^{-\left(\frac{1-\alpha}{z}\right)\omega}\right)\left(1-e^{-\left(\frac{\alpha}{z}\right)\omega}\right)}{\alpha\left(1-\alpha\right)\omega\left(1-e^{-\left(\frac{1}{z}\right)\omega}\right)}e^{i\omega t}\\
=\frac{z}{2\pi\alpha\left(1-\alpha\right)}\ln\!\left(1+\left(\frac{\sin(\pi\alpha)}{\sinh(\pi zt)}\right)^{2}\right)\eqqcolon p_{\alpha,z}(t).
\end{multline}
\end{cor}

\begin{rem}
We recover the result from \cite[Theorem~19]{Minervini2025} as a
special case of Theorem~\ref{thm:time-evolved-a-z-formula} in the
limit $\alpha\to1$ for fixed $z>0$ or in the limit $z\to\infty$
for fixed $\alpha\in\left(0,1\right)$. Indeed, Lemma \ref{lem:fourier-trans-a-z}
states that the characteristic function of $p_{\alpha,z}(t)$ converges
to one everywhere in these limits, so that the expression
\begin{equation}
\frac{\zeta_{\alpha,z}\!\left(e^{-\left(\mu_{k}-\mu_{\ell}\right)},1\right)\left(1-e^{-\left(\mu_{k}-\mu_{\ell}\right)}\right)}{\mu_{k}-\mu_{\ell}}
\end{equation}
in \eqref{eq:evol-mach-gen-proof-step-2} is equal to one for all
$\mu_{k},\mu_{\ell}\in\mathbb{R}$ in either of these limits. Then
proceeding through the rest of the analysis after \eqref{eq:evol-mach-gen-proof-step-2},
one finds that
\begin{equation}
\left[I_{\KM}(\phi)\right]_{i,j}=\lim_{\alpha\to1}\left[I_{\alpha,z}(\phi)\right]_{i,j}=\lim_{z\to\infty}\left[I_{\alpha,z}(\phi)\right]_{i,j}=\left\langle \left[\left[\Psi_{\phi}^{\dag}(H_{j}),G\right],\Psi_{\phi}^{\dag}(H_{i})\right]\right\rangle _{\rho},
\end{equation}
consistent with the expression reported in \cite[Theorem~19]{Minervini2025}.
\begin{rem}
By using the same methods use to prove Theorem~\ref{thm:time-evolved-a-z-formula},
the results from \cite{Minervini2025a}, regarding information matrices
for parameterized quantum circuits, also generalize to $\alpha$-$z$
information matrices. Indeed, by adopting the same notation used in
\cite{Minervini2025a}, we arrive at the following statement: Let
$\rho$ be an initial thermal state of the form in \eqref{eq:initial-fixed-thermal-state}.
For a general layered parameterized circuit of the form in \cite[Eq.~(1)]{Minervini2025a},
the following equality holds for all $\alpha\in\left(0,1\right)$
and $z>0$:
\begin{equation}
\left[I_{\alpha,z}(\phi)\right]_{i,j}=\left\langle \left[\Phi_{p_{\alpha,z}}\!\left(\mathcal{U}_{R_{i}}^{\dag}\!\left(H_{i}\right)\right),\left[G,\mathcal{U}_{R_{j}}^{\dag}(H_{j})\right]\right]\right\rangle _{\rho},
\end{equation}
where $\mathcal{U}_{R_{i}}^{\dag}$ denotes the unitary channel defined
in \cite[Eq.~(8)]{Minervini2025a} and $\Phi_{p_{\alpha,z}}$ is the
channel defined in \eqref{eq:time-evolved-p-a-z-channel}.
\end{rem}

\end{rem}

\subsection{General formula for information matrices of time-evolved states}
\begin{lem}
\label{lem:general-fisher-time-evolved-states}Let $\zeta(x,y)$ be
a function defined for $x,y>0$ satisfying the properties stated in
\eqref{eq:zeta-func-prop-1-later}--\eqref{eq:zeta-kappa-normalization},
and suppose that the Fourier transform of the function
\begin{equation}
\omega\mapsto\frac{\zeta\!\left(e^{-\omega},1\right)\left(1-e^{-\omega}\right)}{\omega}
\end{equation}
exists, where $\omega\in\mathbb{R}$. Then, for time-evolved states
of the form in \eqref{eq:time-evolved-state}, the following equality
holds:
\begin{equation}
\left[I_{\zeta}(\phi)\right]_{i,j}=\left\langle \left[\Phi_{g}\!\left(\Psi_{\phi}\!\left(H_{i}\right)\right),\left[G,\Psi_{\phi}(H_{j})\right]\right]\right\rangle _{\rho},
\end{equation}
where $I_{\zeta}(\phi)$ is defined in \eqref{eq:gen-fisher-info-repeated},
the channel $\Psi_{\phi}$ is defined in \eqref{eq:time-evolved-channel},
\begin{equation}
\Phi_{g}\!\left(X\right)\coloneqq\int_{-\infty}^{\infty}dt\:g(t)e^{-iGt}Xe^{iGt},
\end{equation}
and $g(t)$ is a real-valued function satisfying the following for
all $\omega\in\mathbb{R}$:
\begin{align}
\int_{-\infty}^{\infty}dt\:g(t)e^{it\omega} & =\frac{\zeta\!\left(e^{-\omega},1\right)\left(1-e^{-\omega}\right)}{\omega},\\
\int_{-\infty}^{\infty}dt\:g(t) & =\kappa.\label{eq:normalization-f-fourier-1}
\end{align}
\end{lem}

\begin{proof}
As a consequence of Proposition \ref{prop:deriv-exp}, consider that
\begin{align}
 & \frac{\partial}{\partial\phi}\sigma(\phi)\nonumber \\
 & =\frac{\partial}{\partial\phi}\left(e^{-iH(\phi)}\rho e^{iH(\phi)}\right)\\
 & =\left(\frac{\partial}{\partial\phi}e^{-iH(\phi)}\right)\rho e^{iH(\phi)}+e^{-iH(\phi)}\rho\left(\frac{\partial}{\partial\phi}e^{iH(\phi)}\right)\\
 & =-i\int_{0}^{1}dt\ e^{-iH(\phi)t}H_{i}e^{-iH(\phi)\left(1-t\right)}\rho e^{iH(\phi)}+ie^{-iH(\phi)}\rho\int_{0}^{1}dt\ e^{iH(\phi)t}H_{i}e^{iH(\phi)\left(1-t\right)}\\
 & =-i\int_{0}^{1}dt\ e^{-iH(\phi)t}H_{i}e^{iH(\phi)t}e^{-iH(\phi)}\rho e^{iH(\phi)}+ie^{-iH(\phi)}\rho\int_{0}^{1}dt\ e^{iH(\phi)\left(1-t\right)}H_{i}e^{iH(\phi)t}\\
 & =-i\Psi_{\phi}^{\dag}(H_{i})\sigma(\phi)+i\sigma(\phi)\Psi_{\phi}^{\dag}(H_{i})\\
 & =i\left[\sigma(\phi),\Psi_{\phi}^{\dag}(H_{i})\right].\label{eq:deriv-evolution-machine}
\end{align}
Let $G=\sum_{k}\mu_{k}\widetilde{\Pi}_{k}$ be a spectral decomposition
of $G$, where $\mu_{k}$ is an eigenvalue and $\widetilde{\Pi}_{k}$
is an eigenprojection. Then $\rho=\sum_{k}\frac{e^{-\mu_{k}}}{Z}\widetilde{\Pi}_{k}$
is a spectral decomposition of $\rho$, and set $\lambda_{k}\coloneqq\frac{e^{-\mu_{k}}}{Z}$.
Defining $\Pi_{k}\coloneqq e^{-iH(\phi)}\widetilde{\Pi}_{k}e^{iH(\phi)}$,
it follows that $\sigma(\phi)=\sum_{k}\lambda_{k}\Pi_{k}$ is a spectral
decomposition of $\sigma(\phi)$. Applying \eqref{eq:deriv-evolution-machine},
we conclude that
\begin{align}
\Pi_{k}\left(\partial_{i}\sigma(\phi)\right)\Pi_{\ell} & =\Pi_{k}i\left[\sigma(\phi),\Psi_{\phi}^{\dag}(H_{i})\right]\Pi_{\ell}\\
 & =i\Pi_{k}\sigma(\phi)\Psi_{\phi}^{\dag}(H_{i})\Pi_{\ell}-i\Pi_{k}\Psi_{\phi}^{\dag}(H_{i})\sigma(\phi)\Pi_{\ell}\\
 & =i\lambda_{k}\Pi_{k}\Psi_{\phi}^{\dag}(H_{i})\Pi_{\ell}-i\lambda_{\ell}\Pi_{k}\Psi_{\phi}^{\dag}(H_{i})\Pi_{\ell}\\
 & =i\left(\lambda_{k}-\lambda_{\ell}\right)\Pi_{k}\Psi_{\phi}^{\mathbf{\dag}}(H_{i})\Pi_{\ell}.
\end{align}
Then, considering \eqref{eq:gen-fisher-info-repeated}, we find that
\begin{align}
\left[I_{\zeta}(\phi)\right]_{i,j} & =\sum_{k,\ell}\zeta(\lambda_{k},\lambda_{\ell})\Tr\!\left[\Pi_{k}\left(\partial_{i}\sigma(\phi)\right)\Pi_{\ell}\left(\partial_{j}\sigma(\phi)\right)\right]\\
 & =\sum_{k,\ell}\zeta(\lambda_{k},\lambda_{\ell})i\left(\lambda_{k}-\lambda_{\ell}\right)\Tr\!\left[\Pi_{k}\Psi_{\phi}^{\dag}(H_{i})\Pi_{\ell}\left(\partial_{j}\sigma(\phi)\right)\right]\\
 & =\sum_{k,\ell}\zeta(\lambda_{k},\lambda_{\ell})i\left(\lambda_{k}-\lambda_{\ell}\right)\Tr\!\left[\Psi_{\phi}(H_{i})\Pi_{\ell}\left(\partial_{j}\sigma(\phi)\right)\Pi_{k}\right]\\
 & =\sum_{k,\ell}\zeta(\lambda_{k},\lambda_{\ell})i\left(\lambda_{k}-\lambda_{\ell}\right)i\left(\lambda_{\ell}-\lambda_{k}\right)\Tr\!\left[\Psi_{\phi}^{\dag}(H_{i})\Pi_{\ell}\Psi_{\phi}^{\dag}(H_{j})\Pi_{k}\right]\\
 & =\sum_{k,\ell}\zeta(\lambda_{k},\lambda_{\ell})\left(\lambda_{k}-\lambda_{\ell}\right)^{2}\Tr\!\left[\Pi_{k}\Psi_{\phi}^{\dag}(H_{i})\Pi_{\ell}\Psi_{\phi}^{\dag}(H_{j})\right]\label{eq:evol-mach-gen-proof-step-1}
\end{align}
Note that
\begin{align}
\Tr\!\left[\Pi_{k}\Psi_{\phi}^{\dag}(H_{i})\Pi_{\ell}\Psi_{\phi}^{\dag}(H_{j})\right] & =\Tr\!\left[e^{-iH(\phi)}\widetilde{\Pi}_{k}e^{iH(\phi)}\Psi_{\phi}^{\dag}(H_{i})e^{-iH(\phi)}\widetilde{\Pi}_{\ell}e^{iH(\phi)}\Psi_{\phi}^{\dag}(H_{j})\right]\\
 & =\Tr\!\left[\widetilde{\Pi}_{k}e^{iH(\phi)}\Psi_{\phi}^{\dag}(H_{i})e^{-iH(\phi)}\widetilde{\Pi}_{\ell}e^{iH(\phi)}\Psi_{\phi}^{\dag}(H_{j})e^{-iH(\phi)}\right]\\
 & =\Tr\!\left[\widetilde{\Pi}_{k}\Psi_{\phi}(H_{i})\widetilde{\Pi}_{\ell}\Psi_{\phi}(H_{j})\right],
\end{align}
where we applied \eqref{eq:time-evolved-channel}. Then, continuing
from \eqref{eq:evol-mach-gen-proof-step-1},
\begin{align}
 & \left[I_{\zeta}(\phi)\right]_{i,j}\nonumber \\
 & =\sum_{k,\ell}\zeta(\lambda_{k},\lambda_{\ell})\left(\lambda_{k}-\lambda_{\ell}\right)^{2}\Tr\!\left[\widetilde{\Pi}_{k}\Psi_{\phi}(H_{i})\widetilde{\Pi}_{\ell}\Psi_{\phi}(H_{j})\right]\\
 & =\sum_{k,\ell}\zeta\!\left(\frac{e^{-\mu_{k}}}{Z},\frac{e^{-\mu_{\ell}}}{Z}\right)\left(\frac{e^{-\mu_{k}}}{Z}-\frac{e^{-\mu_{\ell}}}{Z}\right)^{2}\Tr\!\left[\widetilde{\Pi}_{k}\Psi_{\phi}(H_{i})\widetilde{\Pi}_{\ell}\Psi_{\phi}(H_{j})\right]\\
 & =\sum_{k,\ell}Z\zeta\!\left(e^{-\mu_{k}},e^{-\mu_{\ell}}\right)\frac{1}{Z^{2}}\left(e^{-\mu_{k}}-e^{-\mu_{\ell}}\right)^{2}\Tr\!\left[\widetilde{\Pi}_{k}\Psi_{\phi}(H_{i})\widetilde{\Pi}_{\ell}\Psi_{\phi}(H_{j})\right]\\
 & =\sum_{k,\ell}\zeta\!\left(e^{-\mu_{k}},e^{-\mu_{\ell}}\right)\frac{1}{Z}\left(e^{-\mu_{k}}-e^{-\mu_{\ell}}\right)^{2}\Tr\!\left[\widetilde{\Pi}_{k}\Psi_{\phi}(H_{i})\widetilde{\Pi}_{\ell}\Psi_{\phi}(H_{j})\right]\\
 & =\sum_{k,\ell}e^{\mu_{\ell}}\zeta\!\left(e^{-\left(\mu_{k}-\mu_{\ell}\right)},1\right)\frac{e^{-2\mu_{\ell}}}{Z}\left(e^{-\left(\mu_{k}-\mu_{\ell}\right)}-1\right)^{2}\Tr\!\left[\widetilde{\Pi}_{k}\Psi_{\phi}(H_{i})\widetilde{\Pi}_{\ell}\Psi_{\phi}(H_{j})\right]\\
 & =\sum_{k,\ell}\zeta\!\left(e^{-\left(\mu_{k}-\mu_{\ell}\right)},1\right)\frac{e^{-\mu_{\ell}}}{Z}\left(1-e^{-\left(\mu_{k}-\mu_{\ell}\right)}\right)^{2}\Tr\!\left[\widetilde{\Pi}_{k}\Psi_{\phi}(H_{i})\widetilde{\Pi}_{\ell}\Psi_{\phi}(H_{j})\right]\\
 & =\sum_{k,\ell}\frac{\zeta\!\left(e^{-\left(\mu_{k}-\mu_{\ell}\right)},1\right)\left(1-e^{-\left(\mu_{k}-\mu_{\ell}\right)}\right)}{\mu_{k}-\mu_{\ell}}\frac{e^{-\mu_{\ell}}}{Z}\left(\mu_{k}-\mu_{\ell}\right)\left(1-e^{-\left(\mu_{k}-\mu_{\ell}\right)}\right)\Tr\!\left[\widetilde{\Pi}_{k}\Psi_{\phi}(H_{i})\widetilde{\Pi}_{\ell}\Psi_{\phi}(H_{j})\right]\\
 & =\sum_{k,\ell}\frac{\zeta\!\left(e^{-\left(\mu_{k}-\mu_{\ell}\right)},1\right)\left(1-e^{-\left(\mu_{k}-\mu_{\ell}\right)}\right)}{\mu_{k}-\mu_{\ell}}\left(\mu_{k}-\mu_{\ell}\right)\left(\frac{e^{-\mu_{\ell}}}{Z}-\frac{e^{-\mu_{k}}}{Z}\right)\Tr\!\left[\widetilde{\Pi}_{k}\Psi_{\phi}(H_{i})\widetilde{\Pi}_{\ell}\Psi_{\phi}(H_{j})\right]\label{eq:evol-mach-gen-proof-step-2}
\end{align}
Now suppose that $g(t)$ is a function satisfying 
\begin{equation}
\int_{-\infty}^{\infty}dt\ g(t)e^{it\omega}=\frac{\zeta\!\left(e^{-\omega},1\right)\left(1-e^{-\omega}\right)}{\omega}.
\end{equation}
To establish that $g(t)$ is real-valued, it suffices to prove that
its Fourier transform is even. To this end, consider that
\begin{align}
-\omega & \mapsto\frac{\zeta\!\left(e^{-\left(-\omega\right)},1\right)\left(1-e^{-\left(-\omega\right)}\right)}{-\omega}\\
 & =\frac{\zeta\!\left(e^{\omega},1\right)\left(1-e^{\omega}\right)}{-\omega}\\
 & =\frac{\zeta\!\left(1,e^{\omega}\right)\left(1-e^{\omega}\right)}{-\omega}\\
 & =\frac{e^{-\omega}\zeta\!\left(e^{-\omega},1\right)\left(1-e^{\omega}\right)}{-\omega}\\
 & =\frac{\zeta\!\left(e^{-\omega},1\right)\left(e^{-\omega}-1\right)}{-\omega}\\
 & =\frac{\zeta\!\left(e^{-\omega},1\right)\left(1-e^{-\omega}\right)}{\omega}.
\end{align}
Then, continuing from \eqref{eq:evol-mach-gen-proof-step-2}, 
\begin{align}
 & \left[I_{\zeta}(\phi)\right]_{i,j}\nonumber \\
 & =\sum_{k,\ell}\int_{-\infty}^{\infty}dt\ g(t)e^{it\left(\mu_{k}-\mu_{\ell}\right)}\left(\mu_{k}-\mu_{\ell}\right)\left(\frac{e^{-\mu_{\ell}}}{Z}-\frac{e^{-\mu_{k}}}{Z}\right)\Tr\!\left[\widetilde{\Pi}_{k}\Psi_{\phi}(H_{i})\widetilde{\Pi}_{\ell}\Psi_{\phi}(H_{j})\right]\\
 & =\int_{-\infty}^{\infty}dt\ g(t)\sum_{k,\ell}\left(\mu_{k}-\mu_{\ell}\right)\left(\frac{e^{-\mu_{\ell}}}{Z}-\frac{e^{-\mu_{k}}}{Z}\right)\Tr\!\left[e^{it\mu_{k}}\widetilde{\Pi}_{k}\Psi_{\phi}(H_{i})e^{-it\mu_{\ell}}\widetilde{\Pi}_{\ell}\Psi_{\phi}(H_{j})\right]\\
 & =-\int_{-\infty}^{\infty}dt\ g(t)\sum_{k,\ell}\mu_{k}\frac{e^{-\mu_{k}}}{Z}\Tr\!\left[e^{it\mu_{k}}\widetilde{\Pi}_{k}\Psi_{\phi}(H_{i})e^{-it\mu_{\ell}}\widetilde{\Pi}_{\ell}\Psi_{\phi}(H_{j})\right]\nonumber \\
 & \qquad+\int_{-\infty}^{\infty}dt\ g(t)\sum_{k,\ell}\mu_{k}\frac{e^{-\mu_{\ell}}}{Z}\Tr\!\left[e^{it\mu_{k}}\widetilde{\Pi}_{k}\Psi_{\phi}(H_{i})e^{-it\mu_{\ell}}\widetilde{\Pi}_{\ell}\Psi_{\phi}(H_{j})\right]\nonumber \\
 & \qquad+\int_{-\infty}^{\infty}dt\ g(t)\sum_{k,\ell}\mu_{\ell}\frac{e^{-\mu_{k}}}{Z}\Tr\!\left[e^{it\mu_{k}}\widetilde{\Pi}_{k}\Psi_{\phi}(H_{i})e^{-it\mu_{\ell}}\widetilde{\Pi}_{\ell}\Psi_{\phi}(H_{j})\right]\nonumber \\
 & \qquad-\int_{-\infty}^{\infty}dt\ g(t)\sum_{k,\ell}\mu_{\ell}\frac{e^{-\mu_{\ell}}}{Z}\Tr\!\left[e^{it\mu_{k}}\widetilde{\Pi}_{k}\Psi_{\phi}(H_{i})e^{-it\mu_{\ell}}\widetilde{\Pi}_{\ell}\Psi_{\phi}(H_{j})\right]\\
 & =-\int_{-\infty}^{\infty}dt\ g(t)\Tr\!\left[\left(\sum_{k}\mu_{k}\frac{e^{-\mu_{k}}}{Z}e^{it\mu_{k}}\right)\widetilde{\Pi}_{k}\Psi_{\phi}(H_{i})\left(\sum_{\ell}e^{-it\mu_{\ell}}\widetilde{\Pi}_{\ell}\right)\Psi_{\phi}(H_{j})\right]\nonumber \\
 & \qquad+\int_{-\infty}^{\infty}dt\ g(t)\Tr\!\left[\left(\sum_{k}\mu_{k}e^{it\mu_{k}}\widetilde{\Pi}_{k}\right)\Psi_{\phi}(H_{i})\left(\sum_{\ell}e^{-it\mu_{\ell}}\frac{e^{-\mu_{\ell}}}{Z}\widetilde{\Pi}_{\ell}\right)\Psi_{\phi}(H_{j})\right]\nonumber \\
 & \qquad+\int_{-\infty}^{\infty}dt\ g(t)\Tr\!\left[\left(\sum_{k}\frac{e^{-\mu_{k}}}{Z}e^{it\mu_{k}}\widetilde{\Pi}_{k}\right)\Psi_{\phi}(H_{i})\left(\sum_{\ell}e^{-it\mu_{\ell}}\mu_{\ell}\widetilde{\Pi}_{\ell}\right)\Psi_{\phi}(H_{j})\right]\nonumber \\
 & \qquad-\int_{-\infty}^{\infty}dt\ g(t)\Tr\!\left[\left(\sum_{k}e^{it\mu_{k}}\widetilde{\Pi}_{k}\right)\Psi_{\phi}(H_{i})\left(\sum_{\ell}e^{-it\mu_{\ell}}\mu_{\ell}\frac{e^{-\mu_{\ell}}}{Z}\widetilde{\Pi}_{\ell}\right)\Psi_{\phi}(H_{j})\right]\\
 & =-\int_{-\infty}^{\infty}dt\ g(t)\Tr\!\left[G\rho e^{itG}\Psi_{\phi}(H_{i})e^{-itG}\Psi_{\phi}(H_{j})\right]\nonumber \\
 & \qquad+\int_{-\infty}^{\infty}dt\ g(t)\Tr\!\left[Ge^{itG}\Psi_{\phi}(H_{i})e^{-itG}\rho\Psi_{\phi}(H_{j})\right]\nonumber \\
 & \qquad+\int_{-\infty}^{\infty}dt\ g(t)\Tr\!\left[\rho e^{itG}\Psi_{\phi}(H_{i})e^{-itG}G\Psi_{\phi}(H_{j})\right]\nonumber \\
 & \qquad-\int_{-\infty}^{\infty}dt\ g(t)\Tr\!\left[e^{itG}\Psi_{\phi}(H_{i})e^{-itG}G\rho\Psi_{\phi}(H_{j})\right]\\
 & =-\Tr\!\left[G\rho\,\Phi_{g}\!\left(\Psi_{\phi}(H_{i})\right)\Psi_{\phi}(H_{j})\right]+\Tr\!\left[G\,\Phi_{g}\!\left(\Psi_{\phi}(H_{i})\right)\rho\Psi_{\phi}(H_{j})\right]\nonumber \\
 & \qquad+\Tr\!\left[\rho\,\Phi_{g}\!\left(\Psi_{\phi}(H_{i})\right)G\Psi_{\phi}(H_{j})\right]-\Tr\!\left[\Phi_{g}\!\left(\Psi_{\phi}(H_{i})\right)G\rho\Psi_{\phi}(H_{j})\right]\\
 & =-\Tr\!\left[\Phi_{g}\!\left(\Psi_{\phi}(H_{i})\right)\Psi_{\phi}(H_{j})G\rho\right]+\Tr\!\left[\Psi_{\phi}(H_{j})G\,\Phi_{g}\!\left(\Psi_{\phi}(H_{i})\right)\rho\right]\nonumber \\
 & \qquad+\Tr\!\left[\Phi_{g}\!\left(\Psi_{\phi}(H_{i})\right)G\Psi_{\phi}(H_{j})\rho\right]-\Tr\!\left[G\Psi_{\phi}(H_{j})\Phi_{g}\!\left(\Psi_{\phi}(H_{i})\right)\rho\right]\label{eq:long-proof-time-evolved-general}\\
 & =\Tr\!\left[\Phi_{g}\!\left(\Psi_{\phi}(H_{i})\right)\left[G,\Psi_{\phi}(H_{j})\right]\rho\right]-\Tr\!\left[\left[G,\Psi_{\phi}(H_{j})\right]\Phi_{g}\!\left(\Psi_{\phi}(H_{i})\right)\rho\right]\\
 & =\Tr\!\left[\left[\Phi_{g}\!\left(\Psi_{\phi}(H_{i})\right),\left[G,\Psi_{\phi}(H_{j})\right]\right]\rho\right]\\
 & =\left\langle \left[\Phi_{g}\!\left(\Psi_{\phi}(H_{i})\right),\left[G,\Psi_{\phi}(H_{j})\right]\right]\right\rangle _{\rho}.
\end{align}
For the equality in \eqref{eq:long-proof-time-evolved-general}, we
used that $\left[\rho,G\right]=0$. Finally, the claim in \eqref{eq:normalization-f-fourier-1}
follows because
\begin{align}
\int_{-\infty}^{\infty}dt\:g(t) & =\lim_{\omega\to0}\frac{\zeta\!\left(e^{-\omega},1\right)\left(1-e^{-\omega}\right)}{\omega}\\
 & =\zeta\!\left(1,1\right)\left(\lim_{\omega\to0}\frac{1-e^{-\omega}}{\omega}\right)\\
 & =\kappa\left(-\left.\frac{d}{d\omega}e^{-\omega}\right|_{\omega=0}\right)\\
 & =\kappa,
\end{align}
thus concluding the proof.
\end{proof}

\section{Information matrices of classical--quantum states}

\label{sec:c-q-states}In this section, I show how the various information
matrices considered in this paper decompose whenever the underlying
parameterized family of states have a classical--quantum structure,
generalizing the previous finding reported for the single-parameter
case in \cite[Proposition~7]{Katariya2021}. In particular, the various
information matrices decompose as a sum of a classical part and a
quantum part. For the purposes of this section, I adopt a revised
notation in which the family on which the information matrix is being
evaluated is explicitly written. That is, rather than the abbreviated
notation used in \eqref{eq:fisher-info-from-div-quantum}, I write
the following instead:
\begin{equation}
\left[I_{\boldsymbol{D}}(\theta;(\rho(\theta))_{\theta\in\Theta})\right]_{i,j}\coloneqq\left.\frac{\partial^{2}}{\partial\varepsilon_{i}\partial\varepsilon_{j}}\boldsymbol{D}(\rho(\theta)\|\rho(\theta+\varepsilon))\right|_{\varepsilon=0},
\end{equation}
where $\Theta\subseteq\mathbb{R}^{L}$ is open.

A parameterized family $\left(\rho_{XA}(\theta)\right)_{\theta\in\Theta}$
of classical-quantum states has the following form:
\begin{equation}
\rho_{XA}(\theta)\coloneqq\sum_{x\in\mathcal{X}}p_{\theta}(x)|x\rangle\!\langle x|\otimes\rho_{x}(\theta),\label{eq:cq-family}
\end{equation}
where $\left(p_{\theta}\right)_{\theta\in\Theta}$ is a parameterized
family of probability distributions and, for all $x\in\mathcal{X}$,
$\left(\rho_{x}(\theta)\right)_{\theta\in\Theta}$ is a parameterized
family of quantum states.
\begin{thm}
\label{thm:cq-decomposition}Let $\left(\rho_{XA}(\theta)\right)_{\theta\in\Theta}$
be a parameterized family of second-order differentiable, positive
definite, classical-quantum states, of the form in \eqref{eq:cq-family}.
For all $\alpha\in\left(0,1\right)\cup\left(1,\infty\right)$ and
$z>0$, the following equality holds:
\begin{equation}
I_{\alpha,z}(\theta;(\rho_{XA}(\theta))_{\theta\in\Theta})=I_{F}(\theta;(p_{\theta})_{\theta\in\Theta})+\sum_{x\in\mathcal{X}}p_{\theta}(x)I_{\alpha,z}(\theta;(\rho_{x}(\theta))_{\theta\in\Theta}),\label{eq:cq-a-z}
\end{equation}
where the classical Fisher information matrix $I_{F}(\theta;(p_{\theta})_{\theta\in\Theta})$
is defined in \eqref{eq:classical-fisher-information-matrix}. Also,
the following equalities hold:
\begin{align}
I_{\KM}(\theta;(\rho_{XA}(\theta))_{\theta\in\Theta}) & =I_{F}(\theta;(p_{\theta})_{\theta\in\Theta})+\sum_{x\in\mathcal{X}}p_{\theta}(x)I_{\KM}(\theta;(\rho_{x}(\theta))_{\theta\in\Theta}),\label{eq:cq-KM}\\
I_{\RLD}(\theta;(\rho_{XA}(\theta))_{\theta\in\Theta}) & =I_{F}(\theta;(p_{\theta})_{\theta\in\Theta})+\sum_{x\in\mathcal{X}}p_{\theta}(x)I_{\RLD}(\theta;(\rho_{x}(\theta))_{\theta\in\Theta}).\label{eq:cq-RLD}
\end{align}
\end{thm}

\begin{proof}
By applying \eqref{eq:first-reduction-a-z}, it follows that
\begin{multline}
\left.\frac{\partial^{2}}{\partial\varepsilon_{i}\partial\varepsilon_{j}}D_{\alpha,z}(\rho_{XA}(\theta)\|\rho_{XA}(\theta+\varepsilon))\right|_{\varepsilon=0}\\
=\frac{1}{\alpha-1}\left[\left.\frac{\partial^{2}}{\partial\varepsilon_{i}\partial\varepsilon_{j}}\Tr\!\left[\left(\rho_{XA}(\theta)^{\frac{\alpha}{2z}}\rho_{XA}(\theta+\varepsilon)^{\frac{1-\alpha}{z}}\rho_{XA}(\theta)^{\frac{\alpha}{2z}}\right)^{z}\right]\right|_{\varepsilon=0}\right].
\end{multline}
Now consider that
\begin{align}
 & \Tr\!\left[\left(\rho_{XA}(\theta)^{\frac{\alpha}{2z}}\rho_{XA}(\theta+\varepsilon)^{\frac{1-\alpha}{z}}\rho_{XA}(\theta)^{\frac{\alpha}{2z}}\right)^{z}\right]\nonumber \\
 & =\Tr\!\left[\left(\begin{array}{c}
\left(\sum_{x}|x\rangle\!\langle x|\otimes p_{\theta}(x)\rho_{x}(\theta)\right)^{\frac{\alpha}{2z}}\left(\sum_{x'}|x'\rangle\!\langle x'|\otimes p_{\theta+\varepsilon}(x')\rho_{x'}(\theta+\varepsilon)\right)^{\frac{1-\alpha}{z}}\times\\
\left(\sum_{x''}|x''\rangle\!\langle x''|\otimes p_{\theta}(x'')\rho_{x''}(\theta)\right)^{\frac{\alpha}{2z}}
\end{array}\right)^{z}\right]\\
 & =\Tr\!\left[\left(\begin{array}{c}
\left(\sum_{x}|x\rangle\!\langle x|\otimes\left[p_{\theta}(x)\rho_{x}(\theta)\right]^{\frac{\alpha}{2z}}\right)\left(\sum_{x'}|x'\rangle\!\langle x'|\otimes\left[p_{\theta+\varepsilon}(x')\rho_{x'}(\theta+\varepsilon)\right]^{\frac{1-\alpha}{z}}\right)\times\\
\left(\sum_{x''}|x''\rangle\!\langle x''|\otimes\left[p_{\theta}(x'')\rho_{x''}(\theta)\right]^{\frac{\alpha}{2z}}\right)
\end{array}\right)^{z}\right]\\
 & =\Tr\!\left[\left(\begin{array}{c}
\sum_{x}|x\rangle\!\langle x|\otimes\left(\left[p_{\theta}(x)\rho_{x}(\theta)\right]^{\frac{\alpha}{2z}}\left[p_{\theta+\varepsilon}(x)\rho_{x}(\theta+\varepsilon)\right]^{\frac{1-\alpha}{z}}\left[p_{\theta}(x)\rho_{x}(\theta)\right]^{\frac{\alpha}{2z}}\right)\end{array}\right)^{z}\right]\\
 & =\Tr\!\left[\sum_{x}|x\rangle\!\langle x|\otimes\left(\left[p_{\theta}(x)\rho_{x}(\theta)\right]^{\frac{\alpha}{2z}}\left[p_{\theta+\varepsilon}(x)\rho_{x}(\theta+\varepsilon)\right]^{\frac{1-\alpha}{z}}\left[p_{\theta}(x)\rho_{x}(\theta)\right]^{\frac{\alpha}{2z}}\right)^{z}\right]\\
 & =\sum_{x}\Tr\!\left[\left(\left[p_{\theta}(x)\rho_{x}(\theta)\right]^{\frac{\alpha}{2z}}\left[p_{\theta+\varepsilon}(x)\rho_{x}(\theta+\varepsilon)\right]^{\frac{1-\alpha}{z}}\left[p_{\theta}(x)\rho_{x}(\theta)\right]^{\frac{\alpha}{2z}}\right)^{z}\right]\\
 & =\sum_{x}p_{\theta}(x)^{\alpha}p_{\theta+\varepsilon}(x)^{1-\alpha}\Tr\!\left[\left(\left[\rho_{x}(\theta)\right]^{\frac{\alpha}{2z}}\left[\rho_{x}(\theta+\varepsilon)\right]^{\frac{1-\alpha}{z}}\left[\rho_{x}(\theta)\right]^{\frac{\alpha}{2z}}\right)^{z}\right].\nonumber 
\end{align}
Then we find that
\begin{align}
 & \frac{\partial^{2}}{\partial\varepsilon_{i}\partial\varepsilon_{j}}\Tr\!\left[\left(\rho_{XA}(\theta)^{\frac{\alpha}{2z}}\rho_{XA}(\theta+\varepsilon)^{\frac{1-\alpha}{z}}\rho_{XA}(\theta)^{\frac{\alpha}{2z}}\right)^{z}\right]\nonumber \\
 & =\frac{\partial^{2}}{\partial\varepsilon_{i}\partial\varepsilon_{j}}\left(\sum_{x}p_{\theta}(x)^{\alpha}p_{\theta+\varepsilon}(x)^{1-\alpha}\Tr\!\left[\left(\left[\rho_{x}(\theta)\right]^{\frac{\alpha}{2z}}\left[\rho_{x}(\theta+\varepsilon)\right]^{\frac{1-\alpha}{z}}\left[\rho_{x}(\theta)\right]^{\frac{\alpha}{2z}}\right)^{z}\right]\right)\\
 & =\frac{\partial}{\partial\varepsilon_{i}}\left(\begin{array}{c}
\sum_{x}p_{\theta}(x)^{\alpha}\frac{\partial}{\partial\varepsilon_{j}}p_{\theta+\varepsilon}(x)^{1-\alpha}\Tr\!\left[\left(\left[\rho_{x}(\theta)\right]^{\frac{\alpha}{2z}}\left[\rho_{x}(\theta+\varepsilon)\right]^{\frac{1-\alpha}{z}}\left[\rho_{x}(\theta)\right]^{\frac{\alpha}{2z}}\right)^{z}\right]\\
+\sum_{x}p_{\theta}(x)^{\alpha}p_{\theta+\varepsilon}(x)^{1-\alpha}\frac{\partial}{\partial\varepsilon_{j}}\Tr\!\left[\left(\left[\rho_{x}(\theta)\right]^{\frac{\alpha}{2z}}\left[\rho_{x}(\theta+\varepsilon)\right]^{\frac{1-\alpha}{z}}\left[\rho_{x}(\theta)\right]^{\frac{\alpha}{2z}}\right)^{z}\right]
\end{array}\right)\\
 & =\sum_{x}p_{\theta}(x)^{\alpha}\left(\frac{\partial^{2}}{\partial\varepsilon_{i}\partial\varepsilon_{j}}p_{\theta+\varepsilon}(x)^{1-\alpha}\right)\Tr\!\left[\left(\left[\rho_{x}(\theta)\right]^{\frac{\alpha}{2z}}\left[\rho_{x}(\theta+\varepsilon)\right]^{\frac{1-\alpha}{z}}\left[\rho_{x}(\theta)\right]^{\frac{\alpha}{2z}}\right)^{z}\right]\nonumber \\
 & \qquad+\sum_{x}p_{\theta}(x)^{\alpha}\left(\frac{\partial}{\partial\varepsilon_{j}}p_{\theta+\varepsilon}(x)^{1-\alpha}\right)\left(\frac{\partial}{\partial\varepsilon_{i}}\Tr\!\left[\left(\left[\rho_{x}(\theta)\right]^{\frac{\alpha}{2z}}\left[\rho_{x}(\theta+\varepsilon)\right]^{\frac{1-\alpha}{z}}\left[\rho_{x}(\theta)\right]^{\frac{\alpha}{2z}}\right)^{z}\right]\right)\nonumber \\
 & \qquad+\sum_{x}p_{\theta}(x)^{\alpha}p_{\theta+\varepsilon}(x)^{1-\alpha}\left(\frac{\partial^{2}}{\partial\varepsilon_{i}\partial\varepsilon_{j}}\Tr\!\left[\left(\left[\rho_{x}(\theta)\right]^{\frac{\alpha}{2z}}\left[\rho_{x}(\theta+\varepsilon)\right]^{\frac{1-\alpha}{z}}\left[\rho_{x}(\theta)\right]^{\frac{\alpha}{2z}}\right)^{z}\right]\right).
\end{align}
Now evaluating this last expression at $\varepsilon=0$ gives 
\begin{align}
 & \left.\left(\sum_{x}p_{\theta}(x)^{\alpha}\left(\frac{\partial^{2}}{\partial\varepsilon_{i}\partial\varepsilon_{j}}p_{\theta+\varepsilon}(x)^{1-\alpha}\right)\Tr\!\left[\left(\left[\rho_{x}(\theta)\right]^{\frac{\alpha}{2z}}\left[\rho_{x}(\theta+\varepsilon)\right]^{\frac{1-\alpha}{z}}\left[\rho_{x}(\theta)\right]^{\frac{\alpha}{2z}}\right)^{z}\right]\right)\right|_{\varepsilon=0}\nonumber \\
 & \qquad+\left.\left(\sum_{x}p_{\theta}(x)^{\alpha}\left(\frac{\partial}{\partial\varepsilon_{j}}p_{\theta+\varepsilon}(x)^{1-\alpha}\right)\left(\frac{\partial}{\partial\varepsilon_{i}}\Tr\!\left[\left(\left[\rho_{x}(\theta)\right]^{\frac{\alpha}{2z}}\left[\rho_{x}(\theta+\varepsilon)\right]^{\frac{1-\alpha}{z}}\left[\rho_{x}(\theta)\right]^{\frac{\alpha}{2z}}\right)^{z}\right]\right)\right)\right|_{\varepsilon=0}\nonumber \\
 & \qquad+\left.\left(\sum_{x}p_{\theta}(x)^{\alpha}p_{\theta+\varepsilon}(x)^{1-\alpha}\left(\frac{\partial^{2}}{\partial\varepsilon_{i}\partial\varepsilon_{j}}\Tr\!\left[\left(\left[\rho_{x}(\theta)\right]^{\frac{\alpha}{2z}}\left[\rho_{x}(\theta+\varepsilon)\right]^{\frac{1-\alpha}{z}}\left[\rho_{x}(\theta)\right]^{\frac{\alpha}{2z}}\right)^{z}\right]\right)\right)\right|_{\varepsilon=0}\\
 & =\sum_{x}p_{\theta}(x)^{\alpha}\left(\left.\frac{\partial^{2}}{\partial\varepsilon_{i}\partial\varepsilon_{j}}p_{\theta+\varepsilon}(x)^{1-\alpha}\right|_{\varepsilon=0}\right)\Tr\!\left[\left(\left[\rho_{x}(\theta)\right]^{\frac{\alpha}{2z}}\left[\rho_{x}(\theta)\right]^{\frac{1-\alpha}{z}}\left[\rho_{x}(\theta)\right]^{\frac{\alpha}{2z}}\right)^{z}\right]\nonumber \\
 & \qquad+\sum_{x}p_{\theta}(x)^{\alpha}\left(\left.\frac{\partial}{\partial\varepsilon_{j}}p_{\theta+\varepsilon}(x)^{1-\alpha}\right|_{\varepsilon=0}\right)\left(\left.\frac{\partial}{\partial\varepsilon_{i}}\Tr\!\left[\left(\left[\rho_{x}(\theta)\right]^{\frac{\alpha}{2z}}\left[\rho_{x}(\theta+\varepsilon)\right]^{\frac{1-\alpha}{z}}\left[\rho_{x}(\theta)\right]^{\frac{\alpha}{2z}}\right)^{z}\right]\right|_{\varepsilon=0}\right)\nonumber \\
 & \qquad+\sum_{x}p_{\theta}(x)^{\alpha}p_{\theta}(x)^{1-\alpha}\left(\left.\frac{\partial^{2}}{\partial\varepsilon_{i}\partial\varepsilon_{j}}\Tr\!\left[\left(\left[\rho_{x}(\theta)\right]^{\frac{\alpha}{2z}}\left[\rho_{x}(\theta+\varepsilon)\right]^{\frac{1-\alpha}{z}}\left[\rho_{x}(\theta)\right]^{\frac{\alpha}{2z}}\right)^{z}\right]\right|_{\varepsilon=0}\right)\\
 & =\left[I(\theta;(p_{\theta})_{\theta\in\Theta})\right]_{i,j}+\sum_{x}p_{\theta}(x)\left[I_{\alpha,z}(\theta;(\rho_{x}(\theta))_{\theta\in\Theta})\right]_{i,j}.
\end{align}
Proofs of \eqref{eq:cq-KM} and \eqref{eq:cq-RLD} are similar.
\end{proof}
A simple consequence of Theorem~\ref{thm:cq-decomposition} is that
information matrices that satisfy \eqref{eq:cq-a-z}--\eqref{eq:cq-RLD}
are convex whenever the data-processing inequality holds:
\begin{cor}[Convexity]
\label{cor:convexity}For all $x\in\mathcal{X}$, let $\left(\rho_{x}(\theta)\right)_{\theta\in\Theta}$
be a parameterized family of second-order differentiable, positive-definite
states, let $\left(p(x)\right)_{x\in\mathcal{X}}$ be a probability
distribution, and let $\overline{\rho}(\theta)\coloneqq\sum_{x\in\mathcal{X}}p(x)\rho_{x}(\theta)$.
Then, for all $\alpha,z>0$ satisfying the conditions stated in Fact
\ref{fact:a-z-data-proc}, the following convexity inequality holds:
\begin{equation}
\sum_{x\in\mathcal{X}}p(x)I_{\alpha,z}(\theta;(\rho_{x}(\theta))_{\theta\in\Theta})\geq I_{\alpha,z}(\theta;(\overline{\rho}(\theta))_{\theta\in\Theta}).
\end{equation}
Additionally,
\begin{align}
\sum_{x\in\mathcal{X}}p(x)I_{\KM}(\theta;(\rho_{x}(\theta))_{\theta\in\Theta}) & \geq I_{\KM}(\theta;(\overline{\rho}(\theta))_{\theta\in\Theta}),\label{eq:convexity-KM}\\
\sum_{x\in\mathcal{X}}p(x)I_{\RLD}(\theta;(\rho_{x}(\theta))_{\theta\in\Theta}) & \geq I_{\RLD}(\theta;(\overline{\rho}(\theta))_{\theta\in\Theta}).\label{eq:convexity-RLD}
\end{align}
\end{cor}

\begin{proof}
This is a consequence of Theorem~\ref{thm:cq-decomposition} and
the data-processing inequality. Indeed, we can apply Theorem~\ref{thm:cq-decomposition}
to the following classical-quantum state family:
\begin{equation}
\sigma_{XA}(\theta)\coloneqq\sum_{x\in\mathcal{X}}p(x)|x\rangle\!\langle x|\otimes\rho_{x}(\theta),
\end{equation}
to conclude that
\begin{align}
I_{\alpha,z}(\theta;(\sigma_{XA}(\theta))_{\theta\in\Theta}) & =I_{F}(\theta;(p)_{\theta\in\Theta})+\sum_{x\in\mathcal{X}}p(x)I_{\alpha,z}(\theta;(\rho_{x}(\theta))_{\theta\in\Theta})\\
 & =\sum_{x\in\mathcal{X}}p(x)I_{\alpha,z}(\theta;(\rho_{x}(\theta))_{\theta\in\Theta}),
\end{align}
while noting that $I_{F}(\theta;(p)_{\theta\in\Theta})=0$ because
the family $(p)_{\theta\in\Theta}$ has no dependence on the parameter
vector $\theta$. Applying the data-processing inequality for $I_{\alpha,z}$,
we then find that
\begin{align}
I_{\alpha,z}(\theta;(\sigma_{XA}(\theta))_{\theta\in\Theta}) & \geq I_{\alpha,z}(\theta;(\Tr_{X}[\sigma_{XA}(\theta)])_{\theta\in\Theta})\\
 & =I_{\alpha,z}(\theta;(\overline{\rho}(\theta))_{\theta\in\Theta}),
\end{align}
because $\Tr_{X}[\sigma_{XA}(\theta)]=\overline{\rho}(\theta)$. Proofs
for \eqref{eq:convexity-KM} and \eqref{eq:convexity-RLD} follow
similarly.
\end{proof}

\section{Conclusion}

\label{sec:Conclusion}

In summary, this paper explored quantum generalizations of the Fisher
information matrix that are derived from quantum R\'enyi relative
entropies, including the log-Euclidean, geometric, and $\alpha$-$z$
R\'enyi relative entropies. Among the various results of this paper,
key contributions include a detailed proof of the formula in Theorem~\ref{thm:fisher-info-from-alpha-z}
for the $\alpha$-$z$ information matrix, the formula in Theorem~\ref{thm:QBM-a-z-formula}
for the $\alpha$-$z$ information matrix of parameterized thermal
states, the formula in Theorem~\ref{thm:time-evolved-a-z-formula}
for the $\alpha$-$z$ information matrix of time-evolved states,
and orderings for the Petz-- and sandwiched R\'enyi information
matrices (Theorem~\ref{thm:ordering-Petz-Renyi} and Theorem~\ref{thm:ordering-sandwiched-Renyi}).

Going forward from here, it is an interesting open question to determine
the full range of $\alpha,z>0$ for which the data processing inequality
holds for the $\alpha$-$z$ information matrix. Fact~\ref{fact:a-z-data-proc}
provides sufficient conditions on $\alpha,z>0$ for which it holds,
but Theorem~\ref{thm:log-Euclidean-information-matrix} implies that
these conditions are not necessary. By \cite[Theorems~3 and 5]{Petz1996},
determining the full range of $\alpha,z>0$ for which data processing
holds is equivalent to determining the full range of $\alpha,z>0$
for which the function in Corollary \ref{cor:operator-monotone-a-z}
is operator monotone. It is also an intriguing open question to determine
if there is an appealing expression for the $\alpha$-$z$ information
matrices of parameterized thermal states and time-evolved states,
like those in Theorem~\ref{thm:QBM-a-z-formula} and Theorem~\ref{thm:time-evolved-a-z-formula},
for $\alpha>1$.

\medskip

\textit{Acknowledgements}---I am grateful to Milan Mosonyi for helpful
comments on the manuscript, and I thank Hami Mehrabi for a helpful
discussion. I am also grateful to Michele Minervini and Dhrumil Patel
for our collaborations on \cite{Patel2024a,Patel2024,Minervini2025,Minervini2025a},
which were foundational for the present paper. I acknowledge support
from the National Science Foundation under grant no.~2329662 and
from the Cornell School of Electrical and Computer Engineering.

\bibliographystyle{alphaurl}
\bibliography{Ref}

\appendix

\section{Proofs of Equations \eqref{eq:classical-renyi-fisher} and \eqref{eq:classical-rel-ent-fisher}}

\label{app:classical-fisher-from-renyi}This appendix provides short
proofs of \eqref{eq:classical-renyi-fisher} and \eqref{eq:classical-rel-ent-fisher},
when the parameterized family $\left(p_{\theta}\right)_{\theta\in\Theta}$
of probability distributions are in the interior of the simplex (i.e.,
$p_{\theta}(x)>0$ for all $x\in\mathcal{X}$).

Beginning with \eqref{eq:classical-renyi-fisher}, consider that
\begin{align}
 & \frac{\partial^{2}}{\partial\varepsilon_{i}\partial\varepsilon_{j}}D_{\alpha}(p_{\theta}\|p_{\theta+\varepsilon})\nonumber \\
 & =\frac{\partial^{2}}{\partial\varepsilon_{i}\partial\varepsilon_{j}}\frac{1}{\alpha-1}\ln\!\left(\sum_{x\in\mathcal{X}}p_{\theta}(x)^{\alpha}p_{\theta+\varepsilon}(x)^{1-\alpha}\right)\\
 & =\frac{1}{\alpha-1}\frac{\partial}{\partial\varepsilon_{i}}\left(\frac{\sum_{x\in\mathcal{X}}p_{\theta}(x)^{\alpha}\frac{\partial}{\partial\varepsilon_{j}}\left(p_{\theta+\varepsilon}(x)^{1-\alpha}\right)}{\sum_{x\in\mathcal{X}}p_{\theta}(x)^{\alpha}p_{\theta+\varepsilon}(x)^{1-\alpha}}\right)\\
 & =\frac{\sum_{x\in\mathcal{X}}p_{\theta}(x)^{\alpha}\frac{\partial^{2}}{\partial\varepsilon_{i}\partial\varepsilon_{j}}\left(p_{\theta+\varepsilon}(x)^{1-\alpha}\right)}{\left(\alpha-1\right)\sum_{x\in\mathcal{X}}p_{\theta}(x)^{\alpha}p_{\theta+\varepsilon}(x)^{1-\alpha}}\nonumber \\
 & \qquad-\frac{\left[\sum_{x\in\mathcal{X}}p_{\theta}(x)^{\alpha}\frac{\partial}{\partial\varepsilon_{i}}\left(p_{\theta+\varepsilon}(x)^{1-\alpha}\right)\right]\left[\sum_{x\in\mathcal{X}}p_{\theta}(x)^{\alpha}\frac{\partial}{\partial\varepsilon_{j}}\left(p_{\theta+\varepsilon}(x)^{1-\alpha}\right)\right]}{\left(\alpha-1\right)\left(\sum_{x\in\mathcal{X}}p_{\theta}(x)^{\alpha}p_{\theta+\varepsilon}(x)^{1-\alpha}\right)^{2}}.
\end{align}
It then follows that
\begin{align}
 & \left.\frac{\partial^{2}}{\partial\varepsilon_{i}\partial\varepsilon_{j}}D_{\alpha}(p_{\theta}\|p_{\theta+\varepsilon})\right|_{\varepsilon=0}\nonumber \\
 & =\left.\frac{\sum_{x\in\mathcal{X}}p_{\theta}(x)^{\alpha}\frac{\partial^{2}}{\partial\varepsilon_{i}\partial\varepsilon_{j}}\left(p_{\theta+\varepsilon}(x)^{1-\alpha}\right)}{\left(\alpha-1\right)\sum_{x\in\mathcal{X}}p_{\theta}(x)^{\alpha}p_{\theta+\varepsilon}(x)^{1-\alpha}}\right|_{\varepsilon=0}\nonumber \\
 & \qquad-\left.\frac{\left[\sum_{x\in\mathcal{X}}p_{\theta}(x)^{\alpha}\frac{\partial}{\partial\varepsilon_{i}}\left(p_{\theta+\varepsilon}(x)^{1-\alpha}\right)\right]\left[\sum_{x\in\mathcal{X}}p_{\theta}(x)^{\alpha}\frac{\partial}{\partial\varepsilon_{j}}\left(p_{\theta+\varepsilon}(x)^{1-\alpha}\right)\right]}{\left(\alpha-1\right)\left(\sum_{x\in\mathcal{X}}p_{\theta}(x)^{\alpha}p_{\theta+\varepsilon}(x)^{1-\alpha}\right)^{2}}\right|_{\varepsilon=0}\\
 & =\frac{\sum_{x\in\mathcal{X}}p_{\theta}(x)^{\alpha}\left.\frac{\partial^{2}}{\partial\varepsilon_{i}\partial\varepsilon_{j}}\left(p_{\theta+\varepsilon}(x)^{1-\alpha}\right)\right|_{\varepsilon=0}}{\left(\alpha-1\right)\sum_{x\in\mathcal{X}}p_{\theta}(x)^{\alpha}p_{\theta}(x)^{1-\alpha}}\nonumber \\
 & \qquad-\frac{\left[\sum_{x\in\mathcal{X}}p_{\theta}(x)^{\alpha}\left.\frac{\partial}{\partial\varepsilon_{i}}\left(p_{\theta+\varepsilon}(x)^{1-\alpha}\right)\right|_{\varepsilon=0}\right]\left[\sum_{x\in\mathcal{X}}p_{\theta}(x)^{\alpha}\left.\frac{\partial}{\partial\varepsilon_{j}}\left(p_{\theta+\varepsilon}(x)^{1-\alpha}\right)\right|_{\varepsilon=0}\right]}{\left(\alpha-1\right)\left(\sum_{x\in\mathcal{X}}p_{\theta}(x)^{\alpha}p_{\theta}(x)^{1-\alpha}\right)^{2}}\\
 & =\frac{1}{\alpha-1}\sum_{x\in\mathcal{X}}p_{\theta}(x)^{\alpha}\frac{\partial^{2}}{\partial\theta_{i}\partial\theta_{j}}\left(p_{\theta}(x)^{1-\alpha}\right)\nonumber \\
 & \qquad-\frac{1}{\alpha-1}\left[\sum_{x\in\mathcal{X}}p_{\theta}(x)^{\alpha}\frac{\partial}{\partial\theta_{i}}\left(p_{\theta}(x)^{1-\alpha}\right)\right]\left[\sum_{x\in\mathcal{X}}p_{\theta}(x)^{\alpha}\frac{\partial}{\partial\theta_{j}}\left(p_{\theta}(x)^{1-\alpha}\right)\right].
\end{align}
Now consider that
\begin{align}
\sum_{x\in\mathcal{X}}p_{\theta}(x)^{\alpha}\frac{\partial}{\partial\theta_{i}}\left(p_{\theta}(x)^{1-\alpha}\right) & =\left(1-\alpha\right)\sum_{x\in\mathcal{X}}p_{\theta}(x)^{\alpha}p_{\theta}(x)^{-\alpha}\frac{\partial}{\partial\theta_{i}}\left(p_{\theta}(x)\right)\\
 & =\left(1-\alpha\right)\sum_{x\in\mathcal{X}}\frac{\partial}{\partial\theta_{i}}p_{\theta}(x)\\
 & =\left(1-\alpha\right)\frac{\partial}{\partial\theta_{i}}\sum_{x\in\mathcal{X}}p_{\theta}(x)\\
 & =\left(1-\alpha\right)\frac{\partial}{\partial\theta_{i}}(1)\\
 & =0.
\end{align}
Furthermore, 
\begin{align}
 & \sum_{x\in\mathcal{X}}p_{\theta}(x)^{\alpha}\frac{\partial^{2}}{\partial\theta_{i}\partial\theta_{j}}\left(p_{\theta}(x)^{1-\alpha}\right)\nonumber \\
 & =\left(1-\alpha\right)\sum_{x\in\mathcal{X}}p_{\theta}(x)^{\alpha}\frac{\partial}{\partial\theta_{i}}\left(p_{\theta}(x)^{-\alpha}\frac{\partial}{\partial\theta_{j}}p_{\theta}(x)\right)\\
 & =\left(1-\alpha\right)\sum_{x\in\mathcal{X}}p_{\theta}(x)^{\alpha}\left(\frac{\partial}{\partial\theta_{i}}p_{\theta}(x)^{-\alpha}\frac{\partial}{\partial\theta_{j}}p_{\theta}(x)+p_{\theta}(x)^{-\alpha}\frac{\partial^{2}}{\partial\theta_{i}\partial\theta_{j}}p_{\theta}(x)\right)\\
 & =\alpha\left(\alpha-1\right)\sum_{x\in\mathcal{X}}p_{\theta}(x)^{\alpha}p_{\theta}(x)^{-\alpha-1}\left(\frac{\partial}{\partial\theta_{i}}p_{\theta}(x)\right)\left(\frac{\partial}{\partial\theta_{j}}p_{\theta}(x)\right)\nonumber \\
 & \qquad+\left(1-\alpha\right)\sum_{x\in\mathcal{X}}p_{\theta}(x)^{\alpha}p_{\theta}(x)^{-\alpha}\frac{\partial^{2}}{\partial\theta_{i}\partial\theta_{j}}p_{\theta}(x)\\
 & =\alpha\left(\alpha-1\right)\sum_{x\in\mathcal{X}}p_{\theta}(x)^{-1}\left(\frac{\partial}{\partial\theta_{i}}p_{\theta}(x)\right)\left(\frac{\partial}{\partial\theta_{j}}p_{\theta}(x)\right)+\left(1-\alpha\right)\sum_{x\in\mathcal{X}}\frac{\partial^{2}}{\partial\theta_{i}\partial\theta_{j}}p_{\theta}(x)\\
 & =\alpha\left(\alpha-1\right)\sum_{x\in\mathcal{X}}p_{\theta}(x)^{-1}\left(\frac{\partial}{\partial\theta_{i}}p_{\theta}(x)\right)\left(\frac{\partial}{\partial\theta_{j}}p_{\theta}(x)\right)+\left(1-\alpha\right)\frac{\partial^{2}}{\partial\theta_{i}\partial\theta_{j}}\sum_{x\in\mathcal{X}}p_{\theta}(x)\\
 & =\alpha\left(\alpha-1\right)\sum_{x\in\mathcal{X}}p_{\theta}(x)^{-1}\left(\frac{\partial}{\partial\theta_{i}}p_{\theta}(x)\right)\left(\frac{\partial}{\partial\theta_{j}}p_{\theta}(x)\right).
\end{align}
Thus, we conclude that
\begin{equation}
\left.\frac{\partial^{2}}{\partial\varepsilon_{i}\partial\varepsilon_{j}}D_{\alpha}(p_{\theta}\|p_{\theta+\varepsilon})\right|_{\varepsilon=0}=\alpha\sum_{x\in\mathcal{X}}p_{\theta}(x)^{-1}\left(\frac{\partial}{\partial\theta_{i}}p_{\theta}(x)\right)\left(\frac{\partial}{\partial\theta_{j}}p_{\theta}(x)\right),
\end{equation}
which is equivalent to the desired equality in \eqref{eq:classical-renyi-fisher}.

Now let us prove \eqref{eq:classical-rel-ent-fisher}. Consider that
\begin{align}
\frac{\partial^{2}}{\partial\varepsilon_{i}\partial\varepsilon_{j}}D(p_{\theta}\|p_{\theta+\varepsilon}) & =\frac{\partial^{2}}{\partial\varepsilon_{i}\partial\varepsilon_{j}}\sum_{x\in\mathcal{X}}p_{\theta}(x)\ln\!\left(\frac{p_{\theta}(x)}{p_{\theta+\varepsilon}(x)}\right)\\
 & =\frac{\partial^{2}}{\partial\varepsilon_{i}\partial\varepsilon_{j}}\left(\sum_{x\in\mathcal{X}}p_{\theta}(x)\ln\!\left(p_{\theta}(x)\right)-\sum_{x\in\mathcal{X}}p_{\theta}(x)\ln\!\left(p_{\theta+\varepsilon}(x)\right)\right)\\
 & =-\sum_{x\in\mathcal{X}}p_{\theta}(x)\frac{\partial^{2}}{\partial\varepsilon_{i}\partial\varepsilon_{j}}\ln\!\left(p_{\theta+\varepsilon}(x)\right)\\
 & =-\sum_{x\in\mathcal{X}}p_{\theta}(x)\frac{\partial}{\partial\varepsilon_{i}}\left(\frac{\frac{\partial}{\partial\varepsilon_{j}}p_{\theta+\varepsilon}(x)}{p_{\theta+\varepsilon}(x)}\right)\\
 & =\sum_{x\in\mathcal{X}}p_{\theta}(x)\left(-\frac{\frac{\partial^{2}}{\partial\varepsilon_{i}\partial\varepsilon_{j}}p_{\theta+\varepsilon}(x)}{p_{\theta+\varepsilon}(x)}+\frac{\left(\frac{\partial}{\partial\varepsilon_{i}}p_{\theta+\varepsilon}(x)\right)\left(\frac{\partial}{\partial\varepsilon_{j}}p_{\theta+\varepsilon}(x)\right)}{\left(p_{\theta+\varepsilon}(x)\right)^{2}}\right).
\end{align}
Then it follows that
\begin{align}
 & \left.\frac{\partial^{2}}{\partial\varepsilon_{i}\partial\varepsilon_{j}}D(p_{\theta}\|p_{\theta+\varepsilon})\right|_{\varepsilon=0}\\
 & =\left.\sum_{x\in\mathcal{X}}p_{\theta}(x)\left(-\frac{\frac{\partial^{2}}{\partial\varepsilon_{i}\partial\varepsilon_{j}}p_{\theta+\varepsilon}(x)}{p_{\theta+\varepsilon}(x)}+\frac{\left(\frac{\partial}{\partial\varepsilon_{i}}p_{\theta+\varepsilon}(x)\right)\left(\frac{\partial}{\partial\varepsilon_{j}}p_{\theta+\varepsilon}(x)\right)}{\left(p_{\theta+\varepsilon}(x)\right)^{2}}\right)\right|_{\varepsilon=0}\\
 & =\sum_{x\in\mathcal{X}}p_{\theta}(x)\left(-\frac{\left.\frac{\partial^{2}}{\partial\varepsilon_{i}\partial\varepsilon_{j}}p_{\theta+\varepsilon}(x)\right|_{\varepsilon=0}}{p_{\theta}(x)}+\frac{\left(\left.\frac{\partial}{\partial\varepsilon_{i}}p_{\theta+\varepsilon}(x)\right|_{\varepsilon=0}\right)\left(\left.\frac{\partial}{\partial\varepsilon_{j}}p_{\theta+\varepsilon}(x)\right|_{\varepsilon=0}\right)}{\left(p_{\theta}(x)\right)^{2}}\right)\\
 & =\sum_{x\in\mathcal{X}}p_{\theta}(x)\left(-\frac{\frac{\partial^{2}}{\partial\theta_{i}\partial\theta_{j}}p_{\theta}(x)}{p_{\theta}(x)}+\frac{\left(\frac{\partial}{\partial\theta_{i}}p_{\theta}(x)\right)\left(\frac{\partial}{\partial\theta_{j}}p_{\theta}(x)\right)}{\left(p_{\theta}(x)\right)^{2}}\right)\\
 & =-\sum_{x\in\mathcal{X}}\frac{\partial^{2}}{\partial\theta_{i}\partial\theta_{j}}p_{\theta}(x)+\sum_{x\in\mathcal{X}}p_{\theta}(x)^{-1}\left(\frac{\partial}{\partial\theta_{i}}p_{\theta}(x)\right)\left(\frac{\partial}{\partial\theta_{j}}p_{\theta}(x)\right)\\
 & =-\frac{\partial^{2}}{\partial\theta_{i}\partial\theta_{j}}\sum_{x\in\mathcal{X}}p_{\theta}(x)+\sum_{x\in\mathcal{X}}p_{\theta}(x)^{-1}\left(\frac{\partial}{\partial\theta_{i}}p_{\theta}(x)\right)\left(\frac{\partial}{\partial\theta_{j}}p_{\theta}(x)\right)\\
 & =\sum_{x\in\mathcal{X}}p_{\theta}(x)^{-1}\left(\frac{\partial}{\partial\theta_{i}}p_{\theta}(x)\right)\left(\frac{\partial}{\partial\theta_{j}}p_{\theta}(x)\right),
\end{align}
thus completing the proof of \eqref{eq:classical-rel-ent-fisher}.

\section{Review: Matrix derivatives from divided differences}

\label{app:Review:-Matrix-derivatives}This appendix provides a review
of the method of divided differences for calculating matrix derivatives.
See also \cite[Section~V.3]{Bhatia1997} and \cite[Sections~2.2 and 2.3]{Hiai2010}.

For a general $d$-dimensional matrix $A(x)$, the matrix derivative
$\frac{\partial}{\partial x}A(x)$ is defined as
\begin{equation}
\frac{\partial}{\partial x}A(x)\coloneqq\begin{bmatrix}\frac{\partial}{\partial x}a_{1,1}(x) & \frac{\partial}{\partial x}a_{1,2}(x) & \cdots & \frac{\partial}{\partial x}a_{1,d}(x)\\
\frac{\partial}{\partial x}a_{2,1}(x) & \frac{\partial}{\partial x}a_{2,2}(x) &  & \frac{\partial}{\partial x}a_{2,d}(x)\\
\vdots &  & \ddots & \vdots\\
\frac{\partial}{\partial x}a_{d,1}(x) & \frac{\partial}{\partial x}a_{d,2}(x) & \cdots & \frac{\partial}{\partial x}a_{d,d}(x)
\end{bmatrix}.
\end{equation}
where the matrix elements of $A(x)$ are denoted by $a_{i,j}(x)$.
From this, one can deduce that the product rule holds for matrices
$A(x)$ and $B(x)$:
\begin{equation}
\frac{\partial}{\partial x}\left[A(x)B(x)\right]=\left[\frac{\partial}{\partial x}A(x)\right]B(x)+A(x)\left[\frac{\partial}{\partial x}B(x)\right].
\end{equation}

\subsection{First derivative}

Let $A(x)$ be a Hermitian matrix parameterized by $x\in\mathbb{R}$.
Suppose that $A(x)$ has a spectral decomposition as follows:
\begin{equation}
A(x)=\sum_{\ell}\lambda_{\ell}\Pi_{\ell},\label{eq:spectral-decomp-A-x}
\end{equation}
where we have suppressed the dependence of each eigenvalue $\lambda_{\ell}$
and eigenprojection $\Pi_{\ell}$ on the parameter $x$. 
\begin{thm}
\label{thm:divided-difference-matrix-deriv}Let $x\mapsto A(x)$ be
a Hermitian operator-valued function, and let $f$ be an analytic
function that has a power series expansion convergent on an open interval
$I\subseteq\mathbb{R}$, such that, for all $x$, all of the eigenvalues
of $A(x)$ are contained in $I$. Then the following equality holds:
\begin{equation}
\frac{\partial}{\partial x}f(A(x))=\sum_{\ell,m}f^{[1]}(\lambda_{\ell},\lambda_{m})\Pi_{\ell}\left(\frac{\partial}{\partial x}A(x)\right)\Pi_{m},
\end{equation}
where the first divided difference function $f^{[1]}(y_{1},y_{2})$
is defined for $y_{1},y_{2}\in I$ as
\begin{equation}
f^{[1]}(y_{1},y_{2})\coloneqq\begin{cases}
f'(y_{1}) & :y_{1}=y_{2}\\
\frac{f(y_{1})-f(y_{2})}{y_{1}-y_{2}} & :y_{1}\neq y_{2}
\end{cases}.
\end{equation}
\end{thm}

\begin{proof}
By applying the product rule for differentiation, the following chain
of equalities holds for all $n\in\mathbb{N}$: 
\begin{align}
\frac{\partial}{\partial x}\left(A(x)^{n}\right) & =\sum_{k=0}^{n-1}A(x)^{k}\left(\frac{\partial}{\partial x}A(x)\right)A(x)^{n-k-1}\\
 & =\sum_{k=0}^{n-1}\left(\sum_{\ell}\lambda_{\ell}\Pi_{\ell}\right)^{k}\left(\frac{\partial}{\partial x}A(x)\right)\left(\sum_{m}\lambda_{m}\Pi_{m}\right)^{n-k-1}\\
 & =\sum_{k=0}^{n-1}\left(\sum_{\ell}\lambda_{\ell}^{k}\Pi_{\ell}\right)\left(\frac{\partial}{\partial x}A(x)\right)\left(\sum_{m}\lambda_{m}^{n-k-1}\Pi_{m}\right)\\
 & =\sum_{\ell,m}\left(\sum_{k=0}^{n-1}\lambda_{\ell}^{k}\lambda_{m}^{n-k-1}\right)\Pi_{\ell}\left(\frac{\partial}{\partial x}A(x)\right)\Pi_{m}\\
 & =\sum_{\ell,m}f_{x^{n}}^{\left[1\right]}(\lambda_{\ell},\lambda_{m})\Pi_{\ell}\left(\frac{\partial}{\partial x}A(x)\right)\Pi_{m}\label{eq:power-function-simple-derivative}
\end{align}
where the last equality follows because
\begin{equation}
\sum_{k=0}^{n-1}\lambda_{\ell}^{k}\lambda_{m}^{n-k-1}=f_{x^{n}}^{\left[1\right]}(\lambda_{\ell},\lambda_{m})\coloneqq\begin{cases}
n\lambda_{\ell}^{n-1} & :\lambda_{\ell}=\lambda_{m}\\
\frac{\lambda_{\ell}^{n}-\lambda_{m}^{n}}{\lambda_{\ell}-\lambda_{m}} & :\lambda_{\ell}\neq\lambda_{m}
\end{cases}.
\end{equation}
The notation $f_{x^{n}}^{\left[1\right]}(\lambda_{\ell},\lambda_{m})$
means that this function is the first divided difference for $x\mapsto x^{n}$.
To see this, if $\lambda_{\ell}=\lambda_{m}$, then
\begin{align}
\sum_{k=0}^{n-1}\lambda_{\ell}^{k}\lambda_{m}^{n-k-1} & =\sum_{k=0}^{n-1}\lambda_{\ell}^{n-1}=n\lambda_{\ell}^{n-1},
\end{align}
and if $\lambda_{\ell}\neq\lambda_{m}$, then
\begin{align}
\sum_{k=0}^{n-1}\lambda_{\ell}^{k}\lambda_{m}^{n-k-1} & =\lambda_{m}^{n-1}\left[\sum_{k=0}^{n-1}\left(\frac{\lambda_{\ell}}{\lambda_{m}}\right)^{k}\right]\\
 & =\lambda_{m}^{n-1}\left(\frac{1-\left(\frac{\lambda_{\ell}}{\lambda_{m}}\right)^{n}}{1-\frac{\lambda_{\ell}}{\lambda_{m}}}\right)\\
 & =\frac{\lambda_{m}^{n}-\lambda_{\ell}^{n}}{\lambda_{m}-\lambda_{\ell}}.
\end{align}
Thus, $\sum_{k=0}^{n-1}\lambda_{\ell}^{k}\lambda_{m}^{n-k-1}$ is
indeed the first divided difference for the function $x\mapsto x^{n}$.

We now extend this development to an arbitrary analytic function that
has a power series expansion convergent on an open interval $I\subseteq\mathbb{R}$.
Toward this, let $y\in I$, and let $f$ be a function that has the
following power series expansion centered at $c\in I$ and convergent
on $I$:
\begin{equation}
f(y)=\sum_{n=0}^{\infty}a_{n}\left(y-c\right)^{n}.
\end{equation}
Then
\begin{equation}
f(A(x))=\sum_{n=0}^{\infty}a_{n}\left(A(x)-cI\right)^{n}.
\end{equation}
By applying \eqref{eq:power-function-simple-derivative} and noting
that $\frac{\partial}{\partial x}\left(A(x)-cI\right)=\frac{\partial}{\partial x}A(x)$
and $A(x)-cI=\sum_{\ell}\left(\lambda_{\ell}-c\right)\Pi_{\ell}$,
we conclude that
\begin{align}
\frac{\partial}{\partial x}f(A(x)) & =\frac{\partial}{\partial x}\sum_{n=0}^{\infty}a_{n}\left(A(x)-cI\right)^{n}\\
 & =\sum_{n=0}^{\infty}a_{n}\frac{\partial}{\partial x}\left(\left(A(x)-cI\right)^{n}\right)\\
 & =\sum_{n=0}^{\infty}a_{n}\sum_{\ell,m}f_{x^{n}}^{\left[1\right]}(\lambda_{\ell}-c,\lambda_{m}-c)\Pi_{\ell}\left(\frac{\partial}{\partial x}A(x)\right)\Pi_{m}\nonumber \\
 & =\sum_{\ell,m}\left(\sum_{n=0}^{\infty}a_{n}f_{x^{n}}^{\left[1\right]}(\lambda_{\ell}-c,\lambda_{m}-c)\right)\Pi_{\ell}\left(\frac{\partial}{\partial x}A(x)\right)\Pi_{m}\\
 & =\sum_{\ell,m}f^{[1]}(\lambda_{\ell},\lambda_{m})\Pi_{\ell}\left(\frac{\partial}{\partial x}A(x)\right)\Pi_{m}.
\end{align}
The last equality follows because, if $\lambda_{\ell}=\lambda_{m}$,
then
\begin{align}
\sum_{n=0}^{\infty}a_{n}f_{x^{n}}^{\left[1\right]}(\lambda_{\ell}-c,\lambda_{m}-c) & =\sum_{n=0}^{\infty}a_{n}n\left(\lambda_{\ell}-c\right)^{n-1}\\
 & =\sum_{n=0}^{\infty}a_{n}\frac{\partial}{\partial\lambda_{\ell}}\left(\lambda_{\ell}-c\right)^{n}\\
 & =\frac{\partial}{\partial\lambda_{\ell}}\sum_{n=0}^{\infty}a_{n}\left(\lambda_{\ell}-c\right)^{n}\\
 & =\frac{\partial}{\partial\lambda_{\ell}}f(\lambda_{\ell})\\
 & =f'(\lambda_{\ell}),
\end{align}
and if $\lambda_{\ell}\neq\lambda_{m}$, then
\begin{align}
 & \sum_{n=0}^{\infty}a_{n}f_{x^{n}}^{\left[1\right]}(\lambda_{\ell}-c,\lambda_{m}-c)\nonumber \\
 & =\sum_{n=0}^{\infty}a_{n}\left(\frac{\left(\lambda_{m}-c\right)^{n}-\left(\lambda_{\ell}-c\right)^{n}}{\lambda_{m}-\lambda_{\ell}}\right)\\
 & =\frac{\left(\sum_{n=0}^{\infty}a_{n}\left(\lambda_{m}-c\right)^{n}\right)-\left(\sum_{n=0}^{\infty}a_{n}\left(\lambda_{\ell}-c\right)^{n}\right)}{\lambda_{m}-\lambda_{\ell}}\\
 & =\frac{f(\lambda_{m})-f(\lambda_{\ell})}{\lambda_{m}-\lambda_{\ell}}\\
 & =f^{[1]}(\lambda_{\ell},\lambda_{m}).
\end{align}
This concludes the proof.
\end{proof}
\begin{cor}
\label{cor:derivative-in-trace}Let $x\mapsto A(x)$ be a Hermitian
operator-valued function, and let $f$ be an analytic function that
has a power series expansion convergent on an open interval $I\subseteq\mathbb{R}$,
such that, for all $x$, all of the eigenvalues of $A(x)$ are contained
in $I$. Then the following equality holds:
\begin{equation}
\Tr\!\left[\frac{\partial}{\partial x}f(A(x))\right]=\Tr\!\left[f'(A(x))\frac{\partial}{\partial x}A(x)\right].
\end{equation}
\end{cor}

\begin{proof}
Applying Theorem~\ref{thm:divided-difference-matrix-deriv}, consider
that
\begin{align}
\Tr\!\left[\frac{\partial}{\partial x}f(A(x))\right] & =\Tr\!\left[\sum_{\ell,m}f^{[1]}(\lambda_{\ell},\lambda_{m})\Pi_{\ell}\left(\frac{\partial}{\partial x}A(x)\right)\Pi_{m}\right],\\
 & =\sum_{\ell,m}f^{[1]}(\lambda_{\ell},\lambda_{m})\Tr\!\left[\Pi_{m}\Pi_{\ell}\left(\frac{\partial}{\partial x}A(x)\right)\right]\\
 & =\sum_{\ell}f^{[1]}(\lambda_{\ell},\lambda_{\ell})\Tr\!\left[\Pi_{\ell}\left(\frac{\partial}{\partial x}A(x)\right)\right]\\
 & =\sum_{\ell}f'(\lambda_{\ell})\Tr\!\left[\Pi_{\ell}\left(\frac{\partial}{\partial x}A(x)\right)\right]\\
 & =\Tr\!\left[\sum_{\ell}f'(\lambda_{\ell})\Pi_{\ell}\left(\frac{\partial}{\partial x}A(x)\right)\right]\\
 & =\Tr\!\left[f'(A(x))\frac{\partial}{\partial x}A(x)\right],
\end{align}
thus concluding the proof.
\end{proof}
\begin{prop}
The following expression holds for a holomorphic function $f$ and
for a contour $\gamma$ containing all the eigenvalues of the Hermitian
operator-valued function $x\mapsto A(x)$:
\begin{equation}
\frac{\partial}{\partial x}f(A(x))=\frac{1}{2\pi i}\oint_{\gamma}dz\ f(z)\left(zI-A(x)\right)^{-1}\left(\frac{\partial}{\partial x}A(x)\right)\left(zI-A(x)\right)^{-1}.
\end{equation}
\end{prop}

\begin{proof}
The Cauchy integral formula implies that the following equality holds
for a holomorphic function $f$ and a contour $\gamma$ containing
$a\in\mathbb{C}$:
\begin{equation}
f(a)=\frac{1}{2\pi i}\oint_{\gamma}dz\ \frac{f(z)}{z-a}.
\end{equation}
We can then apply this to establish the following formula for the
first divided difference at $a,b\in\mathbb{C}$ and now with $\gamma$
a contour containing $a$ and $b$:
\begin{align}
\frac{f(a)-f(b)}{a-b} & =\left(\frac{1}{a-b}\right)\left(\frac{1}{2\pi i}\oint_{\gamma}dz\ \frac{f(z)}{z-a}-\frac{1}{2\pi i}\oint_{\gamma}dz\ \frac{f(z)}{z-b}\right)\\
 & =\left(\frac{1}{a-b}\right)\frac{1}{2\pi i}\oint_{\gamma}dz\ f(z)\left(\frac{1}{z-a}-\frac{1}{z-b}\right)\\
 & =\left(\frac{1}{a-b}\right)\frac{1}{2\pi i}\oint_{\gamma}dz\ f(z)\left(\frac{z-b-\left(z-a\right)}{\left(z-a\right)\left(z-b\right)}\right)\\
 & =\frac{1}{2\pi i}\oint_{\gamma}dz\ \frac{f(z)}{\left(z-a\right)\left(z-b\right)}.
\end{align}
Applying this and Theorem~\ref{thm:divided-difference-matrix-deriv},
we find that, when $\gamma$ is a contour containing all the eigenvalues
of $A(x)$,
\begin{align}
 & \frac{\partial}{\partial x}f(A(x))\nonumber \\
 & =\sum_{\ell,m}f^{[1]}(\lambda_{\ell},\lambda_{m})\Pi_{\ell}\left(\frac{\partial}{\partial x}A(x)\right)\Pi_{m}\\
 & =\sum_{\ell,m}\left(\frac{1}{2\pi i}\oint_{\gamma}dz\ \frac{f(z)}{\left(z-\lambda_{\ell}\right)\left(z-\lambda_{m}\right)}\right)\Pi_{\ell}\left(\frac{\partial}{\partial x}A(x)\right)\Pi_{m}\\
 & =\frac{1}{2\pi i}\oint_{\gamma}dz\ f(z)\left(\sum_{\ell}\frac{1}{\left(z-\lambda_{\ell}\right)}\Pi_{\ell}\right)\left(\frac{\partial}{\partial x}A(x)\right)\left(\sum_{m}\frac{1}{\left(z-\lambda_{m}\right)}\Pi_{m}\right)\\
 & =\frac{1}{2\pi i}\oint_{\gamma}dz\ f(z)\left(zI-A(x)\right)^{-1}\left(\frac{\partial}{\partial x}A(x)\right)\left(zI-A(x)\right)^{-1}.
\end{align}
This concludes the proof.
\end{proof}
We can also consider the case of operator monotone functions. Recall
from \cite[pp.~144--145]{Bhatia1997} the following representation
theorem regarding operator monotone functions:
\begin{thm}
A function $f\colon\left(0,\infty\right)\to\mathbb{R}$ is operator
monotone on $\left(0,\infty\right)$ if and only if it has the representation
\begin{equation}
f(t)=a+bt+\int_{0}^{\infty}d\mu(v)\ \frac{tv}{t+v},
\end{equation}
where $a\in\mathbb{R}$, $b\geq0$, and $\mu$ is a positive measure
on $\left(0,\infty\right)$ such that
\begin{equation}
\int_{0}^{\infty}d\mu(v)\ \frac{v}{1+v}<+\infty.
\end{equation}
If $f$ is operator monotone on $\left[0,\infty\right)$, then $a=f(0)$.
\end{thm}

We can use this theorem to establish the following representation
for the derivative of operator monotone functions.
\begin{prop}
Let $x\mapsto A(x)$ be a positive operator-valued function. The following
expression holds for the derivative of a function $f$ that is operator
monotone on $\left(0,\infty\right)$:
\begin{equation}
\frac{\partial}{\partial x}f(A(x))=b\frac{\partial}{\partial x}A(x)+\int_{0}^{\infty}d\mu(v)\ v^{2}\left(A(x)+vI\right)^{-1}\left(\frac{\partial}{\partial x}A(x)\right)\left(A(x)+vI\right)^{-1}.
\end{equation}
\end{prop}

\begin{proof}
For $x,y>0$ such that $x\neq y$, consider that
\begin{align}
f(x)-f(y) & =\left(a+bx+\int_{0}^{\infty}d\mu(v)\ \frac{xv}{x+v}\right)\\
 & \qquad-\left(a+by+\int_{0}^{\infty}d\mu(v)\ \frac{yv}{y+v}\right)\\
 & =b\left(x-y\right)+\int_{0}^{\infty}d\mu(v)\ \left(\frac{xv}{x+v}-\frac{yv}{y+v}\right)\\
 & =b\left(x-y\right)+\int_{0}^{\infty}d\mu(v)\ \frac{xv\left(y+v\right)-yv\left(x+v\right)}{\left(x+v\right)\left(y+v\right)}\\
 & =b\left(x-y\right)+\left(x-y\right)\int_{0}^{\infty}d\mu(v)\ \frac{v^{2}}{\left(x+v\right)\left(y+v\right)}\\
 & =\left(x-y\right)\left(b+\int_{0}^{\infty}d\mu(v)\ \frac{v^{2}}{\left(x+v\right)\left(y+v\right)}\right).
\end{align}
Then it follows that
\begin{align}
\frac{f(x)-f(y)}{x-y} & =b+\int_{0}^{\infty}d\mu(v)\ \frac{v^{2}}{\left(x+v\right)\left(y+v\right)},
\end{align}
and furthermore that
\begin{equation}
f'(x)=b+\int_{0}^{\infty}d\mu(v)\ \frac{v^{2}}{\left(x+v\right)^{2}}.
\end{equation}
By applying Theorem~\ref{thm:divided-difference-matrix-deriv}, consider
that
\begin{align}
 & \frac{\partial}{\partial x}f(A(x))\nonumber \\
 & =\sum_{\ell,m}f^{[1]}(\lambda_{\ell},\lambda_{m})\Pi_{\ell}\left(\frac{\partial}{\partial x}A(x)\right)\Pi_{m}\\
 & =\sum_{\ell,m}\left(b+\int_{0}^{\infty}d\mu(v)\ \frac{v^{2}}{\left(\lambda_{\ell}+v\right)\left(\lambda_{m}+v\right)}\right)\Pi_{\ell}\left(\frac{\partial}{\partial x}A(x)\right)\Pi_{m}\\
 & =b\sum_{\ell,m}\Pi_{\ell}\left(\frac{\partial}{\partial x}A(x)\right)\Pi_{m}\nonumber \\
 & \qquad\int_{0}^{\infty}d\mu(v)\ v^{2}\left(\sum_{\ell}\frac{1}{\left(\lambda_{\ell}+v\right)}\Pi_{\ell}\right)\left(\frac{\partial}{\partial x}A(x)\right)\left(\sum_{m}\frac{1}{\left(\lambda_{m}+v\right)}\Pi_{m}\right)\\
 & =b\frac{\partial}{\partial x}A(x)+\int_{0}^{\infty}d\mu(v)\ v^{2}\left(A(x)+vI\right)^{-1}\left(\frac{\partial}{\partial x}A(x)\right)\left(A(x)+vI\right)^{-1},
\end{align}
thus concluding the proof.
\end{proof}

\subsection{Examples}

We can verify the consistency of the first divided difference expression
for the derivative with other known expressions for derivatives. As
a first example, let us verify Duhamel's formula:
\begin{prop}[Exponential function]
\label{prop:deriv-exp}For a Hermitian operator-valued function $x\mapsto A(x)$,
the following identity holds:
\begin{equation}
\frac{\partial}{\partial x}e^{A(x)}=\int_{0}^{1}dt\ e^{tA(x)}\left(\frac{\partial}{\partial x}A(x)\right)e^{\left(1-t\right)A(x)}.
\end{equation}
\end{prop}

\begin{proof}
Suppose that $A(x)$ has a spectral decomposition as in \eqref{eq:spectral-decomp-A-x}.
Then
\begin{align}
 & \int_{0}^{1}dt\ e^{tA(x)}\left(\frac{\partial}{\partial x}A(x)\right)e^{\left(1-t\right)A(x)}\nonumber \\
 & =\int_{0}^{1}dt\ \left(\sum_{\ell}e^{t\lambda_{\ell}}\Pi_{\ell}\right)\left(\frac{\partial}{\partial x}A(x)\right)\left(\sum_{m}e^{\left(1-t\right)\lambda_{m}}\Pi_{m}\right)\\
 & =\sum_{\ell,m}\int_{0}^{1}dt\ e^{t\lambda_{\ell}}e^{\left(1-t\right)\lambda_{m}}\Pi_{\ell}\left(\frac{\partial}{\partial x}A(x)\right)\Pi_{m}\\
 & =\sum_{\ell,m}e^{\lambda_{m}}\int_{0}^{1}dt\ e^{t\left(\lambda_{\ell}-\lambda_{m}\right)}\Pi_{\ell}\left(\frac{\partial}{\partial x}A(x)\right)\Pi_{m}\\
 & =\sum_{\ell}e^{\lambda_{m}}\Pi_{\ell}\left(\frac{\partial}{\partial x}A(x)\right)\Pi_{\ell}\nonumber \\
 & \qquad+\sum_{\ell\neq m}e^{\lambda_{m}}\int_{0}^{1}dt\ e^{t\left(\lambda_{\ell}-\lambda_{m}\right)}\Pi_{\ell}\left(\frac{\partial}{\partial x}A(x)\right)\Pi_{m}.\label{eq:duhamel-verify-proof}
\end{align}
Now consider that
\begin{align}
 & \sum_{\ell\neq m}e^{\lambda_{m}}\int_{0}^{1}dt\ e^{t\left(\lambda_{\ell}-\lambda_{m}\right)}\Pi_{\ell}\left(\frac{\partial}{\partial x}A(x)\right)\Pi_{m}\nonumber \\
 & =\sum_{\ell\neq m}e^{\lambda_{m}}\left[\left.\frac{e^{t\left(\lambda_{\ell}-\lambda_{m}\right)}}{\lambda_{\ell}-\lambda_{m}}\right|_{0}^{1}\right]\Pi_{\ell}\left(\frac{\partial}{\partial x}A(x)\right)\Pi_{m}\\
 & =\sum_{\ell\neq m}e^{\lambda_{m}}\left[\frac{e^{\left(\lambda_{\ell}-\lambda_{m}\right)}}{\lambda_{\ell}-\lambda_{m}}-\frac{1}{\lambda_{\ell}-\lambda_{m}}\right]\Pi_{\ell}\left(\frac{\partial}{\partial x}A(x)\right)\Pi_{m}\\
 & =\sum_{\ell\neq m}\frac{e^{\lambda_{\ell}}-e^{\lambda_{m}}}{\lambda_{\ell}-\lambda_{m}}\Pi_{\ell}\left(\frac{\partial}{\partial x}A(x)\right)\Pi_{m}.
\end{align}
Substituting back into \eqref{eq:duhamel-verify-proof}, we conclude
that
\begin{align}
 & \int_{0}^{1}dt\ e^{tA(x)}\left(\frac{\partial}{\partial x}A(x)\right)e^{\left(1-t\right)A(x)}\nonumber \\
 & =\sum_{\ell}e^{\lambda_{m}}\Pi_{\ell}\left(\frac{\partial}{\partial x}A(x)\right)\Pi_{\ell}+\sum_{\ell\neq m}\frac{e^{\lambda_{\ell}}-e^{\lambda_{m}}}{\lambda_{\ell}-\lambda_{m}}\Pi_{\ell}\left(\frac{\partial}{\partial x}A(x)\right)\Pi_{m}\\
 & =\sum_{\ell,m}f_{e^{x}}^{\left[1\right]}(\lambda_{\ell},\lambda_{m})\Pi_{\ell}\left(\frac{\partial}{\partial x}A(x)\right)\Pi_{m}\\
 & =\frac{\partial}{\partial x}e^{A(x)},
\end{align}
where $f_{e^{x}}^{\left[1\right]}$ denotes the first divided difference
for the function $x\mapsto e^{x}$ and the last equality follows from
Theorem~\ref{thm:divided-difference-matrix-deriv} and the fact that
the Taylor series expansion $e^{x}=\sum_{n=0}^{\infty}\frac{x^{n}}{n!}$
converges for all $x\in\mathbb{R}$.
\end{proof}
We can also verify a well known expression for the derivative of the
logarithm:
\begin{prop}[Logarithmic function]
\label{prop:derivative-log}For a positive operator-valued function
$x\mapsto A(x)$, the following identity holds:
\begin{equation}
\frac{\partial}{\partial x}\ln A(x)=\int_{0}^{\infty}ds\ \left(A(x)+sI\right)^{-1}\left(\frac{\partial}{\partial x}A(x)\right)\left(A(x)+sI\right)^{-1}.
\end{equation}
\end{prop}

\begin{proof}
Suppose that $A(x)$ has a spectral decomposition as in \eqref{eq:spectral-decomp-A-x}.
Then
\begin{align}
 & \int_{0}^{\infty}ds\ \left(A(x)+sI\right)^{-1}\left(\frac{\partial}{\partial x}A(x)\right)\left(A(x)+sI\right)^{-1}\nonumber \\
 & =\int_{0}^{\infty}ds\ \left(\sum_{\ell}\left(\frac{1}{\lambda_{\ell}+s}\right)\Pi_{\ell}\right)\left(\frac{\partial}{\partial x}A(x)\right)\left(\sum_{m}\left(\frac{1}{\lambda_{m}+s}\right)\Pi_{m}\right)\\
 & =\sum_{\ell,m}\left[\int_{0}^{\infty}ds\ \left(\frac{1}{\lambda_{\ell}+s}\right)\left(\frac{1}{\lambda_{m}+s}\right)\right]\Pi_{\ell}\left(\frac{\partial}{\partial x}A(x)\right)\Pi_{m}\\
 & =\sum_{\ell}\left[\int_{0}^{\infty}ds\ \frac{1}{\left(\lambda_{\ell}+s\right)^{2}}\right]\Pi_{\ell}\left(\frac{\partial}{\partial x}A(x)\right)\Pi_{\ell}\nonumber \\
 & \qquad+\sum_{\ell\neq m}\left[\int_{0}^{\infty}ds\ \left(\frac{1}{\lambda_{\ell}+s}\right)\left(\frac{1}{\lambda_{m}+s}\right)\right]\Pi_{\ell}\left(\frac{\partial}{\partial x}A(x)\right)\Pi_{m}\\
 & =\sum_{\ell}\frac{1}{\lambda_{\ell}}\Pi_{\ell}\left(\frac{\partial}{\partial x}A(x)\right)\Pi_{\ell}+\sum_{\ell\neq m}\frac{\ln\lambda_{\ell}-\ln\lambda_{m}}{\lambda_{\ell}-\lambda_{m}}\Pi_{\ell}\left(\frac{\partial}{\partial x}A(x)\right)\Pi_{m}\\
 & =\sum_{\ell,m}f_{\ln x}^{\left[1\right]}(\lambda_{\ell},\lambda_{m})\Pi_{\ell}\left(\frac{\partial}{\partial x}A(x)\right)\Pi_{m}\\
 & =\frac{\partial}{\partial x}\ln A(x),
\end{align}
where $f_{\ln x}^{\left[1\right]}$ denotes the first divided difference
for the function $x\mapsto\ln x$. The penultimate equality is a consequence
of the following integral formulas, holding for all $x,y>0$:
\begin{align}
\int_{0}^{\infty}ds\ \frac{1}{\left(x+s\right)^{2}} & =\frac{1}{x},\\
\int_{0}^{\infty}ds\ \left(\frac{1}{x+s}\right)\left(\frac{1}{y+s}\right) & =\frac{\ln x-\ln y}{x-y},
\end{align}
and the last equality follows from Theorem~\ref{thm:divided-difference-matrix-deriv}
and the fact that the logarithm has the following power series expansion
that converges for all $y\in\left(0,2c\right)$, where $c>0$:
\begin{equation}
\ln y=\ln c+\sum_{n=1}^{\infty}\frac{\left(-1\right)^{n-1}}{n}\left(\frac{y-c}{c}\right)^{n}.
\end{equation}
To see this, we can write
\begin{align}
\ln y & =\ln\!\left(c\left(1+\frac{y-c}{c}\right)\right)\\
 & =\ln c+\ln\!\left(1+\frac{y-c}{c}\right)
\end{align}
and then apply the standard Taylor expansion for $\ln(1+y)$ that
converges for all $\left|y\right|<1$. Substituting $y\to\frac{y-c}{c}$,
the convergence condition becomes $\left|\frac{y-c}{c}\right|<1$,
which is equivalent to $y\in\left(0,2c\right)$. In order to apply
Theorem~\ref{thm:divided-difference-matrix-deriv}, we take $c>0$
to be larger than the largest eigenvalue of $A(x)$.
\end{proof}
We now use Propositions \ref{prop:deriv-exp} and \ref{prop:derivative-log}
to find an integral formula for the derivative of all power functions. 
\begin{prop}[Power function]
\label{prop:power-function-all-powers}Let $x\mapsto A(x)$ be a
positive operator-valued function. For all $r\in\mathbb{R}$, the
following equality holds:
\begin{equation}
\frac{\partial}{\partial x}\left(A(x)^{r}\right)=r\int_{0}^{1}dt\int_{0}^{\infty}ds\ \frac{A(x)^{rt}}{A(x)+sI}\left(\frac{\partial}{\partial x}A(x)\right)\frac{A(x)^{r\left(1-t\right)}}{A(x)+sI}.
\end{equation}
\end{prop}

\begin{proof}
Applying Propositions \ref{prop:deriv-exp} and \ref{prop:derivative-log},
consider that
\begin{align}
 & \frac{\partial}{\partial x}\left(A(x)^{r}\right)\nonumber \\
 & =\frac{\partial}{\partial x}e^{r\ln A(x)}\\
 & =\int_{0}^{1}dt\ e^{tr\ln A(x)}\left(\frac{\partial}{\partial x}r\ln A(x)\right)e^{\left(1-t\right)r\ln A(x)}\\
 & =r\int_{0}^{1}dt\ A(x)^{rt}\left(\frac{\partial}{\partial x}\ln A(x)\right)A(x)^{r\left(1-t\right)}\\
 & =r\int_{0}^{1}dt\ A(x)^{rt}\left(\frac{\partial}{\partial x}\ln A(x)\right)A(x)^{r\left(1-t\right)}\\
 & =r\int_{0}^{1}dt\ A(x)^{rt}\int_{0}^{\infty}ds\ \left(A(x)+sI\right)^{-1}\left(\frac{\partial}{\partial x}A(x)\right)\left(A(x)+sI\right)^{-1}A(x)^{r\left(1-t\right)}\\
 & =r\int_{0}^{1}dt\int_{0}^{\infty}ds\ \frac{A(x)^{rt}}{A(x)+sI}\left(\frac{\partial}{\partial x}A(x)\right)\frac{A(x)^{r\left(1-t\right)}}{A(x)+sI},
\end{align}
thus concluding the proof.
\end{proof}
By taking a spectral decomposition of $A(x)$ as in \eqref{eq:spectral-decomp-A-x},
we can further manipulate the expression from Proposition \ref{prop:power-function-all-powers}
to see that it is consistent with the first divided difference expression
from Theorem~\ref{thm:divided-difference-matrix-deriv}:
\begin{align}
 & \frac{\partial}{\partial x}\left(A(x)^{r}\right)\nonumber \\
 & =r\int_{0}^{1}dt\int_{0}^{\infty}ds\ \frac{A(x)^{rt}}{A(x)+sI}\left(\frac{\partial}{\partial x}A(x)\right)\frac{A(x)^{r\left(1-t\right)}}{A(x)+sI}\\
 & =r\int_{0}^{1}dt\int_{0}^{\infty}ds\ \left(\sum_{\ell}\frac{\lambda_{\ell}^{rt}}{\lambda_{\ell}+sI}\Pi_{\ell}\right)\left(\frac{\partial}{\partial x}A(x)\right)\left(\sum_{m}\frac{\lambda_{m}^{r\left(1-t\right)}}{\lambda_{m}+sI}\Pi_{m}\right)\\
 & =r\sum_{\ell,m}\int_{0}^{1}dt\ \lambda_{\ell}^{rt}\lambda_{m}^{r\left(1-t\right)}\int_{0}^{\infty}ds\ \left(\frac{1}{\lambda_{\ell}+sI}\right)\left(\frac{1}{\lambda_{m}+sI}\right)\Pi_{\ell}\left(\frac{\partial}{\partial x}A(x)\right)\Pi_{m}\\
 & =r\sum_{\ell}\lambda_{\ell}^{r}\int_{0}^{\infty}ds\ \frac{1}{\left(\lambda_{\ell}+sI\right)^{2}}\Pi_{\ell}\left(\frac{\partial}{\partial x}A(x)\right)\Pi_{\ell}\nonumber \\
 & \qquad+r\sum_{\ell\neq m}\int_{0}^{1}dt\ \lambda_{\ell}^{rt}\lambda_{m}^{r\left(1-t\right)}\int_{0}^{\infty}ds\ \left(\frac{1}{\lambda_{\ell}+sI}\right)\left(\frac{1}{\lambda_{m}+sI}\right)\Pi_{\ell}\left(\frac{\partial}{\partial x}A(x)\right)\Pi_{m}\\
 & =\sum_{\ell}r\lambda_{\ell}^{r-1}\Pi_{\ell}\left(\frac{\partial}{\partial x}A(x)\right)\Pi_{\ell}\nonumber \\
 & \qquad+r\sum_{\ell\neq m}\left(\frac{\lambda_{\ell}^{r}-\lambda_{m}^{r}}{\ln\lambda_{\ell}^{r}-\ln\lambda_{m}^{r}}\right)\left(\frac{\ln\lambda_{\ell}-\ln\lambda_{m}}{\lambda_{\ell}-\lambda_{m}}\right)\Pi_{\ell}\left(\frac{\partial}{\partial x}A(x)\right)\Pi_{m}\\
 & =\sum_{\ell}r\lambda_{\ell}^{r-1}\Pi_{\ell}\left(\frac{\partial}{\partial x}A(x)\right)\Pi_{\ell}\nonumber \\
 & \qquad+\sum_{\ell\neq m}\left(\frac{\lambda_{\ell}^{r}-\lambda_{m}^{r}}{\lambda_{\ell}-\lambda_{m}}\right)\Pi_{\ell}\left(\frac{\partial}{\partial x}A(x)\right)\Pi_{m}\\
 & =\sum_{\ell,m}f_{x^{r}}^{\left[1\right]}(\lambda_{\ell},\lambda_{m})\Pi_{\ell}\left(\frac{\partial}{\partial x}A(x)\right)\Pi_{m},
\end{align}
where $f_{x^{r}}^{\left[1\right]}$ denotes the first divided difference
for the function $x\mapsto x^{r}$.

There is an alternative simpler expression for the derivative of the
power function when the power $r\in\left(0,1\right)$:
\begin{prop}[Power function]
\label{prop:power-function-r-in-0-1}For a positive operator-valued
function $x\mapsto A(x)$, the following equality holds for all $r\in\left(-1,0\right)\cup\left(0,1\right)$:
\begin{equation}
\frac{\partial}{\partial x}\left(A(x)^{r}\right)=\frac{\sin(r\pi)}{\pi}\int_{0}^{\infty}dt\ t^{r}\left(A(x)+tI\right)^{-1}\left(\frac{\partial}{\partial x}A(x)\right)\left(A(x)+tI\right)^{-1}.\label{eq:integral-rep-deriv-r-minus-1-to-1}
\end{equation}
\end{prop}

\begin{proof}
The following integral representation holds for $r\in(0,1)$ and $x>0$
\cite[Exercise~V.4.20]{Bhatia1997}:
\begin{equation}
x^{r}=\frac{\sin(r\pi)}{\pi}\int_{0}^{\infty}dt\ t^{r-1}\left(\frac{x}{x+t}\right).\label{eq:integral-rep-power-r-0-1}
\end{equation}
From this, we can conclude the following integral representation for
$r\in(0,1)$ and $x,y>0$ such that $x\neq y$:
\begin{align}
\frac{x^{r}-y^{r}}{x-y} & =\frac{\sin(r\pi)}{\pi}\int_{0}^{\infty}dt\ \frac{t^{r}}{\left(x+t\right)\left(y+t\right)}.\label{eq:integral-rep-div-diff-x-r}
\end{align}
Indeed, consider that
\begin{align}
 & \frac{x^{r}-y^{r}}{x-y}\nonumber \\
 & =\left(\frac{1}{x-y}\right)\left(\frac{\sin(r\pi)}{\pi}\int_{0}^{\infty}dt\ t^{r-1}\left(\frac{x}{x+t}\right)-\frac{\sin(r\pi)}{\pi}\int_{0}^{\infty}dt\ t^{r-1}\left(\frac{y}{y+t}\right)\right)\\
 & =\left(\frac{1}{x-y}\right)\frac{\sin(r\pi)}{\pi}\int_{0}^{\infty}dt\ t^{r-1}\left(\frac{x}{x+t}-\frac{y}{y+t}\right)\\
 & =\left(\frac{1}{x-y}\right)\frac{\sin(r\pi)}{\pi}\int_{0}^{\infty}dt\ t^{r-1}\left(\frac{x\left(y+t\right)-y\left(x+t\right)}{\left(x+t\right)\left(y+t\right)}\right)\\
 & =\frac{\sin(r\pi)}{\pi}\int_{0}^{\infty}dt\ t^{r-1}\left(\frac{t}{\left(x+t\right)\left(y+t\right)}\right)\\
 & =\frac{\sin(r\pi)}{\pi}\int_{0}^{\infty}dt\ \frac{t^{r}}{\left(x+t\right)\left(y+t\right)}.
\end{align}
This allows us to write, for $r\in(0,1)$,
\begin{align}
 & \frac{\partial}{\partial x}\left(A(x)^{r}\right)\nonumber \\
 & =\sum_{\ell,m}f_{x^{r}}^{\left[1\right]}(\lambda_{\ell},\lambda_{m})\Pi_{\ell}\left(\frac{\partial}{\partial x}A(x)\right)\Pi_{m}\label{eq:power-func-r-0-1-proof-final-steps-1}\\
 & =\sum_{\ell,m}\left[\frac{\sin(r\pi)}{\pi}\int_{0}^{\infty}dt\ \frac{t^{r}}{\left(\lambda_{\ell}+t\right)\left(\lambda_{m}+t\right)}\right]\Pi_{\ell}\left(\frac{\partial}{\partial x}A(x)\right)\Pi_{m}\\
 & =\frac{\sin(r\pi)}{\pi}\int_{0}^{\infty}dt\ t^{r}\left(\sum_{\ell}\frac{1}{\left(\lambda_{\ell}+t\right)}\Pi_{\ell}\right)\left(\frac{\partial}{\partial x}A(x)\right)\left(\sum_{m}\frac{1}{\left(\lambda_{m}+t\right)}\Pi_{m}\right)\\
 & =\frac{\sin(r\pi)}{\pi}\int_{0}^{\infty}dt\ t^{r}\left(A(x)+tI\right)^{-1}\left(\frac{\partial}{\partial x}A(x)\right)\left(A(x)+tI\right)^{-1}.\label{eq:power-func-r-0-1-proof-final-steps-last}
\end{align}
The first equality follows from applying Theorem~\ref{thm:divided-difference-matrix-deriv}
and the fact that the power function has the following power series
expansion that converges for all $y\in\left(0,2c\right)$, where $c>0$:
\begin{equation}
y^{r}=c^{r}\sum_{n=0}^{\infty}{r \choose n}\left(\frac{y-c}{c}\right)^{n}.
\end{equation}
To see this, we can write
\begin{align}
y^{r} & =\left(c\left(1+\frac{y-c}{c}\right)\right)^{r}\\
 & =c^{r}\left(1+\frac{y-c}{c}\right)^{r}
\end{align}
and then apply the standard binomial series for $(1+y)^{r}$ that
converges for all $\left|y\right|<1$. Substituting $y\to\frac{y-c}{c}$,
the convergence condition becomes $\left|\frac{y-c}{c}\right|<1$,
which is equivalent to $y\in\left(0,2c\right)$. In order to apply
Theorem~\ref{thm:divided-difference-matrix-deriv}, we take $c>0$
to be larger than the largest eigenvalue of $A(x)$. This concludes
the proof of \eqref{eq:integral-rep-deriv-r-minus-1-to-1} for $r\in(0,1)$.

Dividing the integral representation in \eqref{eq:integral-rep-power-r-0-1}
by $x>0$ gives
\begin{equation}
x^{r-1}=\frac{\sin(r\pi)}{\pi}\int_{0}^{\infty}dt\ t^{r-1}\left(\frac{1}{x+t}\right).
\end{equation}
This is then equivalent to the following for $r\in\left(-1,0\right)$
and $x>0$:
\begin{align}
x^{r} & =\frac{\sin(\left(r+1\right)\pi)}{\pi}\int_{0}^{\infty}dt\ t^{r}\left(\frac{1}{x+t}\right).\\
 & =-\frac{\sin(r\pi)}{\pi}\int_{0}^{\infty}dt\ t^{r}\left(\frac{1}{x+t}\right)
\end{align}
Then it follows that, for $r\in\left(-1,0\right)$ and $x,y>0$ such
that $x\neq y$:
\begin{align}
\frac{x^{r}-y^{r}}{x-y} & =\left(\frac{-\frac{\sin(r\pi)}{\pi}}{x-y}\right)\left(\int_{0}^{\infty}dt\ t^{r}\left(\frac{1}{x+t}\right)-\int_{0}^{\infty}dt\ t^{r}\left(\frac{1}{y+t}\right)\right)\\
 & =\left(\frac{\frac{\sin(r\pi)}{\pi}}{y-x}\right)\int_{0}^{\infty}dt\ t^{r}\left(\frac{1}{x+t}-\frac{1}{y+t}\right)\\
 & =\left(\frac{\frac{\sin(r\pi)}{\pi}}{y-x}\right)\int_{0}^{\infty}dt\ t^{r}\left(\frac{y+t-\left(x+t\right)}{\left(x+t\right)\left(y+t\right)}\right)\\
 & =\frac{\sin(r\pi)}{\pi}\int_{0}^{\infty}dt\ t^{r}\frac{1}{\left(x+t\right)\left(y+t\right)}.
\end{align}
Now performing the same manipulations as in \eqref{eq:power-func-r-0-1-proof-final-steps-1}--\eqref{eq:power-func-r-0-1-proof-final-steps-last}
and similar justification, we conclude the following for $r\in\left(-1,0\right)$:
\begin{equation}
\frac{\partial}{\partial x}\left(A(x)^{r}\right)=\frac{\sin(r\pi)}{\pi}\int_{0}^{\infty}dt\ t^{r}\left(A(x)+tI\right)^{-1}\left(\frac{\partial}{\partial x}A(x)\right)\left(A(x)+tI\right)^{-1},
\end{equation}
thus concluding the proof of \eqref{eq:integral-rep-deriv-r-minus-1-to-1}
for $r\in(-1,0)$.
\end{proof}

\section{Proof of Equation \eqref{eq:a-z-eigenvals-simplified}}

\label{app:proof-algebra-a-z-eigenvals-simplify}Consider that for
$x,y>0$, such that $x\neq y$,
\begin{align}
 & \left(xy\right)^{\frac{\alpha}{z}}\left(\frac{x^{\frac{z-1}{z}}-y^{\frac{z-1}{z}}}{x^{\frac{1}{z}}-y^{\frac{1}{z}}}\right)\left(\frac{x^{\frac{1-\alpha}{z}}-y^{\frac{1-\alpha}{z}}}{x-y}\right)^{2}-\left(\frac{x^{\frac{z-1+\alpha}{z}}-y^{\frac{z-1+\alpha}{z}}}{x-y}\right)\left(\frac{x^{\frac{1-\alpha}{z}}-y^{\frac{1-\alpha}{z}}}{x-y}\right)\nonumber \\
 & =\left(\frac{x^{\frac{1-\alpha}{z}}-y^{\frac{1-\alpha}{z}}}{\left(x-y\right)^{2}}\right)\left[\left(xy\right)^{\frac{\alpha}{z}}\left(\frac{x^{\frac{z-1}{z}}-y^{\frac{z-1}{z}}}{x^{\frac{1}{z}}-y^{\frac{1}{z}}}\right)\left(x^{\frac{1-\alpha}{z}}-y^{\frac{1-\alpha}{z}}\right)-\left(x^{\frac{z-1+\alpha}{z}}-y^{\frac{z-1+\alpha}{z}}\right)\right]\\
 & =\left(\frac{x^{\frac{1-\alpha}{z}}-y^{\frac{1-\alpha}{z}}}{x-y}\right)\left(\frac{1}{x-y}\right)\left[\left(xy\right)^{\frac{\alpha}{z}}\left(\frac{x^{\frac{z-1}{z}}-y^{\frac{z-1}{z}}}{x^{\frac{1}{z}}-y^{\frac{1}{z}}}\right)\left(x^{\frac{1-\alpha}{z}}-y^{\frac{1-\alpha}{z}}\right)-\left(x^{\frac{z-1+\alpha}{z}}-y^{\frac{z-1+\alpha}{z}}\right)\right].\label{eq:first-steps-a-z-algebra-simplify}
\end{align}
Now consider that
\begin{align}
 & \left(\frac{1}{x-y}\right)\left[\left(xy\right)^{\frac{\alpha}{z}}\left(\frac{x^{\frac{z-1}{z}}-y^{\frac{z-1}{z}}}{x^{\frac{1}{z}}-y^{\frac{1}{z}}}\right)\left(x^{\frac{1-\alpha}{z}}-y^{\frac{1-\alpha}{z}}\right)-\left(x^{\frac{z-1+\alpha}{z}}-y^{\frac{z-1+\alpha}{z}}\right)\right]\nonumber \\
 & =\left(\frac{\left(xy\right)^{\frac{\alpha}{z}}}{x-y}\right)\left[\left(\frac{x^{\frac{z-1}{z}}-y^{\frac{z-1}{z}}}{x^{\frac{1}{z}}-y^{\frac{1}{z}}}\right)\left(x^{\frac{1-\alpha}{z}}-y^{\frac{1-\alpha}{z}}\right)-\left(x^{\frac{z-1}{z}}y^{-\frac{\alpha}{z}}-y^{\frac{z-1}{z}}x^{-\frac{\alpha}{z}}\right)\right]\\
 & =\left(\frac{\left(xy\right)^{\frac{\alpha}{z}}}{x-y}\right)\left[\left(\frac{x^{\frac{z-1}{z}}-y^{\frac{z-1}{z}}}{x^{\frac{1}{z}}-y^{\frac{1}{z}}}\right)\left(x^{\frac{1-\alpha}{z}}-y^{\frac{1-\alpha}{z}}\right)-\frac{\left(x^{\frac{1}{z}}-y^{\frac{1}{z}}\right)\left(x^{\frac{z-1}{z}}y^{-\frac{\alpha}{z}}-y^{\frac{z-1}{z}}x^{-\frac{\alpha}{z}}\right)}{x^{\frac{1}{z}}-y^{\frac{1}{z}}}\right]\\
 & =\left(\frac{\left(xy\right)^{\frac{\alpha}{z}}}{\left(x-y\right)\left(x^{\frac{1}{z}}-y^{\frac{1}{z}}\right)}\right)\times\nonumber \\
 & \qquad\left[\left(x^{\frac{z-1}{z}}-y^{\frac{z-1}{z}}\right)\left(x^{\frac{1-\alpha}{z}}-y^{\frac{1-\alpha}{z}}\right)-\left(x^{\frac{1}{z}}-y^{\frac{1}{z}}\right)\left(x^{\frac{z-1}{z}}y^{-\frac{\alpha}{z}}-y^{\frac{z-1}{z}}x^{-\frac{\alpha}{z}}\right)\right]\label{eq:a-z-algebra-app-long}
\end{align}
Focusing on the term in the bottom line of \eqref{eq:a-z-algebra-app-long},
observe that
\begin{align}
 & \left(x^{\frac{z-1}{z}}-y^{\frac{z-1}{z}}\right)\left(x^{\frac{1-\alpha}{z}}-y^{\frac{1-\alpha}{z}}\right)-\left(x^{\frac{1}{z}}-y^{\frac{1}{z}}\right)\left(x^{\frac{z-1}{z}}y^{-\frac{\alpha}{z}}-y^{\frac{z-1}{z}}x^{-\frac{\alpha}{z}}\right)\nonumber \\
 & =\left(x^{\frac{z-1}{z}}\left(x^{\frac{1-\alpha}{z}}-y^{\frac{1-\alpha}{z}}\right)-y^{\frac{z-1}{z}}\left(x^{\frac{1-\alpha}{z}}-y^{\frac{1-\alpha}{z}}\right)\right)\nonumber \\
 & \qquad-\left(x^{\frac{1}{z}}\left(x^{\frac{z-1}{z}}y^{-\frac{\alpha}{z}}-y^{\frac{z-1}{z}}x^{-\frac{\alpha}{z}}\right)-y^{\frac{1}{z}}\left(x^{\frac{z-1}{z}}y^{-\frac{\alpha}{z}}-y^{\frac{z-1}{z}}x^{-\frac{\alpha}{z}}\right)\right)\\
 & =\left(x^{\frac{z-1}{z}}x^{\frac{1-\alpha}{z}}-x^{\frac{z-1}{z}}y^{\frac{1-\alpha}{z}}-y^{\frac{z-1}{z}}x^{\frac{1-\alpha}{z}}+y^{\frac{z-1}{z}}y^{\frac{1-\alpha}{z}}\right)\nonumber \\
 & \qquad-\left(x^{\frac{1}{z}}x^{\frac{z-1}{z}}y^{-\frac{\alpha}{z}}-x^{\frac{1}{z}}y^{\frac{z-1}{z}}x^{-\frac{\alpha}{z}}-y^{\frac{1}{z}}x^{\frac{z-1}{z}}y^{-\frac{\alpha}{z}}+y^{\frac{1}{z}}y^{\frac{z-1}{z}}x^{-\frac{\alpha}{z}}\right)\\
 & =\left(x^{\frac{z-\alpha}{z}}-x^{\frac{z-1}{z}}y^{\frac{1-\alpha}{z}}-y^{\frac{z-1}{z}}x^{\frac{1-\alpha}{z}}+y^{\frac{z-\alpha}{z}}\right)\nonumber \\
 & \qquad-\left(xy^{-\frac{\alpha}{z}}-y^{\frac{z-1}{z}}x^{\frac{1-\alpha}{z}}-x^{\frac{z-1}{z}}y^{\frac{1-\alpha}{z}}+yx^{-\frac{\alpha}{z}}\right)\\
 & =\left(xx^{-\frac{\alpha}{z}}+yy^{-\frac{\alpha}{z}}\right)-\left(xy^{-\frac{\alpha}{z}}+yx^{-\frac{\alpha}{z}}\right)\\
 & =\left(x-y\right)\left(x^{-\frac{\alpha}{z}}-y^{-\frac{\alpha}{z}}\right).
\end{align}
Plugging back into \eqref{eq:a-z-algebra-app-long}, we conclude that
\begin{align}
 & \left(\frac{\left(xy\right)^{\frac{\alpha}{z}}}{\left(x-y\right)\left(x^{\frac{1}{z}}-y^{\frac{1}{z}}\right)}\right)\left[\left(x^{\frac{z-1}{z}}-y^{\frac{z-1}{z}}\right)\left(x^{\frac{1-\alpha}{z}}-y^{\frac{1-\alpha}{z}}\right)-\left(x^{\frac{1}{z}}-y^{\frac{1}{z}}\right)\left(x^{\frac{z-1}{z}}y^{-\frac{\alpha}{z}}-y^{\frac{z-1}{z}}x^{-\frac{\alpha}{z}}\right)\right]\nonumber \\
 & =\left(\frac{\left(xy\right)^{\frac{\alpha}{z}}}{\left(x-y\right)\left(x^{\frac{1}{z}}-y^{\frac{1}{z}}\right)}\right)\left(x-y\right)\left(x^{-\frac{\alpha}{z}}-y^{-\frac{\alpha}{z}}\right)\\
 & =\frac{\left(xy\right)^{\frac{\alpha}{z}}\left(x^{-\frac{\alpha}{z}}-y^{-\frac{\alpha}{z}}\right)}{x^{\frac{1}{z}}-y^{\frac{1}{z}}}\\
 & =-\left(\frac{x^{\frac{\alpha}{z}}-y^{\frac{\alpha}{z}}}{x^{\frac{1}{z}}-y^{\frac{1}{z}}}\right).
\end{align}
We have thus now proved that
\begin{multline}
\left(\frac{1}{x-y}\right)\left[\left(xy\right)^{\frac{\alpha}{z}}\left(\frac{x^{\frac{z-1}{z}}-y^{\frac{z-1}{z}}}{x^{\frac{1}{z}}-y^{\frac{1}{z}}}\right)\left(x^{\frac{1-\alpha}{z}}-y^{\frac{1-\alpha}{z}}\right)-\left(x^{\frac{z-1+\alpha}{z}}-y^{\frac{z-1+\alpha}{z}}\right)\right]\\
=-\left(\frac{x^{\frac{\alpha}{z}}-y^{\frac{\alpha}{z}}}{x^{\frac{1}{z}}-y^{\frac{1}{z}}}\right).
\end{multline}
Finally, plugging back into \eqref{eq:first-steps-a-z-algebra-simplify},
we conclude that
\begin{multline}
\left(xy\right)^{\frac{\alpha}{z}}\left(\frac{x^{\frac{z-1}{z}}-y^{\frac{z-1}{z}}}{x^{\frac{1}{z}}-y^{\frac{1}{z}}}\right)\left(\frac{x^{\frac{1-\alpha}{z}}-y^{\frac{1-\alpha}{z}}}{x-y}\right)^{2}-\left(\frac{x^{\frac{z-1+\alpha}{z}}-y^{\frac{z-1+\alpha}{z}}}{x-y}\right)\left(\frac{x^{\frac{1-\alpha}{z}}-y^{\frac{1-\alpha}{z}}}{x-y}\right)\\
=-\left(\frac{x^{\frac{1-\alpha}{z}}-y^{\frac{1-\alpha}{z}}}{x-y}\right)\left(\frac{x^{\frac{\alpha}{z}}-y^{\frac{\alpha}{z}}}{x^{\frac{1}{z}}-y^{\frac{1}{z}}}\right),
\end{multline}
thus completing the proof of \eqref{eq:a-z-eigenvals-simplified}.

\section{Proof of Equation \eqref{eq:limit-for-a-z-eigenval-func}}

\label{app:limit-for-a-z-eigenval-func}Consider that
\begin{align}
\left(\frac{x^{\frac{1-\alpha}{z}}-y^{\frac{1-\alpha}{z}}}{x-y}\right)\left(\frac{x^{\frac{\alpha}{z}}-y^{\frac{\alpha}{z}}}{x^{\frac{1}{z}}-y^{\frac{1}{z}}}\right) & =\left(\frac{y^{\frac{1-\alpha}{z}}}{y}\right)\left(\frac{\left(\frac{x}{y}\right)^{\frac{1-\alpha}{z}}-1}{\frac{x}{y}-1}\right)\left(\frac{y^{\frac{\alpha}{z}}}{y^{\frac{1}{z}}}\right)\left(\frac{\left(\frac{x}{y}\right)^{\frac{\alpha}{z}}-1}{\left(\frac{x}{y}\right)^{\frac{1}{z}}-1}\right)\\
 & =\left(\frac{y^{\frac{1-\alpha}{z}}}{y}\right)\left(\frac{y^{\frac{\alpha}{z}}}{y^{\frac{1}{z}}}\right)\left(\frac{\left(\frac{x}{y}\right)^{\frac{1-\alpha}{z}}-1}{\frac{x}{y}-1}\right)\left(\frac{\left(\frac{x}{y}\right)^{\frac{\alpha}{z}}-1}{\left(\frac{x}{y}\right)^{\frac{1}{z}}-1}\right)\\
 & =\left(\frac{1}{y}\right)\left(\frac{\left(\frac{x}{y}\right)^{\frac{1-\alpha}{z}}-1}{\frac{x}{y}-1}\right)\left(\frac{\left(\frac{x}{y}\right)^{\frac{\alpha}{z}}-1}{\left(\frac{x}{y}\right)^{\frac{1}{z}}-1}\right).
\end{align}
Now consider that
\begin{align}
\lim_{x\to y}\left(\frac{\left(\frac{x}{y}\right)^{\frac{1-\alpha}{z}}-1}{\frac{x}{y}-1}\right)\left(\frac{\left(\frac{x}{y}\right)^{\frac{\alpha}{z}}-1}{\left(\frac{x}{y}\right)^{\frac{1}{z}}-1}\right) & =\lim_{x\to1}\left(\frac{x^{\frac{1-\alpha}{z}}-1}{x-1}\right)\left(\frac{x^{\frac{\alpha}{z}}-1}{x^{\frac{1}{z}}-1}\right)\\
 & =\left(\lim_{x\to1}\frac{x^{\frac{1-\alpha}{z}}-1}{x-1}\right)\left(\lim_{x\to1}\frac{x^{\frac{\alpha}{z}}-1}{x^{\frac{1}{z}}-1}\right)\\
 & =\left(\lim_{x\to1}\frac{\frac{1-\alpha}{z}x^{\frac{1-\alpha}{z}-1}}{1}\right)\left(\lim_{x\to1}\frac{\frac{\alpha}{z}x^{\frac{\alpha}{z}-1}}{\frac{1}{z}x^{\frac{1}{z}-1}}\right)\\
 & =\frac{\left(\frac{1-\alpha}{z}\right)\left(\frac{\alpha}{z}\right)}{\frac{1}{z}}\lim_{x\to1}x^{\frac{1-\alpha}{z}-1+\frac{\alpha}{z}-1-\left(\frac{1}{z}-1\right)}\\
 & =\frac{\alpha\left(1-\alpha\right)}{z}\lim_{x\to1}x^{-1}\\
 & =\frac{\alpha\left(1-\alpha\right)}{z}.
\end{align}
Thus,
\begin{equation}
\lim_{x\to y}\left(\frac{x^{\frac{1-\alpha}{z}}-y^{\frac{1-\alpha}{z}}}{x-y}\right)\left(\frac{x^{\frac{\alpha}{z}}-y^{\frac{\alpha}{z}}}{x^{\frac{1}{z}}-y^{\frac{1}{z}}}\right)=\alpha\left(\frac{1-\alpha}{z}\right)\left(\frac{1}{y}\right).
\end{equation}

\section{Proof of Equation \eqref{eq:fisher-bures-pure-states}}

\label{app:Fisher-Bures-pure}Consider that
\begin{align}
 & \frac{\partial^{2}}{\partial\varepsilon_{i}\partial\varepsilon_{j}}\left(-2\ln F(\psi(\theta),\psi(\theta+\varepsilon))\right)\nonumber \\
 & =\frac{\partial^{2}}{\partial\varepsilon_{i}\partial\varepsilon_{j}}\left(-2\ln\Tr\!\left[\psi(\theta)\psi(\theta+\varepsilon)\right]\right)\\
 & =-2\frac{\partial}{\partial\varepsilon_{i}}\left(\frac{\Tr\!\left[\psi(\theta)\left(\frac{\partial}{\partial\varepsilon_{j}}\psi(\theta+\varepsilon)\right)\right]}{\Tr\!\left[\psi(\theta)\psi(\theta+\varepsilon)\right]}\right)\\
 & =-2\left(\frac{\Tr\!\left[\psi(\theta)\left(\frac{\partial^{2}}{\partial\varepsilon_{i}\partial\varepsilon_{j}}\psi(\theta+\varepsilon)\right)\right]}{\Tr\!\left[\psi(\theta)\psi(\theta+\varepsilon)\right]}-\frac{\Tr\!\left[\psi(\theta)\left(\frac{\partial}{\partial\varepsilon_{i}}\psi(\theta+\varepsilon)\right)\right]\Tr\!\left[\psi(\theta)\left(\frac{\partial}{\partial\varepsilon_{j}}\psi(\theta+\varepsilon)\right)\right]}{\left(\Tr\!\left[\psi(\theta)\psi(\theta+\varepsilon)\right]\right)^{2}}\right).
\end{align}
Then it follows that
\begin{align}
 & \left.\frac{\partial^{2}}{\partial\varepsilon_{i}\partial\varepsilon_{j}}\left(-2\ln F(\psi(\theta),\psi(\theta+\varepsilon))\right)\right|_{\varepsilon=0}\nonumber \\
 & =-2\left.\frac{\Tr\!\left[\psi(\theta)\left(\frac{\partial^{2}}{\partial\varepsilon_{i}\partial\varepsilon_{j}}\psi(\theta+\varepsilon)\right)\right]}{\Tr\left[\psi(\theta)\psi(\theta+\varepsilon)\right]}\right|_{\varepsilon=0}\nonumber \\
 & \qquad+2\left.\frac{\Tr\!\left[\psi(\theta)\left(\frac{\partial}{\partial\varepsilon_{i}}\psi(\theta+\varepsilon)\right)\right]\Tr\!\left[\psi(\theta)\left(\frac{\partial}{\partial\varepsilon_{j}}\psi(\theta+\varepsilon)\right)\right]}{\left(\Tr\left[\psi(\theta)\psi(\theta+\varepsilon)\right]\right)^{2}}\right|_{\varepsilon=0}\\
 & =-2\frac{\Tr\!\left[\psi(\theta)\left(\left.\frac{\partial^{2}}{\partial\varepsilon_{i}\partial\varepsilon_{j}}\psi(\theta+\varepsilon)\right|_{\varepsilon=0}\right)\right]}{\Tr\left[\psi(\theta)\psi(\theta)\right]}\nonumber \\
 & \qquad+2\frac{\Tr\!\left[\psi(\theta)\left(\left.\frac{\partial}{\partial\varepsilon_{i}}\psi(\theta+\varepsilon)\right|_{\varepsilon=0}\right)\right]\Tr\!\left[\psi(\theta)\left(\left.\frac{\partial}{\partial\varepsilon_{j}}\psi(\theta+\varepsilon)\right|_{\varepsilon=0}\right)\right]}{\left(\Tr\left[\psi(\theta)\psi(\theta)\right]\right)^{2}}\\
 & =-2\left(\Tr\!\left[\psi(\theta)\frac{\partial^{2}}{\partial\theta_{i}\partial\theta_{j}}\psi(\theta)\right]-\Tr\!\left[\psi(\theta)\frac{\partial}{\partial\theta_{i}}\psi(\theta)\right]\Tr\!\left[\psi(\theta)\frac{\partial}{\partial\theta_{j}}\psi(\theta)\right]\right)
\end{align}
Now, given that $\Tr[\psi(\theta)\psi(\theta)]=1$, consider that
\begin{align}
0 & =\frac{\partial}{\partial\theta_{i}}\Tr[\psi(\theta)\psi(\theta)]\\
 & =\Tr\!\left[\left(\frac{\partial}{\partial\theta_{i}}\psi(\theta)\right)\psi(\theta)\right]+\Tr\!\left[\psi(\theta)\left(\frac{\partial}{\partial\theta_{i}}\psi(\theta)\right)\right]\\
 & =2\Tr\!\left[\psi(\theta)\left(\frac{\partial}{\partial\theta_{i}}\psi(\theta)\right)\right].
\end{align}
Thus, 
\begin{align}
\left.\frac{\partial^{2}}{\partial\varepsilon_{i}\partial\varepsilon_{j}}\left(-2\ln F(\psi(\theta),\psi(\theta+\varepsilon))\right)\right|_{\varepsilon=0} & =-2\Tr\!\left[\psi(\theta)\frac{\partial^{2}}{\partial\theta_{i}\partial\theta_{j}}\psi(\theta)\right].\\
 & =2\Tr\!\left[\left(\frac{\partial}{\partial\theta_{i}}\psi(\theta)\right)\left(\frac{\partial}{\partial\theta_{j}}\psi(\theta)\right)\right],
\end{align}
where the equality follows because
\begin{align}
0 & =\frac{\partial}{\partial\theta_{i}}\Tr\!\left[\psi(\theta)\left(\frac{\partial}{\partial\theta_{j}}\psi(\theta)\right)\right]\\
 & =\Tr\!\left[\left(\frac{\partial}{\partial\theta_{i}}\psi(\theta)\right)\left(\frac{\partial}{\partial\theta_{j}}\psi(\theta)\right)\right]+\Tr\!\left[\psi(\theta)\frac{\partial^{2}}{\partial\theta_{i}\partial\theta_{j}}\psi(\theta)\right].
\end{align}
Finally, while using $\partial_{i}\equiv\frac{\partial}{\partial\theta_{i}}$,
observe that
\begin{align}
 & \Tr\!\left[\left(\frac{\partial}{\partial\theta_{i}}\psi(\theta)\right)\left(\frac{\partial}{\partial\theta_{j}}\psi(\theta)\right)\right]\\
 & =\Tr\!\left[\left(\frac{\partial}{\partial\theta_{i}}|\psi(\theta)\rangle\langle\psi(\theta)|\right)\left(\frac{\partial}{\partial\theta_{j}}|\psi(\theta)\rangle\langle\psi(\theta)|\right)\right]\\
 & =\Tr\!\left[\left(|\partial_{i}\psi(\theta)\rangle\langle\psi(\theta)|+|\psi(\theta)\rangle\langle\partial_{i}\psi(\theta)|\right)\left(|\partial_{j}\psi(\theta)\rangle\langle\psi(\theta)|+|\psi(\theta)\rangle\langle\partial_{j}\psi(\theta)|\right)\right]\\
 & =\Tr\!\left[|\partial_{i}\psi(\theta)\rangle\langle\psi(\theta)|\partial_{j}\psi(\theta)\rangle\langle\psi(\theta)|\right]+\Tr\!\left[|\partial_{i}\psi(\theta)\rangle\langle\psi(\theta)|\psi(\theta)\rangle\langle\partial_{j}\psi(\theta)|\right]\nonumber \\
 & \qquad+\Tr\!\left[|\psi(\theta)\rangle\langle\partial_{i}\psi(\theta)|\partial_{j}\psi(\theta)\rangle\langle\psi(\theta)|\right]+\Tr\!\left[|\psi(\theta)\rangle\langle\partial_{i}\psi(\theta)|\psi(\theta)\rangle\langle\partial_{j}\psi(\theta)|\right]\\
 & =\langle\psi(\theta)|\partial_{i}\psi(\theta)\rangle\,\langle\psi(\theta)|\partial_{j}\psi(\theta)\rangle+\langle\partial_{j}\psi(\theta)|\partial_{i}\psi(\theta)\rangle\nonumber \\
 & \qquad+\langle\partial_{i}\psi(\theta)|\partial_{j}\psi(\theta)\rangle+\langle\partial_{j}\psi(\theta)|\psi(\theta)\rangle\,\langle\partial_{i}\psi(\theta)|\psi(\theta)\rangle\\
 & =2\Re\left[\langle\partial_{i}\psi(\theta)|\partial_{j}\psi(\theta)\rangle+\langle\psi(\theta)|\partial_{i}\psi(\theta)\rangle\,\langle\psi(\theta)|\partial_{j}\psi(\theta)\rangle\right]\\
 & =2\Re\left[\langle\partial_{i}\psi(\theta)|\partial_{j}\psi(\theta)\rangle-\langle\partial_{i}\psi(\theta)|\psi(\theta)\rangle\,\langle\psi(\theta)|\partial_{j}\psi(\theta)\rangle\right]\\
 & =2\Re\left[\langle\partial_{i}\psi(\theta)|\left(I-|\psi(\theta)\rangle\!\langle\psi(\theta)|\right)\partial_{j}\psi(\theta)\rangle\right],
\end{align}
where the penultimate equality follows because
\begin{equation}
0=\partial_{i}\left(\langle\psi(\theta)|\psi(\theta)\rangle\right)=\langle\partial_{i}\psi(\theta)|\psi(\theta)\rangle+\langle\psi(\theta)|\partial_{i}\psi(\theta)\rangle.
\end{equation}

\end{document}